\documentclass[pdftex,a4paper,openright,bibliography=totoc,toc=listof,10pt]{scrbook}

\usepackage{scrhack}
\usepackage{graphicx}
\usepackage{xcolor}
\usepackage{microtype}
\usepackage{bm}
\usepackage{braket}
\usepackage{amsmath}
\usepackage{amsthm}
\usepackage{amsfonts}
\usepackage{amssymb}
\usepackage{setspace}
\usepackage[numbers]{natbib}
\usepackage{booktabs}
\usepackage{float}
\usepackage{subfig}
\usepackage{wrapfig}
\usepackage{scrpage2}
\usepackage[a4paper,top=3.5cm,bottom=4.8cm,inner=3.9cm,outer=2.4cm,footskip=3cm]{geometry}
\usepackage[subfigure]{tocloft}
\usepackage[T1]{fontenc}
\usepackage[sc]{mathpazo}
\usepackage[pdfpagelabels,plainpages=false,pdfborder={0 0 0},colorlinks,urlcolor=blue,linkcolor=black,citecolor=blue,filecolor=black]{hyperref}

\addtokomafont{sectioning}{\rmfamily\boldmath} 

\newtheorem{thm}{Theorem}[chapter]

\newtheorem{lem}{Lemma}[chapter]
\newtheorem{defini}{Definition}[chapter]

\makeatletter
\renewcommand*\env@matrix[1][*\c@MaxMatrixCols c]{%
  \hskip -\arraycolsep
  \let\@ifnextchar\new@ifnextchar
  \array{#1}}
\makeatother

\newcommand{\changefont}[3]{\fontfamily{#1} \fontseries{#2} \fontshape{#3} \selectfont}
\makeatletter
\renewcommand{\@pnumwidth}{1.75em}
\renewcommand{\@tocrmarg}{2.75em}
\makeatother

\setheadsepline{0.5pt}
\clearscrheadfoot 
\ohead{\pagemark}
\ihead{\headmark}
\cofoot[\pagemark]{} 
\automark[section]{chapter}

\begin{document}

%Schmutztitel
\begin{flushright}
{\large \changefont{pag}{m}{n}\textbf{Karsten Kreis}}

\vspace{1cm}

{\begin{onehalfspace} \large \changefont{pbk}{b}{sc}{Characterizing And Exploiting \\ Hybrid Entanglement} \end{onehalfspace}}
\end{flushright}

\pagenumbering{roman}
\setcounter{page}{1} 
\thispagestyle{empty} 
\cleardoublepage

%Titlepage
\begin{titlepage}

\thispagestyle{empty}

\begin{center}

{\begin{onehalfspace} \Huge \changefont{pbk}{b}{sc}{Characterizing And Exploiting Hybrid Entanglement} \end{onehalfspace}}

\vspace{3.7cm}

{\Huge \changefont{pag}{m}{n}\textbf{D I P L O M A \hspace{0.6cm}T H E S I S}}

\vspace{1cm}

{\large\changefont{pag}{m}{n}{submitted by}}

\vspace{0.5cm}

{\Large \changefont{pag}{m}{n}\textbf{Karsten Kreis}}

\vspace{4cm}

{\begin{onehalfspace} \large \changefont{pag}{m}{n}{carried out at the Max Planck Institute for the Science of Light in the group "Optical Quantum Information Theory",\\supervised by Dr. Peter van Loock} \end{onehalfspace}}

\vspace{0.0cm}
\hspace{-0.0cm} % Anpassen je nach Textbreite
\rule{\textwidth}{.2pt}

\begin{figure}[ht]
\begin{center}
\subfloat{\includegraphics[width=0.20\textwidth]{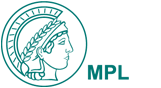}}
\subfloat{\includegraphics[width=0.30\textwidth]{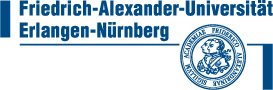}}
\end{center}
\end{figure}

{\large \changefont{pag}{m}{n}{Erlangen, February 2011}}

\end{center}

\end{titlepage}

%Schmutztitelrückseite mit Adresse
\thispagestyle{empty} 
\newpage
\begin{center}
This version of the diploma thesis is a slightly modified version of the original thesis handed in at the University of Erlangen-Nuremberg in February 2011. Only minor corrections and modifications have been incorporated. \newline
\newline
Mainz, November 2012.
\end{center}
\vfill

{\large \textbf{Address:}} \vspace{0.4cm} \\
Karsten Kreis \\
Max Planck Institute for Polymer Research  \\
Theory Group \\
PO Box 3148 \\
55021 Mainz, Germany \vspace{1cm}

{\large \textbf{E-mail:}} \vspace{0.4cm} \\
kreis@mpip-mainz.mpg.de \\
kreis.karsten@gmail.com

%Gutachter
%\thispagestyle{empty} 
%\newpage
%\phantom{oben}
%\vfill
%\begin{flushright}
%{\large \textbf{Erstgutachter:} \hspace{0.7cm} Dr. Peter van Loock \hspace{0.20cm} \vspace{0.5cm} \\
%\textbf{Zweitgutachter:} \hspace{0.36cm} Prof. Dr. Gerd Leuchs}
%\end{flushright}

%Abstract
\thispagestyle{empty} 
\newpage
%\cleardoublepage ÄNDERN GEGEN NEWPAGE WENN MIT GUTACHTER ANGABE
\null
\vfill
\begin{center}
{\huge \textbf{Abstract}}
\end{center}
\vspace{0.8cm}

Quantum information theory is a very young area of research offering a lot of challenging open questions to be tackled by ambitious upcoming physicists. One such problem is addressed in this thesis. Recently, several protocols have emerged which exploit both continuous variables and discrete variables. On the one hand, outperforming many of the established pure continuous variable or discrete variable schemes, these \textit{hybrid} approaches offer new opportunities. However, on the other hand, they also lead to new, intricate, as yet uninvestigated, phenomena. 

An important ingredient of several of these hybrid protocols is a new kind of entanglement: The hybrid entanglement between continuous variable and discrete variable quantum systems, which is studied in detail in this work. An exhaustive analysis of this kind of entanglement is performed, where the focus is on bipartite entanglement. Nevertheless, also issues regarding multipartite hybrid entanglement are briefly discussed. The quintessence of this thesis is a new classification scheme which distinguishes between effective discrete variable hybrid entanglement and so-called \textit{true} hybrid entanglement. However, along the way, also other questions are addressed, which have emerged during the studies. For example, entanglement witnessing is discussed not only for hybrid entangled states, but also for fully continuous variable two-mode Schr\"odinger cat states. Furthermore, subtleties regarding entanglement witnessing in a certain kind of mixed states are examined. Not only theoretical classification and analysis of hybrid entangled states are discussed, but also their generation is presented and a few applications are demonstrated.   
\vfill

%ToC
\thispagestyle{empty} 
\cleardoublepage
%\pagenumbering{roman} 
%\setcounter{page}{1} 
\tableofcontents

%Start mit Chaptern auf neuer Seite und mit neuer Seitenzahlnummerierung
\newpage
\pagenumbering{arabic}
\setcounter{page}{1}

\chapter*{Related Publications}
\addcontentsline{toc}{chapter}{Related Publications}
The results of this diploma thesis, especially the new classification scheme for hybrid entangled quantum states presented in {\textbf{chapter \ref{ch:4}}}, have led to a publication in the scientific journal Physical Review A. All major results of this thesis can be found in this publication (also see arXiv:1111.0478v2 [quant-ph]). \newline
\newline
\textbf{Karsten Kreis} and Peter van Loock. Classifying, quantifying, and witnessing qudit-qumode hybrid entanglement. Phys. Rev. A 85, 032307 (2012).

\chapter{Introduction}
What is quantum information theory (QIT) about?

QIT aims for the combination of information theory with quantum mechanics. Exploiting concepts like the superposition principle yields new striking effects, such as entanglement, unknown from classical information theory. Going from theory to the experimental regime, utilizing these new phenomena, QIT may lead to real-world applications for communication and computation. On the one hand, quantum communication deals with the transmission of quantum states, that is to say the transmission of information, since information is encoded in quantum states in QIT. On the other hand, quantum computation is about the quantum mechanical processing of information. Quantum computers might outperform their classical analogues concerning certain tasks, such as the factorization of large numbers as well as database searches.\footnote[1]{See Shor's algorithm for integer factorization \cite{Sho97}, and Grover's algorithm for searching unsorted databases \cite{Gro96}.} While the construction of quantum computers is quite complicated and not yet realized on large scales, there are already emerging first commercial quantum communication systems, which are used for quantum key distribution (QKD).

It turned out that the impossibility of cloning quantum states as well as the highly intriguing and non-classical phenomenon of entanglement are at the heart of QIT and many applications, it might bring up. Especially entanglement is a very surprising and non-intuitive effect, which was first considered by Einstein, Podolski and Rosen in 1935 in a famous paper, presenting the so-called EPR-paradox \cite{EPR}. Based on these findings, Einstein actually considered quantum mechanics being incomplete and regarded entanglement in correspondence with Max Born famously as "spukhafte Fernwirkung" or "spooky action at a distance" \cite{EinBornLet}. During the same time, also Schr\"odinger considered entangled quantum states, as he pointed out the paradoxal effects of the quantum mechanical superposition principle, when applied on macroscopic objects.\footnote[2]{He considered a cat entangled with a two-level atom, for which the laws of quantum mechanics are valid \cite{SCat}. Due to coupling the atom with a poisoning mechanism, depending on the state of the atom, the cat is dead or alive. Since the atom can be in a coherent superposition of ground and excited state, the cat is necessarily in a superposition state as well. However, a cat which is somehow dead {\textit{and}} alive is contrary to any experience from daily life. As shown later, this problem can be settled by considering decoherence effects immediately destroying any coherence in macroscopic superpositions.} In 1964 Bell considered quantum measurements on an entangled two particle system and could show that entanglement actually leads to non-local quantum correlations which cannot be explained classically \cite{Bell}. This corresponds to today's image of entanglement as a non-local quantum correlation between two or more particles, or, generally speaking, between quantum systems held by different parties. Indeed there are still a lot of interesting and unanswered questions to be tackled, especially concerning entanglement. In this context also note that QIT is actually a very young area of research, being considered as an own branch of physics for only about 20 years.

When searching through publications regarding QIT, one can find that most ideas and schemes are formulated either for quantum systems supported by an infinite-dimensional Hilbert space ({\textit{continuous variables}} or simply {\textit{CV}}), or for systems living in a finite-dimensional space ({\textit{discrete variables}} or simply {\textit{DV}}). Most tasks can be accomplished in both settings, however, either regime has its characterizing advantages and disadvantages. For example DV experiments are nearly always conditional and hence not very efficient, while their fidelities are quite high. On the contrary, CV schemes combine unconditional operation and high efficiencies with lower fidelities. Recently, so-called hybrid approaches emerged trying to combine the benefits from both regimes. However, these hybrid approaches lead to as yet uninvestigated phenomena, such as hybrid entanglement (HE), which is the focus of this diploma thesis. 

The thesis is organized as follows. {\textbf{Chapter \ref{ch:2}}} introduces QIT and provides fundamental concepts, necessary for tackling open questions regarding HE and further discussions. It especially sets out the basics of entanglement theory. An overview over optical hybrid approaches to quantum information is presented in {\textbf{chapter \ref{ch:3}}}, which motivates the following chapter. The main part of the diploma thesis is {\textbf{chapter \ref{ch:4}}} dealing with HE. An introduction and a classification scheme are given in the first section, while the following sections work through different kinds of bipartite HE. In particularly, the distinction between hybrid entangled states supported by finite-dimensional subspaces and so-called {\textit{truly hybrid entangled states}} not being describable in finite-dimensional spaces is pointed out and several exemplary states are investigated. The final section of the chapter provides some analysis regarding multipartite HE. In {\textbf{chapter \ref{ch:5}}} the generation of HE states is briefly discussed. Furthermore some applications of HE are presented. Finally, coming to a conclusion, a brief summary as well as an outlook are given in {\textbf{chapter \ref{ch:6}}}. The {\textbf{appendices}} provide additional material. Appendix A summarizes the abbreviations and notations used in this thesis, and Appendix B consists of proofs of theorems, ancillary calculations, and provides a summary of useful formulas.

\chapter{Basics}\label{ch:2}
In this chapter, the reader is provided with an introduction to QIT and its mathematical concepts, which are necessary to understand the main chapter of the thesis. First, the Hilbert space and the quantum states described therein are introduced from a mathematical point of view. Then it is briefly shown, how to describe DV and CV quantum systems and how to deal with quantum operations. As the focus of this diploma thesis is on entanglement theory, the concept of entanglement is explained and entanglement quantification as well as entanglement detection are discussed.

\section{Quantum States and Hilbert Spaces}
When physicists first investigated quantum phenomena, they had to think about how to actually describe quantum states. There are basically two approaches. On the one hand complex wavefunctions can be used, which are typically functions of position or momentum and give probability amplitudes. Their moduli squared yield probability densities. Therefore, these wavefunctions are easy to interpret and they will be useful if the probability distribution of position, momentum, or some other variable of a particle is of interest. However, wavefunctions are typically quite complicated and if more subtle properties of the quantum system, such as entanglement, are to be examined, it will probably be cumbersome to work with them. Hence, on the other hand, there is the so-called Dirac notation, which employs more abstract vectors in a Hilbert space to describe quantum states.

So, what is a Hilbert space? A Hilbert space over the field of the complex numbers ${\mathcal H}(\mathbb C)$ is a complete vector space on which an inner product is defined. The inner product allows length and angle to be measured and a norm to be defined. Completeness is defined in the following way:
\begin{defini}
A metric space $M$ with a defined norm is called \textup{complete} if and only if every Cauchy sequence in $M$ converges in $M$:
\begin{equation}
M\;\text{complete}\,:\;{(a_n)}_{n\in\mathbb N}\;\text{Cauchy sequence in}\;M\;\Rightarrow\;\lim_{n \to \infty}a_n\in M.
\end{equation}
\end{defini}
\begin{defini}
${(a_n)}_{n\in\mathbb N}$ is called \textup{Cauchy sequence} if and only if 
\begin{equation}\forall\,\epsilon>0\;\exists\,N:\;\|a_m-a_n\|<\epsilon\;\forall\,m,n>N.
\end{equation}
\end{defini}
More intuitively speaking, a space is complete if there are no points "missing" from it. For example $\mathbb Q$ is not complete, since the irrationals are missing, while $\mathbb R$ is complete. 

Hilbert spaces may have any dimension, even infinite. This becomes clear, when the physical quantum states to be described by the Hilbert space vectors are considered. For example, a general spin system may have any finite number of spin states, while a quantum state described by the continuous variables position and momentum lives in an infinite-dimensional Hilbert space. Let the elements of the Hilbert space ${\mathcal H}(\mathbb C)$ be denoted by so-called "ket"-vectors $\ket{\psi}$. Then the according dual space ${\mathcal H'}(\mathbb C)$ contains the "bra"-vector elements $\bra{\psi}\,:\,{\mathcal H}(\mathbb C)\rightarrow\mathbb C$. Therefore, a map into complex numbers between a vector $\bra{\psi}$ and a vector $\ket{\phi}$ can be defined as $\braket{\psi|\phi}\in\mathbb C$.

Consider a continuous basis of a Hilbert space of infinite uncountable dimension $\{\ket{x}:x,y\in\mathbb R,\,\braket{x|y}=\delta(x-y)\}$.\footnote[1]{Continuous variable and discrete variable states, as well as infinite-dimensional and finite-dimensional bases are discussed in more detail in section \ref{states}.} The completeness relation reads:
\begin{equation}
\int\limits_{-\infty}^\infty dx \ket{x}\bra{x}=\mathbb{1}.
\end{equation}
So, writing
\begin{equation}
\ket{\psi}=\int\limits_{-\infty}^\infty dx \braket{x|\psi}\ket{x}=\int\limits_{-\infty}^\infty dx\;\psi(x)\ket{x}
\end{equation}
yields the connection to the previous wavefunction approach, explicitly $\braket{x|\psi}=\psi(x)$. The continuous $\{\ket{x}\}$-basis may for example be the position or the momentum basis and $\psi(x)$ the corresponding wavefunction. If the basis of the Hilbert space is countable or even finite, we have the basis $\{\ket{k}:k,l\in{\mathbb N}_0,\,\braket{k|l}=\delta_{kl}\}$ and the integral becomes an infinite or finite sum, but the argumentation is analogue. The wavefunction, in this case discrete, is just written as $\braket{k|\psi}=\psi_k$.

Again, using the completeness relation, we obtain
\begin{equation}
\braket{\psi|\phi}=\int\limits_{-\infty}^\infty dx\;\psi^\ast(x)\phi(x),
\end{equation}
or in the countable or even finite-dimensional case
\begin{equation}
\braket{\psi|\phi}=\sum_k \;\psi^\ast_k\,\phi_k.
\end{equation}
Hence, every physical Hilbert space is actually the space of square-integrable $L^2$ (uncountable dimension) or square-summable $l^2$ (countable dimension) functions. Fortunately, $l^2$ is its own dual space, and hence, $\{\ket{n}:n\in\mathbb N\}$ being a basis, for any state $\ket{\psi}=\sum_n\,\psi_n\ket{n}$ the dual is given by $\bra{\psi}=\sum_n\,\psi^\ast_n\bra{n}$ (the analogue is valid for $L^2$ and states in Hilbert spaces of uncountable dimension). Therefore, the inner product between a vector $\ket{\psi}$ and a vector $\ket{\phi}$ is given by $\braket{\psi|\phi}$, which is sometimes also called the "overlap" between $\ket{\psi}$ and $\ket{\phi}$. Using this definition of the inner product, the norm is defined by
\begin{equation}
\|\ket{\psi}\|=\sqrt{\braket{\psi|\psi}}.
\end{equation}

In conclusion, every quantum system (in a \textit{pure state}) can be described by its normalized ($\braket{\psi|\psi}=1$) "ket"-vector $\ket{\psi}$ in a Hilbert space ${\mathcal H}_d$. Using a discrete orthonormal basis $\{\ket{k}:k,l\in{\mathbb N}_0,\,\braket{k|l}=\delta_{kl}\}$, it can be written
\begin{equation}
\ket{\psi}=\sum_{k=0}^{d-1} \;\ket{k}\braket{k|\psi}=\sum_{k=0}^{d-1} \;\psi_k\ket{k},
\end{equation}
where $\psi_k\in\mathbb C$ is the discrete wavefunction in the $\{\ket{k}\}$-basis. If a continuous basis is chosen $\{\ket{x}:x,y\in\mathbb R,\,\braket{x|y}=\delta(x-y)\}$, every quantum state can be written as 
\begin{equation}
\ket{\psi}=\int\limits_{-\infty}^\infty dx\ket{x}\braket{x|\psi}=\int\limits_{-\infty}^\infty dx\;\psi(x)\ket{x},
\end{equation}
with the continuous wavefunction $\psi(x)\in\mathbb C$, as previously already shown.

In the following, $\{\ket{k}:k,l\in{\mathbb N}_0,\,\braket{k|l}=\delta_{kl}\}$ will be abbreviated by $\{\ket{k}\}$, and $\{\ket{x}:x,y\in\mathbb R,\,\braket{x|y}=\delta(x-y)\}$ analogously by $\{\ket{x}\}$, for the sake of simplicity.

An expectation value of a Hermitian operator $\hat{A}$ of a pure quantum state is given by
\begin{equation}
\braket{\hat{A}}=\braket{\psi|\hat{A}|\psi}=\sum_m \;a_m\,p_m,
\end{equation}
where the spectral decomposition for the operator $\hat{A}$,
\begin{equation}
\hat{A}=\sum_m \;a_m\ket{m}\bra{m},
\end{equation}
was used \cite{Nielsen}. $a_m$ are real eigenvalues, $\{\ket{m}\}$ is the corresponding eigenbasis, and $p_m=|\braket{m|\psi}|^2$ are the probabilities for obtaining the result $a_m\in\mathbb R$. Note that Hermitian operators are measurable observables. Analogously, for a continuous basis the spectral decomposition is given as
\begin{equation}
\hat{A}=\int\limits_{-\infty}^\infty dx\;a_x\ket{x}\bra{x}
\end{equation}
and the rest follows accordingly.

If quantum mechanics is extended to quantum statistical mechanics, quantum systems may be statistical mixtures of several individual quantum states. Then the whole system is said to be in a \textit{mixed state} and is no longer described by a Hilbert space vector, but by a so-called \textit{density operator} $\hat{\rho}$. A density operator is an operator, exhibiting the following properties: 
\begin{itemize}
\item Linear ($\hat{A}\;\text{linear}\,\Leftrightarrow\,\forall\,\ket{\psi},\ket{\phi}\in\mathcal H\,;\;x,y\in\mathbb C:\;\hat{A}(x\ket{\psi}+y\ket{\phi})=x\hat{A}\ket{\psi}+y\hat{A}\ket{\phi}$)
\item Bounded ($\hat{A}\;\text{bounded}\,\Leftrightarrow\,\exists\,M>0\,\text{such that}\,\;\forall\,\ket{\psi}\in\mathcal H:\;\braket{\psi|\hat{A}^\dagger \hat{A}|\psi}\le M\braket{\psi|\psi}$)
\item Positive semidefinite ($\hat{A}\;\text{positive semidefinite}\,\Leftrightarrow\,\forall\,\ket{\psi}\in\mathcal H:\;\braket{\psi|\hat{A}|\psi}\ge 0$) $\Rightarrow$ also Hermitian ($\hat{A}\;\text{Hermitian}\,\Leftrightarrow\,\hat{A}=\hat{A}^\dagger$)
\item Trace-class ($\hat{A}\;\text{trace-class}\,\Leftrightarrow\,\int\!\!\!\!\!\!\sum_k\braket{k|\hat{A}|k}<\infty$ and independent of basis $\{\ket{k}\}$) with trace $\mathrm{tr}[\hat{\rho}]=1$
\end{itemize} 
Note that the definitions of linearity, boundedness, positive semidefiniteness, and trace-class have been given in the context of the relevant Hilbert spaces with norm $\|\ket{\psi}\|=\sqrt{\braket{\psi|\psi}}$.

Any density operator can be written as a convex combination of rank-1 pure state projectors:
\begin{equation}
\hat{\rho}=\sum_i \;p_i\ket{\psi_i}\bra{\psi_i}\,;\qquad p_i>0\;\forall \;i,\;\sum_i \;p_i=1.
\end{equation}
This decomposition is what actually distinguishes pure states from so-called mixed states. While pure states cannot be written as a convex combination of two or more projectors, mixed states always contain two or more pure state projectors in their convex combination.\footnote[2]{Pure states are therefore extrema in the convex set of states.} Hence, the density operator of a pure state $\ket{\psi}$ is given by $\hat{\rho}_{pure}=\ket{\psi}\bra{\psi}$ and thus:
\begin{itemize}
\item For pure states: $\mathrm{tr}[\hat{\rho}^2]=\mathrm{tr}[\hat{\rho}]=1$.
\item For mixed states: $\mathrm{tr}[\hat{\rho}^2]<\mathrm{tr}[\hat{\rho}]=1$.
\end{itemize}
Furthermore, it is important to mention that there is no unique pure state decomposition for a given density operator. 

Expectation values translate into the density operator formalism simply by
\begin{equation}
\braket{\hat{A}}=\sum_i \;p_i\braket{\psi_i|\hat{A}|\psi_i}=\mathrm{tr}[\hat{\rho}\hat{A}],
\end{equation}
which can be easily seen by inserting $\sum_k \,\ket{k}\bra{k}=\mathbb{1}$ for an arbitrary orthonormal basis $\{\ket{k}\}$ into the term $\braket{\psi_i|\hat{A}|\psi_i}$ and some rearrangements.

\section{Discrete Variable vs. Continuous Variable Quantum States}\label{states}
So far, quantum states are defined as abstract density operators. However, these operators are difficult to handle, since for proper calculations, practical objects such as matrices are required. Therefore, this section discusses how to actually handle quantum states when working with them, and how to experimentally implement them. Two cases have to be distinguished: Either the quantum system is finite- or infinite-dimensional.

\subsection{Discrete Variable Quantum States}
First, we focus on the finite-dimensional case. The finite-dimensional Hilbert space ${\mathcal H}_d$ is spanned by a discrete and finite set of orthonormal basis vectors $\{\ket{n}\,:\,n=0,\ldots,d-1,\,\braket{n|m}=\delta_{nm}\}$. This is the reason why these quantum states are called \textit{discrete variable} (\textit{DV}) quantum states. A pure state of this kind looks like
\begin{equation}
\ket{\psi}=c_0\ket{0}+c_1\ket{1}+\ldots+c_{d-1}\ket{d-1}\;,\qquad\sum_{i=0}^{d-1}|c_i|^2=1.
\end{equation}
For two dimensions this becomes 
\begin{equation}
\ket{\psi}=c_0\ket{0}+c_1\ket{1},
\end{equation}
which is also called a \textit{qubit} in analogy to the classical bit (cbit). This is actually the point, when it becomes clear, why it is interesting to do quantum information and why to do research on quantum computers: In contrast to the cbit, the qubit can be in a superposition of the zero- and the one-state. Therefore, a qubit contains more information than a classical bit and in one computation step more information may be processed. For certain tasks quantum computer algorithms, exploiting the quantum mechanical superposition principle, outperform their classical counterparts. Examples are Peter W. Shor's algorithm for integer factorization \cite{Sho97}, and Lov K. Grover's algorithm for searching unsorted databases \cite{Gro96}. Shor's algorithm factorizes large numbers in polynomial time, while the best classical algorithm only achieves sub-exponential time \cite{Pom}. Its experimental demonstration was accomplished for $N=15$ \cite{Qcom1,Qcom2,Qcom3}. Likewise Grover's algorithm provides a quadratic speedup, requiring $\mathcal O(N^{1/2})$ time for a database with N items, while classically this problem would only be solvable in linear time. However, quantum computation outperforms classical computation only for some special tasks and furthermore a large-scale quantum computer could not be realized yet.

Back to the description of DV quantum states: Motivated by the term for a 2-dimensional quantum state, qubit, a 3-dimensional state is called \textit{qutrit} and a d-dimensional, \textit{qudit}. For such a general qudit in a mixed state $\hat{\rho}$ consider $\rho_{nm}=\braket{n|\hat{\rho}|m}\in\mathbb C$. This defines a proper $d\times d$ matrix containing complex numbers, which is called the \textit{density matrix} of the quantum state $\hat{\rho}$.\footnote[3]{Often the terms "density operator" and "density matrix" are used interchangeably and mostly it is clear from the context whether the actual operator or the matrix is meant. However, they are actually not the same.} These complex-valued density matrices are extremely suited for describing DV quantum states in a manageable way. Several important quantities, such as entropy or, in the multipartite case, entanglement measures, can be calculated from them. 

To illustrate the concepts introduced so far, consider the mixed qubit state
\begin{equation}
\hat{\rho}=\frac{1}{2}\Bigl(\ket{0}\bra{0}+\ket{1}\bra{1}+c\,\mathrm{e}^{i\phi}\ket{1}\bra{0}+c\,\mathrm{e}^{-i\phi}\ket{0}\bra{1}\Bigr)\;,\qquad c<1.
\end{equation}
It possesses the density matrix
\begin{equation}
\rho_{nm}=\frac{1}{2}
\begin{pmatrix} 
1 & c\,\mathrm{e}^{-i\phi} \\
c\,\mathrm{e}^{i\phi} & 1 
\end{pmatrix}.
\end{equation}
A pure state decomposition is given by
\begin{equation}
\hat{\rho}=\frac{1+c}{2}\ket{\psi_+}\bra{\psi_+}\,+\,\frac{1-c}{2}\ket{\psi_-}\bra{\psi_-}\;,\qquad\ket{\psi_\pm}=\frac{1}{\sqrt{2}}\Bigl(\ket{0}\pm \mathrm{e}^{i\phi}\ket{1}\Bigr).
\end{equation}

How can DV quantum states be experimentally realized? More precisely, how can logical qubits or qudits be experimentally encoded? An apparent example is simply to use the ground and excited states of a 2-level atom, $\ket{0}_L=\ket{g}$ and $\ket{1}_L=\ket{e}$, to implement a qubit. Furthermore there are a lot of quantum optical encoding techniques. For example, in a finite-dimensional subspace of the infinite-dimensional Fock space, photons in a single mode can be utilized (\textit{single-rail encoding}). For a qudit this becomes $\{\ket{0}_L=\ket{0}_{Fock},\,\ket{1}_L=\ket{1}_{Fock},\,\ldots\,,\ket{d-1}_L=\ket{d-1}_{Fock}\}$.\footnote[4]{Fock states and coherent states are introduced in more detail in the next subsection, when the quantum harmonic oscillator is discussed.} However, it is more practical to use encodings, where each state holds the same number of photons. In the \textit{multiple-rail encoding} a qudit is represented by a single photon in $d$ different modes: $\{\ket{0}_L=\ket{10\ldots0},\,\ket{1}_L=\ket{010\ldots0},\,\ldots\,,\ket{d-1}_L=\ket{0\ldots01}\}$. Unfortunately, this kind of encoding is not scalable, if several qubits or qudits are used. A compromise is offered by the \textit{dual-rail encoding} which uses only two modes and a constant number of photons distributed between these two modes: $\{\ket{0}_L=\ket{0,d-1},\,\ket{1}_L=\ket{1,d-2},\,\ldots\,,\ket{d-2}_L=\ket{d-2,1},\,\ket{d-1}_L=\ket{d-1,0}\}$. These modes can be realized as spatial modes or as modes of orthogonal polarization. 

Additionally there are optical encodings utilizing coherent states instead of single photons. A logical qubit can be encoded with even and odd \textit{Schr\"odinger-cat states} (also called \textit{coherent-state superpositions}, \textit{CSS}s).
\begin{itemize}
\item Even cat state: $\ket{\psi_+}=\frac{1}{\sqrt{{\mathcal N}_+}}(\ket{\alpha}+\ket{-\alpha})$.
\item Odd cat state: $\ket{\psi_-}=\frac{1}{\sqrt{{\mathcal N}_-}}(\ket{\alpha}-\ket{-\alpha})$.
\end{itemize}
${\mathcal N}_\pm$ are normalization constants. The names \textit{even} and \textit{odd} rely on the fact that the even cat state contains only even photon number states, while the odd cat state contains only odd photon number states. Since $\braket{n|m}=\delta_{nm}$ for Fock states $\ket{n}$ and $\ket{m}$, the overlap between the even and odd cat state is zero: $\braket{\psi_+|\psi_-}=0$. Hence, a qubit can be represented as $\{\ket{0}_L=\ket{\psi_+},\,\ket{1}_L=\ket{\psi_-}\}$ \cite{cats3}.

\subsection{Continuous Variable Quantum States}
This thesis focusses on QIT using quantum optical systems. It is straightforward to show that the quantized electromagnetic field is described by a set of quantum harmonic oscillators. The Hamiltonian for a single mode is (unit mass)
\begin{equation}
\hat{H}=\frac{1}{2}\Bigl(\hat{p}^2+\omega^2\hat{x}^2\Bigr)=\hbar\omega\Bigl(\hat{a}^\dagger\hat{a}+\frac{1}{2}\Bigr).
\end{equation}
Here, $\hat{x}$ and $\hat{p}$, originally the oscillator's position and momentum operators, are now the so-called $\hat{x}$- and $\hat{p}$-\textit{quadrature operators} of the field. $\hat{a}$ and $\hat{a}^\dagger$ are the \textit{annihilation} and \textit{creation} operators of the electromagnetic field, also called \textit{ladder operators} or sometimes \textit{mode operators}, which destroy or create excitations of the field, the \textit{photons}. They are given by
\begin{align}\label{eq:ladders1}
\hat{a} & = \frac{1}{\sqrt{2\hbar\omega}}(\omega\hat{x}+i\hat{p}), \\ \label{eq:ladders2}
\hat{a}^\dagger & = \frac{1}{\sqrt{2\hbar\omega}}(\omega\hat{x}-i\hat{p}).
\end{align}
These operators define the so-called \textit{Fock space} of \textit{Fock states} or \textit{photon number states} $\{\ket{n}:n\in{\mathbb N}_0,\,\braket{n|m}=\delta_{nm}\}$. Fock states are eigenstates of the \textit{number operator} $\hat{n}=\hat{a}^\dagger\hat{a}$.
\begin{align}\label{eq:ladders3}
\hat{a}\ket{n} & = \sqrt{n}\ket{n-1}, \\ \label{eq:ladders4}
\hat{a}^\dagger\ket{n} & = \sqrt{n+1}\ket{n+1}, \\
\hat{n}\ket{n} & = n\ket{n},
\end{align}
where $\ket{n}$ denotes a state which contains $n$ photons, while the energy of a single photon is given by $\hbar\omega$. Therefore, the form of the Hamiltonian becomes comprehensible: $\hat{n}\hbar\omega$ is simply the energy of $n$ photons and additionally there is the vacuum energy $\frac{1}{2}\hbar\omega$.

However, also multimode systems can be considered. Photons in different modes stand for photons of different "fashion", what makes them distinguishable, in contrast to photons in the same mode, which are indistinguishable. These modes might correspond for example to different energies, different polarizations, or to spatial modes. For two modes $j$ and $k$ the above defined operators satisfy the commutation relations
\begin{align}
\lbrack\hat{x}_j,\hat{p}_k\rbrack &= i\hbar\delta_{jk} \label{eq:kom1}, \\
\lbrack\hat{x}_j,\hat{x}_k\rbrack &=  \lbrack\hat{p}_j,\hat{p}_k\rbrack =  0 \label{eq:kom2}, \\
\lbrack\hat{a}_j,\hat{a}_k^\dagger\rbrack &=  \delta_{jk} \label{eq:kom3}, \\
\lbrack\hat{a}_j,\hat{a}_k\rbrack &= \lbrack\hat{a}_j^\dagger,\hat{a}_k^\dagger\rbrack =  0 \label{eq:kom4}.
\end{align}
Equations \eqref{eq:kom1} and \eqref{eq:kom2} are the well-known canonical commutation relations from general quantum mechanics. Equations \eqref{eq:kom3} and \eqref{eq:kom4}, the bosonic commutation relations, simply follow from the definitions of the ladder operators. 

Rewrite equations \eqref{eq:ladders1} and \eqref{eq:ladders2} into
\begin{align}
\hat{x} & = \sqrt{\frac{\hbar}{2\omega}}(\hat{a}+\hat{a}^\dagger), \\
\hat{p} & = -i\sqrt{\frac{\hbar\omega}{2}}(\hat{a}-\hat{a}^\dagger).
\end{align}
Hence, up to normalization factors, the $\hat{x}$- and $\hat{p}$-quadratures are just the real and imaginary parts of the annihilation operator. To get rid of these factors, define new dimensionless quadratures
\begin{align}
\hat{x}' & = \sqrt{\frac{\omega}{\hbar}}\hat{x}= \frac{1}{\sqrt{2}}(\hat{a}+\hat{a}^\dagger),\\
\hat{p}' & = \sqrt{\frac{1}{\hbar\omega}}\hat{p}= \frac{1}{i\sqrt{2}}(\hat{a}-\hat{a}^\dagger),
\end{align}
which obey the commutation relation
\begin{equation}
\lbrack\hat{x}_j,\hat{p}_k\rbrack = i\delta_{jk}.
\end{equation}
Hence, this definition of dimensionless position and momentum operators corresponds to setting $\hbar=1$. Throughout this thesis, the dimensionless quadratures $\hat{x}'$ and $\hat{p}'$ will be used and in the following the prime will be omitted. So, $\hat{x}$ and $\hat{p}$ always stand for a pair of conjugate dimensionless quadratures obeying $\lbrack\hat{x},\hat{p}\rbrack = i$.

The eigenstates corresponding to the $\hat{x}$- and $\hat{p}$-quadratures are the position and momentum eigenstates $\ket{x}$ and $\ket{p}$:
\begin{equation}
\hat{x}\ket{x}_{pos.}=x\ket{x}_{pos.}\;,\qquad\hat{p}\ket{p}_{mom.}=p\ket{p}_{mom.}
\end{equation}
They are orthogonal,
\begin{equation}
\braket{x_1|x_2}=\delta(x_1-x_2)\;,\qquad\braket{p_1|p_2}=\delta(p_1-p_2),
\end{equation} 
and complete,
\begin{equation}
\int\limits_{-\infty}^\infty dx \ket{x}\bra{x}=\mathbb{1}\;,\qquad\int\limits_{-\infty}^\infty dp \ket{p}\bra{p}=\mathbb{1}.
\end{equation} 
Hence, they form bases. However, since they are not normalizable, $\ket{x}$ and $\ket{p}$ are unphysical states. Nevertheless, they are very important and useful to calculate position $\braket{x|\psi}=\psi(x)$ and momentum wavefunctions $\braket{p|\psi}=\psi(p)$, which are indeed well defined. The relation between the position and the momentum bases is given by a Fourier transformation \cite{Leon}.
\begin{align}
\ket{x}_{pos.} & = \frac{1}{\sqrt{2\pi}}\int\limits_{-\infty}^\infty dp\,\mathrm{e}^{-ixp}\ket{p}_{mom.}, \\
\ket{p}_{mom.} & = \frac{1}{\sqrt{2\pi}}\int\limits_{-\infty}^\infty dx\,\mathrm{e}^{+ixp}\ket{x}_{pos.}.
\end{align} 

The Hilbert space of the quantized harmonic oscillator is an infinite-dimensional $\mathcal H_{\infty}$. States living in this Hilbert space are called \textit{qumodes}. They can be either represented in the infinite, but countable Fock basis, or in the continuous x- or p-basis. Note that by qumode only states are denoted which are actually infinite-dimensional. Of course, a single Fock state or a single position eigenstate is also supported by the infinite-dimensional Hilbert space $\mathcal H_{\infty}$, but they only use a finite-dimensional subspace and could hence be characterized as DV qudits.

One of the most prominent examples of qumode states are probably the \textit{coherent states} $\ket{\alpha}$, introduced by Roy Glauber in 1963 \cite{Glauber}.  
\begin{equation}
\ket{\alpha}=\mathrm{e}^{-\frac{|\alpha|^2}{2}}\sum_{n=0}^\infty\,\frac{\alpha^n}{\sqrt{n!}}\ket{n}\;,\qquad\alpha\in{\mathbb C}.
\end{equation} 
Coherent states are eigenstates of the annihilation operator $\hat{a}\ket{\alpha}=\alpha\ket{\alpha}$. They are particularly important because their field dynamics most closely resemble classical sinusoidal waves, such as continuous laser waves. Furthermore, coherent states form an overcomplete basis of the Hilbert space, which is another valuable property, 
\begin{equation}
\frac{1}{\pi}\int_{\mathbb C}d^2\alpha\,\ket{\alpha}\bra{\alpha}=\mathbb{1}\;,\qquad d^2\alpha=d\mathfrak{Im}(\alpha)d\mathfrak{Re}(\alpha).
\end{equation}
Note, however, that the coherent state basis is not orthogonal: 
\begin{equation}
\braket{\beta|\alpha}=\mathrm{e}^{\frac{1}{2}(|\beta|^2+|\alpha|^2-2\beta^\ast\alpha)}.
\end{equation}

As previously shown, DV states can always be expressed by density matrices containing only complex numbers. Unfortunately, these matrices are not practical for describing qumodes. Due to the infinite-dimensionality of the Hilbert space, the density matrix would be of infinite size. Instead, the so-called \textit{phase-space representations} (also denoted as \textit{quasi-probability distributions}) can be employed. They are functions of the real \textit{quadrature variables} $x$ and $p$, from which any expectation values of the quadrature operators $\hat{x}$ and $\hat{p}$ can be calculated. There are different kinds of phase-space representations, which allow the calculation of different types of expectation values \cite{Leon,Scully}. The most famous one is the \textit{Wigner function}, which is suited for calculating expectation values of symmetrically ordered operators \cite{Wigner1}. Furthermore, the \textit{Glauber-Sudarshan P-representation} allows for expectation values of normally ordered operators, while the \textit {Husimi-Q distribution} yields expectation values of antinormally ordered operators \cite{Sudarshan}.\footnote[5]{Also generalized, so-called \textit{s-parameterized}, phase-space representation can be defined. For certain values of $s$ the Wigner-, P- or Q-function can then be retrieved.} Note that from these phase-space representations, the density operator can always be retrieved. Hence, phase-space representations fully describe quantum states and the quantum phase-space picture is equivalent to the density operator approach. Qumodes living in an infinite-dimensional Hilbert space, described by the continuous quadrature variables $x$ and $p$ using phase-space representations, are therefore also called \textit{continuous variable} (\textit{CV}) quantum states.

As an illustration of the phase-space representations the Wigner function is briefly discussed now:
\begin{equation}
W(x,p):=\frac{1}{\pi}\int\limits_{-\infty}^\infty dy\,\braket{x-y|\hat{\rho}|x+y}\mathrm{e}^{2ipy}.
\end{equation}
Its Fourier transform is the \textit{characteristic function}, which is defined as
\begin{equation}
\chi(\xi_1,\xi_2):=\mathrm{tr}[\hat{W}_{\xi}\hat{\rho}],
\end{equation} 
making use of the \textit{Weyl operator}
\begin{equation}
\hat{W}_{\xi}=\mathrm{e}^{i(\xi_1\hat{p}-\xi_2\hat{x})}. 
\end{equation}
The expectation value of a symmetrically ordered operator $\hat{S}_{sym}(\hat{x}^n,\hat{p}^m)$ can be calculated then as
\begin{equation}
\braket{\hat{S}_{sym}(\hat{x}^n,\hat{p}^m)}=\int\limits_{-\infty}^\infty\int\limits_{-\infty}^\infty W(x,p)\,x^np^m\,dx\,dp\,,
\end{equation}
where $\hat{S}_{sym}$ denotes symmetrization.

The most outstanding feature of the Wigner function is that it can become negative in contrast to classical probability functions (which is the reason why these representations are only called "quasi"-probability distributions.). Negativity of the Wigner function is an indication for nonclassicality of the state \cite{Wignernoncl}. Another important property the Wigner function offers is its marginal distributions. Integrating one variable out yields the probability distribution of the remaining one:
\begin{align}
p_x(x) & = \int\limits_{-\infty}^\infty dp\,W(x,p), \\
p_p(p) & = \int\limits_{-\infty}^\infty dx\,W(x,p).
\end{align}
Also note that the Wigner function is normalized,
\begin{equation}
\int\limits_{-\infty}^\infty W(x,p)\,dx\,dp\,=1.
\end{equation}

As an example, figure \ref{fig:WignerCat} shows the Wigner function of the cat state
\begin{equation}\label{eq:WignerCat}
\ket{\psi}=\frac{1}{\sqrt{\mathcal N}}\Bigl(\ket{\alpha \mathrm{e}^{i\phi}}+\ket{\alpha \mathrm{e}^{-i\phi}}\Bigr),
\end{equation}
with $\alpha=6$ and $\phi=\frac{\pi}{6}$. Its negativity is clearly visible, which is due to the high nonclassicality of the state.
\begin{figure}[ht]
\begin{center}
\includegraphics[width=10cm]{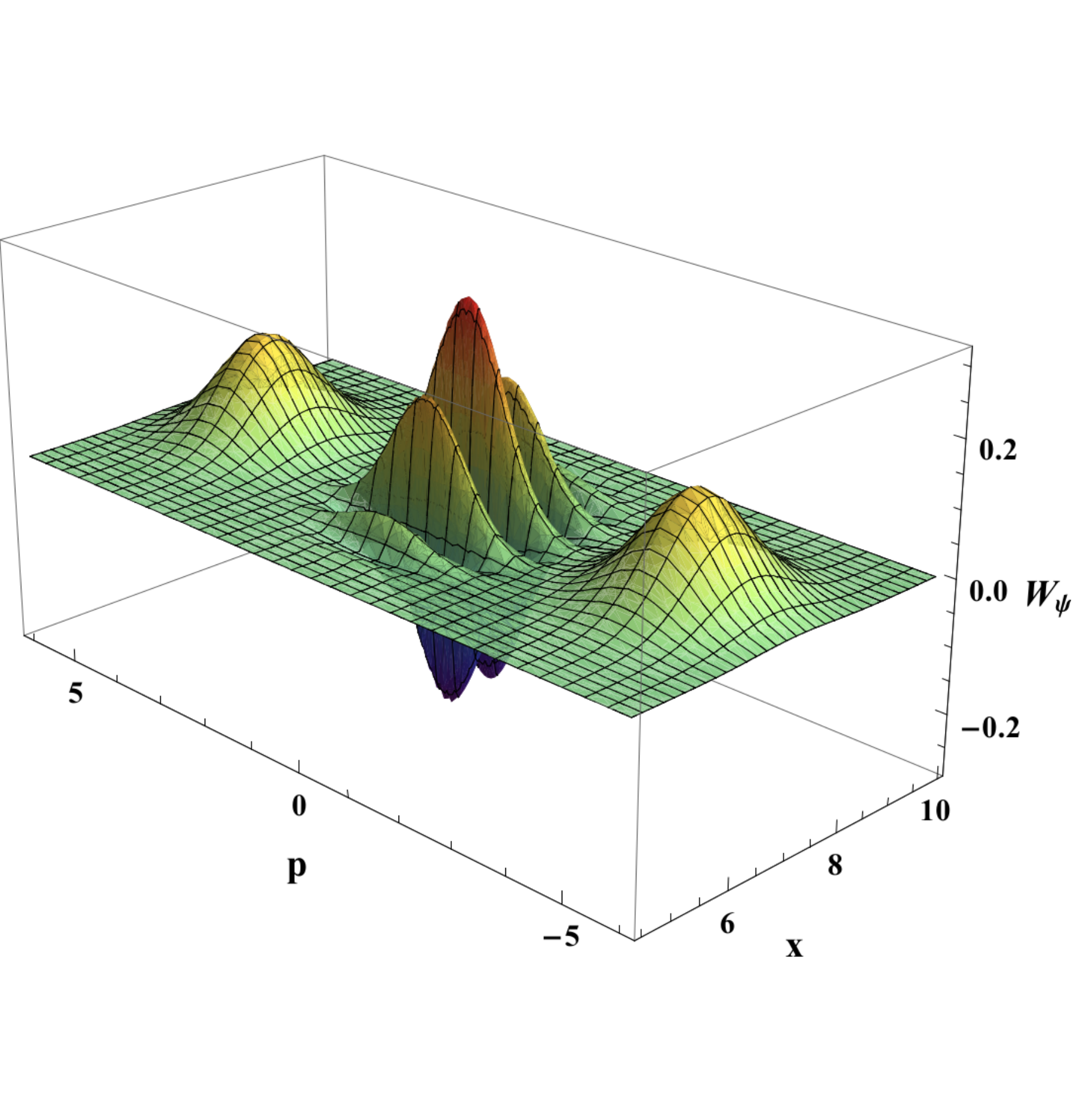}
\caption{Wigner function of the cat state of equation \eqref{eq:WignerCat}.}
\label{fig:WignerCat}
\end{center}
\end{figure}
The definitions of the Wigner function and the characteristic function are easily generalizable to multimode states. For $N$ modes the phase space becomes $2N$-dimensional and the phase space representations become functions of $2N$ variables. 

On the one hand, these quasi-probability distributions offer a neat tool for describing CV quantum states. On the other hand, however, they cannot be used for the calculation of some important quantities such as the entanglement of the state. In contrast, in the DV setting, density matrices are well suited for entanglement quantification. Therefore, now an important subclass of CV states is introduced, the \textit{Gaussian} quantum states, which are more well-behaved.

Gaussian states are defined as those whose characteristic function is a Gaussian distribution function,\footnote[6]{Since the characteristic function and the Wigner function are related by a Fourier transform, Gaussian states also have Gaussian Wigner functions, as Gaussians are Fouriertransformed into Gaussians.} i.e., for N modes,
\begin{equation}
\chi_{Gaussian}(\underline \xi)=\chi_{Gaussian}(0)\mathrm{e}^{-\frac{1}{4}{\underline \xi}^TJ_N^T\gamma J_N\underline \xi+i{\underline \xi}^TJ_N\underline D},
\end{equation}
$\underline\xi\in{\mathbb R}^{2N}$ and $\underline D\in{\mathbb R}^{2N}$ are vectors, $\gamma$ is a real symmetric $2N\times2N$-matrix and the so-called \textit{symplectic matrix} $J_N$ is defined as
\begin{equation}
J_N:=\bigoplus\limits_{i=1}^N J\,,\qquad J:=
\begin{pmatrix} 
0 & 1 \\
-1 & 0 
\end{pmatrix}.
\end{equation}
The vector $\underline D$ contains the first moments of the state, the displacements, while the matrix $\gamma$ contains the covariances and is therefore called \textit{covariance matrix}. Making use of a vector $\hat{\underline R}=(\hat{x}_1,\hat{p}_1,\ldots,\hat{x}_N,\hat{p}_N)=(\hat{R}_1,\ldots,\hat{R}_{2N})$, $\underline D$ and $\gamma$ can be defined by 
\begin{align}
d_i &= \mathrm{tr}[\hat{\rho}\hat{R}_i], \\
\gamma_{ij} &= \mathrm{tr}[\hat{\rho}(\hat{R}_i\hat{R}_j+\hat{R}_j\hat{R}_i)]-2\,\mathrm{tr}[\hat{\rho}\hat{R}_i]\mathrm{tr}[\hat{\rho}\hat{R}_j].
\end{align}
A consequence of the definition of Gaussian states is that they can be solely described by the displacements $\underline D$ and the covariance matrix $\gamma$. However, not every real symmetric $2N\times2N$-matrix corresponds to a valid quantum state, since states also have to obey the Heisenberg uncertainty relation. In the covariance matrix formalism the latter translates into the inequality  
\begin{equation}\label{eq:Heisenberg}
\gamma + iJ_N\geq 0.
\end{equation}
Here the $\geq0$ stands for positive semidefiniteness (in the following, when talking about matrices or operators, $\geq0$ will always mean positive semidefiniteness.) \cite{SimonH}.

In contrast to Gaussian states, general non-Gaussian CV states also require higher moments, i.e. the whole infinite set of moments, for their description. This is the big disadvantage of non-Gaussian states compared to Gaussian ones. A convenient way of describing general CV states using practical ${\mathbb C}$-numbers instead of operators is not known so far.

\section{Quantum Operations}
In this section the formalism behind quantum operations is introduced. Quantum operations are described by linear, \textit{completely positive} (\textit{CP}) maps. 
\begin{defini}
A map $T:S({\mathcal H}_1)\rightarrow S({\mathcal H}_2)$ mapping states $S({\mathcal H}_1)$ of some Hilbert space ${\mathcal H}_1$ onto states $S({\mathcal H}_2)$ of some possibly different Hilbert space ${\mathcal H}_2$ which preserves the positivity of the density operator is called \textup{positive}. Furthermore, if for larger Hilbert spaces also $T\otimes\mathbb{1}_n$ is positive for all $n\in\mathbb N$, $T$ is called \textup{completely positive}. 
\end{defini}
This requirement of complete positivity is clear, since every physical quantum operation should output a valid (and therefore positive) quantum state.

Next it is to be distinguished between completely positive \textit{trace-preserving} (\textit{CPTP}) maps and completely positive \textit{trace-decreasing} maps. Trace-preserving means that for a map $T$, $\mathrm{tr}[T(\hat{\rho})]=\mathrm{tr}[\hat{\rho}]=1$, while trace-preserving denotes $\mathrm{tr}[T(\hat{\rho})]<\mathrm{tr}[\hat{\rho}]=1$ (trace-increasing quantum operations do not exist.). Linear CPTP maps are called quantum channels, while linear CP trace-decreasing maps always involve for measurement operations.

While quantum channels correspond to a deterministic evolution of the state, measurement operations relate to a conditional and hence typically non-deterministic state evolution. How can quantum channels (CPTP maps) actually be described? 
\begin{thm}[Stinespring's dilation theorem \cite{Stinespring}]\label{thm:stine}
Let $\Upsilon:S({\mathcal H}_1)\rightarrow S({\mathcal H}_2)$ be a CPTP map. Then there exists an ancilla Hilbert space ${\mathcal K}$ of dimension $N$ and a joint unitary evolution $\hat{U}$ on ${\mathcal H}\otimes{\mathcal K}$ such that
\begin{equation}\label{eq:Stine}
\Upsilon(\hat{\rho})=\mathrm{tr}_{\mathcal K}[\,\hat{U}(\,\hat{\rho}\;\otimes\;\ket{0}\bra{0}\,)\,\hat{U}^\dagger],
\end{equation}
for all $\hat{\rho}\in S({\mathcal H}_1)$. The ancilla space ${\mathcal K}$ can be chosen such that $\dim[{\mathcal K}]=N\leq \dim^2[{\mathcal H}]$.
\end{thm}
Furthermore, M. Choi has shown that every CPTP map can be given in its \textit{operator-sum decomposition} or \textit{Kraus decomposition}:
\begin{thm}[Kraus decomposition \cite{Choi1,Kraus,Audretsch}]
Every CPTP map $\Upsilon:S({\mathcal H}_1)\rightarrow S({\mathcal H}_2)$ can be written as
\begin{equation}
\Upsilon(\hat{\rho})=\sum_{i=1}^N \hat{K}_i\hat{\rho}\hat{K}_i^\dagger,
\end{equation}
for all $\hat{\rho}\in S({\mathcal H}_1)$.
The $\hat{K}_i:{\mathcal H}_1\rightarrow {\mathcal H}_2$ are called \textit{Kraus operators} and obey the completeness relation $\sum_{i=1}^N {\hat{K}_i}^\dagger\hat{K}_i = \mathbb{1}$.
\end{thm}
The proofs of the theorems can be found in the denoted references. Making use of a basis $\{\ket{n}\}$ of the ancilla Hilbert space ${\mathcal K}$, the Kraus operators can be obtained with the aid of the Stinespring unitary $\hat{U}$ of equation \eqref{eq:Stine} via
\begin{equation}
\hat{K}_n = \braket{n|\hat{U}|0}.
\end{equation}

Now the measurement operations $\mathcal M_m(\hat{\rho})$ are introduced \cite{Audretsch}. Such a measurement operation is defined as a complete ($\sum_m {\hat{M}_m}^\dagger\hat{M}_m = \mathbb{1}$) set of operators $\{\hat{M}_m\}$. Each operator in the set corresponds to a possible measurement outcome $m$, which is obtained with probability $p(m)=\mathrm{tr}[{\hat{M}_m}^\dagger\hat{M}_m\hat{\rho}]<1$. The state after the measurement is 
\begin{equation}
\mathcal M_m(\hat{\rho})=\frac{\hat{M}_m\hat{\rho}{\hat{M}_m}^\dagger}{p(m)}.
\end{equation}
The normalization with $p(m)$ has to be included to compensate the trace-decrease of the actual measurement operation $\hat{M}_m\hat{\rho}{\hat{M}_m}^\dagger$.

To be even more general, each measurement outcome may not just correspond to a single operator $\hat{M}_m$, but to a whole set of Kraus operators $\{\hat{M}_{m,i}\}$. Then $\sum_{m,i} \hat{M}_{m,i}^\dagger\hat{M}_{m,i} = \mathbb{1}$ and the state after a measurement with result $m$ is
\begin{equation}
\mathcal M_m(\hat{\rho})=\frac{\sum_i \hat{M}_{m,i}\hat{\rho}{\hat{M}_{m,i}}^\dagger}{p(m)},
\end{equation}
which is obtained with probability $p(m)=\mathrm{tr}[\sum_{i} \hat{M}_{m,i}^\dagger\hat{M}_{m,i}\hat{\rho}]<1$.

An important class of quantum operations are those which can be implemented by using only linear optical elements. These unitary interactions are described by Hamiltonians of order $\leq2$ in the mode operators and have a linear input-output relation with respect to the mode operators. Actually, there are only 4 distinct such linear optical elements:
\begin{itemize}
\item Displacers: 
\begin{align}
\hat{D}(\alpha) &= \mathrm{e}^{\alpha\hat{a}^\dagger-\alpha^\ast\hat{a}}, \\
\hat{D}^\dagger(\alpha)\hat{a}\hat{D}(\alpha) &= \hat{a}+\alpha.
\end{align}
\item Phase Shifters: 
\begin{align}
\hat{\Gamma}(\phi) & = \mathrm{e}^{i\phi\hat{a}^\dagger\hat{a}}, \\
\hat{\Gamma}^\dagger(\phi)\hat{a}\hat{\Gamma}(\phi) & = \hat{a}\mathrm{e}^{i\phi}.
\end{align}
\item Beam Splitters: 
\begin{align}
\hat{B}_{\phi}(\theta) & = \mathrm{e}^{\theta(\mathrm{e}^{i\phi}{\hat{a}_1}^\dagger\hat{a}_2-\mathrm{e}^{-i\phi}\hat{a}_1{\hat{a}_2}^\dagger)}, \\
\hat{B}^\dagger_{\phi}(\theta)\hat{a}\hat{B}_{\phi}(\theta) & = \cos\,\theta\,\hat{a}_1+\mathrm{e}^{i\phi}\sin\,\theta\,\hat{a}_2.
\end{align}
\item Squeezers: 
\begin{align}
\hat{S}(\theta,r) & = \mathrm{e}^{\frac{1}{2}r(\mathrm{e}^{-i\theta}\hat{a}^2-\mathrm{e}^{i\theta}\hat{a}^{\dagger^2})}, \\
\hat{S}^\dagger(\theta,r)\hat{a}\hat{S}(\theta,r) & = \hat{a}\,\cosh\,r-\hat{a}^\dagger \mathrm{e}^{i\theta}\sinh\,r.
\end{align}
\end{itemize} 

These operations are also called Gaussian operations, since they map Gaussian states onto Gaussian states. Note that for the evaluation of the action of the operations on the mode operators, the Hadamard lemma has been exploited, which reads
\begin{equation}
\mathrm{e}^{\hat{X}} \hat{Y} \mathrm{e}^{-\hat{X}}=\sum_{m=0}^\infty\frac{[\hat{X},\hat{Y}]_m}{m!},
\end{equation}
for bounded operators $\hat{X}$ and $\hat{Y}$, $[\hat{X},\hat{Y}]_m=[\hat{X},[\hat{X},\hat{Y}]_{m-1}]$ and $[\hat{X},\hat{Y}]_0=\hat{Y}$. A proof for the theorem is presented in appendix \ref{Hamamard}.

It is important to know that every pure Gaussian state can be created from the vacuum just with a displacement and a squeezing operation:
\begin{equation}
\ket{\alpha,\xi}\,=\,\hat{D}(\alpha)\,\hat{S}(\xi=r\mathrm{e}^{i\theta})\ket{vac}
\end{equation}

For universal quantum computation it is not sufficient to work only with Gaussian states and operations. A non-Gaussian element is required for universality. In the previous section, it has been already shown that general non-Gaussian CV states are difficult to handle and unfortunately, in contrast to Gaussian operations, non-Gaussian operations are just as well hard to implement efficiently.

\section{Entanglement}
The main part of this thesis deals with entanglement. So what actually is entanglement? 

Several interesting answers to this question are collected in \cite{Bruss}. The most significant ones come from J. Bell, A. Peres and D. Mermin. Bell writes, "entanglement is a correlation that is stronger than any classical correlation". This leads to the fact that entangled states violate the well-known Bell inequalities \cite{Bell}. This violation has the consequence that quantum mechanics cannot be a realistic \textit{and} local theory at once, in line with Mermin, who writes, "entanglement is a correlation that contradicts the theory of elements of reality". So entanglement cannot be explained or simulated classically and hence quantum correlations due to entanglement can be concluded to be stronger than \textit{any} classical ones. \textit{Realism} in the sense of physical theories means that measurements just read off predetermined properties, which are so-called \textit{elements of physical reality} and also exist if they are not measured at all. \textit{Locality} means that for spatially separated particles a measurement on one of the particles cannot instantaneously affect the other particle. Classically such actions propagate at most with the speed of light. Note that in the widely established Copenhagen interpretation of quantum mechanics, both locality and realism are abdicated. However, there exist alternative approaches to quantum mechanics, giving up only one of the two, such as Bohmian mechanics, which is realistic, using hidden variables, but also non-local \cite{Bmech}.

Peres says, "entanglement is a trick that quantum magicians use to produce phenomena that cannot be imitated by classical magicians". An example of such a "trick" is quantum teleportation \cite{Qtel2,Qtel3}, which has its origin in a paper by C. H. Bennett et al. from 1993 \cite{Qtel1}. However, entanglement is not just the main resource for quantum teleportation but for virtually all applications in quantum information, such as long-distance quantum communication and quantum computation.

To come to a conclusion entanglement can be understood as a strong non-classical quantum correlation, which is the basis for most applications in quantum information. Now entanglement is defined and discussed in a mathematical way. The focus will be on bipartite entanglement.\footnote[7]{A very comprehensive and accessible introduction into entangled systems can be found in the book by J. Audretsch \cite{Audretsch}, which was one of the main sources for this chapter.}
\begin{defini}\label{def:ent}
An n-partite quantum state $\hat{\rho}^{1\ldots n}$ on the product Hilbert space ${\mathcal H}^1\otimes\ldots\otimes{\mathcal H}^n$ is called \textup{separable} if and only if it can be written as
\begin{equation}\label{eq:sep}
\hat{\rho}^{1\ldots n}=\sum_i p_i\;\hat{\rho}_i^{1}\otimes\ldots\otimes\hat{\rho}_i^{n}\,,\qquad p_i>0\;\forall \;i,\;\sum_i \;p_i=1.
\end{equation}
Any state which is not separable is called \textup{entangled}, \textup{inseparable} or \textup{quantum correlated}.\footnote[8]{It is worth mentioning that there also exists a classification which goes beyond this two sidedness, introducing the so-called \textit{quantum discord} \cite{discord1}. Quantum discord is a quantum correlation which can be present in separable states but which cannot be simulated classically. It has an operational interpretation in state merging \cite{discord2}. However, this is not relevant for this thesis and the two-sided classification between separability and entanglement is sufficient.}  
\end{defini}
Note that, if the convex combination in equation \eqref{eq:sep} consists of only one term, the state is in a fully uncorrelated \textit{product state} $\hat{\rho}^{1}\otimes\ldots\otimes\hat{\rho}^{n}$. For a convex combination containing several terms there are still classical correlations present \cite{Audretsch}. 

An important class of transformations regarding entanglement are the so-called \textit{LOCC} operations (LOCC stands for \textit{Local Operations and Classical Communication}). Since classical correlations can be created via LOCC, entanglement can also be defined as those correlations which cannot be produced using LOCC \cite{MeasuresIntro}. Such a definition is consistent with the mathematical definition \ref{def:ent}, since a quantum state $\hat{\rho}$ can be generated from a resource of separable states via LOCC if and only if it is separable. It can be also shown that separable states can be created from entangled states using LOCC. More generally, there are states which can be obtained from entangled states via LOCC which are still entangled. It is reasonable to assume that these states are "less entangled". So, what does "less entangled" mean? This leads to the problem of entanglement quantification: How much entangled is a given quantum state?

\subsection{Entanglement Measures}
A bipartite \textit{Entanglement Measure} is a function $E(\hat{\rho})$ which quantifies the entanglement of a given quantum state. What properties should such a function offer \cite{MeasuresIntro,HoroEntIntro}?
\begin{itemize}
\item[1.] For quantum states $\hat{\rho}$, $E(\hat{\rho})$ is a \textit{mapping from density matrices into positive real numbers}
\begin{equation}
\hat{\rho}\rightarrow E(\hat{\rho})\in\mathbb{R}^{+},
\end{equation}
which is \textit{vanishing on separable states}:
\begin{equation}
\hat{\rho}\;\text{separable}\;\Rightarrow\;E(\hat{\rho})=0
\end{equation} 
Actually, $E(\hat{\rho})$ is only required to be minimal for separable states. However, it is reasonable to set this constant to zero.
\item[2.] \textit{Monotonicity under LOCC:} Knowing that entanglement cannot be created via LOCC, it is clear that entanglement cannot increase under LOCC. So, for any LOCC operation $\Lambda$:
\begin{equation}\label{LOCCmon}
E(\Lambda(\hat{\rho}))\leq E(\hat{\rho})
\end{equation}
Unfortunately, the mathematical description of LOCC operations is very intricate \cite{LOCCmath}. Instead so-called \textit{separable operations} can be exploited, which are defined as
\begin{equation}
\Lambda(\hat{\rho})=\sum_i \hat{A}_i\otimes\hat{B}_i(\hat{\rho})\hat{A}_i^\dagger\otimes\hat{B}_i^\dagger\;,\qquad\sum_i \hat{A}_i^\dagger\hat{A}_i\otimes\hat{B}_i^\dagger\hat{B}_i=\mathbb{1}\otimes\mathbb{1}.
\end{equation}
Every LOCC operation can be cast in the form of separable operations. 

Actually, most entanglement measures also satisfy the stronger condition that they do not increase on average under LOCC,
\begin{equation}
\sum_i p_iE(\hat{\varsigma}_i)\leq E(\hat{\rho}),
\end{equation}
where $\{p_i,\hat{\varsigma}_i\}$ denotes the ensemble obtained from the state $\hat{\rho}$ via LOCC. Mostly it is easier to prove this stronger condition than the actually required condition of equation \ref{LOCCmon}.
\end{itemize}
These properties are the essential ones a function has to satisfy in order to quantify entanglement. However, there are some more properties, which can be helpful:
\begin{itemize}
\item[3.] \textit{Normalization:} For maximally entangled states $\ket{\psi_d^{+}}$
\begin{equation}
E(\ket{\psi_d^{+}})=\log d.
\end{equation}
Maximally entangled states only exist in bipartite systems out of two d-{dimensional} (d finite) subsystems in Hilbert spaces of the form ${\mathcal H}_d\otimes{\mathcal H}_d$. They are always pure and can be written as
\begin{equation}
\ket{\psi_d^{+}}=\frac{1}{\sqrt{d}}\Bigl(\ket{0,0}+\ket{1,1}+\ldots+\ket{d-1,d-1}\Bigr).
\end{equation}
The notion of maximal entanglement comes from the fact that any state in ${\mathcal H}_d\otimes{\mathcal H}_d$ can be prepared from such a state with certainty via LOCC.
\item[4.] For pure states, $E(\ket{\psi}\bra{\psi})$ reduces to the \textit{entropy of entanglement} $E_S$, which is defined as
\begin{equation}\label{eq:entropyent}
E_S(\ket{\psi}\bra{\psi}):=S(\mathrm{tr}_A[\ket{\psi}\bra{\psi}])=S(\mathrm{tr}_B[\ket{\psi}\bra{\psi}]),
\end{equation}
where $S$ denotes the \textit{von-Neumann entropy} $S(\hat{\rho})=-\mathrm{tr}[\hat{\rho}\log\hat{\rho}]$ and $\mathrm{tr}_A$ ($\mathrm{tr}_B$) denotes partial tracing over subsystem A (B). 
\end{itemize}
It is important to know that in the literature sometimes a distinction between \textit{entanglement measures} and \textit{entanglement monotones} is made, but the definitions are not always consistent regarding different references. Some authors use the terms interchangeably. As an example here the definitions of \cite{MeasuresIntro} are presented:  
\begin{defini}
A function $E(\hat{\rho})$ is called \textup{entanglement monotone} if and only if it satifies properties 1, 2 (no increase on average under LOCC) and 3.
\end{defini}
\begin{defini}
A function $E(\hat{\rho})$ is called \textup{entanglement measure} if and only if it satifies properties 1, 3 and 4 and does not increase under deterministic LOCC.
\end{defini}
In this thesis, when generally talking about measures and monotones, it is simply used the term "entanglement measures" instead of always explicitly mentioning both.

There is a whole bunch of further properties which entanglement measures may satisfy. For example, there are \textit{convex} entanglement measures and \textit{(fully) additive} ones. Then there are the properties of \textit{(asymptotic) continuity} and \textit{lockability}. Since this is supposed to be a rather brief introduction to entanglement measures, these properties are not explained here. They are not exploited in this thesis anyway. Instead the reader is again referred to \cite{MeasuresIntro,HoroEntIntro}, which provide a very comprehensible introduction to entanglement measures.

Conceptually three different kinds of entanglement measures can be distinguished. First there are the physically motivated \textit{operational entanglement measures}. Then there are the \textit{norm-based} and \textit{distance-based entanglement measures} and finally \textit{axiomatic entanglement measures} can be defined. An operational measure results for example from the answer to the question "Having a supply of $n$ copies of a quantum state $\hat{\rho}$, how many copies $rn$ of a maximally entangled state can be created in the limit $n\rightarrow\infty$?". Then the rate $r$ is defined as the \textit{distillable entanglement} $E_D$. An example for distance-based measures are the \textit{relative entropies of entanglement} $E_R^X$. Based on the quantum relative entropy they compare the correlations in a given quantum state with the correlations of the closest state from a set of states $X$. For $X$, any set of states can be chosen, such as the separable states or the nondistillable states. Finally, an axiomatic measure is just a function which has been constructed in such a way that it obeys the necessary requirements of entanglement measures. However, some axiomatic and norm/distance-based measures were given operational interpretations subsequently.

Since the description of quantum states in the DV and in the CV regime proceeds quite differently, there are also different ways towards entanglement quantification. While for DV states density matrices can be employed, for CV states only covariance matrices can be used. Hence for non-Gaussian CV states general entanglement quantification is yet an unsolved problem. In the following subsections some relevant measures are presented.

\subsection{DV Measures}\label{subsec:DVmeasures}
One of the most important DV entanglement measures is the pure state measure \textit{entropy of entanglement} $E_S$, which has already been defined in equation \eqref{eq:entropyent}. Operationally interpreted $E_S$ specifies the maximal reversible rate $r$ of transforming  $n$ copies of the pure state at hand via LOCC into $rn$ copies of a maximally entangled state and vice versa for $n\rightarrow\infty$. However, for mixed states, these transformations are not reversible anymore. This is the reason, why $E_S$ only works for pure states and is no entanglement measure anymore for mixed states.

When talking about DV pure states, the \textit{Schmidt decomposition} is a very important tool \cite{Audretsch}.
\begin{thm}[Schmidt decomposition]\label{thm:Schmidtdecomp}
Let $\ket{\psi}^{AB}$ be a normalized pure bipartite state in the product Hilbert space ${\mathcal H}^A_{n}\otimes{\mathcal H}^B_{m}$ with $\dim[{\mathcal H}^A_{n}]=n$ and $\dim[{\mathcal H}^B_{m}]=m$. Then there exist orthonormal bases $\{\ket{u_i}^A\}$ and $\{\ket{w_i}^B\}$ such that
\begin{equation}
\ket{\psi}^{AB}=\sum_{i=1}^{r(\psi)} \sqrt{\alpha_i}\ket{u_i}^A\ket{w_i}^B\,,\qquad \alpha_i>0\;\forall \;i,\;\sum_i \alpha_i=1,
\end{equation}
with $r(\psi)\leq min\{n,m\}$. Such a decomposition is called \textup{Schmidt decomposition} of $\ket{\psi}^{AB}$.
\end{thm}
The number $r(\psi)$ is called the \textit{Schmidt rank} of $\ket{\psi}^{AB}$ and $\{\alpha_i\}$ are its \textit{Schmidt coefficients}. What can be seen from the form of the Schmidt decomposition is that the reduced density operators $\hat{\rho}^A=\mathrm{tr}_B[\hat{\rho}^{AB}]$ and $\hat{\rho}^B=\mathrm{tr}_A[\hat{\rho}^{AB}]$ both have the same positive eigenvalues $\{\alpha_i\}$. Furthermore, $\{\ket{u_i}^A\}$ and $\{\ket{w_i}^B\}$ with suitably chosen phases are the orthonormal eigenvectors of $\hat{\rho}^A$ and $\hat{\rho}^B$. They are uniquely determined up to a phase. A proof of the Theorem is provided in appendix \ref{Schmidtproof}. Note that the Schmidt decomposition is essentially a restatement of the singular value decomposition in a different context. For the sake of completeness, a singular value decomposition is defined in the following way:
\begin{thm}[Singular value decomposition]
Let $M$ be a $n\times m$ matrix over the field $K$ which is either the field of the real or complex numbers. Then there exists a factorization of the form
\begin{equation}
M=UAV^\dagger,
\end{equation}
where
\begin{itemize}\label{thm:SVD}
\item $U$ is a $n\times n$ unitary matrix over the field $K$,
\item $V$ is a $m\times m$ unitary matrix over the field $K$,
\item and $A$ is a $n\times m$ diagonal matrix with nonnegative real numbers on the diagonal.
\end{itemize}
Such a decomposition is called \textup{singular value decomposition} of $M$. The diagonal entries of $A$ are called the \textup{singular values} of $M$.
\end{thm}
The singular values correspond to the Schmidt coefficients of the previous theorem \ref{thm:Schmidtdecomp} and the unitary matrices $U$ and $V$ correspond to matrices which transform the state $\ket{\psi}^{AB}$ into its new basis. The proof of theorem \ref{thm:SVD} corresponds to the proof of the Schmidt decomposition, which is presented in appendix \ref{Schmidtproof}.

It can be shown that any pure DV bipartite state is in a separable product state if and only if it has Schmidt rank $r=1$. Furthermore, for states $\ket{\psi}^{AB}$ and $\ket{\phi}^{AB}$ in ${\mathcal H}^A_{n}\otimes{\mathcal H}^B_{m}\,(n\leq m)$ it can be shown that $\ket{\psi}^{AB}$ can be LOCC-transformed into $\ket{\phi}^{AB}$ with unit probability if and only if the Schmidt coefficients $\{\alpha_i^\psi\}$ are majorized by $\{\alpha_i^\phi\}$ when taking them in decreasing order, i.e. $\alpha_1\geq\ldots\geq\alpha_n$ \cite{NielsenTrans}. 
\begin{defini}
$\{\alpha_i^\psi:i=1,\ldots ,n\}$ is said to be \textup{majorized} by $\{\alpha_i^\phi:i=1,\ldots ,n\}$, $\{\alpha_i^\psi\}\prec\{\alpha_i^\phi\}$, if
\begin{equation}
\sum_{i=1}^k \alpha_i^\psi \leq \sum_{i=1}^k \alpha_i^\phi\quad\forall\quad k=1,\ldots,n.
\end{equation}
\end{defini}
If the latter is the case, $\ket{\psi}^{AB}$ can be regarded as more entangled than $\ket{\phi}^{AB}$. This is also the reason why states with equally distributed Schmidt coefficients are maximally entangled. Their set of Schmidt coefficients is majorized by all other sets of Schmidt coefficients. Maximally entangled states in ${\mathcal H}^A_{d_1}\otimes{\mathcal H}^B_{d_2}$ with $d_1\leq d_2$ look like
\begin{equation}
\ket{\psi_{d_1}^{+}}=\frac{1}{\sqrt{d_1}}(\ket{0,0}+\ket{1,1}+\ldots+\ket{d_1-1,d_1-1}).
\end{equation}
However, a consequence of the majorization criterion is that there exist also \textit{incomparable} states, of which neither can be seen as more or less entangled than the other. Other measures have to be employed for such states \cite{MeasuresIntro}.

A very important mixed state entanglement measure is the \textit{entanglement of formation} $E_F$ \cite{MeasuresIntro}. It is defined as
\begin{equation}
E_F(\hat{\rho}):=\inf\{\,\sum_i p_i E_S(\ket{\psi_i}\bra{\psi_i})\;:\;\hat{\rho}=\sum_i p_i\ket{\psi_i}\bra{\psi_i}\,\}.
\end{equation}
It represents the minimal average entanglement over all pure state decompositions of $\hat{\rho}$, where the entropy of entanglement is employed as the pure state measure. Unfortunately, due to the variational problem, the entanglement of formation is extremely difficult to calculate. However, it is a very important measure for two reasons. On the one hand, some open questions regarding $E_F$ have tight connections to major open questions in quantum information. Explicitly full additivity of the entanglement of formation is equivalent to the additivity of the classical communication
capacity of quantum channels \cite{ShoAddi}. However, recently it has been shown by providing a counterexample that these quantities are actually \textit{not} additive \cite{Hastings}.

On the other hand, Wootters could derive a closed analytical solution for the case of bipartite qubit states \cite{Woot1,Woot2}. So, for $\hat{\rho}$ a two-qubit state in ${\mathcal H}^A_2\otimes{\mathcal H}^B_2$,
\begin{equation}
E_F(\hat{\rho})=s\Bigl(\frac{1+\sqrt{1-C^2(\hat{\rho})}}{2}\Bigr),
\end{equation}  
with
\begin{equation}
s(x):=-x\log_2x-(1-x)\log_2(1-x).
\end{equation}  
$C$ is the widely-used \textit{Concurrence}, which is defined as
\begin{equation}
C(\hat{\rho}):=\max\{0,\sqrt{\xi_1}-\sqrt{\xi_2}-\sqrt{\xi_3}-\sqrt{\xi_4}\},
\end{equation}
where $\xi_i$ are the eigenvalues of $\hat{\rho}\tilde{\hat{\rho}}$ in decreasing order. $\tilde{\hat{\rho}}$ is given by
\begin{equation}
\tilde{\hat{\rho}}=(\hat{\sigma}_y\otimes\hat{\sigma}_y)\hat{\rho}^\ast(\hat{\sigma}_y\otimes\hat{\sigma}_y),
\end{equation}
where $\hat{\rho}^\ast$ denotes the elementwise complex conjugate of $\hat{\rho}$ and $\hat{\sigma}_y$ is the Pauli Y matrix. Since the entanglement of formation and the concurrence are monotonically related, many authors prefer to simply quantify entanglement using the concurrence instead of the entanglement of formation itself. For higher dimensions this connection breaks down, as the concurrence is not even properly defined anymore. Note, however, that there actually are approaches for generalizing the concurrence to higher dimensions \cite{Gour}.

Another very important entanglement monotone is the \textit{logarithmic negativity} $E_N$ \cite{Neg}. To define $E_N$, first the notion of \textit{partial transposition} has to be introduced. Consider a bipartite DV quantum state in local orthonormal bases:
\begin{equation}
\ket{\psi}^{AB}=\sum_{i,j,k,l} c_{ijkl} \ket{i}^A\bra{j}\otimes\ket{k}^B\bra{l}.
\end{equation}
Then the partial transposition $\Gamma_B$ with respect to system B is defined as
\begin{equation}
\begin{aligned}
(\hat{\rho}^{AB}){}^{\Gamma_B} &= (\mathbb{1}_A\otimes T_B)\lbrack\hat{\rho}_{AB}\rbrack \\
 &=(\mathbb{1}_A\otimes T_B)(\sum_{i,j,k,l} c_{ijkl} \ket{i}^A\bra{j}\otimes\ket{k}^B\bra{l}) \\
 &= \sum_{i,j,k,l} c_{ijkl} \ket{i}^A\bra{j}\otimes\ket{l}^B\bra{k},
\end{aligned}
\end{equation}
where $T_B$ denotes "normal" transposition of system B. Partial transposition $\Gamma_B$ is a \textit{positive but not completely positive} (\textit{PnCP}) map. So if applied on an entangled quantum state, the output will not necessarily be positive semidefinite. However, if applied on a separable state, the output will actually be positive semidefinite, since the subsystems of separable states fully separate and hence the map acts as two individual maps $\mathbb{1}_A$ and $T_B$ on the state. These individual maps are both positive and hence the output will be positive semidefinite again \cite{Peres}. Therefore, partial transposition yields an inseparability criterion, which is also called the \textit{Peres-Horodecki criterion} \cite{PeresHorodeckiCrit}. 

The concept can be generalized: All PnCP maps distinguish some entangled states from the separable and the other entangled states. If the output of a PnCP operation is not positive semidefinite, it can be concluded that the input was entangled. States which are not positive semidefinite anymore after partial transposition are called \textit{NPT} (\textit{negative partial transpose}) states, while states which are still positive semidefinite are called \textit{PPT} (\textit{positive partial transpose}) states. Note that PPT states may nevertheless contain entanglement - so-called \textit{bound entanglement}. It is just not possible to detect this entanglement via partial transposition. What makes the partial transposition particularly interesting is that it can be shown that "PPTness" coincides with nondistillability \cite{BoundEnt}. Nevertheless, it is not clear yet whether all NPT states are distillable, or whether there exist nondistillable NPT states. 

The \textit{logarithmic negativity} $E_N$ is a measure which attempts to quantify the negativity in the spectrum of the partial transpose. Therefore, it is not able to measure bound entanglement. It is defined as
\begin{equation}
E_N(\hat{\rho}):=\log_2||\hat{\rho}^{\Gamma_B}||=\log_2[1+\sum_i (|\lambda_i|-\lambda_i)],
\end{equation}
where $||X||:=\mathrm{tr}[\sqrt{X^\dagger X}]$ is the \textit{trace norm} and $\lambda_i$ are the eigenvalues of $\hat{\rho}^{\Gamma_B}$. It does not matter whether the partial transposition is performed with respect to system A or B. Sometimes also the related "normal" \textit{negativity} ${\mathcal N}$ is used \cite{Neg}, which is defined as
\begin{equation}
{\mathcal N}(\hat{\rho}):=\frac{||\hat{\rho}^{\Gamma_B}||-1}{2}=\frac{\sum_i (|\lambda_i|-\lambda_i)}{2}.
\end{equation}
While ${\mathcal N}$ is convex, but not additive, $E_N$ is not convex, but fully additive. Since additivity is in the majority of cases more desired than convexity, the logarithmic negativity is the more often used monotone. The major advantage of the negativity quantities is that they can be calculated rather easily. It is just an eigenvalue problem to be solved.

The last monotone, which is briefly presented in this subsection, is the \textit{global robustness of entanglement} $R_g$ \cite{MeasuresIntro,Robust1,Robust2,Robust3}.
\begin{equation}
R_g(\hat{\rho}):=\inf_{\hat{\sigma}} \{ \lambda\,:\,\lambda\geq0,\,\hat{\sigma}\in QS\;\;\text{such that}\;\;(1-\lambda)\hat{\rho}+\lambda\hat{\sigma}\in SEP\},
\end{equation}
where $QS$ denotes the set of all quantum states and $SEP$ the set of the separable quantum states. It quantifies the minimal amount of arbitrary (\textit{global}) noise $\hat{\sigma}$ to be mixed in such that $\hat{\rho}$ becomes separable. However, also other robustness quantities can be defined by considering specific types of noise and drawing $\hat{\sigma}$ for example from the set of separable states or PPT states. 

Robustness monotones can be sometimes calculated, often at least bounded nontrivially. Furthermore they find applications in proofs of theorems and similar argumentations, such as in this thesis in subsection \ref{subsec:Robust}.

\subsection{CV Measures}
Entanglement quantification in the CV regime is a much harder task than in DV. For general states the continuity property of some measures breaks down: It is straightforward to construct pure product states which have arbitrarily highly entangled states in their arbitrarily small neighborhood \cite{MeasuresIntro,ContIssues}. This problem can be solved by applying an energy bound on the states. When only states with bounded mean energy, i.e. $S_M=\{\hat{\rho}\in S:\mathrm{tr}[\hat{\rho}\hat{H}]\leq M\}$, are considered ($S$ denotes the set of all quantum states.), continuity can be recovered. However, there is still no recipe for entanglement quantification in CV even for these energy bounded states. Nevertheless, for specific measures (strongly superadditive ones)\footnote[9]{\textit{Strong superadditivity} of an entanglement measure $E$ means $E(\hat{\rho}^{AB}_{12})\geq E(\hat{\rho}^{AB}_{1})+E(\hat{\rho}^{AB}_{2})$, where $1$ and $2$ refer to two pairs of entangled particles, held by systems $A$ and $B$. See also \cite{MeasuresIntro}.} rough bounds may be calculated with the aid of Gaussian states and their entanglement. Therefore, the set of states is restrained even stronger and only the Gaussian states are considered in the following. Once more, even for this class of states entanglement quantification is an extraordinarily difficult task. However, there do exist at least a few measures which can be calculated.

One such measure is the entropy of entanglement, which features a translation into the pure state Gaussian CV regime. For Alice and Bob holding $n_A$ and $n_B$ harmonic oscillator systems in a Gaussian state,\footnote[10]{In quantum information it is common to give two system $A$ and $B$ specific names - \textit{Alice} and \textit{Bob}. Especially in quantum communication it is very popular to call sender and receiver Alice and Bob.} the entropy of entanglement is
\begin{equation}
E_S(\hat{\rho}^{AB})=\sum_{i=1}^{n_A}\Bigl(\frac{\mu_i+1}{2}\log_2\frac{\mu_i+1}{2}-\frac{\mu_i-1}{2}\log_2\frac{\mu_i-1}{2}\Bigr).
\end{equation}
$\mu_i$ are the \textit{symplectic eigenvalues} of Alice's reduced covariance matrix $\gamma^A$, which is the submatrix referring to Alice's system, of the overall covariance matrix $\gamma^{AB}$. On the one hand, these symplectic eigenvalues are the absolute values of the "normal" eigenvalues of $iJ_N^{-1}\gamma^A$. On the other hand, the general definition of symplectic eigenvalues corresponds to the \textit{Williamson normal form}. When employing covariance matrices no Hermitian matrices and unitary transformations are used anymore as in the case of density matrices. Instead linear Gaussian transformations are described by so-called \textit{symplectic transformations}. Due to the preservation of the canonical commutation relations, symplectic transformations $S$ are defined by their action on the symplectic matrix $J_N$:
\begin{equation}
SJ_NS^\dagger=J_N.
\end{equation}
This is verified by the real $2N\times2N$ matrices $S$, which form the \textit{real, symplectic group} $\mathrm{Sp}(2n,\mathbb R)$. Covariance matrices are then transformed according to $S\gamma S^T$. Williamson proved that for any covariance matrix $\gamma$ on $N$ harmonic oscillators there exists a symplectic transformation $S$ such that
\begin{equation}
S\gamma S^T=\bigoplus\limits_{i=1}^N
\begin{pmatrix} 
\mu_i & 0 \\
0 & \mu_i 
\end{pmatrix}.
\end{equation}
Then the diagonal elements $\mu_i$ are the symplectic eigenvalues of $\gamma$ and the whole set $\{\mu_i\}$ is called the \textit{symplectic spectrum} of $\gamma$ \cite{Williamson1,SimonH}.

Another entanglement measure which can be also employed in the Gaussian CV case is the logarithmic negativity. In contrast to the entropy of entanglement it also works for mixed Gaussian states. Just as in the DV case, where it is the one exception of many measures which is relatively easy to calculate, in the CV case it is the only mixed state measure which in general can be actually calculated at all. Note that for pure states it differs from the entropy of entanglement and hence it rather ought to be called a monotone. If Alice and Bob are in the possession of $N=n_A+n_B$ quantum harmonic oscillators, the logarithmic negativity is
\begin{equation}
E_N(\hat{\rho})=-\sum_{i=1}^N\log_2[\min\{1,\tilde{\mu}_i\}].
\end{equation}
$\{\tilde{\mu}_i\}$ is the symplectic spectrum of the partially transposed state described by the covariance matrix $\gamma^{\Gamma_B}$. On the level of covariance matrices, the partial transposition corresponds to a time reversal. In a system descibed by the canonical quadratures $\hat{x}$ and $\hat{p}$ the time reversal operation can be expressed as $\hat{x}\rightarrow\hat{x}$ and $\hat{p}\rightarrow-\hat{p}$. 

Just as in the DV case partial transposition can be used for the derivation of a necessary condition for separability of CV states \cite{Simon}. Recalling the Heisenberg uncertainty relation on the level of covariance matrices \eqref{eq:Heisenberg}, for every separable state
\begin{equation}
\gamma^{\Gamma_B} + iJ_N\geq 0.
\end{equation}
Note that this is valid for \textit{all} CV states, not just Gaussian ones.

Finally, it has to be pointed out once more that entanglement quantification in the CV regime is much more difficult than in the DV regime. Only for Gaussian states, there exist some rare examples of measures which can be properly defined and possibly calculated. Entanglement quantification for general non-Gaussian CV states is yet an unsolved problem, even for energy bounded states. However, all physically relevant states are of course energy bounded.

\subsection{Entanglement Witnesses}\label{subsec:witnessing}
There are a lot of cases where entanglement quantification is too ambitious, for example, in the non-Gaussian CV regime or simply when matrices have to be evaluated which are too big for eigenvalue problems to be solved. Then it can be asked whether the entanglement, if not quantifiable, can at least be detected or \textit{witnessed}. Furthermore, there are tasks where the amount of entanglement is actually rather secondary and it is only important that entanglement is present at all. For these purposes, so-called \textit{entanglement witnesses} can be used \cite{MeasuresIntro}.
\begin{defini}\label{def:witnesses}
A Hermitian operator $\hat{W}$ is called an \textup{entanglement witness} if and only if:
\begin{align}
\forall\,\hat{\rho}_{sep}\in SEP\;&\;\mathrm{tr}[\hat{W}\hat{\rho}_{sep}]\geq0, \\
\exists\,\hat{\rho}_{ent}\,\;\text{s.t.}\;&\;\mathrm{tr}[\hat{W}\hat{\rho}_{ent}]<0.
\end{align}
\end{defini}
An entanglement witness therefore separates some entangled states from the separable states and the rest of the entangled states. Hence it can be interpreted as a hyperplane in state space as illustrated in figure \ref{fig:witness}.
\begin{figure}[ht]
\begin{center}
\includegraphics[width=10cm]{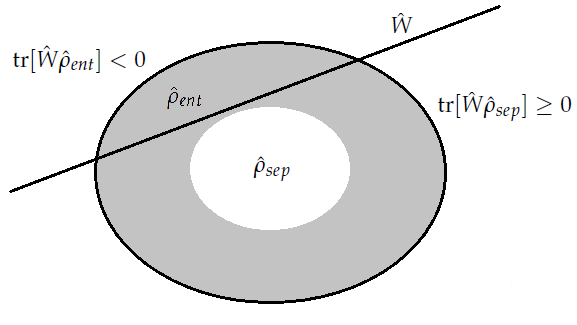}
\caption[Entanglement witnesses as hyperplanes in state space.]{An entanglement witness can be interpreted as a hyperplane in state space. All states to the right of the hyperplane, e.g. all separable states, have nonnegative expectation values of the witness operator and hence, it cannot be concluded whether they are entangled or not. However, the entangled states to the left of the hyperplane are detected by the witness \cite{HoroEntIntro}.}
\label{fig:witness}
\end{center}
\end{figure}
Arguing with this hyperplane interpretation, it becomes clear that a quantum state is separable if and only if it has nonnegative expectation values of all entanglement witnesses. It can be shown that for any entangled state an applicable witness detecting it can be constructed \cite{HoroEntIntro,Audretsch}.

An example of a witness operator for the ${\mathcal H}_d\otimes{\mathcal H}_d$ case is the swap operator
\begin{equation}
\hat{V}=\sum_{i,j=0}^{d-1} \ket{i}\bra{j}\otimes\ket{j}\bra{i}.
\end{equation}
It is straightforward to show that this operator yields positivity for all separable states. However, it can be also shown that it possesses negative eigenvalues, which qualifies it as an entanglement witness \cite{HoroEntIntro}.

A very efficient way of entanglement detection has already been set out and explained: It is not provided by actual witness operators but instead by PnCP maps such as the partial transposition (see subsection \ref{subsec:DVmeasures}). Due to the positivity of the actual map and the added $\mathbb{1}_n$'s, PnCP maps always leave separable states positive. However, their missing complete positivity enables them to detect some entangled states - the NPT states in the case of partial transposition.

In the preceding subsection it has been pointed out that entanglement quantification for general CV states is still a hard problem. So, in this regime entanglement witnesses are especially desirable. Shchukin and Vogel provided a very elegant way for entanglement witnessing in this class of states \cite{SV,SVcomment,SVgeneralize}. A slightly adapted version of their criteria gives rise to one of the central tools used in this thesis. Therefore, their result is discussed in more detail. However, their method also is not based on witness operators, but on so-called \textit{expectation value matrices} (\textit{EVMs}, also \textit{matrices of moments}). For a set of operators $\{\hat{M}_k\}$ the EVM $\chi$ is defined as
\begin{equation}
\chi_{ij}=\braket{\hat{M}_i^\dagger\hat{M}_j},
\end{equation}
where the observables are typically of the form $\{\hat{M}_k\}=\{\hat{A}_k\otimes\hat{B}_k\}$. Then it is straightforward to see that for separable states, $\chi$ also is separable. So when combining the EVM approach with a PnCP such as partial transposition, the separability of the state in question can be investigated based on the matrix of moments \cite{Guehne}. Inseparability, which may result in non-positivity of the mapped state, then also leads to non-positivity of the matrix of moments. Hence criteria based on subdeterminants of the EVM can be derived. However, this will become clear in the derivation of the Shchukin-Vogel inseparability criteria.\footnote[11]{For a nice introduction and a detailed overview of entanglement detection, see G\"uhne and T\'{o}th's review article \cite{Guehne}.} 
\begin{thm}[Shchukin-Vogel inseparability criteria]\label{thm:SV}
A bipartite CV state $\hat{\rho}$ is NPT if and only if there exists a negative principal minor of $M(\hat{\rho}^{\Gamma_B})$, i.e. $\det[M(\hat{\rho}^{\Gamma_B})_{\underline r}]<0$ for some $\underline r\equiv(r_1,\ldots,r_N)$ with $r_1<r_2<...<r_N$.
\end{thm}
The matrix of moments $M$ is here defined as
\begin{equation}
\begin{aligned}
M_{ij}(\hat{\rho}) &= M_{i_1i_2i_3i_4,j_1j_2j_3j_4}(\hat{\rho}) \\
&=\braket{\hat{a}^{\dagger^{i_2}}\hat{a}^{^{i_1}}\hat{a}^{\dagger^{j_1}}\hat{a}^{^{j_2}}\hat{b}^{\dagger^{j_4}}\hat{b}^{^{j_3}}\hat{b}^{\dagger^{i_3}}\hat{b}^{^{i_4}}}_{\hat{\rho}}.
\end{aligned}
\end{equation}
With a convenient arrangement of the indices, i.e. combining $i_1$, $i_2$, $i_3$, $i_4$ to $i$ and $j_1$, $j_2$, $j_3$, $j_4$ to $j$, $M$ can be explicitly written as
\begin{equation}\label{eq:SVmatrix}
M(\hat{\rho}^\Gamma)=\begin{pmatrix} 
1 & \braket{\hat{a}} & \braket{\hat{a}^\dagger} & \braket{\hat{b}^\dagger} & \cdots \\
\braket{\hat{a}^\dagger} & \braket{\hat{a}^\dagger\hat{a}} & \braket{\hat{a}^{\dagger^{2}}} & \braket{\hat{a}^\dagger\hat{b}^\dagger} & \cdots \\
\braket{\hat{a}} & \braket{\hat{a}^2} & \braket{\hat{a}\hat{a}^\dagger} & \braket{\hat{a}\hat{b}^\dagger} & \cdots \\
\braket{\hat{b}} & \braket{\hat{a}\hat{b}} & \braket{\hat{a}^\dagger\hat{b}} & \braket{\hat{b}^\dagger\hat{b}} & \cdots \\
\vdots & \vdots & \vdots & \vdots & \ddots
\end{pmatrix}.
\end{equation}
The process of the index arrangement is not crucial, however, it is explained in detail in the proof. Also note that the matrix of moments $M$ is a matrix of infinite size. What actually are principal minors?
\begin{defini}
Consider the submatrices $M_{\underline r}$ of a square matrix $M$ which are obtained by deleting all rows and columns apart from the the ones with the numbers $\underline r\equiv(r_1,\ldots,r_N)\,;\; r_1<r_2<...<r_N$. The determinants $\det[M_{\underline r}]$ of these submatrices are called \textup{principal minors} of $M$. 
\end{defini}

\begin{proof}[\textbf{Proof of Theorem \ref{thm:SV}.}] \cite{SV}
Recall that a Hermitian operator $\hat{A}$ is called nonnegative if
\begin{equation}
\braket{\psi|\hat{A}|\psi}=\mathrm{tr}[\hat{A}\ket{\psi}\bra{\psi}]\geq0\;\,\forall\,\ket{\psi}.
\end{equation}
The nonnegative operator $\ket{\psi}\bra{\psi}$ can be rewritten as $\hat{f}^\dagger\hat{f}$ with $\hat{f}=\ket{vac}\bra{\psi}$. $\ket{vac}=\ket{00}$ denotes the vacuum in the bipartite case. Any pure bipartite state $\ket{\psi}$ can be expressed as
\begin{equation}
\ket{\psi}=\hat{g}^\dagger\ket{00},
\end{equation}
for an appropriate operator function $\hat{g}=g(\hat{a},\hat{b})$. $\hat{a}$ and $\hat{b}$ denote the annihilation operators of the first and the second mode respectively. Hence
\begin{equation}\label{eq:SVderive1}
\hat{f}=\ket{00}\bra{00}\hat{g}.
\end{equation}
It is straightforward to show, see appendix \ref{VacuumDensityOperator}, that the two-mode vacuum density operator can be expressed in the normally ordered form
\begin{equation}\label{eq:SVderive2}
\ket{00}\bra{00}=\,:\mathrm{e}^{-\hat{a}^\dagger\hat{a}-\hat{b}^\dagger\hat{b}}:,
\end{equation}
where $:\cdots:$ denotes normal ordering. From equations \eqref{eq:SVderive1} and \eqref{eq:SVderive2}, it can be seen that the normally ordered form of $\hat{f}$ exists. Hence, a Hermitian operator $\hat{A}$ is nonnegative if and only if for any operator $\hat{f}$ whose normally ordered form exists,
\begin{equation}
\braket{\hat{f}^\dagger\hat{f}}_{\hat{A}}=\mathrm{tr}[\hat{A}\hat{f}^\dagger\hat{f}]\geq0.
\end{equation}
Now the Peres-Horodecki criterion can be applied. For any separable state $\hat{\rho}$,
\begin{equation}
\braket{\hat{f}^\dagger\hat{f}}_{\hat{\rho}^{\Gamma_B}}=\mathrm{tr}[{\hat{\rho}^{\Gamma_B}}\hat{f}^\dagger\hat{f}]\geq0,
\end{equation}
which is satisfied for any operator $\hat{f}$ whose normally ordered form exists. Due to this existence, $\hat{f}$ can be written as
\begin{equation}
\hat{f}=\sum_{i_1,i_2,i_3,i_4=0}^\infty c_{i_1i_2i_3i_4}\hat{a}^{\dagger^{i_1}}\hat{a}^{^{i_2}}\hat{b}^{\dagger^{i_3}}\hat{b}^{^{i_4}}.
\end{equation}
Then
\begin{equation}\label{eq:SVderive3}
\braket{\hat{f}^\dagger\hat{f}}_{\hat{\rho}^{\Gamma_B}}=\sum_{i_1,i_2,i_3,i_4,j_1,j_2,j_3,j_4=0}^\infty c_{j_1j_2j_3j_4}^\ast c_{i_1i_2i_3i_4}M_{j_1j_2j_3j_4,i_1i_2i_3i_4}\geq0,
\end{equation}
with the moments of the partial transpose
\begin{equation}
M_{j_1j_2j_3j_4,i_1i_2i_3i_4}=\braket{\hat{a}^{\dagger^{j_2}}\hat{a}^{^{j_1}}\hat{a}^{\dagger^{i_1}}\hat{a}^{^{i_2}}\hat{b}^{\dagger^{j_4}}\hat{b}^{^{j_3}}\hat{b}^{\dagger^{i_3}}\hat{b}^{^{i_4}}}_{\hat{\rho}^{\Gamma_B}}.
\end{equation}
The left hand side of inequality \eqref{eq:SVderive3} is a quadratic expression in the coefficients $c_{i_1i_2i_3i_4}$. Sylvester's criterion states that such an inequality holds true for all $c_{i_1i_2i_3i_4}$ if and only if all principal minors of the expression \eqref{eq:SVderive3} are nonnegative \cite{Sylvester}. However, first a relation between the moments of the partially transposed and the moments of the original state has to be derived. This is straightforward:
\begin{equation}
\begin{aligned}
M_{j_1j_2j_3j_4,i_1i_2i_3i_4} &=\braket{\hat{a}^{\dagger^{j_2}}\hat{a}^{^{j_1}}\hat{a}^{\dagger^{i_1}}\hat{a}^{^{i_2}}\hat{b}^{\dagger^{j_4}}\hat{b}^{^{j_3}}\hat{b}^{\dagger^{i_3}}\hat{b}^{^{i_4}}}_{\hat{\rho}^{\Gamma_B}} \\
&=\braket{\hat{a}^{\dagger^{j_2}}\hat{a}^{^{j_1}}\hat{a}^{\dagger^{i_1}}\hat{a}^{^{i_2}}(\hat{b}^{\dagger^{j_4}}\hat{b}^{^{j_3}}\hat{b}^{\dagger^{i_3}}\hat{b}^{^{i_4}})^{\dagger}}_{\hat{\rho}} \\
&=\braket{\hat{a}^{\dagger^{j_2}}\hat{a}^{^{j_1}}\hat{a}^{\dagger^{i_1}}\hat{a}^{^{i_2}}\hat{b}^{\dagger^{i_4}}\hat{b}^{^{i_3}}\hat{b}^{\dagger^{j_3}}\hat{b}^{^{j_4}}}_{\hat{\rho}}.
\end{aligned}
\end{equation}

Furthermore, it is convenient to establish an ordering of indices and number the ordered multi-indices $(i_1,i_2,i_3,i_4)$ by $i$ and $(j_1,j_2,j_3,j_4)$ by $j$. For the inseparability criteria the ordering is irrelevant, however, an example is the following: Order the set of multi-indices $\{(i_1,i_2,i_3,i_4)\}$ in such a way that for two multi-indices $\underline u=(i_1,i_2,i_3,i_4)$ and $\underline v=(j_1,j_2,j_3,j_4)$
\begin{equation}
\underline u < \underline v \Leftrightarrow
\begin{cases}
|\underline u|<|\underline v| & \text{or} \\
|\underline u|=|\underline v| & \text{and}\; \underline u <' \underline v.
\end{cases}
\end{equation}
$|\underline u|=i_1+i_2+i_3+i_4$, and $\underline u <' \underline v$ means that the first nonzero difference $j_3-i_3,j_4-i_4,j_1-i_1,j_2-i_2$ is positive. This is the ordering Shchukin and Vogel applied in their original paper and which results in a matrix of moments as in equation \eqref{eq:SVmatrix}.

With the relation between the moments of the partially transposed and the original state as well as with the multi-index ordering procedure the nonnegativity of all minors of the quadratic expression \eqref{eq:SVderive3} can be expressed as
\begin{equation}
\det[M(\hat{\rho}^{\Gamma_B})]_{\underline r}\geq0\;\;\,\forall\;\,\underline r\equiv(r_1,\ldots,r_N)\,,\qquad r_1<r_2<...<r_N,
\end{equation}
where the matrix of moments is defined by
\begin{equation}
M_{ij}(\hat{\rho}^{\Gamma_B})=\braket{\hat{a}^{\dagger^{i_2}}\hat{a}^{^{i_1}}\hat{a}^{\dagger^{j_1}}\hat{a}^{^{j_2}}\hat{b}^{\dagger^{j_4}}\hat{b}^{^{j_3}}\hat{b}^{\dagger^{i_3}}\hat{b}^{^{i_4}}}_{\hat{\rho}},
\end{equation} 
with the just introduced multi-index ordering. Note that $i$ and $j$ have been exchanged in order to follow the convention that ordinary matrices are denoted with $M_{ij}$ and not with $M_{ji}$. 

Summing up, for every separable state, all principal minors are nonnegative, explicitly $\det[M(\hat{\rho}_{sep}^{\Gamma_B})]_{\underline r} \geq0\;\;\forall\;\,\underline r$. Furthermore a necessary and sufficient condition for the positivity of the partial transpose can be formulated: The partial transpose of a bipartite CV quantum state is nonnegative if and only if all principal minors $\det[(\hat{\rho}^{\Gamma_B})]_{\underline r}$ are greater or equal zero. This also yields the necessary and sufficient condition for NPTness, which has been stated in the theorem. 
\end{proof}

The theorem provides a recipe for determining whether a CV state is NPT (and hence entangled) or not. However, it consists of an infinite set of inequalities to be examined. Therefore, the problem lies in the search for an appropriate principal minor. Nevertheless, for any NPT state such a minor exists. Since this infinite set of criteria is based on partial transposition it can only detect NPT entanglement. PPT entanglement (bound entanglement) remains hidden. It is also worth pointing out that all occuring moments can be measured experimentally \cite{SVmeasure}.

The SV criteria (in the following "Shchukin-Vogel inseparability criteria" will be abbreviated by "SV criteria"; furthermore "SV determinant" stands for a "principal minor of the matrix of moments by Shchukin and Vogel".) are especially outstanding as they provide not just \textit{sufficient} but also \textit{necessary} conditions for NPTness. There have been proposed several sufficient conditions for NPT entanglement before \cite{Simon,Duan,Raymer,Mancini}. However, all of them were given the same footing and they could be rederived with aid of the SV criteria. There are several states whose NPT entanglement can not be detected with aid of these earlier criteria. Mostly these are states whose entanglement only becomes visible when considering moments of sufficiently high order. However, the criteria in \cite{Simon,Duan,Raymer,Mancini} only employ moments up to second order. The SV criteria, on the contrary, make use of moments of any order.

\subsection{Entanglement Witnessing in Two-Mode Cat States}\label{catwitnessing}
To illustrate the power of the SV criteria, entanglement detection in two-mode Schr\"odinger cats is demonstrated. So consider the two-mode cat state
\begin{equation}\label{eq:phicat}
\ket{\psi_\phi}=\frac{1}{\sqrt{{\mathcal N}_\phi}}\Bigl(\ket{\alpha,\alpha}+\mathrm{e}^{i\phi}\ket{-\alpha,-\alpha}\Bigr)
\end{equation}
with
\begin{equation}
{\mathcal N}_\phi=2+2\mathrm{e}^{-4|\alpha|^2}\cos\phi,
\end{equation}
and $0\leq\phi<2\pi$. Ideally, the goal is to find a way to detect entanglement in $\ket{\psi_\phi}$ for \textit{any} $\phi$ with aid of some witness or SV determinant. It is unknown whether such a witness or single determinant exists. However, a set of "witnesses" can be derived from which a feasible one can always be chosen such that entanglement witnessing for any $\ket{\psi_\phi}$ can be performed.

\begin{figure}[ht]
\begin{center}
\includegraphics[width=12cm]{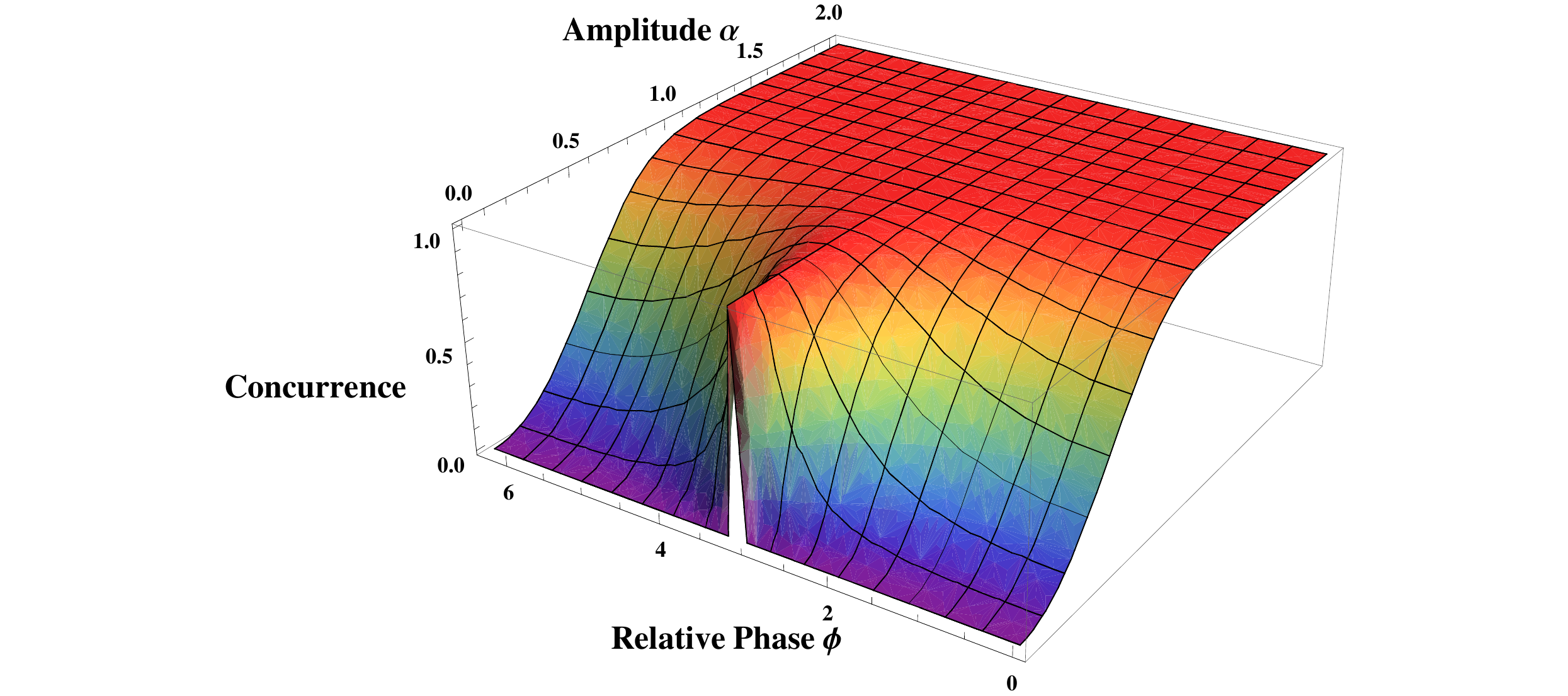}
\caption{Concurrence of the two-mode cat state $\ket{\psi_\phi}$, \eqref{eq:phicat}.}
\label{fig:CatConcurrences}
\end{center}
\end{figure}

To show that the considered state actually is entangled for $\alpha\neq0$ and arbitrary $\phi$, the concurrence can be calculated. It is given by
\begin{equation}
C(\ket{\psi_\phi})=\frac{1-\mathrm{e}^{-4|\alpha|^2}}{1+\mathrm{e}^{-4|\alpha|^2}\cos\phi},
\end{equation} 
and is shown in figure \ref{fig:CatConcurrences}. It is noticeable that the set of states defined by $\ket{\psi_\phi}$ has two-sided properties. The even two-mode cat $\ket{\psi_0}$ behaves other than the rest of the two-mode cats $\ket{\psi_\phi}$ with $\phi\neq\pi$ with regard to entanglement. $\ket{\psi_\pi}$ represents a maximally entangled two-qubit state $\frac{1}{\sqrt{2}}(\ket{01}+\ket{10})$ for \textit{any} amplitude $\alpha\neq0$ \cite{EntNonOrth1,EntNonOrth2}. Hence it reveals highly discontinuous behavior at $\alpha=0$. On the contrary, all other states display continuous behavior.

In \cite{SV} Shchukin and Vogel demonstrated entanglement detection in $\ket{\psi_\pi}$ with aid of the principal minor
\begin{equation}
s_1:=\begin{vmatrix} 
1 & \braket{\hat{b}^\dagger} & \braket{\hat{a}\hat{b}^\dagger}  \\
\braket{\hat{b}} & \braket{\hat{b}^\dagger\hat{b}} & \braket{\hat{a}\hat{b}^\dagger\hat{b}}  \\
\braket{\hat{a}^\dagger\hat{b}} & \braket{\hat{a}^\dagger\hat{b}^\dagger\hat{b}} & \braket{\hat{a}^\dagger\hat{a}\hat{b}^\dagger\hat{b}}
\end{vmatrix}.
\end{equation}
Furthermore, it is known that the inseparability criterion by Duan et al. (\cite{Duan}) is able to detect entanglement in $\ket{\psi_0}$. However, the Duan criterion is just a weaker form of the criterion corresponding to the SV determinant
\begin{equation}
s_2:=\begin{vmatrix} 
1 & \braket{\hat{a}} & \braket{\hat{b}^\dagger}  \\
\braket{\hat{a}^\dagger} & \braket{\hat{a}^\dagger\hat{a}} & \braket{\hat{a}^\dagger\hat{b}^\dagger}  \\
\braket{\hat{b}} & \braket{\hat{a}\hat{b}} & \braket{\hat{b}^\dagger\hat{b}}
\end{vmatrix},
\end{equation}
as Shchukin and Vogel showed. So, what are the witnessing capabilities of these determinants with regard to the whole set of states $\{\ket{\psi_\phi}:0\leq\phi<2\pi\}$? Evaluation of $s_1$ and $s_2$ for general $\phi$ is straightforward:
\begin{align}
s_1 &=-\frac{4|\alpha|^6\mathrm{e}^{4|\alpha|^2}}{(\mathrm{e}^{4|\alpha|^2}+\cos\phi)^3}(1-\mathrm{e}^{4|\alpha|^2}\cos\phi), \\
s_2 &=-\frac{4|\alpha|^4\mathrm{e}^{4|\alpha|^2}}{(\mathrm{e}^{4|\alpha|^2}+\cos\phi)^3}(1+\mathrm{e}^{4|\alpha|^2}\cos\phi).
\end{align}
Hence, determinant $s_1$ is capable of witnessing entanglement in $\{\ket{\psi_\phi}:\frac{\pi}{2}\leq\phi\leq\frac{3\pi}{2}\}$, while $s_2$ witnesses entanglement in $\{\ket{\psi_\phi}:0\leq\phi\leq\frac{\pi}{2}\,;\,\frac{3\pi}{2}\leq\phi<2\pi\}$. This is graphically demonstrated in figure \ref{fig:catdets}.
\begin{figure}[ht]
\subfloat{\includegraphics[width=0.45\textwidth]{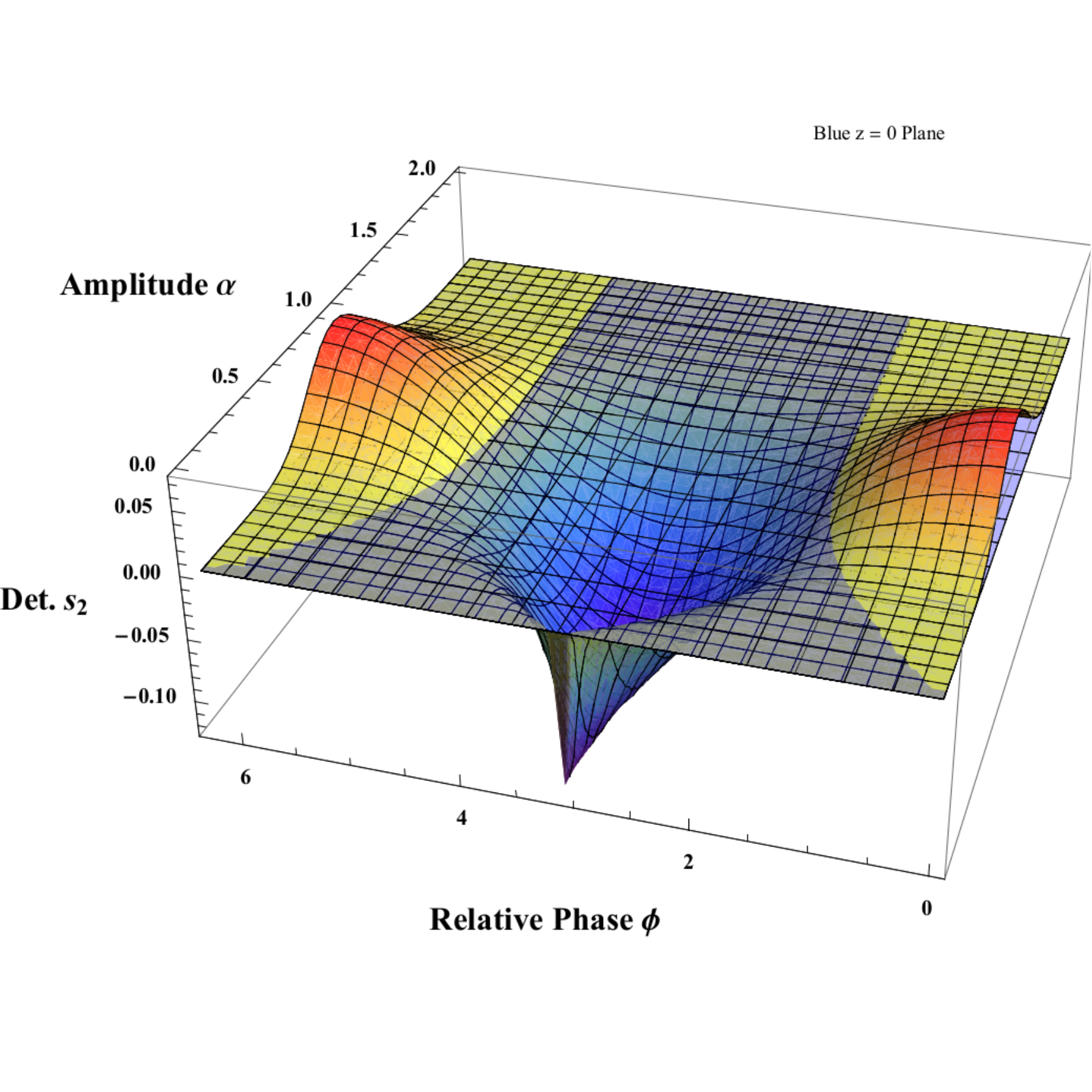}}\hfill
\subfloat{\includegraphics[width=0.45\textwidth]{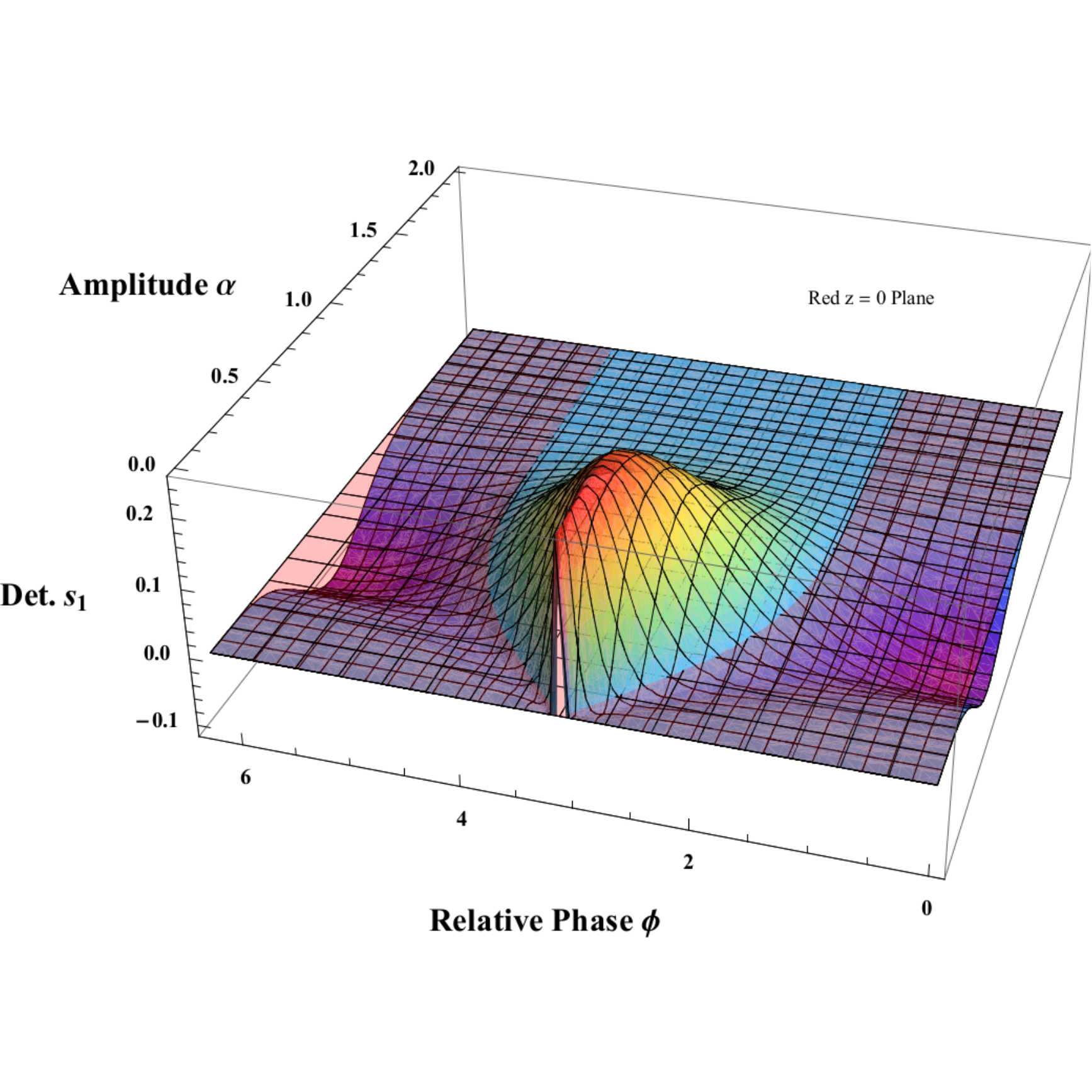}}\hfill
\caption[The SV determinants $s_1$ and $s_2$ for $\ket{\psi_\phi}$.]{The SV determinants $s_1$ and $s_2$ for $\ket{\psi_\phi}$. While $s_1$ performs witnessing for $\frac{\pi}{2}\leq\phi\leq\frac{3\pi}{2}$, $s_2$ does its job for $0\leq\phi\leq\frac{\pi}{2}$ and $\frac{3\pi}{2}\leq\phi<2\pi$. Hence, these two determinants are sufficient for witnessing in $\ket{\psi_\phi}$ for any $\phi$. Note that the blue and the red plane respectively denote the zero-plane.}
\label{fig:catdets}
\end{figure}
So, define a "generalized determinant"
\begin{equation}
\tilde{s}_n=ns_1+(1-n)s_2\,:\,n\in\{0,1\},
\end{equation}
which is capable of entanglement detection in the whole set $\{\ket{\psi_\phi}:0\leq\phi<2\pi\}$. This is meant in such a way that for any $\ket{\psi_\phi}$ there is $n=0\lor1$ such that $\tilde{s}_n<0$, which is equivalent to successful entanglement detection. $\tilde{s}_0$ applies for $0\leq\phi\leq\frac{\pi}{2}$ and $\frac{3\pi}{2}\leq\phi<2\pi$, while $\tilde{s}_1$ is the right choice for $\frac{\pi}{2}\leq\phi\leq\frac{3\pi}{2}$. For $\phi=\frac{\pi}{2}$ and $\phi=\frac{3\pi}{2}$, both $\tilde{s}_0$ and $\tilde{s}_1$ are fine. However, the generalized determinant can be also directly formulated as a function of the phase $\phi$ in $\ket{\psi_\phi}$: For any $\ket{\psi_\phi}$,
\begin{equation}
\tilde{s}'_\phi=\Theta_{\frac{1}{2}}[\cos(\phi+\pi)]\,s_1+\Theta_{\frac{1}{2}}[\cos(\phi)]\,s_2
\end{equation}
performs entanglement witnessing for any $\phi$. Here it has been made use of the \textit{Heaviside step function} $\Theta_{\frac{1}{2}}$, which is defined in appendix \ref{Heaviside}. It can be concluded that the SV criteria are a very strong tool for entanglement detection, when there is an approximate idea of which principal minors to look at. The derived "generalized determinants" yields for the considered states $\ket{\psi_\phi}$ of course no new results, since there is no witnessing necessary at all. The concurrence can be easily calculated for these states, which sets out all entanglement properties of the state. However, when the state is subject to a noisy channel, entanglement quantification is not so easy anymore. Then it is good to have a set of powerful witnesses or determinants in reserve to at least detect entanglement. 

Continuative questions to be tackled are whether the derived generalized determinants also work for more general two-mode CSS states, such as
\begin{equation}\label{eq:generalCats1}
\ket{\psi_{\phi,\theta}}=\frac{1}{\sqrt{{\mathcal N}_{\phi,\theta}}}\Bigl(\ket{\alpha,\alpha}+\mathrm{e}^{i\phi}\ket{\alpha\mathrm{e}^{i\theta},\alpha\mathrm{e}^{i\theta}}\Bigr),
\end{equation}
or even
\begin{equation}\label{eq:generalCats2}
\ket{\psi'_{\phi}}=\frac{1}{\sqrt{{\mathcal N'}_{\phi}}}\Bigl(\ket{\alpha,\beta}+\mathrm{e}^{i\phi}\ket{\gamma,\delta}\Bigr),
\end{equation}
for general coherent states $\ket{\alpha}$, $\ket{\beta}$, $\ket{\gamma}$ and $\ket{\delta}$.

Furthermore, as just mentioned, when states of this kind are subject to lossy or noisy channels, forcing them to become mixed, which are then the appropriate witnesses or determinants? Calculations suggest the conjecture that a witness (in this paragraph witness may refer not just to actual witness operators but also to determinants, etc.) which detects entanglement in a state $\ket{\psi}$ is also the right one for entanglement detection in $\Upsilon(\ket{\psi})$, where $\Upsilon$ denotes some channel (see calculations in subsections \ref{subsecA} and \ref{subsecB}). This would not be so surprising, since physical channels do not essentially reshape or transform entanglement. They can only destroy or preserve it. However, if not all entanglement has been detected by the witness in front of the channel, then it may of course not be able to detect anything afterwards, if the channel has destroyed exactly that part of the entanglement which has been initially used by the witness. But in this case there would exist better suited witnesses for the state at hand anyway. Namely these which would make use of \textit{all} entanglement of the initial state. Admittedly, do such witnesses always exist? If yes, then the state after the channel should be entangled \textit{if and only if} the witness, which is in this sense optimal, takes effect. The other way around, the question is the following: For any state $\hat{\rho}$, does there exist a channel $\Upsilon$, such that a witness which detects entanglement in $\hat{\rho}$ in an optimal way, such that it makes use of all its entanglement, is not able to detect entanglement in $\Upsilon(\hat{\rho})$, when there is actually still entanglement present? If not, the witness would become a necessary and sufficient tool for entanglement detection behind the channel. Or when the witness is in this latter sense not optimal, does there \textit{always} exist a channel which preserves some entanglement, but prevents the witness from detecting it behind the channel?

Concluding, there are a lot of interesting open questions regarding two-mode CSS states as well as general entanglement witnessing issues. However, these are not in the focus of this thesis.

\chapter{Optical Hybrid Approaches to Quantum Information}\label{ch:3}
The term \textit{hybrid} is quite in vogue in quantum information. Several authors use it to describe various experiments and protocols. First, there are proposals considering \textit{hybrid} quantum devices which combine elements from atomic and molecular physics as well as from quantum optics and also solid state physics \cite{Wallquist}. Second, there is the notion of \textit{hybrid entanglement}, when talking about entanglement between different degrees of freedom, for example, the entanglement between spatial and polarization modes \cite{Neves,Gabriel}. Finally, schemes which employ both CV and DV resources are also called \textit{hybrid} \cite{Nadja,Louis,XWang,SLloyd}. This is, how \textit{hybrid} is defined in this thesis. 
\begin{defini}
A quantum information protocol is called \textup{hybrid} if it utilizes both resources from discrete variable quantum information and continuous variable quantum information. These resources may include CV and DV states as well as CV and DV quantum operations and measurement techniques.
\end{defini}
For instance, a protocol is considered hybrid if it uses both qubits and qumodes or if it makes use of both DV photon detection and CV homodyne detection. In the preceding chapter it has been pointed out that the description of CV and DV states proceeds differently. Hence, combining CV and DV resources has its own characterizing and challenging subtleties. So, the question arises, why actually aim for hybrid approaches? Note that this chapter is mainly based on a recent review on \textit{optical hybrid approaches to quantum information} by P. van Loock \cite{PvL}.

\section{Why go Hybrid?}
In CV quantum computation, Gaussian states as well as Gaussian transformations, such as beam splitting and squeezing, are used.\footnote[1]{In this thesis no introduction into quantum computation is given. Hence for such introductions see for example the books by Nielsen and Chuang, as well as by Mermin \cite{Nielsen,Mermin}. A good overview over the concepts is also contained in \cite{PvL}.} Computational universality then is the ability to simulate any Hamiltonian which is expressed as an arbitrary polynomial in the mode operators with arbitrary precision. This kind of universality is called \textit{CV universality}. However, to reach such universality in CV quantum computation just linear Gaussian elements are not sufficient. At least one non-Gaussian component is necessary. Actually, any quantum computer utilizing only linear elements could be efficiently simulated by a classical computer \cite{Bartlett}. Unfortunately, this single non-Gaussian element forms the problem with CV quantum computation. It is very difficult to efficiently realize non-Gaussian transformations on Gaussian states. 

In DV quantum computation the encoding of information takes place in a finite-dimensional subspace of the infinite-dimensional Fock space. Then the weaker \textit{DV universality} refers to the ability to simulate any unitary acting on this working space with arbitrary precision. Just as in the CV case, for deterministic processing, a nonlinear interaction is required to realize DV universality. But when truncating the Fock space only single or few-photon states are left. The drawback on DV quantum computation is that nonlinear interactions on the few-photon level are hardly accomplishable. Note that there actually exists an efficient protocol for universal DV computation with only linear optics by Knill, Laflamme and Milburn (the \textit{KLM} scheme \cite{KLM}), which, however, is probabilistic (or near-deterministic at the expense of complicated states entangled between many photons). 

A way out of the problems of CV and DV quantum computation may provide hybrid approaches. The \textit{GKP} scheme by Gottesman, Kitaev and Preskill can be considered as one of the first hybrid protocols for quantum computation \cite{GKP}. It makes use of CV Gaussian states and transformations in combination with DV photon number measurements. So-called non-Gaussian phase states can be created from Gaussian two-mode squeezed states (TMSSs) with the aid of photon counting measurements. Additionally, the GKP scheme can be considered as hybrid, as it employs the concept of encoding logical DV qubits into CV qumodes.

Gottesman, Kitaev and Preskill exploit so-called \textit{measurement-induced nonlinearities} for the realization of the non-Gaussian element. When talking about hybrid schemes, it has to be distinguished between measurement-induced nonlinearities and \textit{weak nonlinearities}. The latter is an actual nonlinearity, which may be enhanced or even mediated through suffenciently intense light fields. A quantum bus (\textit{qubus}) can be used to mediate the nonlinear interaction between two possibly even distant qubits and act for example as an entangling gate \cite{qubus1,qubus2,qubus3,qubus4}. A measurement-induced nonlinearity is no actual nonlinearity. The state to be nonlinearly transformed is entangled with an ancilla state, which is subsequently measured in such a way that the effect on the residual state corresponds to a nonlinear interaction. This concept is exploited in the seminal works by GKP and especially KLM. 

Concluding, the goal of hybrid approaches is to circumvent the limitations of the practical schemes for quantum computation. Nevertheless, their efficiency and feasibility shall be maintained as well as possible. For this purpose
\begin{itemize}
\item hybrid states, operations and measurements,
\item qubus concepts,
\item weak nonlinearities,
\item as well as measurement-induced nonlinearities
\end{itemize} 
are used.

Besides the quantum computational aspect, there is one more point to mention. CV and DV schemes both have their characterizing advantages and disadvantages. The heralding mechanisms, necessary in DV schemes, make them highly probabilistic and hence rather inefficient. However, the fidelities in case of successful operations are very high, often near unity. In contrast, CV schemes are typically deterministic, as no heralding is required. Their drawback lies in the fact that CV encodings are very sensitive to losses and noise, which yields lower fidelities. Hence, there is typically some kind of trade-off between DV and CV schemes. Hybrid approaches may offer possibilities to benefit from the individual advantages while suppressing the disadvantages.

The next section will present a few exemplary applications which involve hybrid protocols.

\section{Applications involving Hybrid Approaches}\label{subsec:HybridApps}
A first example for a protocol making use of hybrid approaches has been already given in the previous section: The GKP scheme for quantum computation, which will be discussed in slightly more detail now and compared to the KLM scheme. It is a cluster-based quantum computation scheme and makes use of linear Gaussian CV resources and operations in combination with DV photon counting.\footnote[2]{Cluster state quantum computation is a kind of "one-way measurement based quantum computation" which makes use of so-called cluster states. It has been invented by Raussendorf and Briegel, see \cite{Raussendorf}. For an introduction and review see \cite{NielsenCluster}.} A TMSS is produced with a $C_Z$ entangling gate. Then a photon number measurement is performed which yields the so-called cubic phase state. With aid of this phase state, the cubic phase gate can be realized, the non-Gaussian element necessary for universal CV quantum computation. It is important to note that all states and gates can be prepared and performed "offline", prior to the actual computation, which then only consists of measurements. Actually, it is this offline cluster state preparation, which is the most problematic part in the protocol. Growing a sufficienty large cluster state with applicable properties is an extremely difficult task and could not be realized on large scales yet.  

As a comparison the KLM scheme is teleportation based. It is a fully DV model, exploiting linear optics and measurement induced nonlinearities by photon number measurements. The gates are teleported on the processed state, which, however, requires highly nonlinear entangled ancilla states. This is the drawback of the KLM scheme. While it is in principle efficient, it is highly impractical, as the demanded ancilla states are too complicated to be engineered with current technologies.

Besides quantum computation there are several more elementary tasks which can be performed with hybrid approaches more efficiently than with only CV or DV schemes. One example is the realization of POVMs for optimal unambiguous state discrimination \cite{Banaszek,Wittmann}.\footnote[3]{For an introduction to unambiguous state discrimination, where also \textit{positive operator-valued measures} (\textit{POVMs}) are explained, see \cite{Chefles,Helstrom}.} A $50/50$-beam splitter, a coherent-state ancilla $\ket{\alpha}$ and DV photon on/off detectors can be used for the optimal unambiguous discrimination between binary coherent states $\{\ket{\pm\alpha}\}$, see figure \ref{fig:USD}.

Another example for a hybrid scheme is the generation of Schr\"odinger-cat states $\ket{\psi_\pm}={\mathcal N}'_\pm(\ket{\alpha}\pm\ket{-\alpha})$ via DV photon subtraction of CV Gaussian squeezed vacuum. The subtraction \begin{wrapfigure}{r}{0.5\textwidth}
\begin{center}
\includegraphics[width=6cm]{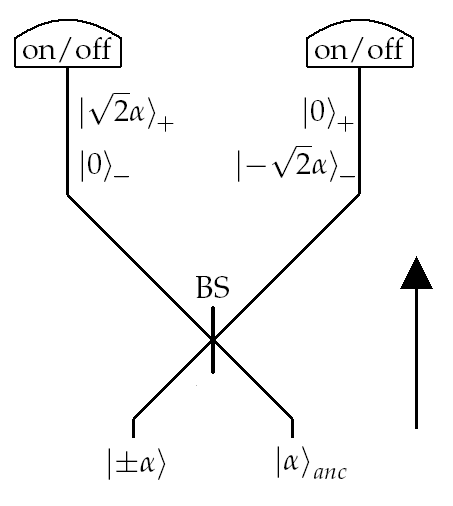}
\caption[Optimal unambiguous state discrimination of equally likely binary coherent states $\{\ket{\pm\alpha}\}$ \cite{PvL}. This figure is a modified version of Fig. 2.2.3 in \cite{PvL}. The figure was used with permission by Peter van Loock.]{Optimal unambiguous state discrimination of equally likely binary coherent states $\{\ket{\pm\alpha}\}$ using a beam splitter, a coherent-state ancilla $\ket{\alpha}$ and DV photon on/off detectors. The inconclusive event here corresponds to the detection of the two-mode vacuum $\mathrm{e}^{-|\alpha|^2}\ket{00}$ with probability $\mathrm{e}^{-2|\alpha|^2}$ \cite{PvL}. This figure is a modified version of Fig. 2.2.3 in \cite{PvL}. The figure was used with permission by Peter van Loock.}
\label{fig:USD}
\end{center}
\end{wrapfigure}
procedure is performed by mixing the squeezed vacuum with one-photon or two-photon Fock states in a beam splitter followed by homodyne detection and postselection depending on the measurement result. One-photon subtraction yields even cat states, while two-photon subtraction results in odd cats \cite{cats1,cats2}. Another way of generating Schr\"odinger-cats is based on hybrid entangled states \cite{cats4}. It is presented in chapter \ref{ch:5}. CSS states find a very important application in fault-tolerant, universal quantum computation \cite{cats3}.

More examples for hybrid protocols can be found in the areas of quantum state characterization and quantum communication. In their theory paper \cite{statechar1} Fiurasek and Cerf show "how to measure squeezing and entanglement of Gaussian states without homodyning". Instead of homodyne detection they apply DV photon counting to characterize the CV Gaussian states. Another theoretical proposal suggests, to measure the Bell nonlocality of CV Einstein-Podolsky-Rosen states (which actually are  TMSSs) via measurement of the DV photon number parity operator \cite{statechar2}. However, even the simple CV homodyne tomography of DV one- (\cite{statechar3}) or two-photon (\cite{statechar4}) Fock states can be already considered as a hybrid experiment. In quantum communication there are entanglement concentration and entanglement distillation subroutines which can be performed more efficiently with hybrid protocols. \cite{quantumcomm1} demonstrates, how to theoretically concentrate entanglement of pure CV TMSSs through DV photon subtraction. \cite{quantumcomm2,quantumcomm3} present the experimental execution. Entanglement distillation of noisy CV TMSSs using beam splitters and on/off DV photon detectors is theoretically set out in \cite{quantumcomm4} and experimentally performed in \cite{quantumcomm5,quantumcomm6}.

A highly important tool being utilized in several models for quantum computation and quantum communication is the hybrid quantum bus concept \cite{qubus1,qubus2,qubus3,qubus4}. The idea is, to use an intense qumode to mediate a nonlinear interaction (even when it is weak) between two qubits. It will be presented in chapter \ref{ch:5}.

One relevant ingredient of several of the mentioned protocols is a new intriguing form of entanglement, the \textit{hybrid entanglement} between CV and DV systems. What distinguishes this kind of entanglement from ordinary DV or CV entanglement? Can, just as in the CV case, where one distinguishes between Gaussian entangled states and non-Gaussian entangled states, different kinds of hybrid entanglement be identified? How is quantification and witnessing of hybrid entanglement achievable? Is there also multipartite hybrid entanglement? And how can hybrid entangled states actually be generated? The rest of this thesis addresses these questions and provides answers to most of them. Also some of the already mentioned applications of hybrid entanglement are presented.

\chapter{Hybrid Entanglement}\label{ch:4}
In this chapter \textit{hybrid entangled} (\textit{HE})\footnote[1]{Depending on the context, HE may denote either "hybrid entangled" or "hybrid entanglement".} states, i.e. entangled states whose constituents are both DV and CV systems are discussed and classified, while the emphasis lies on bipartite hybrid entanglement. The chapter starts with an introductory section \ref{sec:HEintro} focussing on some general issues regarding HE. Afterwards, in section \ref{sec:pure} pure bipartite HE states are investigated. Then we move on to mixed states. Section \ref{sec:mixed1} discusses mixed states which are effectively fully DV while in section \ref{sec:mixed2} so-called \textit{truly hybrid entangled} states are introduced and examined. While the first 4 sections only focus on bipartite HE, finally, in section \ref{sec:multihe} also multipartite HE is explored.

\section{Introduction}\label{sec:HEintro}
An example for a hybrid entangled state is a state like 
\begin{equation}
\ket{\psi}^{AB}=\frac{1}{\sqrt{2}}\Bigl(\ket{0}^A\ket{\psi_0}^B+\ket{1}^A\ket{\psi_1}^B\Bigr),
\end{equation}
where $\ket{0}^A$ and $\ket{1}^A$ are DV qubit states in ${\mathcal H}^A_2$ and $\ket{\psi_1}^B$ as well as $\ket{\psi_2}^B$ are CV qumode states living in the infinite-dimensional Hilbert space ${\mathcal H}^B_\infty$. Already this state shows a combination of DV and CV properties. On the one hand, it can be seen that it has the same normalization constant $\frac{1}{\sqrt{2}}$ as the comparable, fully DV qubit state $\frac{1}{\sqrt{2}}(\ket{00}+\ket{11})$. On the other hand, since $\ket{\psi_1}^B$ and $\ket{\psi_2}^B$ are general CV qumode states and hence not orthogonal, there does not directly exist a Schmidt decomposition of the state. This is also the case for pure fully CV states.

However, one can think of many HE states like $\ket{\psi}^{AB}$, but with a possibly higher dimensional DV system or even mixed states. So, a rigorous definition of hybrid entanglement is necessary.
\begin{defini}
Any entangled bipartite state of the form
\begin{equation}\label{eq:HEdef}
\begin{aligned}
\hat{\rho}^{AB} &=\sum_{n=1}^N p_n\,\ket{\psi_n}_{AB}\bra{\psi_n}\,,\qquad p_n>0\;\forall \;n\,;\;\sum_{n=1}^N \;p_n=1, \\
\ket{\psi_n} &=\sum_{m=0}^{d-1} c_{nm}\ket{m}^A\ket{\psi_{nm}}^B\,,\qquad c_{nm}\in\mathbb C\;;\;\sum_{m=0}^{d-1} \;|c_{nm}|^2=1,
\end{aligned} 
\end{equation}
with generally non-orthogonal qumode states $\ket{\psi_{nm}}^B$, which is living in a Hilbert space of the form ${\mathcal H}^{AB}={\mathcal H}_{d}^A\otimes{\mathcal H}_{\infty}^B$ with finite $d$, is called \textup{hybrid entangled}.
\end{defini}
How can HE states actually be described in a useable way? Can their entanglement be quantified? Facing an overall infinite-dimensional Hilbert space ${\mathcal H}_{d}^A\otimes{\mathcal H}_{\infty}^B$, no density matrices can be employed. So, CV methods have to be attempted. But in CV entanglement theory the only conveniently representable states are the Gaussian ones. In the general non-Gaussian CV regime, for instance, exact entanglement quantification is not achievable. Therefore, do HE Gaussian states exist?
\begin{lem}\label{lem:DVnonG}
Any singlepartite $d$-dimensional quantum state with finite $d$ and $d\geq2$ is non-Gaussian.
\end{lem}
\begin{proof}[\textbf{Proof.}]
Consider such a general $d$-dimensional state in a pure state decomposition and make use of the Fock basis of a $d$-dimensional subspace of the Fock space,
\begin{equation}\label{eq:Lemmaqudit}
\begin{aligned}
\hat{\rho} &=\sum_{n=1}^N p_n\ket{\psi_n}\bra{\psi_n}\,,\qquad p_n>0\;\forall \;n\,;\;\sum_n^N \;p_n=1, \\
\ket{\psi_n} &=\sum_{m=0}^{d-1} c_{nm}\ket{m}\,,\qquad c_{nm}\in\mathbb C\;;\;\sum_m^{d-1} \;|c_{nm}|^2=1.
\end{aligned} 
\end{equation}
Non-Gaussianity corresponds to non-Gaussianity of the characteristic function. This is equivalent to non-Gaussianity of the Wigner function, as (non-)Gaussians stay (non-)Gaussians under Fourier transformations \cite{Bronstein}. Hence, consider the Wigner function of $\hat{\rho}$,
\begin{equation}
\begin{aligned}
W(x,p) & :=\frac{1}{\pi}\int\limits_{-\infty}^\infty dy\,\mathrm{e}^{2ipy}\braket{x-y|\hat{\rho}|x+y} \\
& =\sum_{n=1}^N\frac{p_n}{\pi}\int\limits_{-\infty}^\infty dy\,\mathrm{e}^{2ipy} \sum_{kl=0}^{d-1}c_{nl}c_{nk}^\ast\,\psi_l(x-y)\psi_k^\ast(x+y).
\end{aligned} 
\end{equation}
The position wavefunction of a Fock state $\ket{n}$ is given by
\begin{equation}
\psi_n(x)=\braket{x|n}=\frac{H_n(x)}{\sqrt{2^nn!\sqrt{\pi}}}\mathrm{e}^{-\frac{x^2}{2}},
\end{equation}
where $H_n(x)$ are the \textit{Hermite polynomials} \cite{Leon,Bronstein}. For an overview over some relevant properties of the Hermite polynomials, see appendix \ref{Hermite}. Therefore,
\begin{equation}
\begin{aligned}
W(x,p) & =\sum_{n=1}^N\frac{p_n}{\pi}\int\limits_{-\infty}^\infty dy\,\mathrm{e}^{2ipy}\,\mathrm{e}^{-\frac{(x-y)^2}{2}}\mathrm{e}^{-\frac{(x+y)^2}{2}} \sum_{kl=0}^{d-1}c_{nl}c_{nk}^\ast\frac{H_k(x+y)H_l(x-y)}{\sqrt{2^{k+l}k!l!\pi}} \\
& =\frac{1}{\pi}\int\limits_{-\infty}^\infty dy\,\mathrm{e}^{2ipy}\;\mathrm{e}^{-x^2-y^2}\, \sum_{n=1}^N\,p_n \sum_{kl=0}^{d-1}c_{nl}c_{nk}^\ast \frac{H_k(x+y)H_l(x-y)}{\sqrt{2^{k+l}k!l!\pi}}.
\end{aligned} 
\end{equation}
Now $\frac{1}{\pi}\int\limits_{-\infty}^\infty dy\,\mathrm{e}^{2ipy}$ is just a Fourier transform operator with respect to $y$, and $\mathrm{e}^{-x^2-y^2}$ is a Gaussian. Hence, for $W(x,p)$ being a Gaussian, the sum $P:=\sum_{n=1}^N\,p_n \sum_{kl=0}^{d-1}c_{nl}c_{nk}^\ast \frac{H_k(x+y)H_l(x-y)}{\sqrt{2^{k+l}k!l!\pi}}$ also has to be Gaussian in $x$ as well as in $y$. However, $P$ is just a polynomial. Strictly speaking, it is a polynomial of \textit{finite} order. Hence, $P$ can never be an exponential and even more no Gaussian. Admittedly, there is one exception. $P$ may be a polynomial of order zero, i.e. a constant. 

So, assume that $\hat{\rho}$ is Gaussian and show that this is the case if and only if $c_{nl}=0\;\forall\;l\geq1$. Showing "$\Leftarrow$" is trivial, since for $c_{nl}=0\;\forall\;l\geq1$, $\hat{\rho}=\ket{0}\bra{0}$, which is obviously Gaussian. For "$\Rightarrow$" perform an induction with respect to the dimension $d$. 
\begin{itemize}
\item \textbf{Inductive Basis}: For $d=2$ there is only one term $x^{2(d-1)}=x^{2}$ in the sum $P$, which comes from $H_{d-1}(x+y)H_{d-1}(x-y)=H_{1}(x+y)H_{1}(x-y)$. Its coefficient is $\tilde{c}^2_{1}\sum_{n=1}^N p_n|c_{n,1}|$, where $\tilde{c}_{1}$ is a real constant due to the Hermite polynomials. If the state is Gaussian, $\sum_{n=1}^N p_n|c_{n,1}|=0$. Since $p_n>0\;\forall\;n$ and $|c_{n,1}|\geq0\;\forall\;n$, from $\sum_{n=1}^N p_n|c_{n,1}|=0$ follows $c_{n,1}=0\;\forall\;n$. Hence, $c_{nl}=0\;\forall\;l\geq1$ ($l$ can be only $0$ or $1$ for $d=2$) is proved as a necessary condition for Gaussianity for $d=2$.
\item \textbf{Inductive Step}: Assume validity of the statement for $d=m$ and consider $d=m+1$. Then, there is only one term $x^{2((m+1)-1)}=x^{2m}$ in the sum $P$, which comes from $H_{(m+1)-1}(x+y)H_{(m+1)-1}(x-y)=H_{m}(x+y)H_{m}(x-y)$. Again with $\tilde{c}_{m}$ being a real constant the coefficient is $\tilde{c}^2_{m}\sum_{n=1}^N p_n|c_{n,m}|$. If the state is Gaussian, $\sum_{n=1}^N p_n|c_{n,m}|=0$. Since $p_n>0\;\forall\;n$ and $|c_{n,m}|\geq0\;\forall\;n$, from $\sum_{n=1}^N p_n|c_{n,m}|=0$ follows $c_{n,m}=0\;\forall\;n$. However, from the inductive hypothesis it is known that $c_{nl}=0\;\forall\;1\leq l\leq d-2=m-1$. Hence, $c_{nl}=0\;\forall\;1\leq l\leq d-1=m$, which proves the validity of the statement for $d=m+1$ and hence for all finite $m$. 
\end{itemize}
Concluding, a $d$-dimensional quantum state with finite $d$ and $d\geq2$, written as \eqref{eq:Lemmaqudit}, is Gaussian if and only if $c_{nl}=0\;\forall\;l\geq1$, which corresponds to $\hat{\rho}=\ket{0}\bra{0}$. Therefore, it can be said that $\hat{\rho}$ is Gaussian if and only if $\hat{\rho}=\ket{0}\bra{0}$, which is only one-dimensional. Hence, any $d$-dimensional quantum state with finite $d$ and $d\geq2$ is non-Gaussian.
\end{proof}

\begin{thm}\label{thm:hybridnonG}
Any bipartite hybrid entangled or classically correlated state is non-Gaussian.
\end{thm}
\begin{proof}[\textbf{Proof.}]
If a multipartite quantum state is Gaussian, all its subsystems have to be Gaussian. So, consider a state of the form \eqref{eq:HEdef} and trace out the CV subsystem. What is left over is a $d$-dimensional singlepartite system with finite $d$, which can be described in its Fock basis. It is denoted by $\hat{\rho}$. Due to lemma \ref{lem:DVnonG}, $\hat{\rho}$ is Gaussian if and only if $\hat{\rho}=\ket{0}\bra{0}$, which is a pure state. However, for any entangled or even classically correlated state, the reduced state cannot be pure \cite{Audretsch}. Hence, for Gaussian $\hat{\rho}$, the overall system cannot have been entangled or classically correlated. Therefore, every bipartite hybrid entangled or classically correlated state is non-Gaussian.
\end{proof}

The proof of theorem \ref{thm:hybridnonG} basically relies on lemma \ref{lem:DVnonG}, which states that any DV state with dimension $\geq2$ is non-Gaussian. It is straightforwardly generalized to multipartite systems. Arguing that for Gaussianity \textit{all} subsystems have to be Gaussian, it is sufficient for the non-Gaussianity of any multipartite system which also possesses DV constituent(s) that at least one DV subsystem is of dimension $\geq2$. However, if there shall be entanglement between the DV subsystem and the rest, the state is necessarily non-Gaussian, since for entanglement dimension $\geq2$ is required. So, any multipartite quantum state which involves DV entangled subsystems is non-Gaussian. Even more generally, only 1- or infinite-dimensional systems can be Gaussian. This is of course not surprising, however, it is worth pointing out that theorem \ref{thm:hybridnonG} is not trivial. Since the \textit{overall} Hilbert space of HE systems is very well infinite-dimensional, it is not a priori clear that there do not exist Gaussian hybrid entangled states.

As shown, HE states live in the non-Gaussian, infinite-dimensional Hilbert space regime, which is quite problematic as already known from conventional CV entanglement theory. The states can be neither described with proper density matrices, nor with covariance matrices. Phase space representations are also not very suitable, since one of the subsystems is DV. They are not qualified for investigations regarding entanglement properties anyway. The only known quasi-probability distribution which may in some cases make statements about the entanglement of the state is the Glauber-Sudarshan P-representation. However, it can be easily shown that this function does not even exist for these highly nonclassical states. Of course, it may still be possible to perform entanglement witnessing, but what about quantification? Is there a way out of the dilemma? Yes, there is! For some HE states the unique Hilbert space structure can be exploited in such a way that the states can nevertheless be described by density matrices. These states are therefore called DV-like hybrid entangled. This gives rise to a classification scheme of HE states.

\subsection{Classification}
Consider again a general HE state of the form \eqref{eq:HEdef}. Depending on the number of mix terms $N$ and the dimension of the DV subsystem $d$, there can be maximally $N\times d$ linearly independent CV qumode states $\ket{\psi_{nm}}^B$ in $\hat{\rho}^{AB}$. $d$ is always finite due to the definition of HE. However, $N$ may be either finite or infinite and, hence, the number of linearly independent qumode states is either finite or infinite. Furthermore, if $N=1$, the state is pure.

If the number $N\times d$ of linearly independent CV qumode states is finite, they only span a $N\times d$-dimensional subspace ${\mathcal H}_{N\times d}$ of the initially infinite-dimensional Hilbert space ${\mathcal H}_\infty$. Then, the \textit{Gram-Schmidt process} can be employed to express the qumode states in an orthonormal basis of this finite-dimensional subspace. In this case, the state is effectively DV and all the methods from DV entanglement theory can be applied. 

However, this shall be discussed in more detail. What actually is the Gram-Schmidt process? It is a method for orthonormalizing a finite, linearly independent set of vectors in an inner product space \cite{Schwetlick,Bau3}.\footnote[2]{Note that the Gram-Schmidt process and the Schmidt decomposition are both named after the German mathematician Erhard Schmidt. He developed the Gram-Schmidt process together with Jorgen Pedersen Gram.} For a linearly independent set of vectors $\{\ket{\psi_i}:i=1,\ldots,n\}$ (since the process is to be exploited in the framework of Hilbert spaces, the inner product space is a priorily assumed to be a Hilbert space) a set of pairwise orthonormal vectors $\{\ket{e_i}:i=1,\ldots,n\}$ spanning the same subspace as $\{\ket{\psi_i}:i=1,\ldots,n\}$ is given by
\begin{equation}\label{GSprocess}
\begin{aligned}
\ket{e_1'} &= \ket{\psi_1}\,, & \ket{e_1} &= \frac{\ket{e_1'}}{\sqrt{\braket{e_1'|e_1'}}}\,, \\
\ket{e_2'} &= \ket{\psi_2}-\braket{e_1|\psi_2}\ket{e_1}\,, & \ket{e_2} &= \frac{\ket{e_2'}}{\sqrt{\braket{e_2'|e_2'}}}\,, \\
& \vdots & & \vdots \\
\ket{e_n'} &=\ket{\psi_n}-\sum_{i=1}^{n-1}\braket{e_i|\psi_n}\ket{e_i}\,,\qquad & \ket{e_n} &=\frac{\ket{e_n'}}{\sqrt{\braket{e_n'|e_n'}}}\,.
\end{aligned}
\end{equation}
Making use of this, any finite set of linearly independent qumode states can be expressed in an orthonormal basis $\{\ket{e_i}:i=1,\ldots,n\}$. But strictly speaking, equation \eqref{GSprocess} only tells, how to express the new orthonormal basis in terms of the old non-orthonormal one. What is actually required, is the other way around. To express the qumode states in terms of $\{\ket{e_i}:i=1,\ldots,n\}$ try the following approach, which can be considered as some kind of \textit{inverse Gram-Schmidt process}.
\begin{thm}\label{thm:GSI}
$n$ normalized, in general non-orthogonal, linearly independent states $\{\ket{\psi_i}:i=1,\ldots,n \,;\,0\leq|\braket{\psi_i|\psi_j}|\leq1\,\forall\,i,j\}$ can always be expressed as $\ket{\psi_i}=\sum_{j=1}^ia_{ij}\ket{e_j}$, where $\{\ket{e_i}:i=1,\ldots,n\}$ form an orthonormal basis of the space spanned by $\{\ket{\psi_i}\}$ and $a_{ij}\in\mathbb C$.
\end{thm}
\begin{proof}[\textbf{Proof.}]
Writing out $\ket{\psi_i}=\sum_{j=1}^ia_{ij}\ket{e_j}$, theorem \ref{thm:GSI} states that the $\{\ket{\psi_i}\}$ can be always written as
\begin{equation}\label{GSprocessInverse}
\begin{aligned}
\ket{\psi_1} &= a_{11}\ket{e_1}\,, \\
\ket{\psi_2} &= a_{21}\ket{e_1}+a_{22}\ket{e_2}\,, \\
\ket{\psi_3} &= a_{31}\ket{e_1}+a_{32}\ket{e_2}+a_{33}\ket{e_3}\,, \\
& \vdots  \\
\ket{\psi_n} &= a_{n1}\ket{e_1}+\ldots+a_{nn}\ket{e_n}=\sum_{i=1}^{n}a_{ni}\ket{e_i}\,.
\end{aligned}
\end{equation}
To prove it, it has to be shown that 1.) $\ket{\psi_i}=\sum_{j=1}^ia_{ij}\ket{e_j}$ corresponds to a valid basis transformation and 2.) it actually performs the right mapping.

1.) Write the transformation as
\begin{equation}
\ket{\psi_i}=\sum_j U_{ij}\ket{e_j},
\end{equation}
with the transformation matrix 
\begin{equation}
U=\begin{pmatrix} a_{11} & 0 & \cdots & 0 \\ 
a_{21} & a_{22} & & \vdots \\ 
\vdots & & \ddots & \vdots \\ 
a_{n1} & a_{n2} & \cdots & a_{nn}
\end{pmatrix}.
\end{equation}
Due to normalization of the initial and the new vectors $\sum_{j=1}^i |a_{ij}|^2=1$ is known and due to the linear independence of the $\{\ket{\psi_i}\}$ also $a_{ii}\neq0\,\forall\,i$ is necessary.
\begin{itemize}
\item[$\Rightarrow$] $\det[U] =\prod_{i=1}^n a_{ii}\neq0$. 
\item[$\Rightarrow$] $U$ is invertible.
\item[$\Rightarrow$] $U$ is a valid basis transformation.
\end{itemize}

2.) To show that the lower triangular structure of the basis transformation $U$ in combination with the orthonormal basis $\{\ket{e_i}\}$ is sufficient to actually express the $\{\ket{\psi_i}\}$ accurately in terms of $\{\ket{e_i}\}$, it is demonstrated that the $\frac{N(N+1)}{2}$ parameters $a_{ij}$ can be chosen such that all overlaps $\braket{\psi_i|\psi_j}$ are preserved when the transformation is applied.

On the one hand, there are $n^2$ such overlaps in total and $\frac{n(n+1)}{2}$ ones with potentially differing absolute values. On the other hand, there are $\frac{n(n+1)}{2}$ complex parameters $a_{ij}$. From the structure of the basis transformation and the fact that $\{a_{ij}\}$ are complex, it is clear that if $a_{ij}$ can be chosen such that the $\frac{n(n+1)}{2}$-element set of $\{\braket{\psi_i|\psi_j}:i\leq j\}$ can be preserved, also the rest of the overlaps is preserved, since they are only complex conjugates of the former. Hence, it already becomes reasonable that the $a_{ij}$ can be chosen appropriately.

However, a proper proof is performed by induction in $n$:
\begin{itemize}
\item \textbf{Inductive Basis}:

$n=1$: There is only one overlap to be preserved. 
\begin{equation}
\braket{\psi_1|\psi_1}=1\stackrel{!}{=}\braket{e_1|a_{11}^\ast a_{11}|e_1}=|a_{11}|^2.
\end{equation}
Hence, choose $a_{11}=1$, which preserves the overlap $\braket{\psi_1|\psi_1}$.
\item \textbf{Inductive Step}: 

Assume $a_{ij}$ have been calculated for $i\leq n-1$ such that all overlaps $\{\braket{\psi_i|\psi_j}:i,j\leq n-1\}$ are preserved. We show that then also all $a_{nj}$ can be chosen such that the overlaps $\{\braket{\psi_i|\psi_n}:i=1,\ldots,n\}$ are preserved. The complex conjugated overlaps follow automatically as argued before.

Applying the basis transformation, the overlaps $\{\braket{\psi_i|\psi_n}:i=1,\ldots,n-1\}$ are
\begin{equation}
\begin{aligned}
\braket{\psi_1|\psi_n} &= a_{11}^\ast a_{n1}\,, \\
\braket{\psi_2|\psi_n} &= a_{21}^\ast a_{n1}+a_{22}^\ast a_{n2}\,, \\
& \vdots  \\
\braket{\psi_{n-1}|\psi_n} &= a_{n-1,1}^\ast a_{n1} + \ldots + a_{n-1,n-1}^\ast a_{n,n-1}  \,.
\end{aligned}
\end{equation}
As $a_{ij}$ have been calculated for $i\leq n-1$ due to the inductive hypothesis, this is just a system of linear equations, which can be written as an augmented matrix:
\begin{equation}\label{GSIaug}
\begin{pmatrix}[c|cccc]
\braket{\psi_1|\psi_n}  & 1 & 0 & \hdots & 0 \\
\braket{\psi_2|\psi_n}  & a_{21}^\ast & a_{22}^\ast & & \vdots \\
 & \vdots & & \ddots & \vdots  \\
\braket{\psi_{n-1}|\psi_n}  & a_{n-1,1}^\ast & a_{n-1,2}^\ast & \hdots & a_{n-1,n-1}^\ast \end{pmatrix}.
\end{equation}
Since $a_{ii}\neq0\,\forall\,i$, this system of equations is exactly solvable. Hence, $\{a_{nj}:j=1,\ldots,n-1\}$ can be chosen such that $\{\braket{\psi_i|\psi_n}:i=1,\ldots,n-1\}$ are preserved. Therefore, there are only one free parameter $a_{nn}$ and one overlap $\braket{\psi_n|\psi_n}$ to be preserved left:
\begin{equation}
\braket{\psi_n|\psi_n}=1\stackrel{!}{=}\sum_{j=1}^N |a_{Nj}|^2=\sum_{j=1}^{N-1} |a_{Nj}|^2 + |a_{NN}|^2.
\end{equation}
From $\sum_{j=1}^n |a_{nj}|^2=1$ and $a_{nn}\neq 0$, which is already known, $\sum_{j=1}^{N-1} |a_{Nj}|^2<1$ follows, and hence $a_{nn}$ can be chosen as  
\begin{equation}
a_{NN}=\sqrt{1-\sum_{j=1}^{N-1} |a_{Nj}|^2}.
\end{equation}
In the end, also $\braket{\psi_n|\psi_n}$ can be preserved. Therefore, the theorem is valid for $n$ under the assumption of validity for $n-1$.  
\end{itemize}
Concluding, with the inductive basis it is valid for all $n$.
\end{proof}

The proof has been performed in such a great detail, because it also sets out, how to actually compute $\{a_{ij}\}$ for a given set of qumode states. Recalling equations \eqref{GSprocessInverse} and \eqref{GSIaug}, the $\{a_{ij}\}$ can be calculated successively one after another by considering successive overlaps. A parameter $a_{i1}$ is directly obtained from the overlap $\braket{\psi_1|\psi_i}$. Then, $a_{i2}$ follows from $\braket{\psi_2|\psi_i}$ together with the known $a_{i1}$. Likewise, $a_{i3}$ is calculated from $\braket{\psi_3|\psi_i}$, $a_{i1}$ and $a_{i2}$. For the other parameters just go on like this. Hence, the inverse Gram-Schmidt process can be efficiently implemented and computed. 

As an example consider the normalized qumode states $\{\ket{\psi_i}:i=1,\ldots,3\}$ with overlaps
\begin{equation}
\begin{aligned}
\braket{\psi_1|\psi_2} &= c_1, \\
\braket{\psi_1|\psi_3} &= c_2, \\
\braket{\psi_2|\psi_3} &= c_3.
\end{aligned}
\end{equation} 
They can be expressed in the orthonormal basis $\{\ket{e_i}:i=1,\ldots,3\}$ as
\begin{equation}
\begin{aligned}
\ket{\psi_1} &= \ket{e_1}, \\
\ket{\psi_2} &= c_1\ket{e_1}+\sqrt{1-|c_1|^2}\ket{e_2}, \\
\ket{\psi_3} &= c_2\ket{e_1}+\frac{c_3-c_1^\ast c_2}{\sqrt{1-|c_1|^2}}\ket{e_2}+\sqrt{1-|c_2|^2-\frac{|c_3-c_1^\ast c_2|^2}{1-|c_1|^2}}\ket{e_3}.
\end{aligned}
\end{equation}
The Gram-Schmidt process and the inverse Gram-Schmidt process respectively are possibly the most important tools with respect to HE quantum states.

So, once again, consider a state of the form \eqref{eq:HEdef}. For $N=1$ the state is pure. Then it contains only $d$ linearly independent qumode states, which span a d-dimensional subspace ${\mathcal H}_d^B$. Therefore, with aid of the inverse Gram-Schmidt process these qumode states can be expressed in an orthonormal basis. Then density matrices can be employed for the description of the overall state, and also a Schmidt decomposition can be performed or pure state measures such as the entropy of entanglement can be calculated. The state is effectively DV. For $1<N<\infty$, the state is mixed. Nevertheless, it possesses only $N\times d$ qumode states, which can be again casted in an orthonormal basis. Therefore, also states of this kind are effectively DV and the density matrix formalism can be exploited. However, they are not pure anymore and neither pure state measures nor a Schmidt decomposition can be applied. Finally, there is $N=\infty$. Here, by $N=\infty$ states those states are meant which can be expressed with infinite $N$ only and not also with finite $N$. Remember that pure state decompositions are not unique. Anyway, these states hold an infinite number of qumode states, which has the effect that the Gram-Schmidt process cannot be applied anymore. Hence, they are not describable with density matrices and therefore not effectively DV. Nevertheless, one subsystem is DV. These are the states which are called \textit{truly hybrid entangled}.

Summing up, for bipartite HE states in pure state decomposition with $N$ denoting the number of mix terms in the convex combination of pure state projectors, there is the following classification scheme:
\begin{itemize}
\item $\boldsymbol{N=1:}$ Pure hybrid entangled states
\begin{itemize}
\item Supported by finite-dimensional subspace. 
\item $\,\Rightarrow\,$ \textbf{DV-like entanglement}.
\item Schmidt decomposition compatible.
\item DV Pure state measures (Entropy of Entanglement) applicable.
\end{itemize}
\item $\boldsymbol{1<N<\infty:}$ Mixed hybrid entangled states
\begin{itemize}
\item Supported by finite-dimensional subspace. 
\item $\,\Rightarrow\,$ \textbf{DV-like entanglement}.
\item DV mixed state measures applicable.
\end{itemize}
\item $\boldsymbol{N=\infty:}$ Mixed hybrid entangled states
\begin{itemize}
\item No support by finite-dimensional subspace.
\item $\,\Rightarrow\,$ \textbf{True hybrid entanglement}.
\item No measures applicable.
\item Witnesses adaptable from CV entanglement theory.
\end{itemize}
\end{itemize}
At first, it has to be pointed out that the possibility of applying DV methods on such a wide class of hybrid entangled states is quite remarkable. An initially non-Gaussian infinite-dimensional quantum state, which seems rather awkward at first sight, can finally be described with neat density matrices and in the pure state case even a Schmidt decomposition can be performed. Nevertheless, there is also the class of states which stay \textit{truly hybrid entangled} and cannot be transformed with the Gram-Schmidt process. As mentioned earlier, in this infinite-dimensional, non-Gaussian regime exact entanglement quantification is up to now impossible. Therefore, witnesses which detect hybrid entanglement are here especially demanded. In the next subsection an adapted version of the SV criteria is presented, which can be used for this purpose. The other sections of the chapter will investigate relevant states from each class of HE states.

\subsection{Adapting Shchukin-Vogel Inseparability Criteria}\label{subsec:AdSVcrit}
In subsection \ref{subsec:witnessing} the SV inseparability criteria have been introduced. They are well defined for bipartite CV states. However, HE states are not fully CV, but contain a DV subsystem. Nevertheless, the SV criteria can be applied when appropriately adapted to the new situation. Actually, there are two ways for the adaption.

The first way is to simply interpret the DV subsystem as living in a subspace of an infinite-dimensional Hilbert space. The DV qudit is interpreted as a CV system supported by ${\mathcal H}_\infty$ and being encoded in a Fock basis. However, it only makes use of a finite number of the basis vectors. Then the SV criteria can be applied just as usual.

The second way is not to adapt the state, but the criteria. The mode operators for CV states, which SV make use of, are defined by equations \eqref{eq:ladders3} and \eqref{eq:ladders4}. Now, however, assume the system which the operators $\hat{a}$ and $\hat{a}^\dagger$ belong to to be $d$-dimensional. Then the orthonormal Hilbert space basis vectors $\ket{m}$ can be written as column vectors with a "$1$" in row $m$ with $m=0,\ldots,d-1$:
\begin{equation}
\ket{m}_d=(0_{\,0}, \ldots , 0_{m-1} ,1_m ,0_{m+1},\ldots ,0_{d-1})^T.
\end{equation}
With this notation the new \textit{qudit mode operators} $\hat{a}_d$ can be defined as proper $d\times d$ matrices:
\begin{equation}
\hat{a}_d = \begin{pmatrix} 
0 & \sqrt{1} & \cdots & \cdots & 0 \\
\vdots & 0 & \sqrt{2} &   & \vdots \\
\vdots &   & 0 & \ddots  & \vdots \\
\vdots &  &  & \ddots & \sqrt{d-1} \\
0 & \cdots & \cdots & \cdots & 0
\end{pmatrix}\,, \qquad \hat{a}_d^\dagger = \begin{pmatrix} 
0 & \cdots & \cdots & \cdots & 0 \\
\sqrt{1} & 0 &  &   & \vdots \\
\vdots &  \sqrt{2} & 0 &   & \vdots \\
\vdots &  & \ddots & \ddots & \vdots \\
0 & \cdots & \cdots & \sqrt{d-1} & 0
\end{pmatrix}.
\end{equation}
They have the properties
\begin{align}
\hat{a}_d\ket{n} & = \sqrt{n}\ket{n-1}, \\
\hat{a}_d^\dagger\ket{n} & = (1-\delta_{dm})\sqrt{n+1}\ket{n+1}, \\ \label{eq:AdMod1}
(\hat{a}_d)^d & =0, \\ \label{eq:AdMod2}
(\hat{a}_d^\dagger)^d & =0.
\end{align}
Furthermore, the commutator changes:
\begin{equation}
[\hat{a}_d,\hat{a}_d^\dagger]=\begin{pmatrix} 
1 &  &  &  \\
 & \ddots &  &  \\
 &  & 1 &  \\
 &  &  & -(d-1)
\end{pmatrix}_{d\times d}=\begin{pmatrix} 
{\mathbb 1}_{d-1} &    \\
 &  -(d-1)
\end{pmatrix}_{d\times d}.
\end{equation}
Since any d-dimensional qudit $\ket{\psi}_d$ can still be written as $\ket{\psi}_d=g^\dagger(\hat{a}_d)\ket{0}_d$ and also the projection operator on $\ket{0}_d\bra{0}$ nevertheless is $:e^{-\hat{a}_d^\dagger\hat{a}_d}:$, the SV criteria can be also derived for hybrid systems with aid of these new operators. Structurally the same criteria are obtained, just that for certain $i_1,i_2,j_1,j_2$ the moments $M_{ij}(\hat{\rho})=\braket{\hat{a}^{\dagger^{i_2}}\hat{a}^{^{i_1}}\hat{a}^{\dagger^{j_1}}\hat{a}^{^{j_2}}\hat{b}^{\dagger^{j_4}}\hat{b}^{^{j_3}}\hat{b}^{\dagger^{i_3}}\hat{b}^{^{i_4}}}_{\hat{\rho}}$ are zero. Due to the properties \eqref{eq:AdMod1} and \eqref{eq:AdMod2} any combination of $i_1,i_2,j_1,j_2$ such that $(\hat{a}_d)^k$ or $(\hat{a}_d^\dagger)^k$ with $k\geq d$ occurs in the moments, nullifies them. In the end, in the matrix of moments the rows and columns corresponding to these combinations of $i_1,i_2,j_1,j_2$ are simply missing.

When actually applying the SV criteria it rarely makes a difference whether the first or the second approach to the adaption of the criteria is chosen. Only for moments involving terms like $(\hat{a}_d)^k(\hat{a}_d^\dagger)^l$ with $k,l\in\mathbb N$ the approaches may yield diverging results. However, for all determinants and moments used throughout this thesis, they lead to the same results.

As an example for the application of the SV criteria on HE, entanglement detection in the HE state
\begin{equation}\label{eq:HEstate1}
\ket{\psi}=\frac{1}{\sqrt{2}}\Bigl(\ket{0}\ket{\alpha}+\ket{1}\ket{-\alpha}\Bigr)
\end{equation}
is shortly demonstrated. In subsection \ref{catwitnessing} the determinants 
\begin{equation}
s_1=\begin{vmatrix} 
1 & \braket{\hat{b}^\dagger} & \braket{\hat{a}\hat{b}^\dagger}  \\
\braket{\hat{b}} & \braket{\hat{b}^\dagger\hat{b}} & \braket{\hat{a}\hat{b}^\dagger\hat{b}}  \\
\braket{\hat{a}^\dagger\hat{b}} & \braket{\hat{a}^\dagger\hat{b}^\dagger\hat{b}} & \braket{\hat{a}^\dagger\hat{a}\hat{b}^\dagger\hat{b}}
\end{vmatrix}\,,\quad \text{and} \quad
s_2=\begin{vmatrix} 
1 & \braket{\hat{a}} & \braket{\hat{b}^\dagger}  \\
\braket{\hat{a}^\dagger} & \braket{\hat{a}^\dagger\hat{a}} & \braket{\hat{a}^\dagger\hat{b}^\dagger}  \\
\braket{\hat{b}} & \braket{\hat{a}\hat{b}} & \braket{\hat{b}^\dagger\hat{b}}
\end{vmatrix},
\end{equation}
were employed to detect entanglement in two-mode CSSs, which are structurally comparable. Hence, it is reasonable to try the same criteria on the HE state \eqref{eq:HEstate1}. The mode operators $\hat{a}$ and $\hat{a}^\dagger$ shall correspond to the CV subsystem, while $\hat{b}$ and $\hat{b}^\dagger$ shall act on the qudit. This results in
\begin{align}
s_1 &=-\frac{|\alpha|^2\mathrm{e}^{-4|\alpha|^2}}{2}<0\qquad\quad\forall\,\alpha\neq0, \\
s_2 &=\frac{|\alpha|^2}{2}(1-\mathrm{e}^{-4|\alpha|^2})\geq0\quad\;\forall\,\alpha.
\end{align}
As $s_2$ does not perform any entanglement detection, it can be concluded that second moments are probably not sufficient for the detection of entanglement in \eqref{eq:HEstate1}. However, $s_1$, which goes up to fourth order moments, does the job very well, as it detects entanglement for all $\alpha\neq0$. For $\alpha=0$, the state becomes a product state and hence it is not entangled. The determinant $s_1$ is actually the main tool for entanglement detection in hybrid entangled states in this thesis. In the pure state case it is of course not so relevant, but especially for the detection of true HE $s_1$ proves to be very valuable, as seen in section \ref{sec:mixed2}.

\section{Pure States}\label{sec:pure}
In this section the first, pure state class of HE states is investigated. Some representative bipartite pure HE states are discussed and the mentioned tools are exploited. A general pure bipartite HE state in ${\mathcal H}^A_d\otimes{\mathcal H}^B_\infty$ looks like 
\begin{equation}
\ket{\psi}^{AB} =\sum_{i=0}^{d-1} c_{i}\ket{i}^A\ket{\psi_{i}}^B\,,\qquad c_{i}\in\mathbb C\;,\qquad\sum_{i=0}^{d-1} \;|c_{i}|^2=1.
\end{equation}
Due to the finite number of qumode states, a Gram-Schmidt process can be always applied. Afterwards, also a Schmidt decomposition can be carried out. Finally, the state's entanglement properties can be also analyzed with aid of a pure state measure such as the entropy of entanglement.

\subsection{Pure Qubit - Qumode Entanglement}\label{subsec:PQbQm}
Consider the case where the DV system is a qubit:
\begin{equation}
\ket{\psi}_{HE}^{AB} =c_{0}\ket{0}^A\ket{\psi_{0}}^B+c_{1}\ket{1}^A\ket{\psi_{1}}^B.
\end{equation}
Since only the relative phase in the superposition is important and due to normalization, this can be written as
\begin{equation}\label{eq:pQbQmS}
\ket{\psi}_{HE}^{AB} =\sqrt{c}\ket{0}^A\ket{\psi_{0}}^B+\mathrm{e}^{i\phi}\sqrt{1-c}\ket{1}^A\ket{\psi_{1}}^B,
\end{equation}
with real constant $c\in[0,1]$. Performing an inverse Gram-Schmidt process on system B yields
\begin{align}\label{eq:pQbQmSgs1}
\ket{\psi_{0}}^B &=\ket{0}^B, \\ \label{eq:pQbQmSgs2}
\ket{\psi_{1}}^B &=N\ket{0}^B+\sqrt{1-|N|^2}\ket{1}^B,
\end{align}
where the overlap $\braket{\psi_0|\psi_1}$ is denoted by $N$. Making use of equations \eqref{eq:pQbQmSgs1} and \eqref{eq:pQbQmSgs2} $\ket{\psi_{i}}^B$ can be expressed in the new orthogonal basis and DV methods can be applied to calculate the entropy of entanglement of $\ket{\psi}_{HE}^{AB}$, which is shown in figure \ref{fig:pQbQmS}.\footnote[3]{Of course the Gram-Schmidt procedure is not really necessary here, since the CV subsystem can be simply traced out, and the entropy of the residual DV state can be ordinarily computed. For the calculation of the entropy of entanglement it is sufficient that one of the subsystems is DV.}
\begin{equation}
E_S(\ket{\psi}_{HE}^{AB})=-\lambda_+\log_2\lambda_+-\lambda_-\log_2\lambda_-\,,\qquad \lambda_\pm=\frac{1}{2}\Bigl(1\pm\sqrt{1-4c(1-c)(1-|N|^2)}\Bigr).
\end{equation}

\begin{figure}[ht]
\begin{center}
\includegraphics[height=6cm]{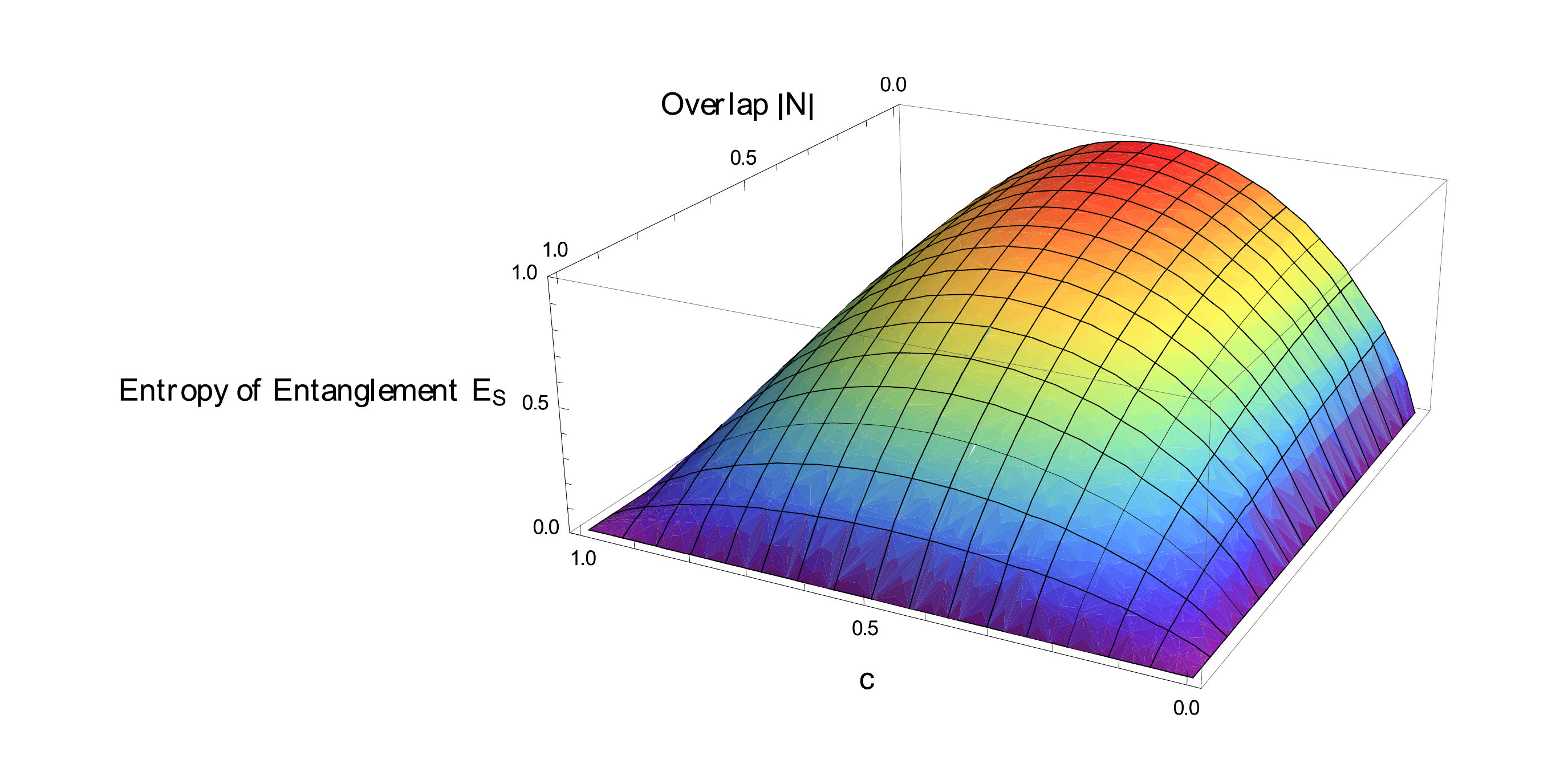}
\caption[Entropy of entanglement of the state $\ket{\psi}_{HE}^{AB}$, \eqref{eq:pQbQmS}.]{Entropy of entanglement of the state $\ket{\psi}_{HE}^{AB}$, \eqref{eq:pQbQmS}. $N$ denotes the overlap $\braket{\psi_0|\psi_1}$. Maximal entanglement $E_S=1$ is achieved for orthogonal qumode states ($N=0$) and $c=\frac{1}{2}$. Note that the entanglement only depends on the absolute value of the overlap $N$, not on its phase.}
\label{fig:pQbQmS}
\end{center}
\end{figure}
Based on this result, compare the HE state $\ket{\psi}^{AB}$ to the structurally comparable, fully DV and CV states 
\begin{align}
\ket{\psi}_{DV}^{AB} &=\sqrt{c}\ket{0}^A\ket{0}^B+\mathrm{e}^{i\phi}\sqrt{1-c}\ket{1}^A\ket{1}^B, \\
\ket{\psi}_{CV}^{AB} &=\frac{1}{\sqrt{{\mathcal N}_{\phi}}}\Bigl(\sqrt{c}\ket{\psi_{0}}^A\ket{\Phi_{0}}^B+\mathrm{e}^{i\phi}\sqrt{1-c}\ket{\psi_{1}}^A\ket{\Phi_{1}}^B\Bigr),
\end{align}
where ${\mathcal N}_{\phi}$ is a normalization constant. It is known that $\ket{\psi}_{CV}^{AB}$'s entanglement depends on the phase $\phi$, as also seen in subsection \ref{catwitnessing}. On the contrary $\ket{\psi}_{DV}^{AB}$ is equally entangled for any phase $\phi$. This the hybrid state has in common with the DV state. Its entanglement also shows no dependence on $\phi$. A common property of all the states is that the entanglement is highest for $c=\frac{1}{2}$. In this case, the DV state becomes a maximally entangled Bell state. Furthermore, note that actually $\ket{\psi}_{HE}^{AB}$ and $\ket{\psi}_{DV}^{AB}$ are only special cases of $\ket{\psi}_{CV}^{AB}$: For orthogonal $\ket{\psi_{i}}^A$'s the HE state is obtained, while both orthogonal $\ket{\psi_{i}}^A$'s and $\ket{\Phi_{i}}^B$'s result in $\ket{\psi}_{DV}^{AB}$.

\subsection{Squeezed Hybrid Entangled States}\label{subsec:sqeez}
In subsection \ref{subsec:AdSVcrit} the state
\begin{equation}
\ket{\psi}^{AB}=\frac{1}{\sqrt{2}}\Bigl(\ket{0}^A\ket{\alpha}^B+\ket{1}^A\ket{-\alpha}^B\Bigr)
\end{equation}
has been investigated and its entanglement has been detected with aid of the SV determinant $s_1$. This actually is an example for a typical HE state which also occurs in QKD experiments (see chapter \ref{ch:5}). Analyzing this state is pretty trivial, as it is just a special case of the HE state considered in the previous subsection. Regarding the SV criteria, so far only their advantages have been pointed out. Now it is demonstrated how easy the criteria can be forced to fail in the detection of entanglement. To do this, consider again $\ket{\psi}^{AB}$, but this time squeezed in mode B:
\begin{equation}
\hat{S}^B(\xi)\ket{\psi}^{AB}=\frac{1}{\sqrt{2}}\Bigl(\ket{0}^A\hat{S}^B(\xi)\ket{\alpha}^B+\ket{1}^A\hat{S}^B(\xi)\ket{-\alpha}^B\Bigr),
\end{equation}
with the squeezing parameter $\xi=r\mathrm{e}^{i\theta}$. This state is still a special case of the general qubit-qumode states $\ket{\psi}_{HE}^{AB}$ considered in the last subsection. Its entanglement quantification is straightforward.

The state is obtained by just applying a unitary transformation, i.e. the squeezing operation, on $\ket{\psi}^{AB}$. As unitary transformations are a special kind of LOCC and also invertible, entanglement has to be preserved under them. Therefore, $\ket{\psi}^{AB}$ necessarily has the same amount of entanglement as $\hat{S}^B(\xi)\ket{\psi}^{AB}$. Instead of the latter, squeezed state, regarding entanglement, it is sufficient to examine the former, unsqueezed state. As the SV determinant $s_1$ detects entanglement in $\ket{\psi}^{AB}$ for all $\alpha\neq0$, the same would be expected, when evaluating $s_1$ with respect to $\hat{S}^B(\xi)\ket{\psi}^{AB}$, for any squeezing $\xi$.

However,
\begin{equation}
s_1(\hat{S}^B(\xi)\ket{\psi}^{AB})=\frac{1}{4}\sinh^2r-\frac{\mathrm{e}^{-4|\alpha|^2}}{2}|\alpha|^2\cosh^2r-\frac{\mathrm{e}^{-4|\alpha|^2}}{8}\sinh^2r,
\end{equation}
which is obviously not below zero for all $r$ and $\alpha\neq0$. Also note the independence of the squeezing phase $\theta$. See figure \ref{fig:SqeezDet} for a visualization of $s_1(\hat{S}^B(\xi)\ket{\psi}^{AB})$. 
\begin{figure}[ht]
\subfloat{\includegraphics[width=0.5\textwidth]{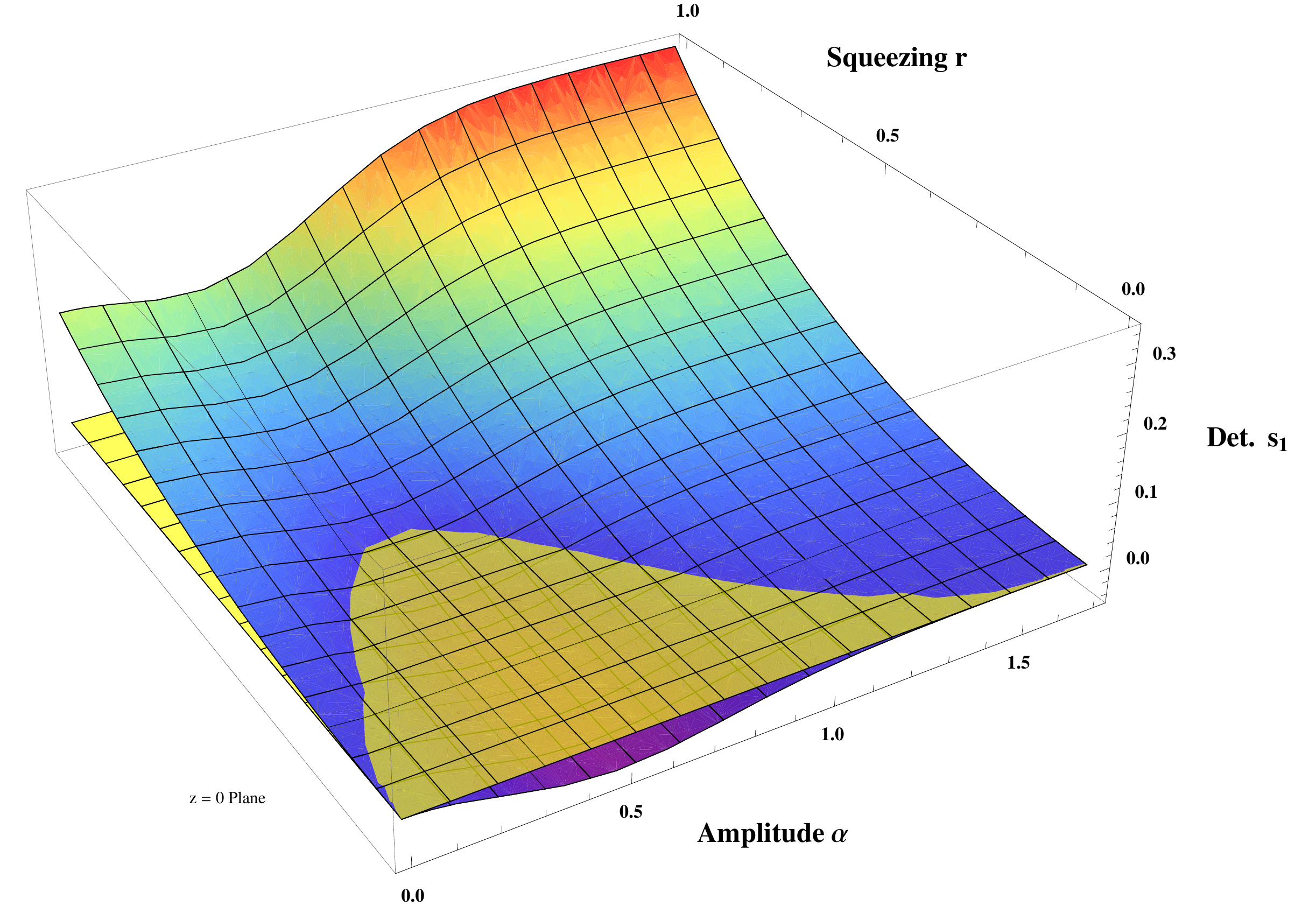}}\hfill
\subfloat{\includegraphics[width=0.5\textwidth]{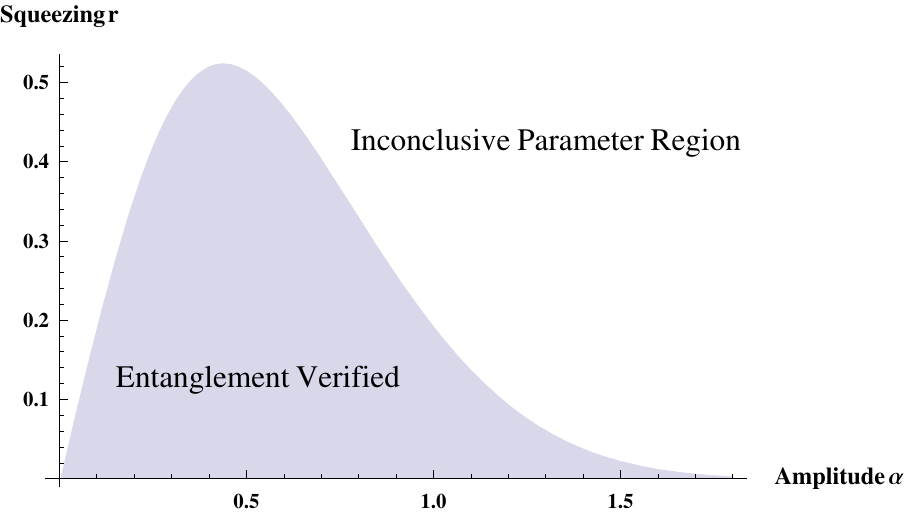}}\hfill
\caption[SV determinant $s_1$ of $\hat{S}^B(\xi)\ket{\psi}^{AB}$, and parameter regime of entanglement witnessing.]{The left-hand graphic shows the SV determinant $s_1$ of $\hat{S}^B(\xi)\ket{\psi}^{AB}$. The yellow plane denotes the zero-plane. For values $(r,\alpha)$ such that $s_1(r,\alpha)<0$ entanglement is verified. The two regimes are plotted in the right-hand diagram.}
\label{fig:SqeezDet}
\end{figure}
What is even more surprising is that, on the one side, the higher $\alpha$ is the more difficult it is to witness entanglement. But on the other side, the entanglement does actually \textit{increase} as $\alpha$ gets higher, since $\hat{S}(\xi)^B\ket{\alpha}^B$ and $\hat{S}(\xi)^B\ket{-\alpha}^B$ approach orthogonality in this case. This is of course no contradiction to earlier results. Failing to detect entanglement does not mean there is not any. Nevertheless, it is remarkable that a simple unitary transformation, which does not change the entanglement at all, is able to cause the SV criteria to fail. While for $r=0$ unambiguous witnessing for all $\alpha\neq0$ is obtained, for $r\neq0$ there are regions, where the SV determinant breaks down. It could be argued that the unitary transformation does on the one hand not destroy entanglement, but on the other hand possibly rearranges the existing entanglement in such a way that initially successful witnesses fail to detect the rearranged entanglement. However, this is just speculation.

In $\hat{S}^B(\xi)\ket{\psi}^{AB}$, the terms $\hat{S}^B(\xi)\ket{\pm\alpha}^B$ can be expressed as $\hat{S}^B(\xi)\hat{D}^B(\pm\alpha)\ket{0}^B$. But what about states which involve squeezed qumode states in the form $\hat{D}^B(\pm\alpha)\hat{S}^B(\xi)\ket{0}^B$ instead of $\hat{S}^B(\xi)\hat{D}^B(\pm\alpha)\ket{0}^B\,$? Since $\hat{S}^B(\xi)^\dagger\hat{D}^B(\alpha)\hat{S}^B(\xi)=\hat{D}^B(\beta)$ with $\beta=\alpha\cosh r+\alpha^\ast\mathrm{e}^{i\theta}\sinh r$, and therefore $\hat{D}^B(\alpha)\hat{S}^B(\xi)=\hat{S}^B(\xi)\hat{D}^B(\beta)$, these states can then be rewritten in the former way \cite{Aravind}. Therefore, they do not provide any new insights, as both scenarios are equivalent.

\subsection{Pure Qutrit - Qumode Entanglement}
As a final example of pure bipartite HE, an entangled state of a qutrit and a qumode subsystem is considered. Using this example all the mentioned tools for pure bipartite HE states are successively demonstrated. Consider
\begin{equation}\label{eq:pQtQmS}
\ket{\psi}_{HE}^{AB} =\frac{1}{\sqrt{3}}\Bigl(\ket{0}^A\ket{vac}^B+\ket{1}^A\ket{\alpha}^B+\ket{2}^A\ket{-\alpha}^B\Bigr).
\end{equation}
\begin{itemize}
\item An \textbf{Inverse Gram-Schmidt process} with respect to subsystem B yields
\begin{align}
\ket{vac}^B & = \;\;\,\ket{0}^B,  \\
\ket{\alpha}^B & = x\ket{0}^B + \quad\sqrt{1-x^2}\ket{1}^B, \\
\ket{-\alpha}^B & = x\ket{0}^B - x^2\sqrt{1-x^2}\ket{1}^B + \sqrt{1-x^2-x^4+x^6}\ket{2}^B,
\end{align}
with $x=e^{-\frac{1}{2}|\alpha|^2}$. Hence
\begin{equation}
\begin{aligned}
\ket{\phi}^{AB}& =\frac{1}{\sqrt{3}}\Bigl(\ket{0}^A\ket{0}^B+x\ket{1}^A\ket{0}^B+\sqrt{1-x^2}\ket{1}^A\ket{1}^B+x\ket{2}^A\ket{0}^B \\ & -x^2\sqrt{1-x^2}\ket{2}^A\ket{1}^B+ \sqrt{1-x^2-x^4+x^6}\ket{2}^A\ket{2}^B\Bigr),
\end{aligned}
\end{equation}
which is the effective DV form of the state.
\item The pure-state \textbf{entropy of entanglement} of the state can be calculated. Its analytic form is rather lengthy. Hence, only the graphical version is shown in figure \ref{fig:pQtQmS},
\begin{figure}[ht]
\begin{center}
\includegraphics[height=5cm]{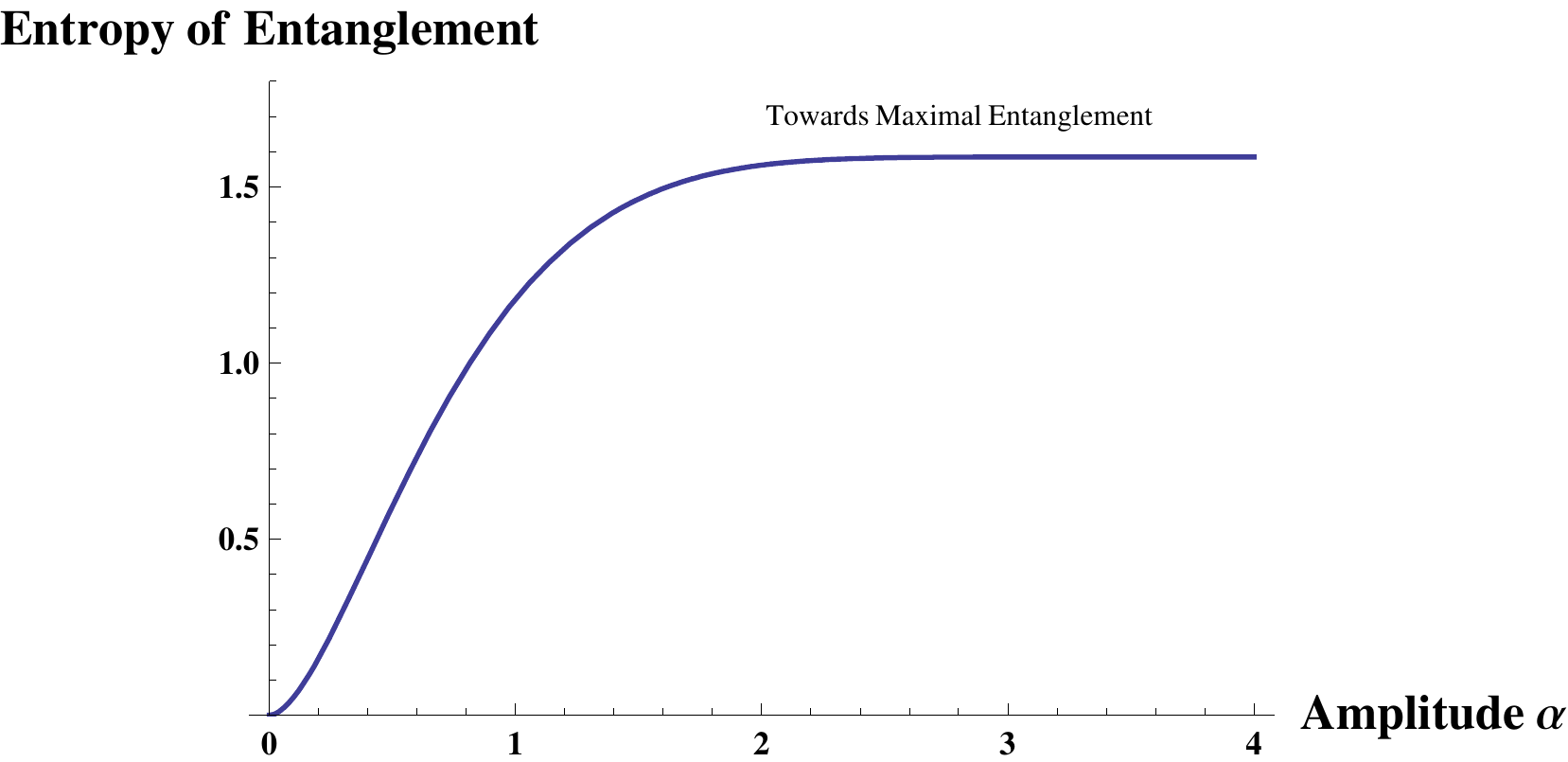}
\caption[Entropy of entanglement of the state $\ket{\psi}^{AB}$, \eqref{eq:pQtQmS}.]{Entropy of entanglement of the state $\ket{\psi}^{AB}$, \eqref{eq:pQtQmS}. For $\alpha\rightarrow\infty$ the state becomes maximally entangled.}
\label{fig:pQtQmS}
\end{center}
\end{figure}
from which it can be seen that for $\alpha\rightarrow\infty$ the state becomes maximally entangled. Note that, while the state deals with qutrit entanglement, the entropy of entanglement has been calculated in qubit entanglement units, so-called \text{e-bits}.\footnote[4]{A maximally entangled two-qubit Bell state $\ket{\Phi}_{Bell}:=\frac{1}{\sqrt{2}}(\ket{00}+\mathrm{e}^{i\phi}\ket{11})$ has exactly 1 \textit{e-bit} of entanglement. E-bits are therefore a unit for entanglement quantification, especially for two-qubit entanglement.} If the logarithm in the von-Neumann entropy had been calculated with respect to base $3$, the maximal qutrit entanglement would actually be 1 in these qutrit entanglement units and in the plot, the graph would converge against 1. 
\item Finally, set $x=\frac{1}{2}$, which corresponds to $\alpha=\sqrt{2\ln2}\approx1.18$, and calculate the \textbf{Schmidt decomposition}:
\begin{equation}
\ket{\phi}^{AB}=0.76\ket{0}^A\ket{0}^B+0.56\ket{1}^A\ket{1}^B+0.33\ket{2}^A\ket{2}^B.
\end{equation}
\end{itemize}
What can be concluded from not only considering qubits, but also qutrits as the DV subsystem of HE states, is that focussing on higher dimensional DV subsystems will probably not bring new insights. The methods applied are always the same whether the DV system is 2-, 3- or higher-dimensional. Hence, in the following mostly qubit subsystems are employed.

\section{Mixed States supported by a Finite-Dimensional Subspace}\label{sec:mixed1}
In this section mixed HE states which are effectively DV are investigated. That is to say those states with $1<N<\infty$ according to equation \eqref{eq:HEdef}. One the one hand, states are discussed which feature the characteristic properties of this kind of HE entanglement. Some typical DV tools are then applied on them. For example the very important entanglement measure logarithmic negativity is calculated. On the other hand, adequate mixed HE states which are of physical relevance are discussed. This leads to the question, how states actually become mixed. The answer are quantum channels which cause them to decohere. Therefore, along the way, also some interesting properties of quantum channels are presented, especially regarding entanglement evolution. 

What can be generally said about bipartite, effectively DV, mixed HE states is that there always an inverse Gram-Schmidt process can be performed, which yields a density matrix description. This enables the application of DV tools, such as typical DV entanglement measures.

\subsection{Amplitude Damped Qubit - Qumode Entanglement}\label{subsecA}
The state
\begin{equation}
\ket{\psi}^{AB}=\frac{1}{\sqrt{2}}\Bigl(\ket{0}^A\ket{\psi_0}^B+\ket{1}^A\ket{\psi_1}^B\Bigr)
\end{equation}
is not only a very simple and illustrating example of HE, but also physically relevant, as already mentioned. Therefore, simply send this state through some decoherence-causing channel and obtain a HE mixed state. Then it depends on the form of the channel whether the output is either effectively DV HE or true HE. As bipartite systems are employed, the question arises whether both subsystems or only one of them shall be subject to a channel. As the focus here lies not on complicated states and deep mathematics, but on the identification of characteristic effects and behavior, only one-sided channels $\hat{\rho}'{}^{AB}=({\mathbb 1}\otimes\Upsilon)\hat{\rho}^{AB}$ are considered. Two-sided channels will just make the outcome more intricate and the analysis more complicated. Furthermore, in a typical physical application, such as particular cases of QKD, it is only the CV subsystem which undergoes the action of the channel. Hence, consider channels of the form $\hat{\rho}'{}^{AB}=({\mathbb 1}^A\otimes\Upsilon^B)\hat{\rho}^{AB}$, where subsystem A is DV and subsystem B is CV.
\begin{equation}
\begin{aligned}
\hat{\rho}'{}^{AB}&=({\mathbb 1}^A\otimes\Upsilon^B)\frac{1}{2}\Bigl(\ket{0}^A\bra{0}\otimes\ket{\psi_0}^B\bra{\psi_0} \\ 
&+\ket{1}^A\bra{1}\otimes\ket{\psi_1}^B\bra{\psi_1} \\
&+\ket{0}^A\bra{1}\otimes\ket{\psi_0}^B\bra{\psi_1} \\
&+\ket{1}^A\bra{0}\otimes\ket{\psi_1}^B\bra{\psi_0}\Bigr).
\end{aligned}
\end{equation}
Now write the state as a proper density matrix with respect to subsystem A, which is not subject to the channel, i.e. $\hat{\rho}'{}^{AB}_{ij}=\braket{i^A|\hat{\rho}'{}^{AB}|j^A}$.
\begin{equation}
\hat{\rho}'{}^{AB}=\begin{pmatrix}
\Upsilon^B(\ket{\psi_0}^B\bra{\psi_0}) & \Upsilon^B(\ket{\psi_0}^B\bra{\psi_1}) \\ 
\Upsilon^B(\ket{\psi_1}^B\bra{\psi_0}) & \Upsilon^B(\ket{\psi_1}^B\bra{\psi_1})
\end{pmatrix}.
\end{equation}
What can be seen is that for these one-sided channels on qubit-qumode HE states, a nice block structure is obtained when applying the density matrix formalism on subsystem A. This can be generalized to universal pure, qudit-qumode entangled states. For a general state of the form 
\begin{equation}
\ket{\psi}^{AB} =\sum_{i=0}^{d-1} c_{i}\ket{i}^A\ket{\psi_{i}}^B\,,\qquad c_{i}\in\mathbb C\;,\qquad\sum_{i=0}^{d-1} \;|c_{i}|^2=1,
\end{equation}
the transformed state $\hat{\rho}'{}^{AB} =({\mathbb 1}^A\otimes\Upsilon^B)\ket{\psi}^{AB}\bra{\psi}^{AB}$, in the density matrix formalism for subsystem A, looks like
\begin{equation}
\hat{\rho}'{}^{AB}= \begin{pmatrix}
|c_0|^2\Upsilon^B(\ket{\psi_0}^B\bra{\psi_0}) & \cdots & c_0 c_{d-1}^\ast\Upsilon^B(\ket{\psi_0}^B\bra{\psi_{d-1}}) \\ 
\vdots & \ddots & \vdots \\
c_{d-1} c_0^\ast\Upsilon^B(\ket{\psi_{d-1}}^B\bra{\psi_0}) & \cdots & |c_{d-1}|^2\Upsilon^B(\ket{\psi_{d-1}}^B\bra{\psi_{d-1}})
\end{pmatrix},
\end{equation}
or in index notation
\begin{equation}
\hat{\rho}'{}^{AB}_{ij}= c_i c_j^\ast\Upsilon^B(\ket{\psi_i}^B\bra{\psi_j}).
\end{equation}
One step further, for completely general bipartite mixed HE states of the form \eqref{eq:HEdef}, the output of one-sided channels $({\mathbb 1^A}\otimes\Upsilon^B)\hat{\rho}^{AB}$ is
\begin{equation}
\hat{\rho}'{}^{AB}_{ij}= \sum_{n=1}^N p_n c_{ni} c_{nj}^\ast\Upsilon^B(\ket{\psi_{ni}}^B\bra{\psi_{nj}}).
\end{equation}
Consider a specific, physically relevant example. Choose as the state subject to the channel the already known
\begin{equation}\label{eq:lossyQbQm}
\ket{\psi}^{AB}=\frac{1}{\sqrt{2}}\Bigl(\ket{0}^A\ket{\alpha}^B+\ket{1}^A\ket{-\alpha}^B\Bigr),
\end{equation}
and as the one-sided channel itself the amplitude damping channel (also called photon loss channel) acting on subsystem B. Exploiting Stinespring's dilation theorem \ref{thm:stine}, write
\begin{equation}
\hat{\rho}^{AB}=tr_E[\hat{U}^{BE}\;\bigl( \ket{\phi}^{AB}\bra{\phi}\otimes\ket{0}^E\bra{0} \bigr) \;\hat{U}^{{BE}^\dagger}].
\end{equation}
The ancilla Hilbert space is the environment E, which is in a vacuum state, and into which photons may leak out of the qumode system. The coupling to the environment is decribed by a beam splitter $\hat{U}^{BE}=e^{\theta(\hat{a}_E^\dagger \hat{a}_B-\hat{a}_B^\dagger \hat{a}_E)}$, where the beam splitter transmissivity $\eta$ is given as $\cos^2\theta=\eta$. For a graphical illustration of the channel modeling, see figure \ref{fig:ADampCh}.
\begin{figure}[ht]
\begin{center}
\includegraphics[width=12cm]{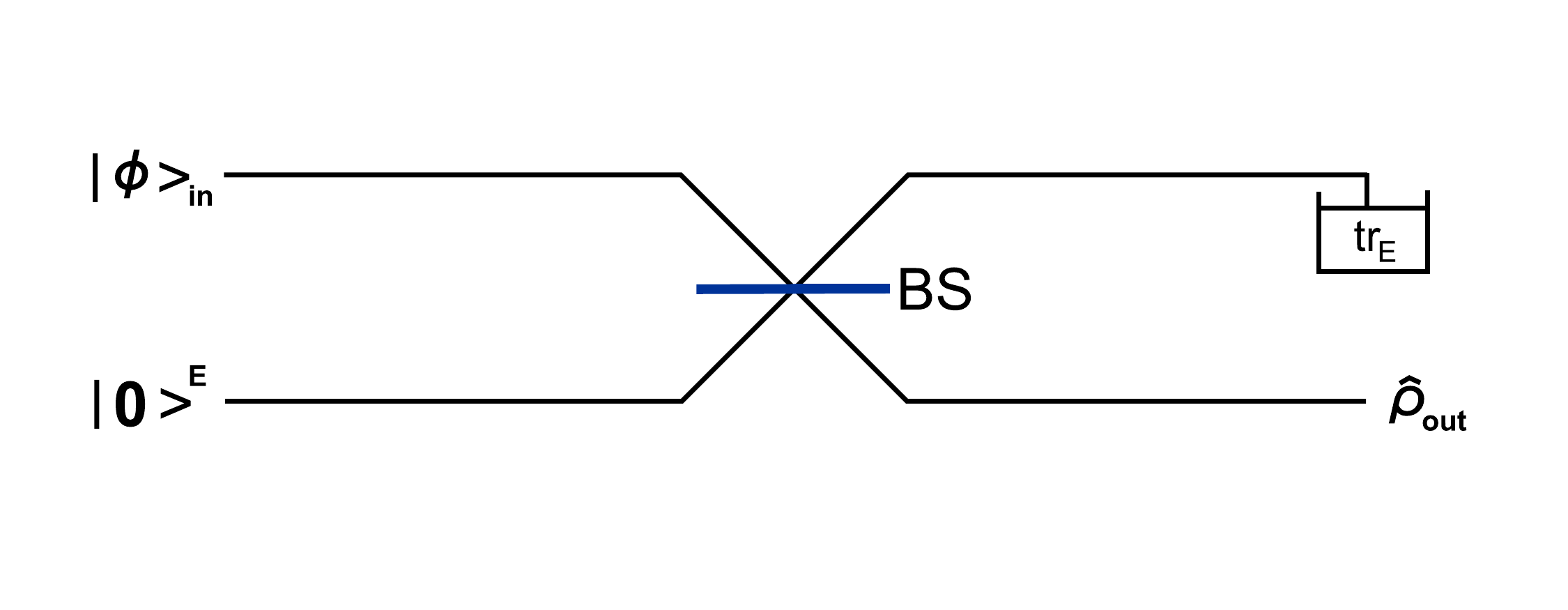}
\caption[Illustration of the amplitude damping channel.]{Illustration of the amplitude damping channel. The output $\hat{\rho}_{out}$ is obtained by using a beam splitter to mix the input state $\ket{\psi}_{in}$ with vacuum environment. Afterwards the environment is traced out.}
\label{fig:ADampCh}
\end{center}
\end{figure}
For the actual calculation make use of $\hat{U}^{{BE}^\dagger}\hat{U}^{BE}=\mathbb{1}$, $\hat{U}^{BE}D^B(\alpha)\hat{U}^{{BE}^\dagger}=D^B(\sqrt{\eta}\alpha)D^E(\sqrt{1-\eta}\alpha)$ as well as $\hat{U}^{BE}\ket{0}^B\ket{0}^E=\ket{0}^B\ket{0}^E$ and obtain
\begin{equation}
\begin{aligned}
\hat{\rho}^{AB} &=\frac{1}{2}\Bigl(\ket{0}^A\bra{0}\otimes\ket{\sqrt{\eta}\alpha}^B\bra{\sqrt{\eta}\alpha} \\
&+ \ket{1}^A\bra{1}\otimes\ket{-\sqrt{\eta}\alpha}^B\bra{-\sqrt{\eta}\alpha} \\
&+\mathrm{e}^{-2(1-\eta)|\alpha|^2}\, \ket{0}^A\bra{1}\otimes\ket{\sqrt{\eta}\alpha}^B\bra{-\sqrt{\eta}\alpha} \\
&+\mathrm{e}^{-2(1-\eta)|\alpha|^2}\, \ket{1}^A\bra{0}\otimes\ket{-\sqrt{\eta}\alpha}^B\bra{\sqrt{\eta}\alpha}\Bigr).
\end{aligned}
\end{equation}
This clearly is a mixed state. The channel causes the state to decohere, as can be inferred from the damping of its off-diagonal coherence terms. It can be also expressed in its pure state decomposition. Defining the damping $\tau:=\mathrm{e}^{-2(1-\eta)|\alpha|^2}$ and the overlap $\lambda:=\braket{-\sqrt{\eta}\alpha|\sqrt{\eta}\alpha}=\mathrm{e}^{-2\eta|\alpha|^2}$, 
\begin{equation}\label{eq:dampedQQ}
\begin{aligned}
\hat{\rho}^{AB} & = \frac{1+\tau}{2}\;\ket{\phi_+}^{AB}\bra{\phi_+}\;+\;\frac{1-\tau}{2}\;\ket{\phi_-}^{AB}\bra{\phi_-}, \\ 
\ket{\phi_\pm}^{AB} & = \frac{1}{\sqrt{2}}\;\biggl(\ket{0}^A\ket{\sqrt{\eta}\alpha}^B\;\pm\;\ket{1}^A\ket{-\sqrt{\eta}\alpha}^B\biggr).
\end{aligned}
\end{equation}
This state contains obviously only two qumode states $\ket{\pm\sqrt{\eta}\alpha}^B$, which can be put into an orthogonal basis.
\begin{align}
\ket{\sqrt{\eta}\alpha}^B & = \ket{0}^B, \\
\ket{-\sqrt{\eta}\alpha}^B & = \lambda\ket{0}^B+\sqrt{1-\lambda^2}\ket{1}^B.
\end{align} 
Therefore, $\hat{\rho}^{AB}$ is an effective two-qubit state with density matrix
\begin{equation}
\rho^{AB}= \frac{1}{2}\begin{pmatrix}
1 & 0 & \lambda\tau & \tau\sqrt{1-\lambda^2} \\
0 & 0 & 0 & 0 \\
\lambda\tau & 0 & \lambda^2 & \lambda\sqrt{1-\lambda^2} \\
\tau\sqrt{1-\lambda^2} & 0 & \lambda\sqrt{1-\lambda^2} & 1-\lambda^2
\end{pmatrix}.
\end{equation}
When dealing with mixed two-qubit states, the obvious choice for entanglement quantification is the concurrence $C$. Its calculation is straightforward:
\begin{equation}
C(\hat{\rho}^{AB})=\frac{1}{2}\sqrt{1-\mathrm{e}^{-4\eta|\alpha|^2}}\Bigl(\sqrt{1+3\mathrm{e}^{-4(1-\eta)|\alpha|^2}}-\sqrt{1-\mathrm{e}^{-4(1-\eta)|\alpha|^2}}\Bigr).
\end{equation} 
The plot is shown in figure \ref{fig:CADampCh}.
\begin{figure}[ht]
\begin{center}
\includegraphics[width=12cm]{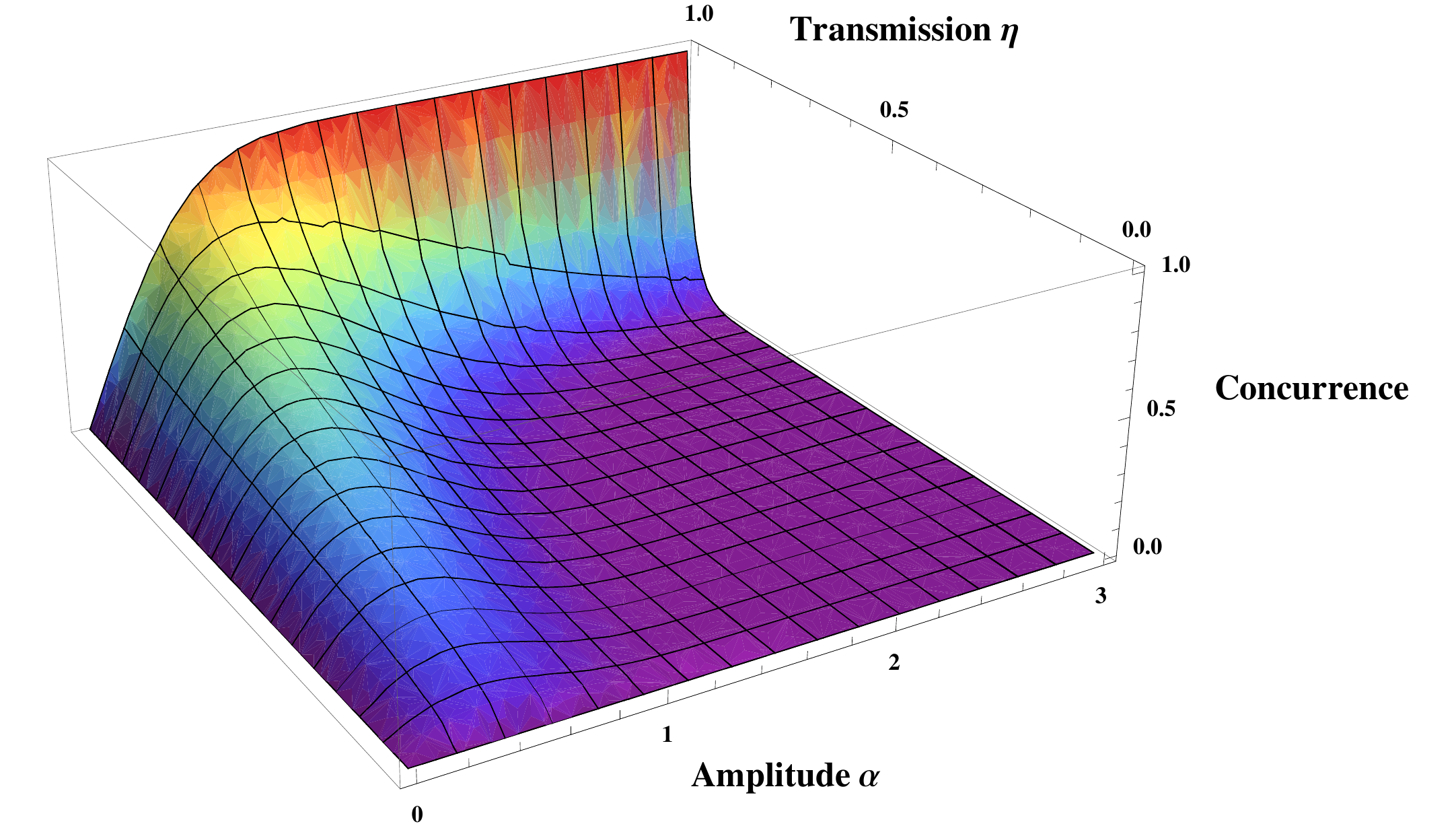}
\caption[Concurrence of the state $\hat{\rho}^{AB}$, \eqref{eq:dampedQQ}, as a function of transmissivity $\eta$ and amplitude $\alpha$.]{Concurrence of the state $\hat{\rho}^{AB}$, \eqref{eq:dampedQQ}, as a function of transmissivity $\eta$ and amplitude $\alpha$. On the one hand, the higher $\alpha$ the greater the initial entanglement, but also the more sensitive the state is against photon loss. On the other hand, for low $\alpha$ there is only little initial entanglement. However, the state is more robust against losses. Hence, there is some trade-off behavior and depending on $\eta$ an optimal $\alpha^{opt}_\eta$ exists for which the output entanglement is largest.}
\label{fig:CADampCh}
\end{center}
\end{figure}
It shows an interesting trade-off behavior. On the one side, a higher $\alpha$ results in a greater initial entanglement. But it also causes the state to be more sensitive against photon loss. On the other side, a low $\alpha$ yields only little initial entanglement, but the state is more robust against loss. This leads to an optimal $\eta$-dependent $\alpha^{opt}_\eta$ for which the largest output entanglement is obtained. It is worth mentioning that a similar trade-off is observed in two-mode cat states $\frac{1}{\sqrt{{\mathcal N}_\pm}}(\ket{\alpha}^A\ket{\alpha}^B\pm\ket{-\alpha}^A\ket{-\alpha}^B)$, when one of the modes suffers photon loss.

Finally, note that just as expected in the limit $\eta=1$ the concurrence becomes  $C(\hat{\rho}^{AB})=\sqrt{1-\mathrm{e}^{-4|\alpha|^2}}$, which is the concurrence of the initial undisturbed state 
\begin{equation}
\ket{\psi}^{AB}=\frac{1}{\sqrt{2}}\Bigl(\ket{0}^A\ket{\alpha}^B+\ket{1}^A\ket{-\alpha}^B\Bigr).
\end{equation}

\subsection{The Choi-Jamiolkowski Isomorphism and Entanglement Evolution}
This subsection can be regarded as some kind of bonus. As quantum channels play an important role in this thesis and have just been used for the first time, an elegant way to work with them and study them is presented. Consider a linear map between two finite dimensional Hilbert spaces:
\begin{equation}
\Upsilon :S({\mathcal H}_{d_1})\rightarrow S({\mathcal H}_{d_2}).
\end{equation}
Then there can be defined an isomorphism between this map and an operator living in ${\mathcal H}_{d_1}\otimes{\mathcal H}_{d_2}$ by
\begin{equation}
\hat{\rho}_{\Upsilon}=({\mathbb 1}^A\otimes\Upsilon^B)\ket{\Phi^+_{d_1}}^{AB}\bra{\Phi^+_{d_1}},
\end{equation}
where $\ket{\Phi^+_{d_1}}^{AB}$ denotes the maximally entangled state in ${\mathcal H}_{d_1}\otimes{\mathcal H}_{d_1}$,
\begin{equation}
\ket{\Phi^+_{d_1}}^{AB}=\frac{1}{\sqrt{d_1}}\sum_{i=0}^{d_1-1}\ket{i}^A\ket{i}^B.
\end{equation}
This \textit{Choi-Jamiolkowski isomorphism} holds for all linear maps and operators. If, additionally, the linear map is completely positive, it is a quantum channel. The isomorphism was first suggested by Jamiolkowski in 1972 \cite{Jamiol}. 3 years later Choi employed it to derive the operator-sum decomposition of CPTP maps (which had been found independently by Kraus already in 1971) \cite{Choi1,Kraus}. This isomorphism proves very useful for the study of quantum channels. If the physics of quantum channels are to be studied, instead simply the associated quantum states can be investigated. Then advantage can be taken of the fact that much is already known about quantum states. Their theory is well developed. 

Exploiting this isomorphism, several remarkable connections between maps and the associated linear operators can be derived. For example consider the set of Hermiticity-preserving linear maps. Then the associated set of linear operators is the set of Hermitian operators. If the set of CP maps is considered, i.e. the set of quantum channels, the associated operators are the positive-semidefinite operators, i.e. all quantum states. An interesting link occurs when considering entanglement breaking maps. Then the associated linear operators are the separable quantum states. Furthermore, with aid of the isomorphism, relations between theorems on states and theorems on channels and operations respectively can be derived. For example the operator-sum decomposition directly translates into the spectral decomposition of quantum states \cite{Isom1,Isom2}. 

Another very important application of the Choi-Jamiolkowski isomorphism can be found in entanglement evolution theory \cite{Konrad}. Consider a two-qubit state $\ket{\chi}^{AB}=\sqrt{w}\ket{0}^A\ket{0}^B+\sqrt{1-w}\ket{1}^A\ket{1}^B$ subject to a one-sided channel $\hat{\rho}'=({\mathbb 1}^A\otimes\Upsilon^B)\ket{\chi}^{AB}\bra{\chi}$. Due to the isomorphism this can be also expressed in the dual picture, where the roles of channel and state are interchanged:
\begin{equation}
\frac{1}{p'}({\mathbb 1}^A\otimes\Upsilon^B)\ket{\chi}^{AB}\bra{\chi}=\frac{1}{p''}(\Upsilon_{\chi}^A\otimes{\mathbb 1}^B)\hat{\rho}_{\Upsilon}^{AB}.
\end{equation}
$\Upsilon_{\chi}^A$ denotes the channel associated with the state $\ket{\chi}^{AB}$ and $\hat{\rho}_{\Upsilon}^{AB}$ is the state which corresponds to the channel $\Upsilon^B$. $p'=\mathrm{tr}[({\mathbb 1}^A\otimes\Upsilon^B)\ket{\chi}^{AB}\bra{\chi}]$ and $p''=\mathrm{tr}[(\Upsilon_{\chi}^A\otimes{\mathbb 1}^B)\hat{\rho}_{\Upsilon}^{AB}]$ are normalization constants due to a possible trace decrease. This dual picture can be given a very descriptive interpretation using quantum teleportation, which is explained in detail and in a very demonstrative way in \cite{Konrad}. With this teleportation interpretation, $\Upsilon_{\chi}^A$ and $\hat{\rho}_{\Upsilon}^{AB}$ can be expressed in terms of the known variables $\ket{\chi}^{AB}$ and $\Upsilon^B$. Finally, entanglement quantification is performed in the dual picture with aid of the concurrence and what comes out is a neat entanglement evolution equation for pure two-qubit states:
\begin{equation}
C[({\mathbb 1}^A\otimes\Upsilon^B)\ket{\chi}^{AB}\bra{\chi}]=C[({\mathbb 1}^A\otimes\Upsilon^B)\ket{\Phi^+_{2}}^{AB}\bra{\Phi^+_{2}}]\,C[\ket{\chi}^{AB}].
\end{equation}
As the concurrence is a convex measure, for mixed input states the evolution equation becomes an inequality, only providing upper bounds on the output entanglement.
\begin{equation}
C[({\mathbb 1}^A\otimes\Upsilon^B)\hat{\rho}_0^{AB}]\leq C[({\mathbb 1}^A\otimes\Upsilon^B)\ket{\Phi^+_{2}}^{AB}\bra{\Phi^+_{2}}]\,C[\hat{\rho}_0^{AB}].
\end{equation}
In the derivation of this formula, the explicit form of the concurrence has been exploited. As this monotone is only defined for two-qubit states, a generalization to higher dimensional systems is not trivial. Nevertheless it exists, utilizing the \textit{G-concurrence} $G_d$ by Gilad Gour \cite{Gour}. It is a generalized concurrence monotone for ${\mathcal H}_d\otimes{\mathcal H}_d$ finite-dimensional systems, which coincides with the "classical" concurrence for $d=2$. Therewith, a generalized entanglement evolution equation can be derived in the same way, again making use of the Choi-Jamiolkowski isomorphism \cite{Tiersch}.
\begin{align}
G_d[({\mathbb 1}^A\otimes\Upsilon^B)\ket{\chi_d}^{AB}\bra{\chi_d}]&=G_d[({\mathbb 1}^A\otimes\Upsilon^B)\ket{\Phi^+_d}^{AB}\bra{\Phi^+_d}]\,G_d[\ket{\chi_d}^{AB}], \\
G_d[({\mathbb 1}^A\otimes\Upsilon^B)\hat{\rho}_d^{AB}]&\leq G_d[({\mathbb 1}^A\otimes\Upsilon^B)\ket{\Phi^+_d}^{AB}\bra{\Phi^+_d}]\,G_d[\hat{\rho}_d^{AB}],
\end{align}
where $\ket{\chi_d}^{AB}=\frac{1}{\sqrt{d}}\sum_{i,j=0}^{d-1}A_{ij}\ket{i}^A\ket{j}^B$ denotes a bipartite $d\times d$-dimensional state with finite $d$ and constants $A_{ij}\in{\mathbb C}$. $\hat{\rho}_d^{AB}$ denotes some mixed state in the same Hilbert space.

These entanglement evolution equations provide an elegant way for entanglement analysis of finite-dimensional states which are subject to a quantum channel. Unfortunately, they cannot be employed for HE states, as these live in an overall infinite-dimensional Hilbert space. The question may arise whether it at least works for the class of effectively DV HE states. The problem is that the Gram-Schmidt process necessary for the DV description has to be performed after the channel. Admittedly, it could theoretically also be performed prior to the action of the channel. However, in this case also the operators describing the channel have to be transformed and expressed in the new basis. In practice, this is often a rather cumbersome undertaking. Hence, for HE states entanglement evolution equations do not work very well as long as they are formulated for finite-dimensional systems. General CV entanglement evolution equations could not be derived yet.

Two of the most important entanglement monotones for DV as well as for Gaussian CV quantum states are the negativity $\mathcal N$ and the related logarithmic negativity $E_N$. An apparent question is whether a comparable entanglement evolution equation can be also derived for one of these measures. By a simple counterexample it can be shown that this is not the case. Consider again the state 
\begin{equation}
\ket{\chi}^{AB}=\sqrt{w}\ket{0}^A\ket{0}^B+\sqrt{1-w}\ket{1}^A\ket{1}^B.
\end{equation}
Apply a photon number encoding, where $\ket{0}$ corresponds to zero photons and $\ket{1}$ to one photon. The channel applied is the one-sided photon loss channel, which has been presented in the previous subsection. For qubits, its Kraus operators are given by \cite{Nielsen}
\begin{align}
\hat{K}_0 &=\ket{0}\bra{0}+\sqrt{\eta}\ket{1}\bra{1}, \\
\hat{K}_1 &=\sqrt{1-\eta}\ket{0}\bra{1}.
\end{align}
The output state $\hat{\rho}^{AB}$ is given by
\begin{equation}
\begin{aligned}
\hat{\rho}^{AB} &=({\mathbb 1}^A\otimes\Upsilon_{loss}^B)\ket{\chi}^{AB}\bra{\chi} \\
&=\frac{1}{2}\Bigl(\ket{0}^A\bra{0}\otimes\ket{0}^B\bra{0}+\sqrt{\eta}\ket{1}^A\bra{0}\otimes\ket{1}^B\bra{0} \\
&+\sqrt{\eta}\ket{0}^A\bra{1}\otimes\ket{0}^B\bra{1}+(1-\eta)\ket{1}^A\bra{1}\otimes\ket{0}^B\bra{0}+\eta\ket{1}^A\bra{1}\otimes\ket{1}^B\bra{1}       \Bigr).
\end{aligned}
\end{equation}
Its partial transpose (PT) has one negative eigenvalue, i.e. 
\begin{equation}
\lambda_-=\frac{1-w}{2}\Bigl(1-\eta-\sqrt{(1-\eta)^2+4\eta\frac{w}{1-w}}\Bigr)\leq0\,\forall\,\eta,w\in[0,1]. 
\end{equation}
This yields the negativity
\begin{equation}
{\mathcal N}(\hat{\rho}^{AB})=\frac{1-w}{2}\Bigl(\sqrt{(1-\eta)^2+4\eta\frac{w}{1-w}}-(1-\eta)\Bigr).
\end{equation}
It is also known that
\begin{align}
{\mathcal N}&[\ket{\chi}^{AB}] =\sqrt{w(1-w)}, \\
{\mathcal N}&[({\mathbb 1}^A\otimes\Upsilon_{loss}^B)\ket{\Phi^+_{2}}^{AB}\bra{\Phi^+_{2}}] =\frac{\eta}{2}.
\end{align} 
However,
\begin{equation}
{\mathcal N}(\hat{\rho}^{AB})\neq\frac{\eta}{2}\sqrt{w(1-w)}={\mathcal N}[({\mathbb 1}^A\otimes\Upsilon_{loss}^B)\ket{\Phi^+_{2}}^{AB}\bra{\Phi^+_{2}}]\,{\mathcal N}[\ket{\chi}^{AB}].
\end{equation}
Hence, there does not exist an entanglement evolution equation in a comparable form to the one derived in \cite{Konrad}. Instead the output entanglement ${\mathcal N}(\hat{\rho}^{AB})$ is actually a complex function of ${\mathcal N}[\ket{\chi}^{AB}]$ and ${\mathcal N}[({\mathbb 1}^A\otimes\Upsilon_{loss}^B)\ket{\Phi^+_{2}}^{AB}\bra{\Phi^+_{2}}]$, and this for the rather approachable photon loss channel. It follows that $E_N$ also offers no such entanglement evolution equation, as it is directly related to the negativity $\mathcal N$ by $E_N=\log_2[1+2{\mathcal N}]$.

\subsection{Logarithmic Negativity of a Nontrivial Hybrid Entangled State}
After this short excursion into entanglement evolution theory, HE states shall be discussed again. The first example of the class of mixed, but effectively DV HE states in the last but one subsection was the state \eqref{eq:lossyQbQm}, subject to photon loss. The output state only contained two qumode states and hence, after a Gram-Schmidt process, the state's entanglement was quantified with aid of the concurrence. However, the concurrence only works for two-qubit states. Therefore, another more representative state is to be discussed, which has to be analyzed with a more universal measure. The example is more artificially constructed, though. 

Consider the state
\begin{equation}\label{eq:2x3Mstate1}
\begin{aligned}
\hat{\rho}^{AB} &= p\,\ket{\phi_+}^{AB}\bra{\phi_+}+(1-p)\,\ket{\phi_-}^{AB}\bra{\phi_-}, \\ 
\ket{\phi_\pm}^{AB} & = \frac{1}{\sqrt{2}}\Bigl(\ket{0}^A\ket{vac}^B+\ket{1}^A\ket{\pm\alpha}^B\Bigr),
\end{aligned}
\end{equation}
which contains three qumode states $\ket{vac}^B$ and $\ket{\pm\alpha}^B$ with $\alpha\in{\mathbb R}$ (see figure \ref{fig:3x2state} for a visualization).
\begin{figure}[ht]
\begin{center}
\includegraphics[width=10cm]{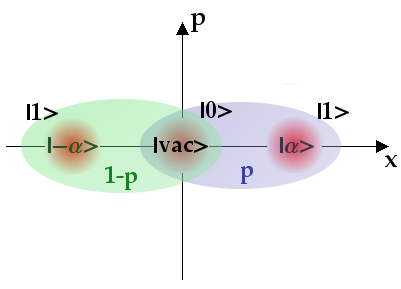}
\caption[Visualization of the HE state $\hat{\rho}^{AB}$, see \eqref{eq:2x3Mstate1}, in phase space.]{Visualization of the HE state $\hat{\rho}^{AB}$, see \eqref{eq:2x3Mstate1}, in phase space. The blue region corresponds to the first pure state in the convex combination with probability $p$, while the green region represents the pure state obtained with probability $1-p$. The additional $\ket{0}$ and $\ket{1}$ vectors denote the qubit states associated with the qumode states $\ket{vac}$ and $\ket{\pm\alpha}$.}
\label{fig:3x2state}
\end{center}
\end{figure}
Perform a Gram-Schmidt process and obtain a qubit-qutrit entangled state in ${\mathcal H}_2^A\otimes{\mathcal H}_3^B$. The concurrence does not work for states of this kind anymore. The pure states in the convex combination of \eqref{eq:2x3Mstate1}, after the inverse Gram-Schmidt process in an orthonormal basis, look like
\begin{align}
\ket{\phi_+}^{AB} &=  \frac{1}{\sqrt{2}}\Bigl(\ket{0}^A\ket{0}^B+x\ket{1}^A\ket{0}^B+\sqrt{1-x^2}\ket{1}^A\ket{1}^B\Bigr),\\
\ket{\phi_-}^{AB} & = \frac{1}{\sqrt{2}}\Bigl(\ket{0}^A\ket{0}^B+x\ket{1}^A\ket{0}^B-x^2\sqrt{1-x^2}\ket{1}^A\ket{1}^B+\sqrt{1-x^2-x^4+x^6}\ket{1}^A\ket{2}^B\Bigr),
\end{align}
with $x=\mathrm{e}^{-\frac{1}{2}|\alpha|^2}$. This time entanglement quantification is tried with the logarithmic negativity $E_N$. As, for general bipartite mixed states with higher dimension than $2\times 2$, it is basically the only actually calculable monotone, it is also the most important one for mixed, effectively DV HE states. The analytic expression for $E_N(\hat{\rho}^{AB})$ is very long. It is omitted here. The graphical version can be observed in figure \ref{fig:3x2LogNeg}.
\begin{figure}[ht]
\begin{center}
\includegraphics[width=8cm]{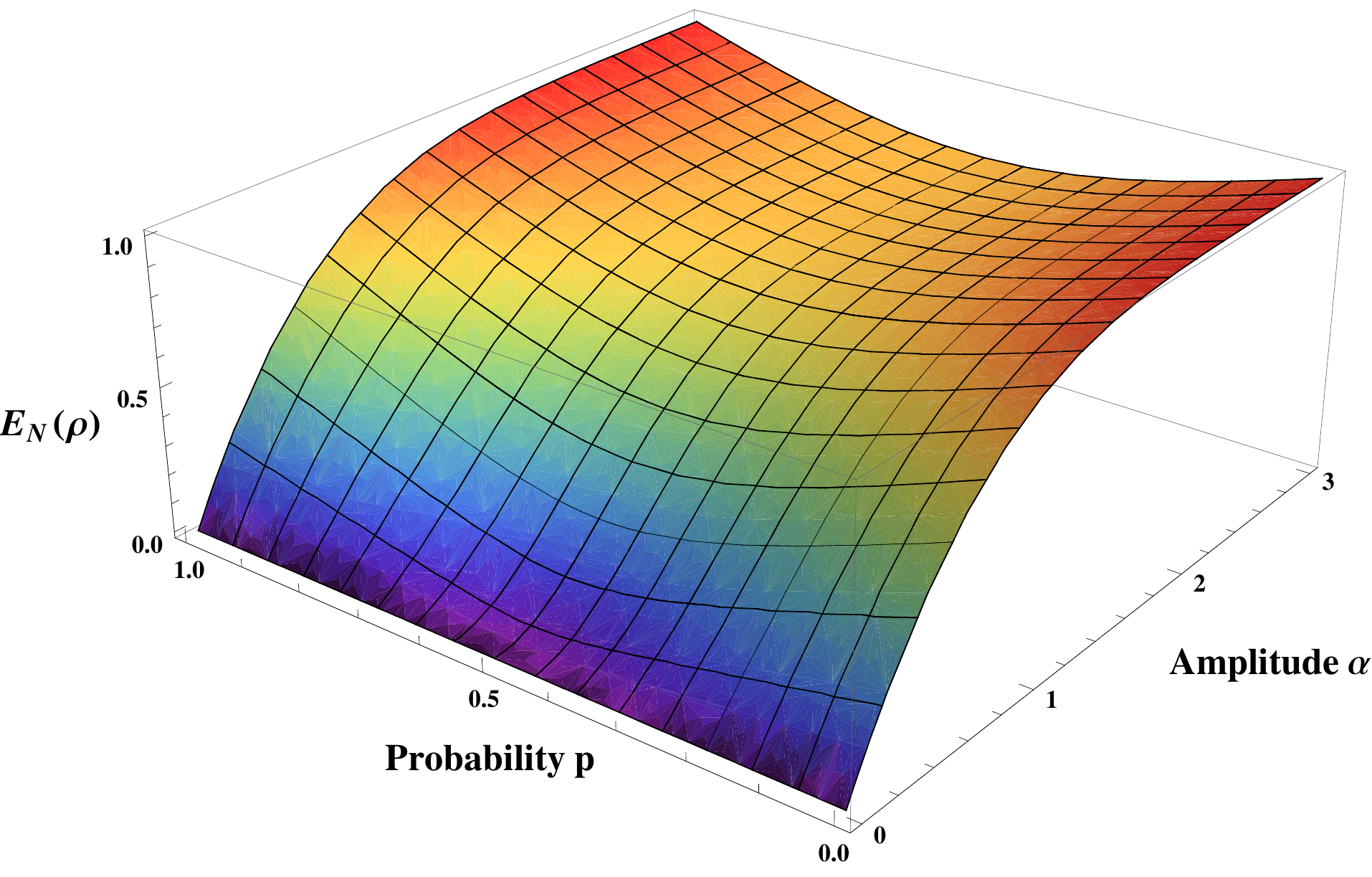}
\caption[Logarithmic negativity of the effective qubit-qutrit HE state $\hat{\rho}^{AB}$, \eqref{eq:2x3Mstate1}, as a function of probability $p$ and amplitude $\alpha$.]{Logarithmic negativity of the effective qubit-qutrit HE state $\hat{\rho}^{AB}$, \eqref{eq:2x3Mstate1}, as a function of probability $p$ and amplitude $\alpha$. For given $\alpha$ maximal mixing $p=\frac{1}{2}$ yields the lowest entanglement. Furthermore, higher $\alpha$ also results in greater entanglement. For $p=0\lor1$ and $\alpha\rightarrow\infty$ the state becomes maximally entangled.}
\label{fig:3x2LogNeg}
\end{center}
\end{figure}
It can be seen that for given $\alpha$ the lowest entanglement is obtained for maximal mixing $p=\frac{1}{2}$. Furthermore, the higher $\alpha$, the more entangled the state is. This is what has been expected. Finally, for $p=0\lor1$ and $\alpha\rightarrow\infty$ the state approaches maximal entanglement.

Of course, the discussion of this state does not bring much new insight. However, it is a good demonstration of the quantification of HE with aid of the important logarithmic negativity.

\subsection{A State containing \texorpdfstring{$N\times d$}{Nxd} Qumode States}
To round off this subsection of mixed, but effectively DV HE states, a state which actually contains $N\times d$ qumode states shall be briefly discussed, namely the state
\begin{equation}\label{eq:2x4Mstate1}
\begin{aligned}
\hat{\rho}^{AB} &= p\,\ket{\phi_1}^{AB}\bra{\phi_1}+(1-p)\,\ket{\phi_2}^{AB}\bra{\phi_2}, \\
\ket{\phi_1}^{AB} & = \frac{1}{\sqrt{2}}\Bigl(\ket{0}^A\ket{\alpha}^B+\ket{1}^A\ket{-\alpha}^B\Bigr), \\
\ket{\phi_2}^{AB} & = \frac{1}{\sqrt{2}}\Bigl(\ket{0}^A\ket{i\alpha}^B+\ket{1}^A\ket{-i\alpha}^B\Bigr),
\end{aligned}
\end{equation}
which contains $N\times d=2\times2=4$ qumode states $\ket{\pm\alpha}^B$ and $\ket{\pm i\alpha}^B$. Assume $\alpha\in{\mathbb R}$. Figure \ref{fig:4x2state} provides a graphical illustration of the state.
\begin{figure}[ht]
\begin{center}
\includegraphics[width=10cm]{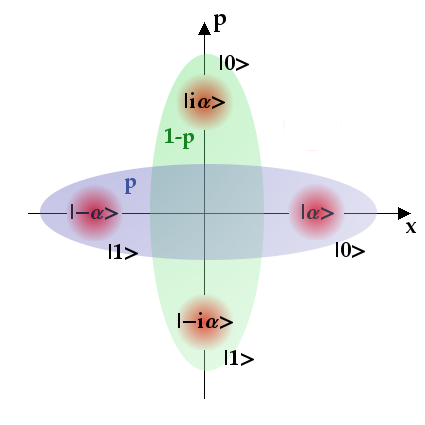}
\caption[Phase space illustration of the HE state $\hat{\rho}^{AB}$ corresponding to equation \eqref{eq:2x4Mstate1}.]{Phase space illustration of the HE state $\hat{\rho}^{AB}$ corresponding to equation \eqref{eq:2x4Mstate1}. The blue and green regions correspond to the first and second pure states $\ket{\phi_1}^{AB}$ and $\ket{\phi_2}^{AB}$ obtained with probability $p$ and $1-p$. The $\ket{0}$ and $\ket{1}$ state vectors denote the qubit states affiliated to the qumode states.}
\label{fig:4x2state}
\end{center}
\end{figure}
Of course again an inverse Gram-Schmidt process could be performed, which would yield a $8\times8$ density matrix in an effective ${\mathcal H}^A_2\otimes{\mathcal H}^B_4$ Hilbert space. The appropriate measure to be calculated would again be the negativity. Unfortunately, already for this $8\times8$ density matrix, the calculation of the eigenvalues of the PT state is not trivial. Since the negativity has already been calculated in the previous subsection and as the focus does not rest on heavy eigenvalue computations, this time witnessing with aid of the SV criteria is discussed instead.  

Try to witness entanglement with the well-known SV determinant $s_1$.
\begin{equation}
\begin{aligned}
s_1 & =\begin{vmatrix} 
1 & \braket{\hat{b}^\dagger} & \braket{\hat{a}\hat{b}^\dagger}  \\
\braket{\hat{b}} & \braket{\hat{b}^\dagger\hat{b}} & \braket{\hat{a}\hat{b}^\dagger\hat{b}}  \\
\braket{\hat{a}^\dagger\hat{b}} & \braket{\hat{a}^\dagger\hat{b}^\dagger\hat{b}} & \braket{\hat{a}^\dagger\hat{a}\hat{b}^\dagger\hat{b}}
\end{vmatrix} \\
& =\begin{vmatrix} 
1 & \frac{\mathrm{e}^{-2\alpha^2}}{2} & \frac{\alpha\mathrm{e}^{-2\alpha^2}}{2}(p+i(1-p))  \\
\frac{\mathrm{e}^{-2\alpha^2}}{2} & \frac{1}{2} & \frac{-\alpha}{2}(p+i(1-p))  \\
\frac{\alpha\mathrm{e}^{-2\alpha^2}}{2}(p-i(1-p)) & \frac{-\alpha}{2}(p-i(1-p)) & \frac{\alpha^2}{2}
\end{vmatrix} \\
& = \frac{\alpha^2}{2}\biggl[p(1-p)-\mathrm{e}^{-4\alpha^2}\Bigl(1-\frac{3p}{2}(1-p)\Bigr)\biggr].
\end{aligned}
\end{equation}
The graphical evaluation of this result is shown in figure \ref{fig:DetQubitQuquad}.
\begin{figure}[ht]
\subfloat{\includegraphics[width=0.5\textwidth]{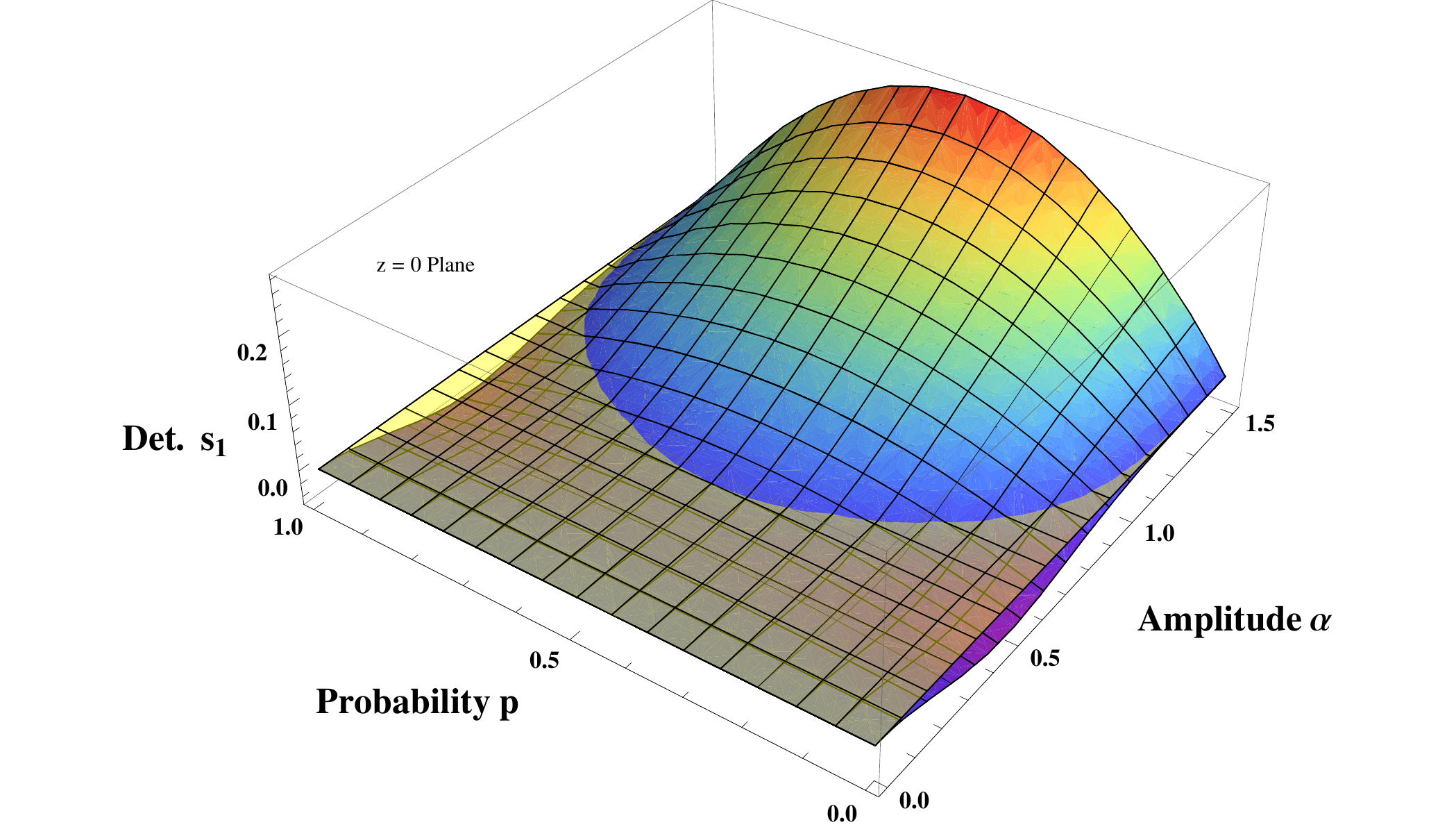}}\hfill
\subfloat{\includegraphics[width=0.5\textwidth]{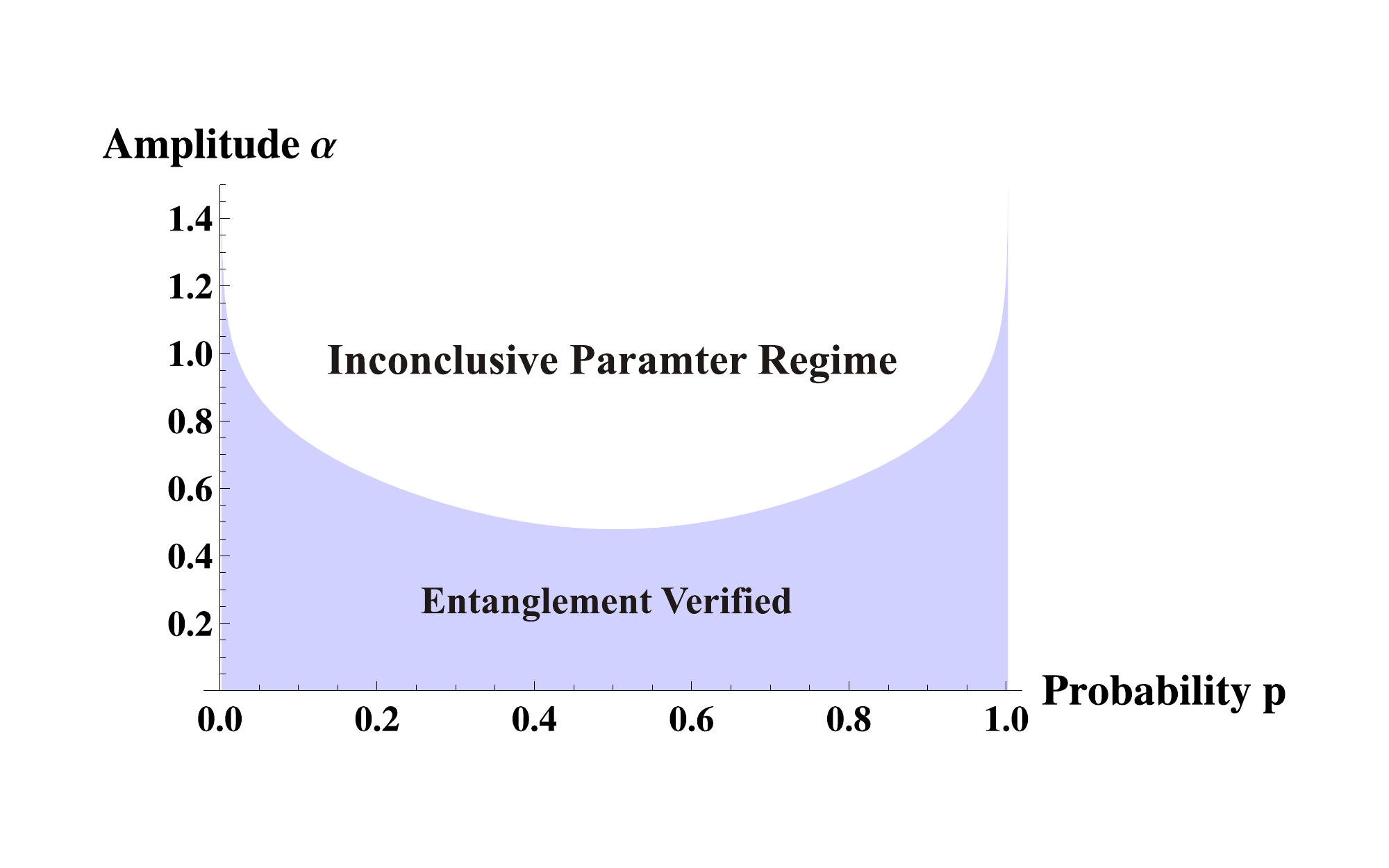}}\hfill
\caption[SV determinant $s_1$ evaluated for the state of equation \eqref{eq:2x4Mstate1}, and parameter regime of entanglement witnessing.]{The left diagram displays the SV determinant $s_1$ evaluated for the state of equation \eqref{eq:2x4Mstate1}. While for small $\alpha\neq0$ entanglement can be always detected, large $\alpha$ causes the determinant to fail. Note that the yellow plane denotes zero. On the right hand side the two regimes are plotted. For $p=0\lor1$ the state is pure and the determinant witnesses its entanglement for any $\alpha\neq0$.}
\label{fig:DetQubitQuquad}
\end{figure}
It shows that for small $\alpha\neq0$ entanglement is verified, whereas for large $\alpha$ the determinant fails. For $\alpha=0$ the state is separable and hence nothing can be detected. Furthermore, the more mixing occurs, i.e. $p\rightarrow\frac{1}{2}$, the harder to observe the entanglement. For $p=0\lor1$ the state is pure and the determinant witnesses its entanglement for any $\alpha\neq0$.

Concluding, it can be stated once again that the determinant $s_1$ is a great tool for the detection of HE, even though it has been effective DV HE in this case. With this statement the section is to be finished and finally the remarkable true hybrid entanglement is investigated.

\section{True Hybrid Entanglement}\label{sec:mixed2}
Finally, the third class of bipartite HE states is considered, the \textit{truly hybrid entangled} states. Recall their definition corresponding to equation \eqref{eq:HEdef} with $N=\infty$:
\begin{equation}\label{eq:THEdef}
\begin{aligned}
\hat{\rho}^{AB} &=\sum_{n=1}^\infty p_n\,\ket{\psi_n}_{AB}\bra{\psi_n}\,,\qquad p_n>0\;\forall \;n\,,\;\;\;\sum_{n=1}^\infty \;p_n=1, \\
\ket{\psi_n} &=\sum_{m=0}^{d-1} c_{nm}\ket{m}^A\ket{\psi_{nm}}^B\,,\qquad c_{nm}\in\mathbb C\,,\;\;\;\sum_{m=0}^{d-1} \;|c_{nm}|^2=1.
\end{aligned} 
\end{equation}
As can be seen from these equations, truly HE states possess an infinite number of qumode states $\ket{\psi_{nm}}^B$. Therefore, the Gram-Schmidt process does not work anymore and the states have to stay in a Hilbert space of the form ${\mathcal H}^A_d\otimes{\mathcal H}^B_\infty$. Hence, they are not effectively DV, but instead real combined $\text{DV}\otimes\text{CV}$ states and therefore \textit{truly} hybrid. This is where the name actually comes from. Unfortunately, this true "hybridness" has the effect that the states live in an overall infinite-dimensional Hilbert space in the non-Gaussian regime. This makes exact entanglement quantification impossible. Hence, the focus is only on the \textit{detection} of true hybrid entanglement.

In this section two examples for truly HE states are discussed, one more physically relevant and one artificially constructed state. Finally, general issues regarding entanglement quantification and witnessing in mixed states with an infinite number of pure state projectors are discussed.

\subsection{Noisy Qubit - Qumode Entanglement}\label{subsecB}
Once again, consider the state 
\begin{equation}\label{eq:noisyQbQm1}
\ket{\psi}^{AB}=\frac{1}{\sqrt{2}}\Bigl(\ket{0}^A\ket{\alpha}^B+\ket{1}^A\ket{-\alpha}^B\Bigr).
\end{equation}
In subsection \ref{subsecA}, this state was discussed when being transmitted through a one-sided amplitude damping channel of the form ${\mathbb 1}^A\otimes\Upsilon_{loss}^B$. Now the channel is taken as not only lossy but also afflicted with thermal noise. Writing this out, 
\begin{equation}\label{eq:noisyQbQm2}
\hat{\rho}'{}^{AB}=({\mathbb 1}^A\otimes\Upsilon_{thermal}^B)\ket{\psi}^{AB}\bra{\psi}
\end{equation}
is the state to be investigated. The photon noise channel can be modelled in a similar way as the amplitude damping channel. The coupled environment is taken not simply in the vacuum state but in a thermal state (see figure \ref{fig:NoiseCh} for a visualization).
\begin{figure}[ht]
\begin{center}
\includegraphics[width=12cm]{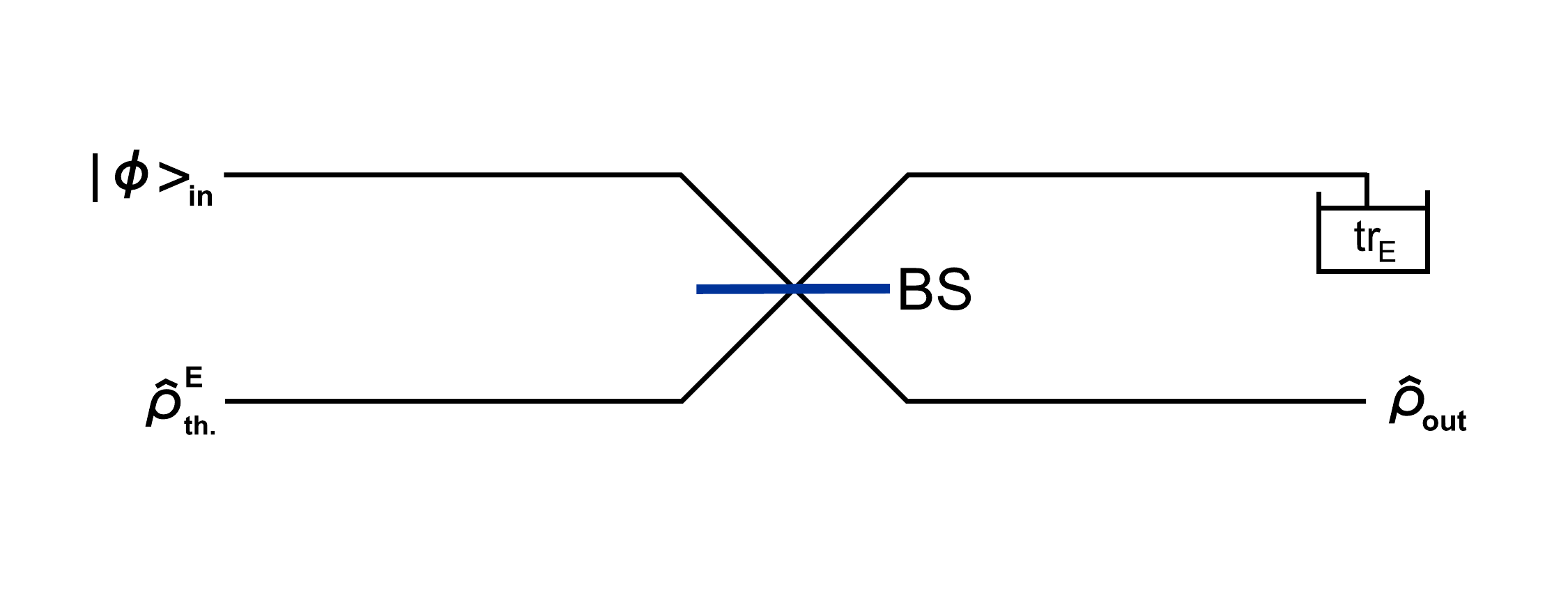}
\caption[Visualization of the photon noise channel.]{Visualization of the photon noise channel. The environment mode is in a thermal state and coupled via a beam splitter to the input state to be transmitted. Subsequently the environment mode is traced out and the decohered output state is obtained.}
\label{fig:NoiseCh}
\end{center}
\end{figure}
It is seen that the channel can again be written with aid of an ancilla Hilbert space (the environment E), and joint unitary evolution on the overall Hilbert space. Finally, a trace operation over the ancilla Hilbert space has to be performed. Hence,
\begin{equation}
\hat{\rho}'{}^{AB}=\mathrm{tr}_E[\hat{U}^{BE}\;\bigl( \ket{\phi}^{AB}\bra{\phi}\otimes\hat{\rho}^E_{thermal}\bigr) \;\hat{U}^{{BE}^\dagger}],
\end{equation}
with
\begin{itemize}
\item beam splitter unitary $\hat{U}^{BE}=e^{\theta(\hat{a}_E^\dagger \hat{a}_B-\hat{a}_B^\dagger \hat{a}_E)}$,
\item environment thermal state $\hat{\rho}^E_{thermal}=\sum_{n=0}^\infty \frac{\braket{n_{th.}}^n}{(1+\braket{n_{th.}})^{n+1}}\,\ket{n}_{E}\bra{n}$ \cite{Scully},
\item mean thermal photon number $\braket{n_{th.}}$, 
\item and beam splitter transmissivity $\eta=\cos^2\theta$.
\end{itemize}
The two main parameters characterizing the channel are the mean thermal photon number $\braket{n_{th.}}$ and the beam splitter transmissivity $\eta$. For $\braket{n_{th.}}=0$, the photon loss channel is obtained, which is therefore just a limiting case of this channel. 

For the output, 
\begin{equation}\label{eq:noisyQbQm3}
\begin{aligned}
\hat{\rho}'{}^{AB} &=({\mathbb 1}^A\otimes\Upsilon_{thermal}^B)\ket{\psi}^{AB}\bra{\psi} \\
&=\mathrm{tr}_E[\hat{U}^{BE}\;\bigl( \ket{\phi}^{AB}\bra{\phi}\otimes\hat{\rho}^E_{thermal}\bigr) \;\hat{U}^{{BE}^\dagger}] \\
&=\frac{1}{2}\sum_{n=0}^\infty \rho_n^{th.} \sum_{k,l=0}^n f_{nk}(\eta)f_{nl}(\eta)\biggl(A_{nkl}^{\alpha\alpha}(\eta)\;\ket{0}^A\bra{0}\,\otimes\,\hat{a}^{\dagger^k}\ket{\sqrt{\eta}\alpha}^B\bra{\sqrt{\eta}\alpha}\hat{a}^l  \\
&+ A_{nkl}^{-\alpha-\alpha}(\eta)\;\ket{1}^A\bra{1}\,\otimes\,\hat{a}^{\dagger^k}\ket{-\sqrt{\eta}\alpha}^B\bra{-\sqrt{\eta}\alpha}\hat{a}^l   \\
&+ A_{nkl}^{\alpha-\alpha}(\eta)\;\ket{0}^A\bra{1}\,\otimes\,\hat{a}^{\dagger^k}\ket{\sqrt{\eta}\alpha}^B\bra{-\sqrt{\eta}\alpha}\hat{a}^l \\
&+ A_{nkl}^{-\alpha\alpha}(\eta)\;\ket{1}^A\bra{0}\,\otimes\,\hat{a}^{\dagger^k}\ket{-\sqrt{\eta}\alpha}^B\bra{\sqrt{\eta}\alpha}\hat{a}^l\biggr)
\end{aligned}
\end{equation}
is obtained. $\hat{a}$ and $\hat{a}^\dagger$ are the mode operators of system B. $\rho_n^{th.}$ denotes the thermal photon distribution $\frac{\braket{n_{th.}}^n}{(1+\braket{n_{th.}})^{n+1}}$, and $f_{nk}(\eta)$ and $A_{nkl}^{\alpha\beta}(\eta)$ are defined as
\begin{align}
f_{nk}(\eta)&:=\frac{1}{\sqrt{n!}}{\binom n k}\sqrt{\eta}^{n-k}(-\sqrt{1-\eta})^k, \\
A_{nkl}^{\alpha\beta}(\eta)&:=\braket{\sqrt{1-\eta}\beta|\hat{a}^{{}^{n-k}}\hat{a}^{\dagger^{n-l}}|\sqrt{1-\eta}\alpha}.
\end{align}
For the calculation, the following relations have been exploited.
\begin{align} \label{equse1}
\hat{U}^{{BE}^\dagger}\hat{U}^{BE} & =\mathbb{1}, \\ \label{equse2}
\hat{U}^{BE}(\hat{a}^{E^\dagger})^n\hat{U}^{{BE}^\dagger} & =(\sqrt{\eta}\hat{a}^{E^\dagger}-\sqrt{1-\eta}\hat{a}^{B^\dagger})^n, \\ \label{equse3} \hat{U}^{BE}D^B(\alpha)\hat{U}^{{BE}^\dagger} & =D^B(\sqrt{\eta}\alpha)D^E(\sqrt{1-\eta}\alpha), \\ \label{equse4}
\hat{U}^{BE}\ket{0}^B\ket{0}^E & =\ket{0}^B\ket{0}^E.
\end{align}
It can be infered from equation \eqref{eq:noisyQbQm3} that the state $\hat{\rho}'{}^{AB}$ clearly is truly HE, as it contains an infinite number of qumode states $\{\hat{a}^{\dagger^k}\ket{\pm\sqrt{\eta}\alpha}_B:k=0,1,\ldots,\infty\}$. Hence, \textit{true hybrid qubit-qumode entanglement} is obtained. 

Furthermore, writing the state in this form, \eqref{eq:noisyQbQm3}, illustrates the effects of the thermal photon noise channel in a very concrete way. On the one hand, the damping effect due to the beam splitter is clearly visible in form of $\sqrt{\eta}$ in the states $\hat{a}^{\dagger^k}\ket{\pm\sqrt{\eta}\alpha}$. On the other hand, there is not only damping, but also thermal photon noise, which becomes manifest in the creation operators $\hat{a}^{\dagger^k}$ in the states $\hat{a}^{\dagger^k}\ket{\pm\sqrt{\eta}\alpha}$. Thermal photons "leak into the system" and are "created" in the damped coherent states. Finally, it is clearly visible, how each term of $\hat{\rho}^E_{thermal}=\sum_{n=0}^\infty \rho_n^{th.}\,\ket{n}_{E}\bra{n}$ results in the creation of maximally $n$ noise photons in the coherent states. However, as descriptive \eqref{eq:noisyQbQm3} is, as inapplicable it is for further calculations. To perform entanglement witnessing in this state, the moments required for the SV determinant $s_1$ have to be computed. Unfortunately, when this is tried, making use of the state as written in equation \eqref{eq:noisyQbQm3}, intractable infinite sums are obtained whose convergence behavior is impossible to be worked out exactly. Of course, truncation at some very large $n$ could be performed, which would probably result in very accurate outcomes. However, such a procedure is opposed to the actual intention of analyzing true HE, since a truncated state is not truly hybrid entangled anymore. This is a point which makes the investigation of true HE particularly challenging. Infinite sums or integrals emerge, which have to be calculated \textit{exactly}.

At least in this case, there is a way out of the dilemma. Instead of utilizing $\hat{\rho}^E_{thermal}=\sum_{n=0}^\infty \rho_n^{th.}\,\ket{n}_{E}\bra{n}$ in the Fock basis, make use of the coherent state basis. Sufficiently classical states $\hat{\rho}$ can be written as $\hat{\rho}=\int_{\mathbb C}d^2\alpha\, P_{\hat{\rho}}(\alpha,\alpha^\ast)\ket{\alpha}\bra{\alpha}$, where $P_{\hat{\rho}}(\alpha,\alpha^\ast)$ denotes the Glauber-Sudarshan P-representation of $\hat{\rho}$. The thermal state is of course sufficiently classical and hence its P-representation exists \cite{Scully}:
\begin{equation}
P_{\hat{\rho}_{thermal}}(\alpha,\alpha^\ast)=\frac{1}{\pi\braket{n_{th.}}}\mathrm{e}^{-\frac{|\alpha|^2}{\braket{n_{th.}}}}.
\end{equation}
Therefore,
\begin{equation}
\hat{\rho}_{thermal}=\frac{1}{\pi\braket{n_{th.}}}\int_{\mathbb C} d^2\alpha\,\mathrm{e}^{-\frac{|\alpha|^2}{\braket{n_{th.}}}}\ket{\alpha}\bra{\alpha}.
\end{equation}
Expressing the thermal state in this form, the action of the thermal channel on an element $\ket{\alpha}\bra{\beta}$ looks like
\begin{equation}
\begin{aligned}
\Upsilon_{thermal}(\ket{\alpha}^B\bra{\beta}) &=\frac{1}{\pi\braket{n_{th.}}}\mathrm{tr}_E[\hat{U}^{BE}\;\bigl( \ket{\alpha}^B\bra{\beta}\otimes\int_{\mathbb C} d^2\gamma\,\mathrm{e}^{-\frac{|\gamma|^2}{\braket{n_{th.}}}}\ket{\gamma}^E\bra{\gamma}\bigr) \;\hat{U}^{{BE}^\dagger}] \\
&= \frac{1}{\pi\braket{n_{th.}}}\int_{\mathbb C} d^2\gamma\,\mathrm{e}^{-\frac{|\gamma|^2}{\braket{n_{th.}}}}\mathrm{tr}_E[\hat{U}^{BE}\;\bigl( \ket{\alpha}^B\bra{\beta}\otimes\ket{\gamma}^E\bra{\gamma}\bigr) \;\hat{U}^{{BE}^\dagger}].
\end{aligned}
\end{equation}
Define $\hat{\chi}:=\hat{U}^{BE}\;\bigl( \ket{\alpha}^B\bra{\beta}\otimes\ket{\gamma}^E\bra{\gamma}\bigr) \;\hat{U}^{{BE}^\dagger}$ and calculate
\begin{equation}
\begin{aligned}
\hat{\chi} & =\hat{U}^{BE}\;\bigl( \ket{\alpha}^B\bra{\beta}\otimes\ket{\gamma}^E\bra{\gamma}\bigr) \;\hat{U}^{{BE}^\dagger} \\
& = \hat{U}^{BE}\hat{D}^B(\alpha)\hat{D}^E(\gamma)\biggl( \ket{0}^B\bra{0}\otimes\ket{0}^E\bra{0}\biggr)\hat{D}^{E^\dagger}(\gamma)\hat{D}^{B^\dagger}(\beta)\hat{U}^{{BE}^\dagger} \\
& = (\hat{U}^{BE}\hat{D}^B(\alpha)\hat{U}^{{BE}^\dagger})\,(\hat{U}^{BE}\hat{D}^E(\gamma)\hat{U}^{{BE}^\dagger})\,\hat{U}^{BE}\biggl( \ket{0}^B\bra{0} \\
& \quad\otimes\ket{0}^E\bra{0}\biggr)\hat{U}^{{BE}^\dagger}(\hat{U}^{BE}\hat{D}^{E^\dagger}(\gamma)\hat{U}^{{BE}^\dagger})(\hat{U}^{BE}\hat{D}^{B^\dagger}(\beta)\hat{U}^{{BE}^\dagger}) \\
& =\hat{D}^B(\sqrt{\eta}\alpha)\hat{D}^E(\sqrt{1-\eta}\alpha)\hat{D}^E(\sqrt{\eta}\gamma)\hat{D}^B(-\sqrt{1-\eta}\gamma)\biggl( \ket{0}^B\bra{0} \\ 
& \quad\otimes\ket{0}^E\bra{0}\biggr)\hat{D}^{B^\dagger}(-\sqrt{1-\eta}\gamma)\hat{D}^{E^\dagger}(\sqrt{\eta}\gamma)\hat{D}^{E^\dagger}(\sqrt{1 -\eta}\beta)\hat{D}^{B^\dagger}(\sqrt{\eta}\beta) \\
& =\ket{\sqrt{\eta}\alpha-\sqrt{1-\eta}\gamma}^B\bra{\sqrt{\eta}\beta-\sqrt{1-\eta}\gamma}\otimes\ket{\sqrt{1-\eta}\alpha+\sqrt{\eta}\gamma}^E\bra{\sqrt{1-\eta}\beta +\sqrt{\eta}\gamma}, 
\end{aligned}
\end{equation}
where once again the relations \eqref{equse1}, \eqref{equse3}, \eqref{equse4} have been exploited. Additionally use has been made of $\hat{U}^{BE}D^E(\gamma)\hat{U}^{{BE}^\dagger} =D^E(\sqrt{\eta}\gamma)D^B(-\sqrt{1-\eta}\gamma)$. The derived expression for $\hat{\chi}$ yields 
\begin{equation}
\begin{aligned}
\Upsilon_{thermal}(\ket{\alpha}^B\bra{\beta}) & = \frac{1}{\pi\braket{n_{th.}}}\int_{\mathbb C} d^2\gamma\,\mathrm{e}^{-\frac{|\gamma|^2}{\braket{n_{th.}}}}\braket{\sqrt{1-\eta}\beta +\sqrt{\eta}\gamma|\sqrt{1-\eta}\alpha+\sqrt{\eta}\gamma} \\
& \quad\cdot\; \ket{\sqrt{\eta}\alpha-\sqrt{1-\eta}\gamma}^B\bra{\sqrt{\eta}\beta-\sqrt{1-\eta}\gamma}.
\end{aligned}
\end{equation}
Making use of this auxiliary calculation, the overall state $\hat{\rho}'{}^{AB}=({\mathbb 1}^A\otimes\Upsilon_{thermal}^B)\ket{\psi}^{AB}\bra{\psi}$ with $\ket{\psi}^{AB}=\frac{1}{\sqrt{2}}\Bigl(\ket{0}^A\ket{\alpha}^B+\ket{1}^A\ket{-\alpha}^B\Bigr)$ is given by
\begin{equation}\label{eq:noisyQbQm4}
\begin{aligned}
\hat{\rho}'{}^{AB} & =({\mathbb 1}^A\otimes\Upsilon_{thermal}^B)\ket{\psi}^{AB}\bra{\psi} \\
& = \frac{1}{2\pi\braket{n_{th.}}}\int_{\mathbb C} d^2\gamma\,\mathrm{e}^{-\frac{|\gamma|^2}{\braket{n_{th.}}}}\biggl( \\
&\quad\qquad\quad\,\qquad\,  \ket{0}^A\bra{0}\otimes\ket{\sqrt{\eta}\alpha-\sqrt{1-\eta}\gamma}^B\bra{\sqrt{\eta}\alpha-\sqrt{1-\eta}\gamma}  \\
&\quad+\qquad\quad\,\,\quad \ket{1}^A\bra{1}\otimes\ket{-\sqrt{\eta}\alpha-\sqrt{1-\eta}\gamma}^B\bra{-\sqrt{\eta}\alpha-\sqrt{1-\eta}\gamma}  \\
&\quad+ \tilde{A}^{\alpha-\alpha\gamma}(\eta)\ket{0}^A\bra{1}\otimes\ket{\sqrt{\eta}\alpha-\sqrt{1-\eta}\gamma}^B\bra{-\sqrt{\eta}\alpha-\sqrt{1-\eta}\gamma}  \\ 
&\quad+ \tilde{A}^{-\alpha\alpha\gamma}(\eta)\ket{1}^A\bra{0}\otimes\ket{-\sqrt{\eta}\alpha-\sqrt{1-\eta}\gamma}^B\bra{\sqrt{\eta}\alpha-\sqrt{1-\eta}\gamma}  \biggr), 
\end{aligned}
\end{equation}
with
\begin{equation}
\begin{aligned}
\tilde{A}^{\alpha\beta\gamma}(\eta):&=\braket{\sqrt{1-\eta}\beta +\sqrt{\eta}\gamma|\sqrt{1-\eta}\alpha+\sqrt{\eta}\gamma} \\
&=\mathrm{exp}\Bigl[-\frac{1}{2}|\sqrt{1-\eta}\beta +\sqrt{\eta}\gamma|^2-\frac{1}{2}|\sqrt{1-\eta}\alpha+\sqrt{\eta}\gamma|^2 \\
&\quad\;\,\qquad+(\sqrt{1-\eta}\beta^{\ast} +\sqrt{\eta}\gamma^{\ast})(\sqrt{1-\eta}\alpha+\sqrt{\eta}\gamma)\Bigr].
\end{aligned}
\end{equation}
This form of the state $\hat{\rho}'{}^{AB}$, \eqref{eq:noisyQbQm4}, is not as insightful as \eqref{eq:noisyQbQm3}. However, it is mathematically much more controllable. Instead of infinite sums, in this form just integrals which can be easily calculated occur. As an example, consider the moment $\braket{\hat{a}^\dagger\hat{a}\hat{b}^\dagger\hat{b}}$ (actually $\hat{a}$ and $\hat{a}^\dagger$ are the mode operators corresponding to system B and $\hat{b}$ and $\hat{b}^\dagger$ those of system A.):
\begin{equation}
\begin{aligned}
\braket{\hat{a}^\dagger\hat{a}\hat{b}^\dagger\hat{b}}&=\frac{1}{2\pi\braket{n_{th.}}}\int_{\mathbb C} d^2\gamma\,\mathrm{e}^{-\frac{|\gamma|^2}{\braket{n_{th.}}}}\braket{-\sqrt{\eta}\alpha-\sqrt{1-\eta}\gamma|\hat{a}^\dagger\hat{a}|-\sqrt{\eta}\alpha-\sqrt{1-\eta}\gamma} \\
&=\frac{1}{2\pi\braket{n_{th.}}}\int_{\mathbb C} d^2\gamma\,\mathrm{e}^{-\frac{|\gamma|^2}{\braket{n_{th.}}}}\biggl(\eta|\alpha|^2+(1-\eta)|\gamma|^2+\sqrt{\eta(1-\eta)}(\alpha\gamma^\ast+\alpha^\ast\gamma)\biggr) \\
&=\frac{1}{2\pi\braket{n_{th.}}}\int_{\mathbb R}\int_{\mathbb R} dx\,dy\,\mathrm{e}^{-\frac{x^2+y^2}{\braket{n_{th.}}}}\biggl(\eta|\alpha|^2+(1-\eta)(x^2+y^2) \\
& \qquad +\sqrt{\eta(1-\eta)}\bigl(\alpha(x-iy)+\alpha^\ast(x+iy)\bigr)\biggr) \\
&=\frac{\eta|\alpha|^2}{2}+\frac{1-\eta}{\sqrt{\pi\braket{n_{th.}}}}\int_{\mathbb R} dx\,x^2\mathrm{e}^{-\frac{x^2}{\braket{n_{th.}}}} \\
&=\frac{\eta|\alpha|^2}{2}+\frac{1-\eta}{2}\braket{n_{th.}},
\end{aligned}
\end{equation}
where $\gamma=x+iy$ has been substituted. Additionally, some integral identities have been exploited, which are presented in appendix \ref{Apx:Ints}.

The calculation of the other moments in the SV determinant $s_1$ proceeds similarly and finally, 
\begin{equation}
s_{1_{\hat{\rho}'{}^{AB}}}(\alpha,\eta,\braket{n_{th.}})=\frac{1-\eta}{4}\braket{n_{th.}}\biggl(1-\frac{\mathrm{e}^{-4|\alpha|^2}}{2}\biggr)-\frac{\eta |\alpha|^2}{2}\mathrm{e}^{-4|\alpha|^2}
\end{equation}
is obtained. A graphical illustration for $\eta=\frac{2}{3}$ is shown in figure \ref{fig:ThQbQm1}.
\begin{figure}[ht]
\subfloat{\includegraphics[width=0.55\textwidth]{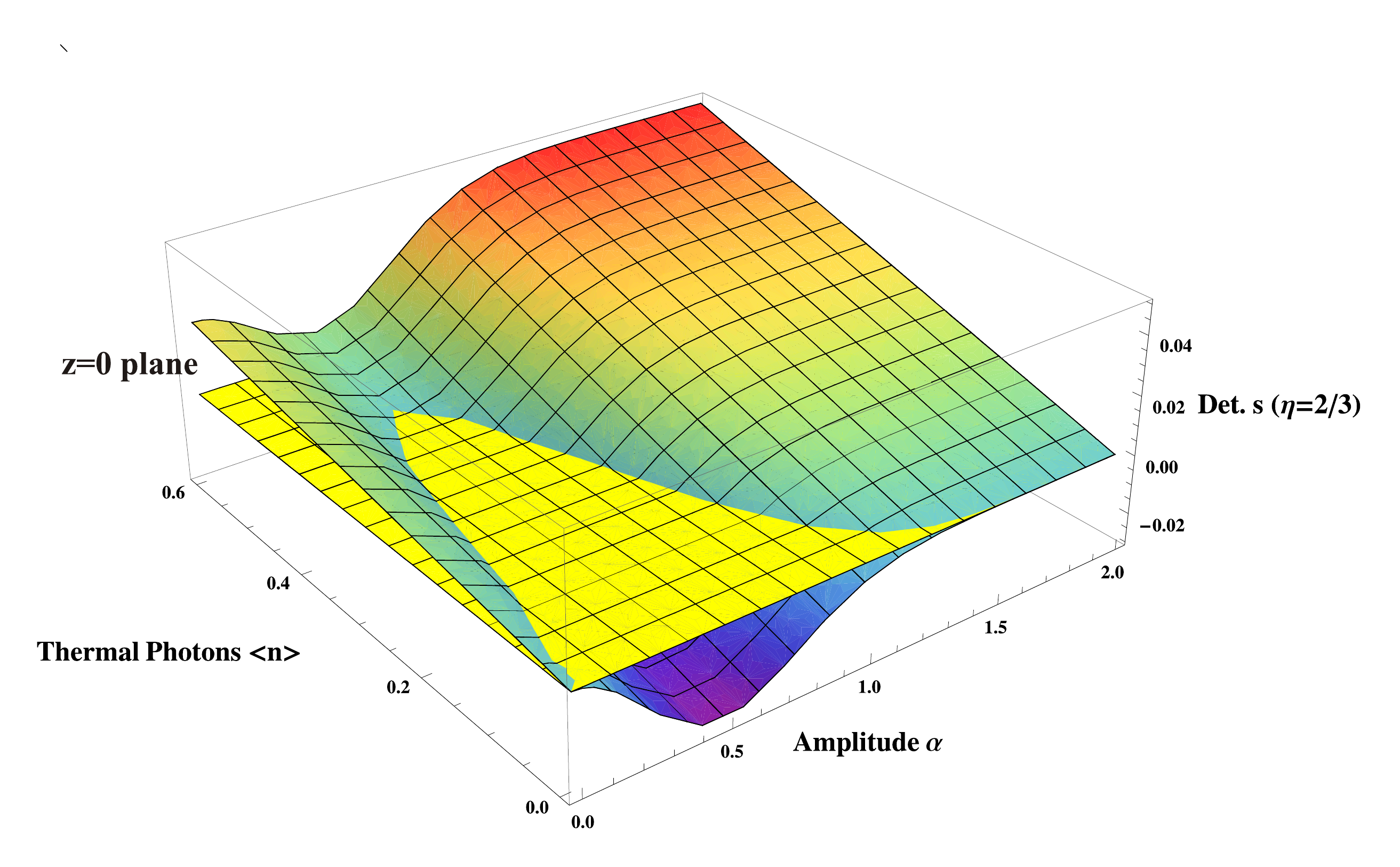}}\hfill
\subfloat{\includegraphics[width=0.45\textwidth]{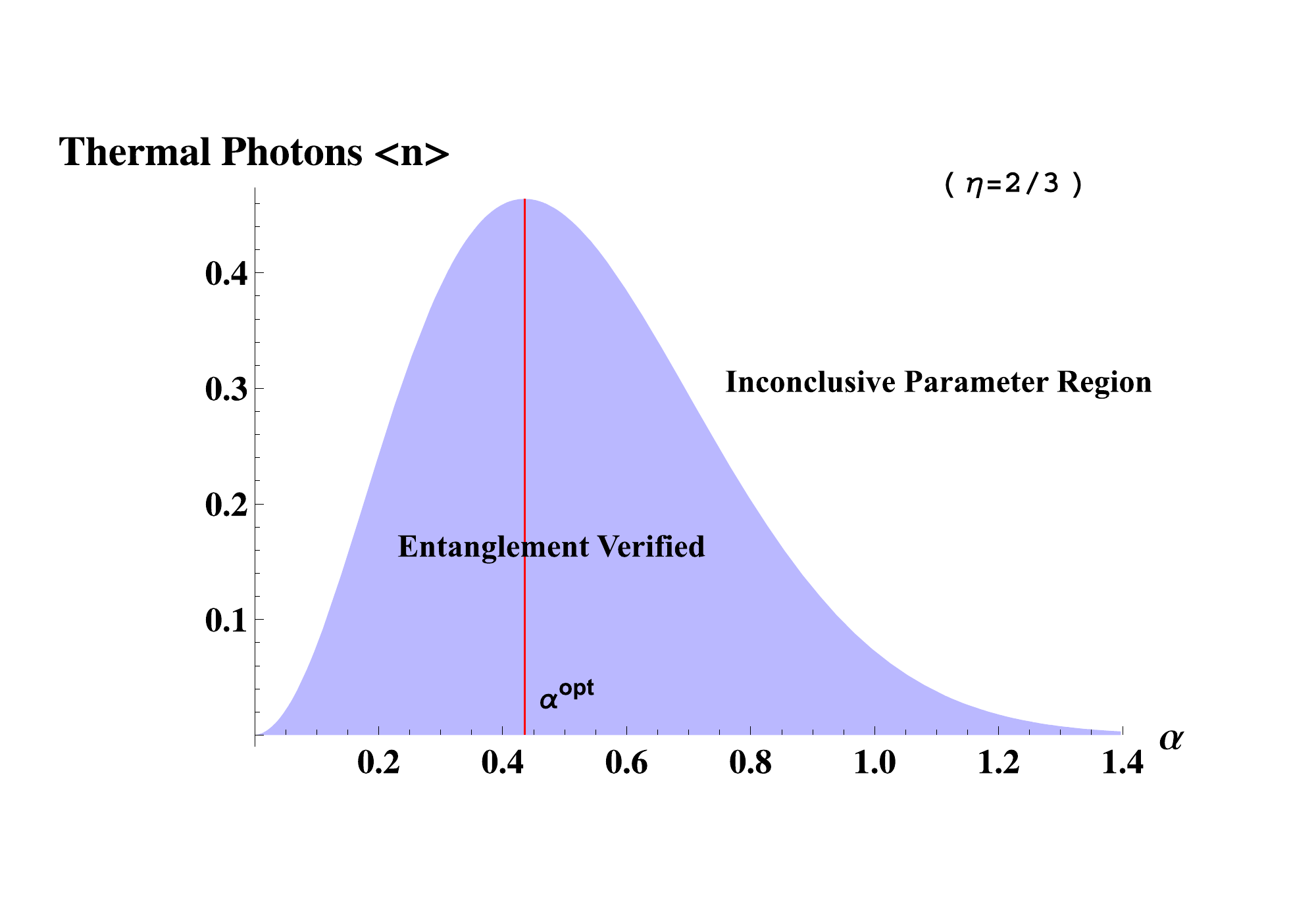}}\hfill
\caption[SV determinant $s_{1_{\hat{\rho}'{}^{AB}}}(\alpha,\eta=\frac{2}{3},{\braket{n_{th.}}})$ for the state $\hat{\rho}'{}^{AB}$ with $\eta=\frac{2}{3}$, and parameter regimes of entanglement witnessing.]{The left hand graph displays the SV determinant $s_{1_{\hat{\rho}'{}^{AB}}}(\alpha,\eta=\frac{2}{3},{\braket{n_{th.}}})$ for the state $\hat{\rho}'{}^{AB}$ with $\eta=\frac{2}{3}$. Without loss of generality $\alpha\in\mathbb R$ has been assumed. On the right hand side the the two regimes are plotted. There can be clearly identified a region, where the determinant is below zero. Hence, true HE can be witnessed. Furthermore, there is again a trade-off behavior, which results in the existence of an optimal $\alpha^{opt}$. This $\alpha^{opt}$ corresponds to the most robust state $\hat{\rho}'{}^{AB}_{\alpha^{opt}}$ regarding entanglement witnessing for fixed $\eta=\frac{2}{3}$ and varying mean thermal photon number $\braket{n_{th.}}$.}
\label{fig:ThQbQm1}
\end{figure}
It can be observed that there clearly is a parameter region in which entanglement can be detected. Also note the trade-off behavior and the optimal $\alpha^{opt}$ which corresponds to the most robust state $\hat{\rho}'{}^{AB}_{\alpha^{opt}}$ regarding entanglement witnessing for fixed $\eta=\frac{2}{3}$ and varying mean thermal photon number $\braket{n_{th.}}$. Furthermore, it is worth pointing out that the witnessed entanglement in this case is actually true HE, in contrast to the entanglement witnessed in previous sections. It can be concluded that the SV determinants provide a suitable tool for the detection of true HE.

When setting $s_{1_{\hat{\rho}'{}^{AB}}}(\alpha,\eta,\braket{n_{th.}})<0$ and solving this equation for $\braket{n_{th.}}$,
\begin{equation}\label{eq:ThQbQmCut}
\braket{n_{th.}}<\frac{4\eta|\alpha|^2}{(1-\eta)(2\mathrm{e}^{4|\alpha|^2}-1)}
\end{equation}
is obtained. For parameters $(\alpha,\eta,\braket{n_{th.}})$ satisfying this inequality, entanglement is detected. Furthermore, the inequality can be used to define a surface. The parameters $(\alpha,\eta,\braket{n_{th.}})_{ent}$ for which entanglement is verified lie below this surface (see figure \ref{fig:ThQbQm2}).
\begin{figure}[ht]
\begin{center}
\includegraphics[width=11cm]{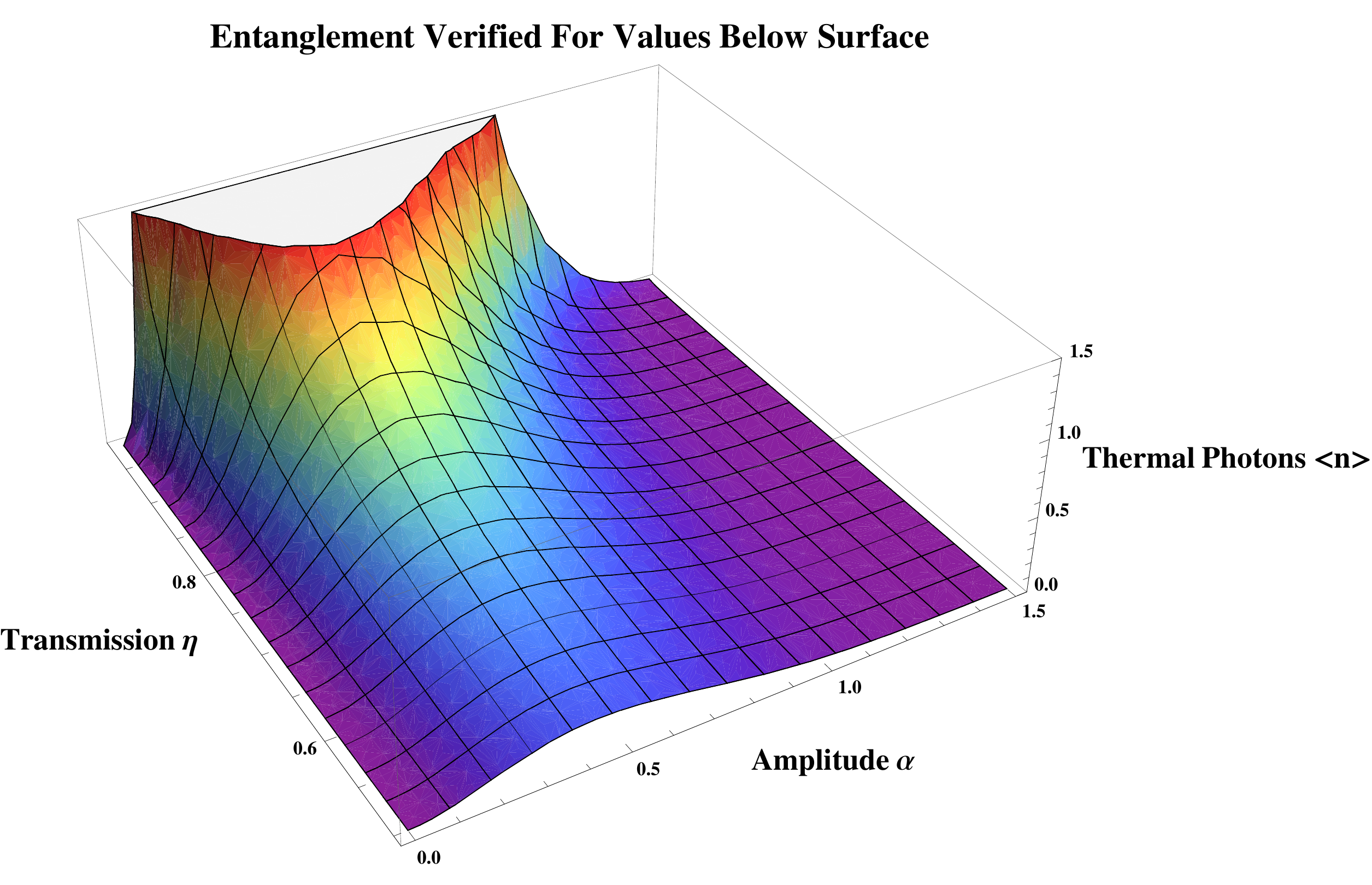}
\caption[Surface defined by equation \eqref{eq:ThQbQmCut} describing entanglement witnessing in the truly HE state $\hat{\rho}'{}^{AB}$ corresponding to equation \eqref{eq:noisyQbQm4}.]{Surface defined by equation \eqref{eq:ThQbQmCut} ($\alpha\in\mathbb R$). For parameter triples $(\alpha,\eta,\braket{n_{th.}})$ lying below it, entanglement is witnessed. Once again, the trade-off behavior and an optimal $\alpha^{opt}$ for which entanglement can be detected for the strongest possible noise can be recognized.}
\label{fig:ThQbQm2}
\end{center}
\end{figure}
However, it is rather cumbersome to read off exact parameters in such a 3D plot. Hence, regions of successful entanglement detection are plotted in figure \ref{fig:ThQbQm3} for different values of $\eta$. 
\begin{figure}[ht]
\begin{center}
\includegraphics[width=11cm]{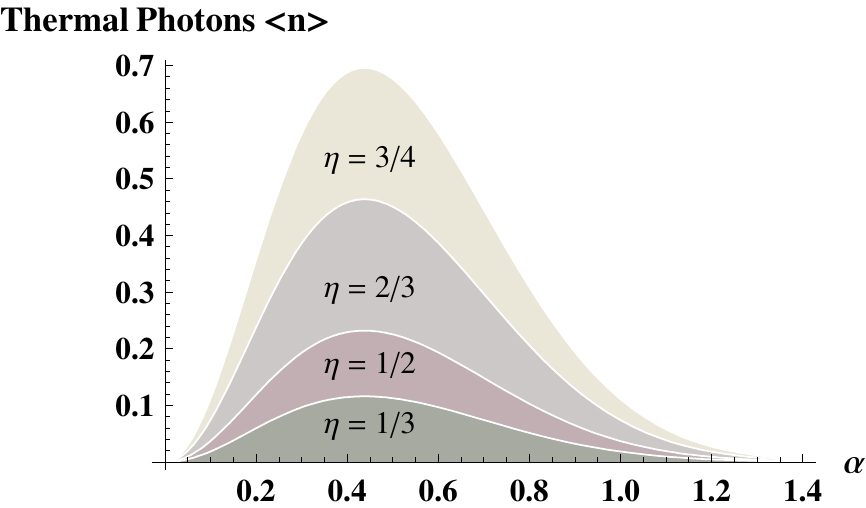}
\caption[Regions of entanglement detection in the truly HE state $\hat{\rho}'{}^{AB}$ corresponding to equation \eqref{eq:noisyQbQm4} for different values of transmissivity $\eta$.]{Regions of entanglement detection for a few different values of transmissivity $\eta$. It is interesting to note that the optimal $\alpha^{opt}$ seems to have a fixed value (gain $\alpha\in\mathbb R$).}
\label{fig:ThQbQm3}
\end{center}
\end{figure}
As expected, the higher the transmissivity $\eta$ the greater the parameter regions of entanglement detection. Furthermore, the plot shows that the optimal $\alpha^{opt}$ seems not to depend on $\eta$. Actually, it is straightforward to calculate $\alpha^{opt}$. Simply solve
\begin{equation}
\frac{\partial}{\partial\alpha}\,\frac{4\eta|\alpha|^2}{(1-\eta)(2\mathrm{e}^{4|\alpha|^2}-1)}=0
\end{equation} 
for $\alpha$. This yields $\alpha^{opt}\approx 0.44$, independent of $\eta$. Furthermore, also a similar equation to \eqref{eq:ThQbQmCut} can be written down when solving $s_{1_{\hat{\rho}'{}^{AB}}}(\alpha,\eta,\braket{n_{th.}})<0$ for $\eta$:
\begin{equation}
\eta>\frac{\braket{n_{th.}}}{\braket{n_{th.}}+\frac{4|\alpha|^2}{(2\mathrm{e}^{4|\alpha|^2}-1)}}.
\end{equation}
Calculating $\alpha^{opt}$ with aid of this equation yields of course, as expected, the same result $\alpha^{opt}\approx 0.44$, independent of $\braket{n_{th.}}$. It is actually quite interesting that the optimal amplitude $\alpha^{opt}$ regarding entanglement witnessing with the SV determinant $s_1$ does not depend on the channel at all. $\alpha^{opt}$ seems to be only determined by the choice of the SV determinant.

Compare this to figure \ref{fig:CADampCh}, which shows the concurrence of the state when setting $\braket{n_{th.}}=0$. There, the optimal $\tilde{\alpha}^{opt}$ does very well depend on the transmissivity $\eta$: The greater the transmissivity, the higher $\tilde{\alpha}^{opt}$. This is quite remarkable. It can be infered that the ability to detect entanglement in $\hat{\rho}'{}^{AB}$, depending only on the choice of the SV determinant, does not scale as the entanglement of $\hat{\rho}'{}^{AB}$ itself, which of course depends on the thermal channel's parameters $\eta$ and $\braket{n_{th.}}$.

Concluding, a first example of true HE has been presented and entanglement detection has been performed. Some interesting properties regarding the comparison between entanglement witnessing and entanglement quantification have been identified. Furthermore, the exemplary state is additionally also physically relevant, as it occurs in quantum information applications such as QKD. 

Finally, it is worth to mention that the calculations regarding the thermal photon noise channel have led to the derivation of a set of corresponding Kraus operators. They are presented in appendix \ref{Apx:KrausOpTh}.

\subsection{Another Truly Hybrid Entangled State}\label{subsec:OtherTrueHE}
Another example for a truly HE state which admits entanglement detection shall be presented. However, this state is rather artificially constructed and probably not of great physical relevance. Nevertheless, it is truly HE and enables nice HE witnessing:
\begin{equation}\label{eq:trueHE2}
\begin{aligned}
\hat{\rho}^{AB}&=\sum_{n=1}^{\infty}p_n\ket{\psi_n}^{AB}\bra{\psi_n}, \\
\ket{\psi_n}^{AB}&=\frac{1}{\sqrt{2}}\Bigl(\ket{0}^A\ket{\sqrt{n}\alpha}^B+\mathrm{e}^{i\phi}\ket{1}^A\ket{-\sqrt{n}\alpha}^B\Bigr), \\
p_n&=\frac{1-x}{x}x^n\,,\qquad 0<x<1\,,\qquad \alpha\in\mathbb R.
\end{aligned}
\end{equation}
As this state contains an infinite set of qumode states $\{\ket{\pm\sqrt{n}\alpha}^B:n=1,2,\ldots,\infty\}$, it is clearly truly HE, strictly speaking it shows true HE between a qumode and a qubit. 

Now, we try to witness entanglement in this state with aid of the SV determinant $s_1$. For the calculation of the moments, the identities
\begin{align}
\sum_{n=1}^\infty x^n &= \frac{x}{1-x}, \\ 
\sum_{n=1}^\infty n\,x^n &= \frac{x}{(1-x)^2},
\end{align}
which are derived in appendix \ref{Apx:Series}, are exploited. As an example, the calculation of the moment $\braket{\hat{a}^\dagger\hat{a}\hat{b}^\dagger\hat{b}}$ is presented. Note that the mode operators $\hat{a}$ and $\hat{a}^\dagger$ correspond to system B, while $\hat{b}$ and $\hat{b}^\dagger$ denote the mode operators of system A.
\begin{equation}
\braket{\hat{a}^\dagger\hat{a}\hat{b}^\dagger\hat{b}} =\frac{1}{2}\sum_{n=1}^\infty p_n n\alpha^2 =\frac{\alpha^2}{2}\frac{1-x}{x}\sum_{n=1}^\infty n\,x^n = \frac{\alpha^2}{2}\frac{1}{1-x}.
\end{equation}
For the whole determinant $s_1$,
\begin{equation}\label{eq:ArtiS}
\begin{aligned}
s_1(x,\alpha)&=\frac{1}{8}\biggl[\frac{2\alpha^2}{1-x}-2\Bigl(\frac{\mathrm{e}^{-2\alpha^2}(1-x)}{1-x\mathrm{e}^{-2\alpha^2}}\Bigr)\Bigl(\frac{\alpha(1-x)}{x}\sum_{n=1}^\infty\sqrt{n}(x\mathrm{e}^{-2\alpha^2})^n\Bigr)\Bigl(\frac{\alpha(1-x)}{x}\sum_{n=1}^\infty\sqrt{n}x^n\Bigr) \\
&\qquad -\Bigl(\frac{\alpha(1-x)}{x}\sum_{n=1}^\infty\sqrt{n}(x\mathrm{e}^{-2\alpha^2})^n\Bigr)^2-2\Bigl(\frac{\alpha(1-x)}{x}\sum_{n=1}^\infty\sqrt{n}x^n\Bigr)^2 \\
&\qquad-\frac{\alpha^2}{1-x}\Bigl(\frac{\mathrm{e}^{-2\alpha^2}(1-x)}{1-x\mathrm{e}^{-2\alpha^2}}\Bigr)^2\biggr]
\end{aligned}
\end{equation}
is obtained. Note that this expression does not depend on the phase $\phi$. This has been expected, since the amount of entanglement in states of this form does not depend on $\phi$, as pointed out earlier. 

It is clear that
\begin{equation}
\sum_{n=1}^\infty x^n<\sum_{n=1}^\infty \sqrt{n}\,x^n<\sum_{n=1}^\infty n\,x^n,
\end{equation}
for $0<x<1$. Hence,
\begin{align}
\frac{x}{1-x}< & \sum_{n=1}^\infty \sqrt{n}\,x^n<\frac{x}{(1-x)^2}, \\
\frac{x\,\mathrm{e}^{-2\alpha^2}}{1-x\,\mathrm{e}^{-2\alpha^2}}< & \sum_{n=1}^\infty \sqrt{n}\,(x\,\mathrm{e}^{-2\alpha^2})^n < \frac{x\,\mathrm{e}^{-2\alpha^2}}{(1-x\,\mathrm{e}^{-2\alpha^2})^2}.
\end{align}
Inserting the lower bounds into the sums of equation \eqref{eq:ArtiS} yields
\begin{equation}
s'_1(x,\alpha)=\frac{\alpha^2}{8}\biggl[\frac{2x}{1-x}-\Bigl(\frac{1-x}{1-x\mathrm{e}^{-2\alpha^2}}\Bigr)^2\mathrm{e}^{-4\alpha^2}\Bigl(3+\frac{1}{1-x}\Bigr)\biggr].
\end{equation}
However, since $s_1(x,\alpha)>s'_1(x,\alpha)\,\forall\,\alpha\in{\mathbb R},\,x\in]0,1[$, from $s_1(x,\alpha)<0$ follows $s'_1(x,\alpha)<0$. Therefore, $s'_1(x,\alpha)<0$ is a sufficient criterion for entanglement detection. It is plotted in figure \ref{fig:artiS}.
\begin{figure}[ht]
\subfloat{\includegraphics[width=0.53\textwidth]{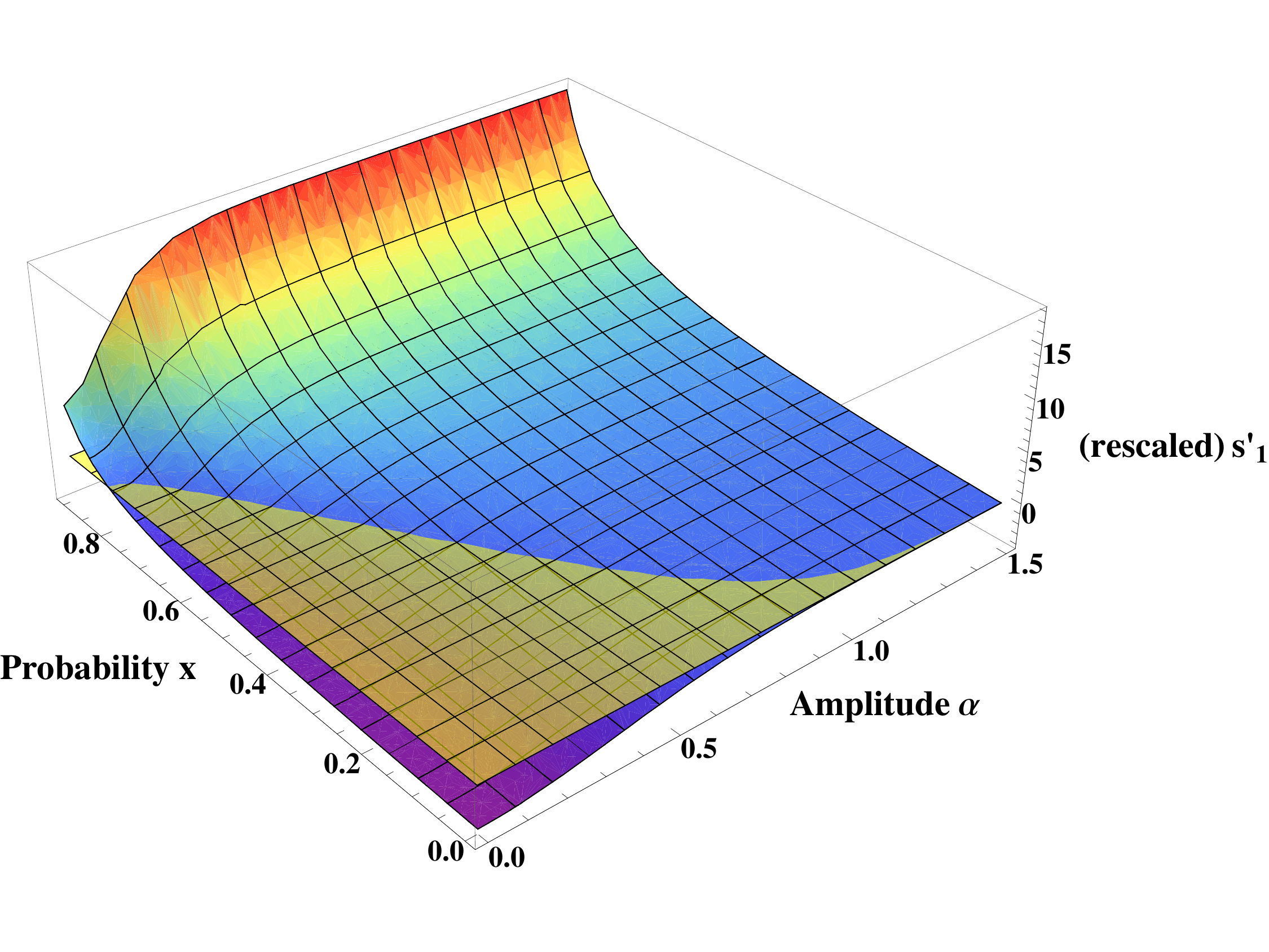}}\hfill
\subfloat{\includegraphics[width=0.47\textwidth]{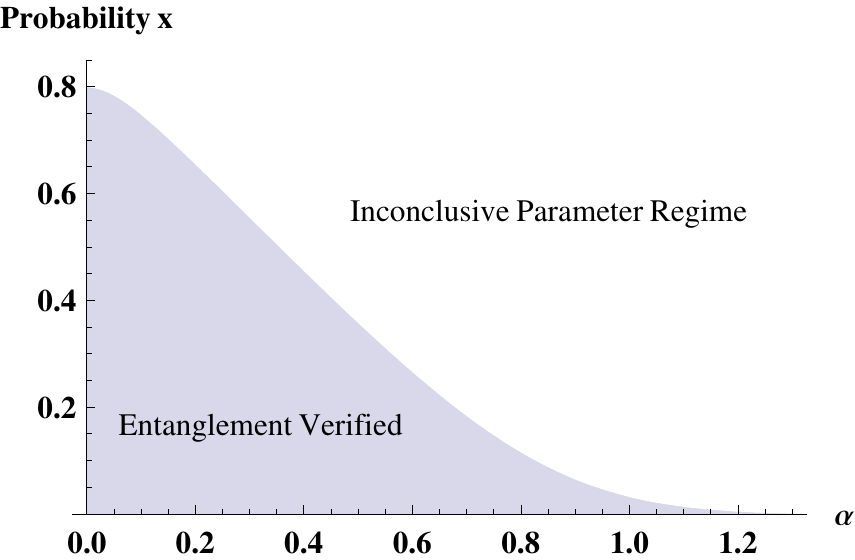}}\hfill
\caption[Entanglement witnessing in the state corresponding to equation \eqref{eq:trueHE2} via $s'_1(x,\alpha)$.]{On the left hand side $s'_1(x,\alpha)$ is shown, which is derived from the original SV determinant $s_1(x,\alpha)$. Note that for the plot the prefactor $\frac{\alpha^2}{8}$ has been omitted, which, of course, does not affect the witnessing region. The yellow plane denotes zero. The right hand side graphic presents the cutting line between $s'_1(x,\alpha)$ and the zero plane. It displays the witnessing region. It also has to be pointed out that for $\alpha=0$ the state is of course not entangled.}
\label{fig:artiS}
\end{figure}
The graphics show that entanglement can be verified for sufficiently small $x$ and $\alpha$. This can be understood in the following way. Small $x$ corresponds to low mixing between the pure states $\ket{\psi_n}^{AB}$, which themselves are rather highly entangled, depending on the amplitude $\sqrt{n}\alpha$. For pure states of the form $\ket{\psi_n}^{AB}$ entanglement witnessing can be performed perfectly with the SV determinant $s_1$. Therefore, efficient entanglement detection for low $x$ is no surprise. Furthermore, already previous calculations (compare subsections \ref{catwitnessing} and \ref{subsec:sqeez}) have exposed that entanglement detection via $s_1$ works best for low $\alpha$ and may fail for high $\alpha$, even if the entanglement itself may be higher.

Concluding, a second example for a truly HE state has been presented, whose true HE could again be verified with aid of the SV criteria. For completion, note that states of the form
\begin{equation}
\ket{\psi_n}^{AB}=\frac{1}{\sqrt{2}}\Bigl(\ket{0}^A\ket{\sqrt{n}\alpha}^B+\mathrm{e}^{i\phi}\ket{1}^A\ket{-\sqrt{n}\alpha}^B\Bigr)
\end{equation}
can be generated via the interaction $\hat{G}$,
\begin{equation}
\hat{G}:=\mathrm{e}^{\frac{1}{2}|\alpha|^2(1-n)}(\hat{\sigma}_z\sqrt{n})^{\hat{a}^\dagger\hat{a}},
\end{equation}
from separable product states of the form
\begin{equation}
\ket{\phi_n}^{AB}_{sep}=\frac{1}{\sqrt{2}}(\ket{0}^A+\mathrm{e}^{i\phi}\ket{1}^A)\ket{\alpha}^B.
\end{equation}
Once again, $\hat{a}$ and $\hat{a}^\dagger$ are the mode operators corresponding to subsystem B, while $\hat{\sigma}_z$ is the Pauli Z matrix acting on subsystem A. Therefore,
\begin{equation}
\begin{aligned}
\hat{G}\ket{\phi_n}^{AB}_{sep}&=\mathrm{e}^{\frac{1}{2}|\alpha|^2(1-n)}(\hat{\sigma}_z\sqrt{n})^{\hat{a}^\dagger\hat{a}}\frac{1}{\sqrt{2}}\Bigl(\ket{0}^A+\mathrm{e}^{i\phi}\ket{1}^A\Bigr)\ket{\alpha}^B \\
&=\frac{1}{\sqrt{2}}\Bigl(\ket{0}^A\ket{\sqrt{n}\alpha}^B+\mathrm{e}^{i\phi}\ket{1}^A\ket{-\sqrt{n}\alpha}^B\Bigr) \\
&=\ket{\psi_n}^{AB},
\end{aligned}
\end{equation}
which becomes obvious when expressing the coherent state $\ket{\alpha}^B$ in the Fock basis.

\subsection{Truncated Mixed States and Entanglement Witnessing}\label{subsec:Robust}
This subsection does not really deal with truly HE states, but with general entangled quantum states of the form
\begin{equation}\label{eq:trunc1state}
\hat{\rho}=\sum_{n=0}^\infty p_n \ket{\psi_n}\bra{\psi_n}\,,\qquad p_n\neq0\,\forall\,n.\footnote[5]{Also states with some $p_n=0$ could be considered as long as the total number of $p_n\neq0$ is still infinite. However, when omitting the $p_n=0$ terms in this case, one ends up exactly at the form \eqref{eq:trunc1state}.}
\end{equation}
The state may be bipartite or multipartite entangled as well as CV, DV or hybrid. Hence, truly HE states are included (note that effectively DV HE states are not included.). The focus here lies on the infinite number of terms. In the previous subsections, it emerged that one of the main problems with true HE is that the occurring infinite sums or infinite integrals have to be calculated without truncation. A truncation would correspond to analysis of a different state, which would not have the important infinite number of mix terms character anymore. 

Imagine for example the following case. Consider a separable state $\sum_{n=0}^\infty p_n \ket{\psi_n}\bra{\psi_n}$ for which, however, the truncated state $\hat{\rho}_N=\sum_{n=0}^N p_n \ket{\psi_n}\bra{\psi_n}$ is entangled for some $N$. Therefore, a truncation in occurring sums or integrals may lead to misinterpretations. Entanglement in $\hat{\rho}_N$ would potentially be detected and wrongly misinterpreted as entanglement of the state with $N=\infty$. 

The questions arising are the following: In which cases is, for an entangled state of the form
\begin{equation}
\hat{\rho}=\sum_{n=0}^\infty p_n \ket{\psi_n}\bra{\psi_n},
\end{equation}
also the state
\begin{equation}
\hat{\rho}_N=\frac{1}{p_N}\sum_{n=0}^N p_n \ket{\psi_n}\bra{\psi_n}\,,\qquad p_N:=\sum_{n=0}^N p_n.
\end{equation}
entangled? The prefactor $\frac{1}{p_N}$ is due to normalization of the truncated state. Or, coming to the important point: When successfully witnessing entanglement in 
\begin{equation}
\hat{\rho}_N=\frac{1}{p_N}\sum_{n=0}^N p_n \ket{\psi_n}\bra{\psi_n},
\end{equation}
in which cases also entanglement in
\begin{equation}
\hat{\rho}=\sum_{n=0}^\infty p_n \ket{\psi_n}\bra{\psi_n}
\end{equation}
can be deduced? An answer to this question is provided in the following.
\begin{thm}\label{thm:TruncMix}
For $\hat{\rho}=\sum_{n=0}^\infty p_n \ket{\psi_n}\bra{\psi_n}$ with $p_n\neq0\,\forall\,n$ define $\hat{\rho}_N=\frac{1}{p_N}\sum_{n=0}^N p_n \ket{\psi_n}\bra{\psi_n}$ for some $N\in{\mathbb N}_0$ with $p_N:=\sum_{n=0}^N p_n$. Then the following statement holds true: If $R_g(\hat{\rho}_N)>1-p_N$, $\hat{\rho}$ is entangled.
\end{thm}
\begin{proof}[\textbf{Proof.}]
\begin{itemize}
\item[] {}
\item $R_g(\hat{\rho}_N)>1-p_N$.
\item $\Rightarrow\;\,p_N\hat{\rho}_N+(1-p_N)\hat{\sigma}$ entangled for any quantum state $\hat{\sigma}$.
\item Choose $\hat{\sigma}=\frac{1}{1-p_N}\sum_{n=N+1}^\infty p_n \ket{\psi_n}\bra{\psi_n}$.
\item $\Rightarrow\;\,p_N\hat{\rho}_N+(1-p_N)\hat{\sigma}=\hat{\rho}$ entangled. \qedhere
\end{itemize}
\end{proof}
The answer to the above question is therefore closely related to the global robustness $R_g$. Unfortunately, the calculation of $R_g$ is mostly quite difficult or even impossible. Nevertheless, it has to be calculated for a state with only finite $N$, which may be possible in certain circumstances. If entanglement detection has to be performed for a state with $N=\infty$ yielding infinite sums or integrals which cannot be calculated without truncation this theorem may be a useful tool if robustnesses can be calculated for the truncated states. However, it is not clear which $N$ has to be chosen. In the case that the state is actually not entangled, such an $N$ with $R_g(\hat{\rho}_N)>1-p_N$ does not exist. Proving this is trivial: Assume that there exists such an $N$ for a separable state. Then this would indicate that the separable state is actually entangled, which is a contradiction to the assumption.

Since there is an infinite amount of possible $N$s, no statement about the state's entanglement can be made, if no such $N$ is found. Working through "all" $N$s is impossible, since there is no finite "all". Nevertheless, for any entangled state of the form
\begin{equation}
\hat{\rho}=\sum_{n=0}^\infty p_n \ket{\psi_n}\bra{\psi_n}\,,\qquad p_n\neq0\,\forall\,n,
\end{equation}
such an $N$ does actually exist. 

To prove this a new entanglement monotone is required, which is introduced now. The \textit{entanglement of subtraction} $E_{Sub}(\hat{\rho})$ is defined as
\begin{equation}
E_{Sub}(\hat{\rho}):=\inf_{\Omega, \Delta}\{b:b\geq0,\,\Delta\in QS,\,\Omega\in SEP\;\;\text{such that}\;\;\hat{\rho}-b\Delta=(1-b)\Omega\},
\end{equation}
where $QS$ denotes the set of all quantum states and $SEP$ the set of the separable quantum states. The interpretation of this monotone is clear. It asks how much of a quantum state must be subtracted from the state in question that it becomes separable. However, for $E_{Sub}(\hat{\rho})$ to be a proper entanglement monotone, first its LOCC monotonicity has to be proved \cite{MeasuresIntro}.

Consider an entangled state $\hat{\rho}$ which is supposed to have $E_{Sub}(\hat{\rho})=R$. Then, its optimal decomposition according to the entanglement of subtraction is given by
\begin{equation}
\hat{\rho}=a\Omega+R\Delta,
\end{equation}
with $a=1-R$. We perform a general LOCC operation on $\hat{\rho}$, which results in the output $\hat{\rho}_i=\frac{\Lambda_i(\hat{\rho})}{q_i}$ with probability $q_i$. The output ensemble is given as
\begin{equation}\label{eq:EoSproof}
\{q_i\; ;\;\frac{a\Lambda_i(\Omega)+R\Lambda_i(\Delta)}{q_i}\}\equiv \{q_i\; ;\;\tilde{a}_i\frac{\Lambda_i(\Omega)}{\mathrm{tr}[\Lambda_i(\Omega)]}+\tilde{R}_i\frac{\Lambda_i(\Delta)}{\mathrm{tr}[\Lambda_i(\Delta)]}\},
\end{equation}
where
\begin{align}
\tilde{a}_i&=\frac{a\,\mathrm{tr}[\Lambda_i(\Omega)]}{q_i}, \\
\tilde{R}_i&=\frac{R\,\mathrm{tr}[\Lambda_i(\Delta)]}{q_i}.
\end{align}
Since $a,R\geq0$, and therefore also $\tilde{a}_i,\tilde{R}_i\geq0$, equation \eqref{eq:EoSproof} is a valid decomposition of $\{q_i\,;\,\Lambda_i(\hat{\rho})\}$ of the form $\{q_i\,;\,\tilde{a}_i\Omega+\tilde{R}_i\Delta\}$. However, it is not necessarily the optimal one. Denote the actual entanglement of subtraction of the output $\Lambda_i(\hat{\rho})$ with $R(\Lambda_i(\hat{\rho}))$. Then the average output entanglement is given by $\sum_i q_iR(\Lambda_i(\hat{\rho}))_i$. Finally,
\begin{equation}
\begin{aligned}
\sum_i q_iR(\Lambda_i(\hat{\rho}))_i \leq \sum_i q_i\tilde{R}_i=R \sum_i\mathrm{tr}[\Lambda_i(\Delta)]=R.
\end{aligned}
\end{equation}
Hence, the entanglement of subtraction does not increase on average under LOCC. Furthermore, it is clear from the construction that $E_{Sub}$ is zero if and only if the state is separable. Therefore, it is an entanglement monotone. Some properties of $E_{Sub}$ shall be listed:
\begin{itemize}
\item $0\leq E_{Sub}\leq1$.
\item $\hat{\rho}$ separable $\Leftrightarrow$ $E_{Sub}=0$.
\item For $b=1$, every state "becomes" separable: $\forall\,\hat{\rho}$, choose $\Delta=\hat{\rho}$ and hence, $\hat{\rho}-\Delta={\mathbb{O}}\in SEP$, where ${\mathbb{O}}$ denotes the zero operator, which is of course separable.
\end{itemize}

Two annotations have to be made regarding this new monotone: The subtraction of quantum states from other quantum states does not necessarily result in an output which is a valid quantum state \cite{Subtra}. However, in the definition of $E_{Sub}$ the variation is performed in such a way that only those subtractions are considered which yield a valid quantum state as output. Furthermore, it has to be pointed out that the entanglement of subtraction is a special monotone of a whole family of monotones, as presented in \cite{MeasuresIntro}: 

Consider two sets of operators $X$ and $Y$ which are convex cones and closed under LOCC operations. Furthermore, the members of the sets can be written in the form $\alpha_{X(Y)}\cdot$positive-semidefinite operator with trace $\alpha_{X(Y)}$. $\alpha_{X(Y)}$ are fixed real constants. Finally, the sets have to be chosen such that any Hermitian operator $\hat{h}$ can be written as $\hat{h}=a\tilde{\Omega}-b\tilde{\Delta}$, where $\tilde{\Omega}\in X$ and $\tilde{\Delta}\in Y$ and $a,b\geq 0$. Then a monotone may be defined as
\begin{equation}
R_{X,Y}(\hat{\rho}):=\inf_{\Omega, \Delta}\{b:a,b\geq0,\,\Delta\in Y,\,\Omega\in X\;\;\text{such that}\;\;\hat{\rho}=a\Omega-b\Delta\}.
\end{equation}
For $\alpha_X=1$ and $\alpha_Y=-1$, the entanglement of subtraction is obtained.

Finally, recall theorem \ref{thm:TruncMix}. With the aid of the entanglement of subtraction it can be proved that for every entangled state $\hat{\rho}$ of this form actually an $N$ exists such that $R_g(\hat{\rho}_N)>1-p_N$.
\begin{thm}\label{thm:TruncMix2}
For any entangled state of the form $\hat{\rho}=\sum_{n=0}^\infty p_n \ket{\psi_n}\bra{\psi_n}$ with $p_n\neq0\,\forall\,n$ there exists an $N\in{\mathbb N}_0$ such that $R_g(\hat{\rho}_N)>1-p_N$ for $\hat{\rho}_N=\frac{1}{p_N}\sum_{n=0}^N p_n \ket{\psi_n}\bra{\psi_n}$ with $p_N:=\sum_{n=0}^N p_n$.
\end{thm}
\begin{proof}[\textbf{Proof.}]
%Assume $\{p_n\}$ to be ordered such that $p_0\geq p_1\geq p_2\geq\ldots$
\begin{itemize}
\item[] {}
\item $\hat{\rho}$ entangled $\;\Rightarrow\;$ $E_{sub}(\hat{\rho})\geq x>0$.
%\item $\hat{\rho}=(1-t)\hat{\rho}_1+t\hat{\rho}_2$ with $\hat{\rho}_1=\frac{1}{1-t}(\hat{\rho}-t\hat{\rho}_2)$ for arbitrary states $\hat{\rho}_1$, $\hat{\rho}_2$ and $0\leq t<1$.
\item $\Rightarrow\;$ $\hat{\rho}_1=\frac{1}{1-t}(\hat{\rho}-t\hat{\rho}_2)$ entangled for $0\leq t<E_{sub}(\hat{\rho})$ and arbitrary states $\hat{\rho}_2$.
\item Furthermore, $\,\sum_{n=0}^\infty p_n=1$ $\;\Rightarrow\;$ $\lim_{n\rightarrow\infty} p_n=0$ $\;\Rightarrow\;$ $\forall\,\epsilon>0\,\exists \,N\in{\mathbb N}_0\,$ such that $\;\sum_{n=N+1}^\infty p_n <\epsilon$.
\item Choose $N$ such that $\sum_{n=N+1}^\infty p_n =1-p_N<x\leq E_{sub}(\hat{\rho})$.
\item Set $t=1-p_M$.
\item $\Rightarrow\;$ $\hat{\rho}_1=\frac{1}{p_M}(\hat{\rho}-(1-p_M)\hat{\rho}_2)$ entangled for any $\hat{\rho}_2$ and $M\geq N$.
\item Choose $\hat{\rho}_2=\frac{1}{1-p_M}\sum_{n=M+1}^\infty p_n \ket{\psi_n}\bra{\psi_n}$.
\item $\Rightarrow\;$ $\hat{\rho}_1=\hat{\rho}_M=\frac{1}{p_M}\sum_{n=0}^M p_n \ket{\psi_n}\bra{\psi_n}$ entangled $\forall\,M\geq N$.
%\item Since $E_{sub}(\hat{\rho})$ is greater than some finite $x$ also the "left" entanglement in states of the form $\hat{\rho}_M=\frac{1}{p_M}\sum_{n=0}^M p_n \ket{\psi_n}\bra{\psi_n}$ for $M\geq N$ and $1-p_N<x$ is greater than some finite amount $x'$ $\forall\,M\geq N$. This of course also holds if measured by the global robustness. Hence, $R_g(\hat{\rho}_M)\geq \tilde{x}>0\,\forall\,M\geq N$.
\item Consider the global robustness of $\hat{\rho}_M$. Define $R_g(\hat{\rho}_M):=\tilde{x}_M$. As $\hat{\rho}_M$ entangled $\forall\,M\geq N$, $\tilde{x}_M>0$ $\forall\,M\geq N$. The important point is that $\tilde{x}_M$ converges definitely \textit{not} against zero: Due to $\hat{\rho}_\infty=\hat{\rho}$ it is known that $\tilde{x}_\infty=R_g(\hat{\rho})>0$, as the initial state $\hat{\rho}$ is entangled. 

Concluding, on the one hand, $\lim_{M\rightarrow\infty}\tilde{x}_M=R_g(\hat{\rho})>0$, while, on the other hand, $\lim_{M\rightarrow\infty}(1-p_M)=0$. 
\item $\Rightarrow\;$ $\exists \,\tilde{M}>N$ such that $R_g(\hat{\rho}_{\tilde{M}})>1-p_{\tilde{M}}$.
%\item Choose $M\geq N$ such that $\tilde{x}>1-p_M$.
\item $\Rightarrow\;$ $\hat{\rho}_{\tilde{M}}=\frac{1}{p_{\tilde{M}}}\sum_{n=0}^{\tilde{M}} p_n \ket{\psi_n}\bra{\psi_n}$ entangled with $R_g(\hat{\rho}_{\tilde{M}})>1-p_{\tilde{M}}$. \qedhere
\end{itemize}
\end{proof}

The consequence of the theorems \ref{thm:TruncMix} and \ref{thm:TruncMix2} is that entanglement witnessing problems regarding states of the form
\begin{equation}
\hat{\rho}=\sum_{n=0}^\infty p_n \ket{\psi_n}\bra{\psi_n}\,,\qquad p_n\neq0\,\forall\,n,
\end{equation}
can be reduced to entanglement quantification problems of states of the form
\begin{equation}
\hat{\rho}_N=\sum_{n=0}^N p_n \ket{\psi_n}\bra{\psi_n},
\end{equation}
with finite $N$, where the employed entanglement monotone has to be the global robustness. However, the applicability of this approach is quite limited.
\begin{itemize}
\item No recipe is given for finding the right $N$. Simply several choices of $N$ have to be tried. Obviously, there is an infinite number of possible choices. Therefore, it is more or less a matter of luck to find an $N$ with $R_g(\hat{\rho}_N)>1-p_N$.
\item Calculating the global robustness for general states $\hat{\rho}_N$ is an extraordinarily difficult task, since it corresponds to a complex variational problem. Additionally, another entanglement measure or necessary separability criterion is required to check whether a state $(1-\lambda)\hat{\rho}_N+\lambda\hat{\sigma}$ actually is separable or not. Powerful necessary separability criteria do not even exist. However, this problem may be avoided by loosening the constraints and requiring for example only PPTness instead of separability for the robustness calculation.
\item Entanglement witnessing is relevant in cases in which it is not known whether a state is entangled or not. But in the case that the state is not entangled at all, no $N$ with $R_g(\hat{\rho}_N)>1-p_N$ can be found. Hence, it is in the first place unknown whether such an $N$ does exist or not. An extensive search for such an $N$ is therefore pointless, as it may not exist at all.
\end{itemize}
Nevertheless, the approach is not fully senseless. In cases, where for the original $N=\infty$ state no calculations can be made without truncating the occurring infinite sums or integrals, it is the only chance to actually witness entanglement with certainty. Any such truncation during the witnessing process for the original $N=\infty$ state $\hat{\rho}$ would correspond to witnessing entanglement in an actually different state. Then, if robustness calculations are possible for the corresponding truncated finite $N$ states $\hat{\rho}_N$, being lucky and finding an $N$ with $R_g(\hat{\rho}_N)>1-p_N$ is the only chance to unerringly witness entanglement in $\hat{\rho}$.

\section{Multipartite Hybrid Entanglement}\label{sec:multihe}
In this section some investigations regarding multipartite hybrid entanglement are presented. In general a HE quantum system may of course not only be built of two subsystems of finite and infinite dimensionality, but also of several finite- and infinite-dimensional subsystems. In general, stepping from the bipartite into the multipartite regime yields a wealth of new remarkable effects as well as difficulties. These shall not be discussed here in detail.\footnote[6]{Brief introductions to multipartite entanglement are given in \cite{MeasuresIntro} as well as in \cite{HoroEntIntro}, which is one of the most thorough reference works regarding quantum entanglement.} However, some interesting phenomena regarding multipartite HE are examined. In the first subsection, some fundamental concepts are briefly introduced which are necessary for further studies on multipartite HE. Afterwards, general multipartite HE issues are discussed, and two specific states are investigated.

\subsection{Multipartite Entanglement Basics}
Several applications in quantum information require the consideration of entanglement between more than two parties. For example, when performing multi-user quantum communication between more participants than just one sender and one receiver, multipartite entangled states have to be established \cite{multicomm1,multicomm2,multicomm3}. Furthermore, the already mentioned cluster states required for certain quantum computation schemes are multipartite entangled states of a very large number of particles \cite{Raussendorf,NielsenCluster}. Additionally, the already mentioned qubus approaches for quantum communication and quantum computation actually involve multipartite HE states. The CV qubus entangles itself consecutively with two qubits, resulting in a tripartite HE state, before it is measured out such that the remaining two-qubit state is left entangled \cite{qubus4,Nadja}.

One of the most crucial differences between bipartite and multipartite entanglement becomes evident when trying to define maximally entangled states from which all other states can be obtained via LOCC with certainty (for the sake of clearness and simplicity consider two-dimensional qubit subsystems in this introduction). Motivated from the bipartite setting, an obvious possibility would be the \textit{Greenberger-Horne-Zeilinger} (\textit{GHZ}) state $\ket{GHZ}$, which looks in a tripartite setting like
\begin{equation}
\ket{GHZ}^{ABC}  =  \frac{1}{\sqrt{2}}\biggl(\ket{0}^A\ket{0}^B\ket{0}^C+\ket{1}^A\ket{1}^B\ket{1}^C\biggr).
\end{equation}
It has some very neat properties. The entanglement of any bipartite cut, for example the entanglement of party A vs. parties B and C together, takes one ebit. Furthermore, when a projection measurement on $\frac{1}{\sqrt{2}}(\ket{0}+\ket{1})$ is performed, always a maximally entangled two-qubit Bell state is left over. Therefore, any bipartite state can be obtained from it with the aid of LOCC transformations. However, there are tripartite states which cannot be generated via the GHZ state and LOCC, as for instance the W-state:
\begin{equation}
\ket{W}^{ABC}  =  \frac{1}{\sqrt{3}}\biggl(\ket{0}^A\ket{0}^B\ket{1}^C+\ket{0}^A\ket{1}^B\ket{0}^C+\ket{1}^A\ket{0}^B\ket{0}^C\biggr).
\end{equation}
This impossibility of interconversion of the GHZ- and the W-state holds also when switching to an asymptotic setting and considering arbitrarily many identical copies of the states \cite{MeasuresIntro}.\footnote[7]{In the bipartite setting, the consistent picture of entanglement quantification via LOCC transformability is actually not obtained until, on the one hand, transformations between arbitrarily many copies of the states are considered, and, on the other hand, asymptotic imperfections are allowed, which vanish when the number of copies tends to infinity \cite{MeasuresIntro}.} Therefore, none of the states can be considered as more entangled than the other. Actually, there are different kinds of multipartite entanglement. On the one side, a tripartite state may just be bipartite entangled between its subsystems, while, on the other side, it may also show "genuine tripartite entanglement". The generalization to higher dimensional systems proceeds accordingly.

This motivates to go beyond a simple distinguishment between separable and entangled states. For a N-partite state \textit{full N-separability} is still defined in such a way that the state can be written as
\begin{equation}
\hat{\rho}^{ABC...}=\sum_i p_i \hat{\rho}^A\otimes\hat{\rho}^B\otimes\hat{\rho}^C\otimes\ldots\;.
\end{equation}
However, there is also the notion of \textit{partial separability} when only some subsystems separate from the remaining state, which may be inseparable. Furthermore, an N-partite state which shows entanglement between maximally $k$ subsystems is said to be \textit{k-entangled}. A tripartite state is for example said to be 2-entangled if it can be written as
\begin{equation}
\hat{\rho}^{ABC}=\sum_i p_i \hat{\rho}^A_i\otimes\hat{\rho}^{BC}_i+\sum_i q_i \hat{\rho}^B_i\otimes\hat{\rho}^{AC}_i+\sum_i r_i \hat{\rho}^C_i\otimes\hat{\rho}^{AB}_i.
\end{equation}
However, when taking many copies of states with, for instance, k-entanglement, also l-en\-tang\-led states with $l>k$ may be created via LOCC \cite{MeasuresIntro}. Therefore, this notion of "k-entanglement" is either not closed regarding LOCC transformations or under taking many copies of the states. 
\begin{figure}[ht]
\begin{center}
\includegraphics[width=13cm]{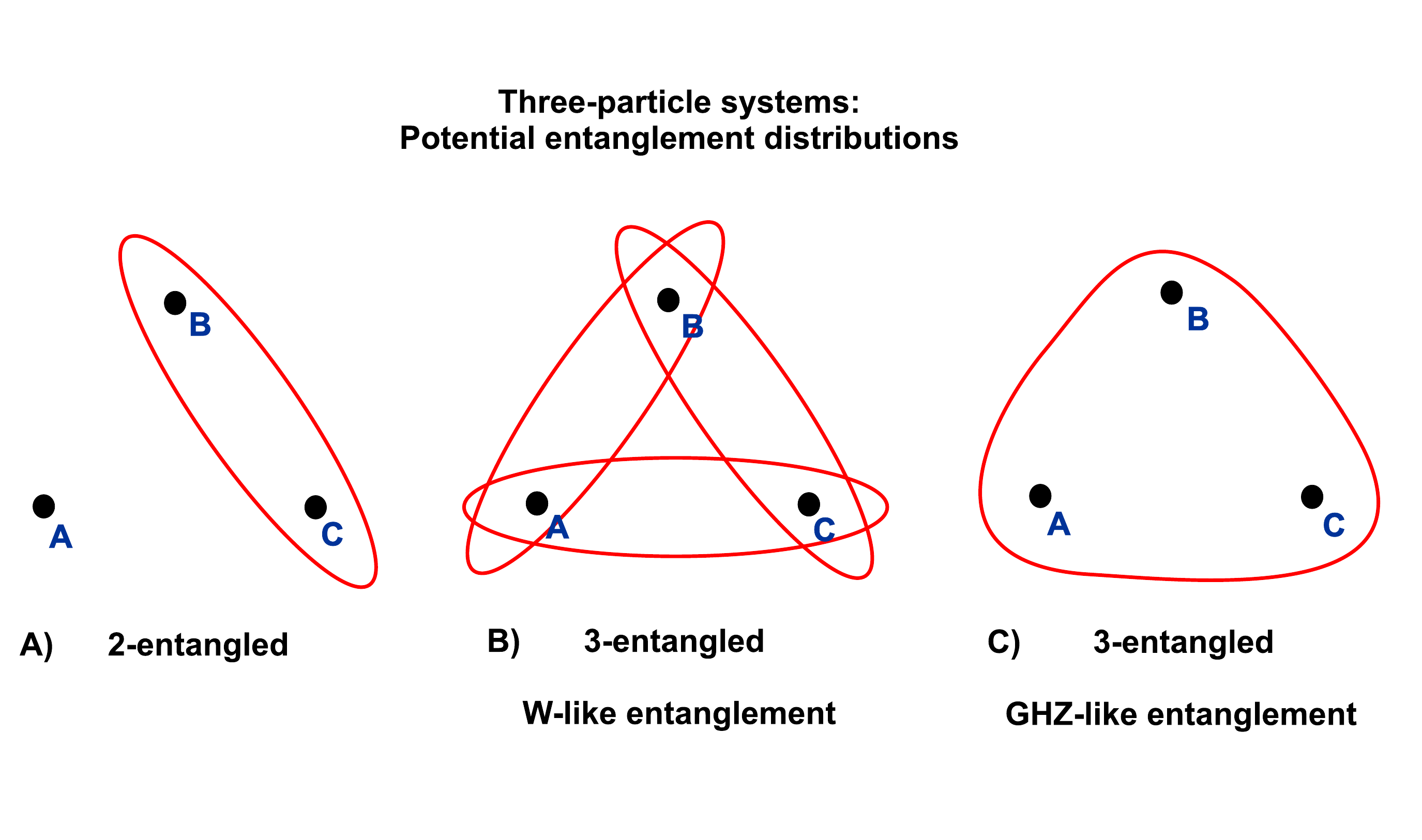}
\caption[Three possibilities for entanglement distribution in a pure three-particle system.]{Three possibilities for entanglement distribution in a pure three-particle system. A) corresponds to a state of the form $\ket{0}^A\otimes\ket{\Phi}^{BC}$, where $\ket{\Phi}$ denotes an entangled two-qubit state. B) is the visualization of the W-state $\ket{W}^{ABC}$, while picture C) presents a GHZ-state $\ket{GHZ}^{ABC}$. It becomes clear that for a system in a W-state losing one particle leaves the remaining particles entangled, while for systems prepared in GHZ-states particle loss destroys \textit{all} entanglement.}
\label{fig:multiEnt}
\end{center}
\end{figure}

This does actually not solve the problem with the tripartite GHZ- and the W- state, which are both 3-entangled. When tracing out one party of the W-state, always an entangled state is left over, while in case of the GHZ-state a trace over one subsystem yields a separable state. Therefore, the GHZ-state seems to be genuinely tripartite entangled, showing no bipartite entanglement, while the W-state seems to offer only bipartite entanglement between its subsystems. Nevertheless, both states are 3-entangled. Figure \ref{fig:multiEnt} visualizes some possible scenarios of entanglement in pure tripartite states. This discussion demonstrates quite well that the characterization of entanglement in multipartite settings is much more involved than in the bipartite regime. Actually, there are a lot of open questions, especially when considering not only tripartite states but systems of higher order containing much more parties.

Interestingly, there is a way to quantify how much GHZ-like, genuine tripartite entanglement a tripartite state contains. Coffman, Kundu and Wootters (CKW) have shown that
\begin{equation}\label{eq:monogamy}
C^2(A|B)+C^2(A|C)\leq C^2(A|BC)
\end{equation}
holds true for any tripartite state \cite{CKW}. $C^2(A|B)$ and $C^2(A|C)$ respectively represent the entanglement between parties A and B as well as A and C, measured by the squared concurrence. For the calculation of $C^2(A|B)$ subsystem C has to be traced out and the ordinary bipartite concurrence has to be applied, analogously for $C^2(A|C)$. Finally, $C^2(A|BC)$ denotes the squared concurrence of the bipartite cut between party A and parties B and C. CKW could show that 
\begin{equation}
C^2(A|BC)=4\det[\rho^A]=4\det[\mathrm{tr}_{BC}[\rho^{ABC}]].
\end{equation}
The relation \eqref{eq:monogamy} is also known as the \textit{mongamy of entanglement}, which says that if two systems are maximally entangled, they cannot be entangled with another third system. This becomes quite clear from the form of the inequality: If A and B are maximally entangled, $C^2(A|B)=1$. As all quantities in the formula are between one and zero, $C^2(A|B)=1$ necessarily yields $C^2(A|BC)=1$ and $C^2(A|C)=0$. Hence, system A cannot be entangled with system C. A similar equation holds also for system B and therefore also system B cannot be entangled with another system. Note that there is a generalization to more parties: 
\begin{equation}
C^2(A|B)+C^2(A|C)+C^2(A|D)+\ldots\,\leq C^2(A|BCD\ldots).
\end{equation}
However, stick to the tripartite case and consider \eqref{eq:monogamy}. For the "$=$" case the overall entanglement between party A and parties B and C $C^2(A|BC)$ is just bipartite entanglement, namely $C^2(A|B)+C^2(A|C)$. However, when the overall entanglement $C^2(A|BC)$ is greater than the amount of bipartite entanglement $C^2(A|B)+C^2(A|C)$, of what kind is the remaining entanglement? Hence, define the so-called \textit{residual entanglement} $\tau_{res}$,
\begin{equation}\label{eq:residualEnt}
\tau_{res}:=C^2(A|BC)-C^2(A|B)-C^2(A|C).
\end{equation}
In the tripartite case, it corresponds to the GHZ-like \textit{genuine tripartite entanglement}, which has already been mentioned regarding the GHZ-state. And actually, as expected,
\begin{align}
\tau_{res}^{GHZ}&=1, \\
\tau_{res}^{W}&=0,
\end{align}
is found for the GHZ- and the W-state. This shows very nicely and actually quantifies the characteristic difference between these two kinds of states. As $\tau_{res}\leq1$ in general, the GHZ state is in the following considered as maximally entangled with regard to its special kind of GHZ-like genuine tripartite entanglement. The W-state is of course also 3-entangled, however, its 3-entanglement is based on 2-entanglement between its subsystems, which can be inferred from the fact that it satisfies the CKW inequality \eqref{eq:monogamy} for the "$=$" case (see also figure \ref{fig:multiEnt}). Therefore, W-like 3-entanglement is not regarded as "genuine tripartite entanglement". Nevertheless, it is a special form of 3-entanglement.    

Note that $\tau_{res}$ does not change, when switching the roles of the systems A, B, and C. In the multipartite case with more than three parties, the analogously defined residual entanglement corresponds to the amount of entanglement in the state which is \textit{not} bipartite or based on bipartite entanglement.

In the following section, this important work by CKW is applied on tripartite HE states.

\subsection{General Multipartite Hybrid Entanglement}
In general, N-partite HE states live in Hilbert spaces of the form ${\mathcal H}_{d_1}\otimes\ldots\otimes{\mathcal H}_{d_N}$, where some $d_i$ are finite and some infinite. As this chapter is primarily supposed to be an outlook, and for the sake of simplicity, mainly tripartite HE states are considered. Then there are two cases: Either the Hilbert space looks like ${\mathcal H}_{d_1}\otimes{\mathcal H}_{d_2}\otimes{\mathcal H}_\infty$ or ${\mathcal H}_{d_1}\otimes{\mathcal H}_\infty\otimes{\mathcal H}_\infty$, with finite $d_1,d_2$. 

Consider the first case, where only one subsystem is CV. A general state in ${\mathcal H}_{d_1}^A\otimes{\mathcal H}_{d_2}^B\otimes{\mathcal H}^C_\infty$ can be written as
\begin{equation}
\ket{\psi}^{ABC}=\sum_{i,j=1}^{d_1,d_2}c_{ij}\ket{i}^A\ket{j}^B\ket{\psi_{ij}}^C\,,\qquad\sum_{i,j}^{d_1,d_2}|c_{ij}|^2=1,
\end{equation}
where $\ket{\psi_{ij}}^C$ represents some qumode state. As there are at most only $d_1\times d_2$ such qumode states, a Gram-Schmidt process can be executed to write the state as a multipartite, effectively DV HE state. Therefore, if all $\ket{\psi_{ij}}^C$ are linearly independent the state effectively lives in a Hilbert space of the form ${\mathcal H}_{d_1}^A\otimes{\mathcal H}_{d_2}^B\otimes{\mathcal H}^C_{d_1\times d_2}$. 

Now move on to the second case, where the initial Hilbert space looks like ${\mathcal H}_{d}^A\otimes{\mathcal H}_\infty^B\otimes{\mathcal H}_\infty^C$. Then the most general states are
\begin{equation}
\ket{\psi}^{ABC}=\sum_{i,j=1}^{d,\infty}c_{ij}\ket{i}^A\ket{\phi_{ij}}^B\ket{\psi_{ij}}^C,
\end{equation}
where both $\ket{\phi_{ij}}^B$ and $\ket{\psi_{ij}}^C$ denote qumode states, and the $c_{ij}$ are chosen such that $\braket{\psi|\psi}=1$. Obviously, in this case an infinite number of qumode states is present for an infinite number of $c_{ij}\neq0$. No Gram-Schmidt process can be performed. Hence, these states show \textit{multipartite true HE}.

Concluding, in the multipartite regime already for pure states, there are characteristic differences between the two different possible configurations. Furthermore, quite remarkably, in contrast to the bipartite setting, there is also pure true HE. When dealing with two parties, true HE only occurs for specific types of mixed states (compare section \ref{sec:mixed2}). Actually, this can be generalized. True hybrid entanglement can be obtained in two cases: Either the system is mixed with an infinite number of mix terms, then it is sufficient that only one subsystem is CV, or two or more subsystems are CV, then the state is not even required to be mixed. However, every mixed state with a finite number of mix terms, and which contains only one CV subsystem is effectively DV. A mixed state of this type in ${\mathcal H}_{d_1}\otimes\ldots\otimes{\mathcal H}_{d_{N-1}}\otimes{\mathcal H}_{\infty}$ can be always described in a Hilbert space of the form ${\mathcal H}_{d_1}\otimes\ldots\otimes{\mathcal H}_{d_{N-1}}\otimes{\mathcal H}_{\Xi}$, where
\begin{equation}
\Xi=M\prod_{i=1}^{N-1} d_i,
\end{equation}
and $M$ is the number of mix terms. In the following subsections two explicit tripartite HE states are investigated, one living in a Hilbert space of the first form ${\mathcal H}_{d_1}\otimes{\mathcal H}_{d_2}\otimes{\mathcal H}_\infty$ and the other one in a Hilbert space of the form ${\mathcal H}_{d_1}\otimes{\mathcal H}_\infty\otimes{\mathcal H}_\infty$, again with finite $d_1,d_2$.

\subsection{A Pure \texorpdfstring{${\mathcal H}_{2}\otimes{\mathcal H}_{2}\otimes{\mathcal H}_{\infty}$}{H2xH2xHinfty} State}
Consider the state
\begin{equation}
\ket{\psi}^{ABC}=\frac{1}{\sqrt{2}}\biggl(\ket{0}^{A}\ket{0}^{B}\ket{\psi_0}^{C}+\ket{1}^{A}\ket{1}^{B}\ket{\psi_1}^{C}\biggr),
\end{equation}
living in a Hilbert space of the form ${\mathcal H}_{2}^A\otimes{\mathcal H}_{2}^B\otimes{\mathcal H}^C_{\infty}$. It contains only two qumode states $\ket{\psi_0}^{C}$ and $\ket{\psi_1}^{C}$. Hence, perform a Gram-Schmidt process:
\begin{align}
\ket{\psi_0}^{C}&=\ket{0}^{C}, \\
\ket{\psi_1}^{C}&={\mathcal Q}\ket{0}^{C} + \sqrt{1-|{\mathcal Q}|^2}\ket{1}^{C},
\end{align}
where the overlap between the qumode states is denoted by ${\mathcal Q}:=\braket{\psi_0|\psi_1}$. This yields
\begin{equation}
\ket{\psi}^{ABC}=\frac{1}{\sqrt{2}}\biggl(\ket{0}^{A}\ket{0}^{B}\ket{0}^{C}+{\mathcal Q}\ket{1}^{A}\ket{1}^{B}\ket{0}^{C}+\sqrt{1-|{\mathcal Q}|^2}\ket{1}^{A}\ket{1}^{B}\ket{1}^{C}\biggr),
\end{equation}
which is a three-qubit state in ${\mathcal H}_{2}^A\otimes{\mathcal H}_{2}^B\otimes{\mathcal H}^C_2$. Hence, it can be easily analyzed with aid of the concurrence. This is the reason why such a HE state which is effectively a qubit triple has actually been chosen. Only in this case the state can be examined regarding its bipartite and residual, genuine tripartite entanglement. For higher dimensional states this would not be possible, as the residual entanglement is defined via concurrences, which only work for qubits.

Tracing out one subsystem and calculating the squared concurrences of the remaining two-qubit states yields the following results:
\begin{align}
C^2(A|B)&=|{\mathcal Q}|^2=|\braket{\psi_0|\psi_1}|^2, \\
C^2(A|C)&=0, \\
C^2(B|C)&=0. 
\end{align}
Furthermore,
\begin{equation}
C^2(C|AB)=1- |{\mathcal Q}|^2=1-|\braket{\psi_0|\psi_1}|^2,
\end{equation}
and hence,
\begin{equation}
\tau_{res}=1- |{\mathcal Q}|^2=1-|\braket{\psi_0|\psi_1}|^2.
\end{equation}
It can be seen that the state offers maximal residual entanglement for ${\mathcal Q}=0$. Since ${\mathcal Q}=0$ corresponds to orthogonal $\ket{\psi_0}^{C}$ and $\ket{\psi_1}^{C}$, in this case the state is simply the GHZ state itself. However, for the other extreme ${\mathcal Q}=1$, and therefore $\ket{\psi_0}^{C}=\ket{\psi_1}^{C}$, the state is only a 2-entangled product state of the form $\ket{\Phi^+_2}^{AB}\ket{\psi_0}^{C}$, which shows maximal bipartite entanglement between systems A and B. Therefore, by tuning the overlap ${\mathcal Q}$ between the qumode states, it can be gradually switched between genuine tripartite entanglement and bipartite 2-entanglement between systems A and B. This corresponds to varying between scenario A) and scenario C) in figure \ref{fig:multiEnt}. Actually, the total entanglement in the state stays always constant and assumes the value of one ebit:
\begin{equation}
C^2(A|B)+C^2(A|C)+C^2(B|C)+\tau_{res}=1.
\end{equation}    
The state is therefore a 1-ebit entangled state for \textit{any} ${\mathcal Q}$. The overlap ${\mathcal Q}$ only determines the distribution between normal bipartite 2-entanglement and genuine tripartite entanglement.\footnote[8]{The W-state actually yields $C^2_W(A|B)+C^2_W(A|C)+C^2_W(B|C)+\tau_{{res}_W}=\frac{4}{3}$, which is greater than one. However, when only permitting GHZ-like 3-entanglement and "normal" 2-entanglement but no W-like 3-entanglement, in general $C^2(A|B)+C^2(A|C)+C^2(B|C)+\tau_{res}\leq1$. Hence, in this sense the state $\ket{\psi}^{ABC}$ can be regarded as maximally entangled.}

Tuning the overlap between two such qumode states should not be very difficult experimentally. The qumode states can be simply realized by two opposed coherent states $\ket{\pm\alpha}^C$, which results in ${\mathcal Q}=\mathrm{e}^{-2|\alpha|^2}$. Then the overlap can be tuned by simply adjusting the amplitude $\alpha$. In this way, relatively small overlaps can be achieved already for low amplitudes. The problem is rather to establish the entanglement itself between the qumode and the qubits.

\subsection{A Pure \texorpdfstring{${\mathcal H}_{2}\otimes{\mathcal H}_{\infty}\otimes{\mathcal H}_{\infty}$}{H2xHinftyxHinfty} State}
Now the related state,
\begin{equation}
\ket{\tilde{\psi}}^{ABC}=\frac{1}{\sqrt{2}}\biggl(\ket{0}^{A}\ket{\phi_0}^{B}\ket{\psi_0}^{C}+\ket{1}^{A}\ket{\phi_1}^{B}\ket{\psi_1}^{C}\biggr),
\end{equation}
which contains two CV subsystems and is supported by the Hilbert space ${\mathcal H}_{2}^A\otimes{\mathcal H}_{\infty}^B\otimes{\mathcal H}^C_{\infty}$, shall be discussed. Defining ${\mathcal Q}_{\phi}:=\braket{\phi_0|\phi_1}$ and ${\mathcal Q}_{\psi}:=\braket{\psi_0|\psi_1}$, a Gram-Schmidt process yields
\begin{align}
\ket{\phi_0}^{C}&=\ket{0}^{B}, \\
\ket{\phi_1}^{C}&={\mathcal Q}_{\phi}\ket{0}^{B} + \sqrt{1-|{\mathcal Q}_{\phi}|^2}\ket{1}^{B}, \\
\ket{\psi_0}^{C}&=\ket{0}^{C}, \\
\ket{\psi_1}^{C}&={\mathcal Q}_{\psi}\ket{0}^{C} + \sqrt{1-|{\mathcal Q}_{\psi}|^2}\ket{1}^{C},
\end{align}
and therefore,
\begin{equation}
\begin{aligned}
\ket{\tilde{\psi}}^{ABC}&=\frac{1}{\sqrt{2}}\biggl(\ket{0}^{A}\ket{0}^{B}\ket{0}^{C}+{\mathcal Q}_{\phi}{\mathcal Q}_{\psi}\ket{1}^{A}\ket{0}^{B}\ket{0}^{C}+\sqrt{(1-|{\mathcal Q}_{\phi}|^2)(1-|{\mathcal Q}_{\psi}|^2)}\ket{1}^{A}\ket{1}^{B}\ket{1}^{C} \\
&\qquad\quad +{\mathcal Q}_{\psi}\sqrt{1-|{\mathcal Q}_{\phi}|^2)}\ket{1}^{A}\ket{1}^{B}\ket{0}^{C}+{\mathcal Q}_{\phi}\sqrt{1-|{\mathcal Q}_{\psi}|^2)}\ket{1}^{A}\ket{0}^{B}\ket{1}^{C}\biggr).
\end{aligned}
\end{equation}
Once again, the state is effectively a three-qubit state in ${\mathcal H}_{2}^A\otimes{\mathcal H}_{2}^B\otimes{\mathcal H}^C_2$. Therefore, calculate the relevant squared concurrences:
\begin{align}
C^2(A|B)&=|{\mathcal Q}_{\psi}|^2(1-|{\mathcal Q}_{\phi}|^2)=|\braket{\psi_0|\psi_1}|^2(1-|\braket{\phi_0|\phi_1}|^2) \\
C^2(A|C)&=|{\mathcal Q}_{\phi}|^2(1-|{\mathcal Q}_{\psi}|^2)=|\braket{\phi_0|\phi_1}|^2(1-|\braket{\psi_0|\psi_1}|^2),  \\
C^2(B|C)&=0, \\ 
C^2(A|BC)&=1-|{\mathcal Q}_{\phi}|^2|{\mathcal Q}_{\psi}|^2=1-|\braket{\phi_0|\phi_1}|^2|\braket{\psi_0|\psi_1}|^2,
\end{align}
Hence,
\begin{equation}
\tau_{res}=(1-|{\mathcal Q}_{\phi}|^2)(1-|{\mathcal Q}_{\psi}|^2)=(1-|\braket{\phi_0|\phi_1}|^2)(1-|\braket{\psi_0|\psi_1}|^2),
\end{equation}
which is plotted in figure \ref{fig:resEnt}. 
\begin{figure}[ht]
\begin{center}
\includegraphics[width=10cm]{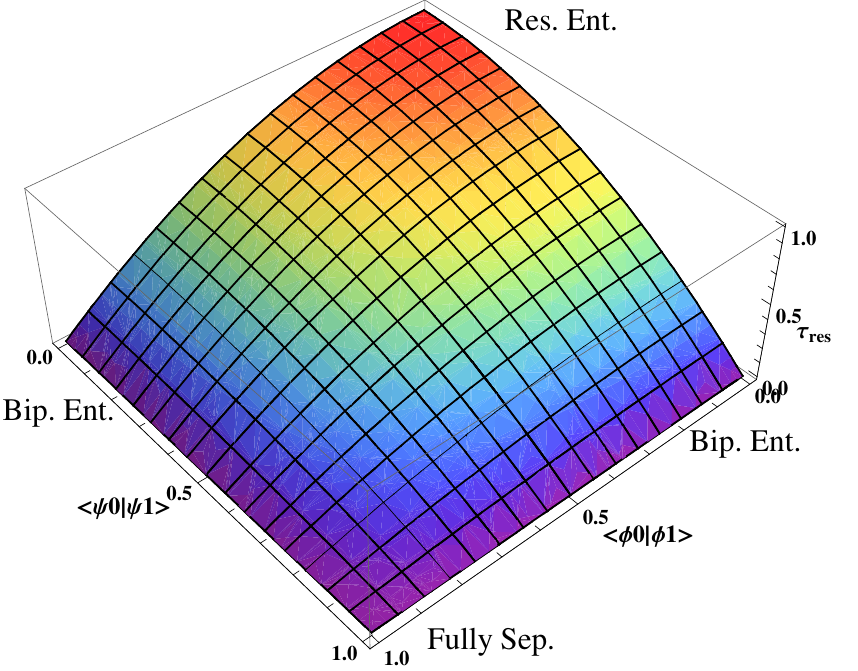}
\caption[The residual entanglement of the state $\ket{\psi}^{ABC}$ as a function of the overlaps $\braket{\phi_0|\phi_1}$ and $\braket{\psi_0|\psi_1}$.]{The residual entanglement of the state $\ket{\psi}^{ABC}$ as a function of the overlaps $\braket{\phi_0|\phi_1}$ and $\braket{\psi_0|\psi_1}$, which were assumed to be real without loss of generality. In the upper corner, for $\tau_{res}=1$, the state is maximally GHZ-like entangled. It shows no bipartite entanglement. In the left and the right corner, the residual entanglement is zero. Nevertheless, the state is maximally entangled, but just bipartite 2-entangled. Finally, in the lower corner, the state is fully separable. For either $\braket{\phi_0|\phi_1}$=0 or $\braket{\psi_0|\psi_1}$=0, the state is always 1-ebit entangled, while for both $\braket{\phi_0|\phi_1}\neq0$ and $\braket{\psi_0|\psi_1}\neq0$ the state has less than 1-ebit entanglement.}
\label{fig:resEnt}
\end{center}
\end{figure}
Furthermore, the total entanglement $C^2_{total}$ is given by
\begin{equation}
C^2_{total}=C^2(A|B)+C^2(A|C)+C^2(B|C)+\tau_{res}=1-|\braket{\phi_0|\phi_1}|^2|\braket{\psi_0|\psi_1}|^2.
\end{equation}
These results provide a basis for interesting discussions. Figure \ref{fig:resEnt} looks quite unspectacular. However, it contains actually interesting information in combination with the other results. First, it can be seen that for both overlaps $\braket{\phi_0|\phi_1}$ and $\braket{\psi_0|\psi_1}$ being zero, the state is once again exactly the GHZ-state, which shows only GHZ-like 3-entanglement but no bipartite entanglement. In the case that only one of the overlaps is zero, the state considered in the previous subsection is obtained. It is always 1-ebit entangled, while it can be tuned between GHZ-like entanglement and common bipartite 2-entanglement. If the overlap which is not zero becomes one, maximal Bell state-like 2-entanglement occurs. Finally, if none of the overlaps is zero the state is not 1-ebit entangled anymore and it stands in a mixture of bipartite 2-entanglement and GHZ-like 3-entanglement. If both overlaps are one, the state becomes simply fully separable, showing no entanglement at all. All this is basically reflected in figure \ref{fig:resEnt}. The upper corner of the graph corresponds to a GHZ-state, while the left and the right corners denote states which are maximally bipartite 2-entangled. Every point in the graph which is not at the upper left or upper right edge corresponds to a less than 1-ebit entangled state, while states lying on these edges are 1-ebit entangled. States on the lower edges of the graph are only 2-entangled, while all other states show also GHZ-like 3-entanglement. Finally, the lower corner corresponds to a fully separable state.

Coming to a conclusion, it is pointed out once again that it is quite remarkable that just by tuning simple overlaps between states, it is possible to gradually switch between all these different entanglement scenarios. Furthermore, tuning overlaps is a relatively easy task, when the qumode states are simply realized by ordinary coherent states. Then only the amplitudes have to be adjusted. However, the experimental preparation of the overall entangled tripartite state $\ket{\tilde{\psi}}^{ABC}$ may cause difficulties. Nevertheless, this idea of tuning between various entanglement configurations by modification of the overlaps of the participating qumode states may be a scheme which could possibly be experimentally realized in future.

\chapter{Hybrid Entanglement in Experiments}\label{ch:5}
In the previous chapter hybrid entanglement has been illuminated from a theoretical point of view. Different classes of HE were identified. Furthermore, approaches regarding the quantification as well as the detection of hybrid entanglement have been discussed in-depth on a theoretical level. Additionally, some peculiarities of multipartite HE have been examined. This is all very nice, but naturally two questions arise. On the one hand, how can HE actually be experimentally generated? And on the other hand, what is it useful for?

These questions are addressed in this chapter. First, an experimental scheme is introduced, which achieves the generation of pure entangled qubit-qumode states. When sending these systems through a channel, mixed HE or even true HE may be obtained. Second, three experimental applications are briefly described. On the one side, there are schemes regarding quantum key distribution as well as the generation of Schr\"odinger-cat states, which both involve pure entangled qubit-qumode states. On the other side, qubus appoaches for quantum communication and computation actually utilize multipartite HE.

\section{Generating Hybrid Entanglement and Cat State Preparation}\label{sec:Generation}
Consider a single 2-level system (qubit) interacting with a single mode of the quantized electromagnetic field. Such a configuration can be easily achieved by, for example, placing an appropriate atom inside an optical cavity which supports a single photonic qumode. The required 2-level "atom" may for instance be realized by a quantum dot \cite{QDot}. Then, in the \textit{Jaynes-Cummings-Paul model} (\textit{JCM}), applying the rotating-wave approximation, the system's Hamiltonian is given by
\begin{equation}
\hat{H}^{JCM}=\hbar\Omega\hat{a}^\dagger\hat{a}+\frac{1}{2}\hbar\omega\hat{\sigma}_z+\hbar g(\hat{\sigma}\hat{a}^\dagger+\hat{\sigma}^\dagger\hat{a}),
\end{equation}
where $\Omega$ is the electromagnetic field's frequency and $\omega$ is the atom's transition frequency. Furthermore, $\hat{a}$ as well as $\hat{a}^\dagger$ are the fields mode operators, while $\hat{\sigma}_z$ is the qubit's Pauli Z operator acting on the atom's ground and excited states, $\ket{0}$ and $\ket{1}$ respectively. Additionally, $\hat{\sigma}=\ket{0}\bra{1}$ as well as $\hat{\sigma}^\dagger=\ket{1}\bra{0}$ are the raising and lowering operators of the atom. Finally, the interaction strength is determined by the position dependent vacuum Rabi frequency
\begin{equation}
g(\vec{r})=\frac{|\vec{\mathfrak{p}}\,\vec{u}(\vec{r})|}{\hbar}{\mathcal E}_0,
\end{equation}
with the vacuum electric field inside the cavity ${\mathcal E}_0$, the atom's dipole moment $\vec{\mathfrak{p}}$, and the mode function $\vec{u}(\vec{r})$ of the resonator \cite{Schleich,JCM1,JCM2}. Note that the atom is assumed to be at rest. Otherwise, additionally a center-of-mass motion term $\frac{\hat{\vec{P}}^2}{2m}$ would occur in the Hamiltonian.

The JCM is one of the most fundamental models in quantum optics when considering atom-light interactions. It is in principle the most elementary model which is nevertheless able to describe most phenomena in cavity quantum electrodynamics. 

The most important term in the Hamiltonian is the interaction term $\hbar g(\hat{\sigma}\hat{a}^\dagger+\hat{\sigma}^\dagger\hat{a})$, which describes the dynamics in the JCM. In the interaction picture,
\begin{equation}
\hat{H}_{int}^{JCM}=\hbar g(\hat{\sigma}\hat{a}^\dagger\mathrm{e}^{i\Delta t}+\hat{\sigma}^\dagger\hat{a}\mathrm{e}^{-i\Delta t})
\end{equation}
is obtained, with the frequency detuning $\Delta:=\Omega-\omega$. Most interesting are the dynamics in the limits of this model: The resonant case $\Delta=0$ and the far off-resonant case with very large $\Delta$. As the far off-resonant case conserves photon numbers as well as atomic populations, and introduces only detuning dependent phase shifts, it is also called the \textit{dispersive limit}. This is also the regime which is of relevance for the preparation of HE states. Hence, we shall no longer consider the resonant case.

More precisely, the dispersive limit is reached if
\begin{equation}
\frac{2g\sqrt{n+1}}{|\Delta|}\ll 1,
\end{equation}
with photon number $n$. In this case the effective Hamiltonian,
\begin{equation}
\hat{H}^{eff}=-\frac{\hbar g^2}{\Delta}\Bigl(\hat{\sigma}_z\hat{a}^\dagger\hat{a}+\frac{1}{2}(\hat{\sigma}_z+{\mathbb 1})\Bigr),
\end{equation} 
can be used to describe the system \cite{Schleich}. Once again, the interesting term, necessary for the generation of HE, is only the interaction term. Hence, just consider
\begin{equation}
\hat{H}^{eff}_{int}=-\hbar\chi\hat{\sigma}_z\hat{a}^\dagger\hat{a}\,,\qquad\chi:=\frac{g^2}{\Delta}.
\end{equation} 
Then the corresponding unitary evolution operator is given by
\begin{equation}
\hat{U}=\mathrm{exp}[-\frac{i}{\hbar}\hat{H}^{eff}_{int}t]=\mathrm{exp}[i\chi\hat{\sigma}_z\hat{a}^\dagger\hat{a}\,t],
\end{equation}
which represents a simple rotation in the qumode's phase space \cite{PvL}. Hence, define $\varphi:=\chi t$ and rewrite this equation as
\begin{equation}
\hat{U}=\hat{R}(\varphi\hat{\sigma}_z)=\mathrm{exp}[i\varphi\hat{\sigma}_z\hat{a}^\dagger\hat{a}].
\end{equation}
$\hat{R}(\varphi\hat{\sigma}_z)$ corresponds to a conditional rotation, acting on the whole Hilbert space of the qubit and the qumode. 
\begin{figure}[ht]
\begin{center}
\includegraphics[width=7cm]{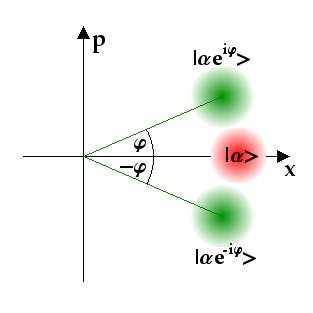}
\caption[A coherent state $\ket{\alpha}$ with $\alpha\in{\mathbb R}$ is conditionally rotated by $\pm\varphi$ in phase space.]{A coherent state $\ket{\alpha}$ with $\alpha\in{\mathbb R}$ is conditionally rotated by $\pm\varphi$ in phase space, depending on the state of the qubit. This procedure results in the qumode states $\ket{\alpha\,\mathrm{e}^{\pm i\varphi}}$.}
\label{fig:GeneratingHE}
\end{center}
\end{figure}
Applying it on a coherent state $\ket{\alpha}$ yields 
\begin{equation}
\begin{aligned}
\hat{R}(\varphi\hat{\sigma}_z)\ket{\alpha}&=\mathrm{exp}[i\varphi\hat{\sigma}_z\hat{n}]\,\mathrm{e}^{-\frac{|\alpha|^2}{2}}\sum_{n=0}^\infty\,\frac{\alpha^n}{\sqrt{n!}}\ket{n} \\
&=\mathrm{e}^{-\frac{|\alpha|^2}{2}}\sum_{n=0}^\infty\,\frac{(\alpha\,\mathrm{e}^{i\varphi\hat{\sigma}_z})^n}{\sqrt{n!}}\ket{n} \\
&=\ket{\alpha\,\mathrm{e}^{i\varphi\hat{\sigma}_z}}.
\end{aligned}
\end{equation}
The state is rotated in phase space by a value of $\pm\varphi$ depending on the state of the qubit (see figure \ref{fig:GeneratingHE}). At this point it becomes obvious that in the dispersive limit only phase shifts are acquired, but no photon numbers or atomic populations are touched. The most interesting case occurs, when the qubit is in a superposition state $\frac{1}{\sqrt{2}}(\ket{0}+\ket{1})$:
\begin{equation}\label{eq:GenerateHE}
\begin{aligned}
\hat{R}(\varphi\hat{\sigma}_z)\ket{\alpha}\otimes\frac{1}{\sqrt{2}}\Bigl(\ket{0}+\ket{1}\Bigr) &=\ket{\alpha\,\mathrm{e}^{i\varphi\hat{\sigma}_z}}\otimes\frac{1}{\sqrt{2}}\Bigl(\ket{0}+\ket{1}\Bigr) \\
&=\frac{1}{\sqrt{2}}\Bigl(\ket{\alpha\,\mathrm{e}^{i\varphi}}\ket{0}+\ket{\alpha\,\mathrm{e}^{-i\varphi}}\ket{1}\Bigr).
\end{aligned}
\end{equation}
Finally, this is a hybrid entangled state between a qubit and a qumode of the form 
\begin{equation}
\ket{\psi}_{HE}^{AB} =c_{0}\ket{0}^A\ket{\psi_{0}}^B+c_{1}\ket{1}^A\ket{\psi_{1}}^B,
\end{equation}
as considered in paragraph \ref{subsec:PQbQm}. Concluding, it can be summarized that the JCM provides an efficient way for the generation of HE between a qumode and a qubit. Of course the qumode may also consecutively interact with several qubits, or similarly, several modes may possibly interact with a single or several qubits. This would result in a multipartite HE state. Since $\varphi=\chi t$ with typically very short interaction times, only small phases can be acquired. Nevertheless, for sufficiently large amplitude $\alpha$, the overlap between the states $\{\ket{\alpha\,\mathrm{e}^{\pm i\varphi}}\}$ can become arbitrarily small.

Finally, from the state of equation \eqref{eq:GenerateHE} CSS states can be created by a suitable projection measurement on the qubit \cite{cats4,PvL}. Simply perform a projection on the conjugate $\hat{\sigma}_x$-basis $\ket{\pm}=\frac{1}{\sqrt{2}}(\ket{0}\pm\ket{1})$ and obtain cat states of the form
\begin{equation}
\frac{1}{\sqrt{N_\pm}}\Bigl(\ket{\alpha\,\mathrm{e}^{i\varphi}}\pm\ket{\alpha\,\mathrm{e}^{-i\varphi}}\Bigr).
\end{equation} 
Hence, besides the hybrid cat state generation technique with aid of CV squeezed vacuum and DV photon subtraction \cite{cats1, cats2}, which has been presented in section \ref{subsec:HybridApps}, this approach constitutes an alternative way that exploits hybrid entanglement.

Summarizing, in this section an efficient scheme for the generation of qubit-qumode HE has been introduced, which also sets out the basis for general multipartite HE creation. Furthermore, a first application of HE has been presented - the CSS state generation.

\section{Quantum Key Distribution}\label{sec:QKD}
Another interesting application of hybrid entangled states is found in \textit{quantum key distribution} (\textit{QKD}) \cite{QKDrev}. In QKD two parties, \textit{Alice} and \textit{Bob}, aim at distributing a shared random key, which only they know, and which can be used for encrypting and decrypting messages. They make use of the fundamental laws of quantum mechanics to guarantee security. An eavesdropper, \textit{Eve}, must somehow measure the quantum states, which Alice and Bob exchange, to gain information. However, in quantum mechanics measurements disturb the system. Therefore, Alice and Bob are able to detect Eve. Furthermore, the amount of intercepted information can be calculated and bounded. If the shared information between Alice and Bob is greater than the information Eve could get, nevertheless a secret key can be distilled via privacy amplification, effectively eliminating Eve's amount of information. In privacy amplification the shared key, which Eve has partial knowledge about, is used to distill a shorter key, which Eve has less knowledge about. This can be repeated continuously to, in principle, completely zero out Eve's information about the final key. Finally, the shared secret key can be used for the execution of a one-time pad ciphering to interchange unconditionally secure messages.\footnote[1]{For a detailed introduction to QKD see the review \cite{QKDrev}.}  

In principle, there are two different approaches towards QKD. On the one hand, there are so-called \textit{prepare and measure} schemes, where Alice prepares an entangled state first, measures her part, and finally sends the remaining part through an insecure domain to Bob. On the other hand, there are \textit{entanglement-based} approaches in which an entangled state is provided by an untrusted third party. Afterwards Alice and Bob perform measurements on this state and basically quantum teleport the information from Alice to Bob. An example for a prepare and measure approach is the well-known BB84 protocol by Bennett and Brassard, which is actually the first suggested QKD protocol at all \cite{BB84}. In 1991, Ekert presented the first entanglement-based scheme, the E91 protocol \cite{E91}.

At the end the \textit{prepare and measure} protocols as well as the \textit{entanglement-based} schemes are more or less similar. The only difference is, where the entangled state is prepared. Hence, for a given prepare and measure protocol always the corresponding entanglement-based setup can be written down, and vice versa. It is the same key point in both configurations which determines whether unconditionally secure QKD can be successfully performed or not: The correlated data between Alice and Bob have to violate some Bell inequality \cite{QKDbell}. Hence, in other words, the data have to correspond to an entangled state. Therefore, security in QKD protocols is in principle an inseparability problem. Actually, there are schemes which involve hybrid entangled states. Such a protocol shall be presented now \cite{QKDh1,QKDh2}.

Suppose Alice prepares randomly entangled states of the form 
\begin{equation}
\ket{\psi}^{AB}_n=\frac{1}{\sqrt{2}}\Bigl(\ket{0}^A\ket{\mathrm{e}^{in\frac{\pi}{2}}\alpha}^B+\ket{1}^A\ket{\mathrm{e}^{i(n+2)\frac{\pi}{2}}\alpha}^B\Bigr),
\end{equation}
with two bases $n=0\lor 1$ (compare to figure \ref{fig:4x2state}). Then she measures her part of the state and sends the remaining state to Bob through an insecure domain controlled by Eve. Bob has to make measurements such that he can distinguish between $\ket{\mathrm{e}^{in\frac{\pi}{2}}\alpha}^B$ and $\ket{\mathrm{e}^{i(n+2)\frac{\pi}{2}}\alpha}^B$ (see figure \ref{fig:QKD}).
\begin{figure}[ht]
\begin{center}
\includegraphics[width=13cm]{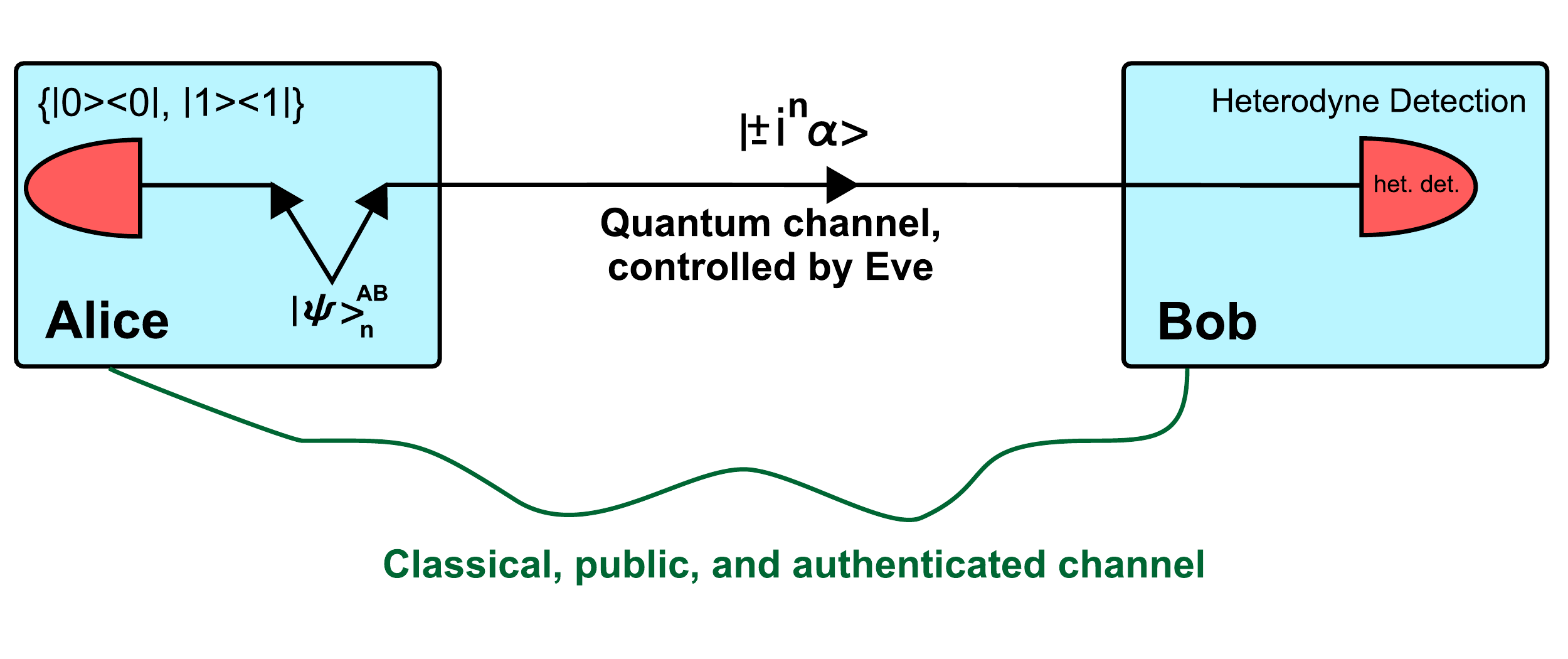}
\caption[The explained QKD scheme: Alice prepares entangled states $\ket{\psi}^{AB}_n$, and measures her part of the state. The remaining coherent states are sent through an unsecure domain to Bob, who performs for example heterodyne measurements on the state \cite{QKDh2}.]{The explained QKD scheme: Alice prepares entangled states $\ket{\psi}^{AB}_n$, and measures her part of the state. The remaining coherent states are sent through an unsecure domain to Bob, who performs for example heterodyne measurements on the state \cite{QKDh2}. Also an unambiguous state discrimination setup as introduced in section \ref{subsec:HybridApps} would be reasonable.}
\label{fig:QKD}
\end{center}
\end{figure}
So he chooses either the $n=0$ or $n=1$ basis, performs appropriate measurements, and distinguishes between $\ket{\mathrm{e}^{in\frac{\pi}{2}}\alpha}^B$ and $\ket{\mathrm{e}^{i(n+2)\frac{\pi}{2}}\alpha}^B$ for the chosen $n$. Since $\alpha$ is chosen such that the states have a significant overlap, the distinguishment between the states can only be performed with a certain error probability. This is a very important feature in this scheme and actually in most QKD protocols. After the transmission Alice and Bob talk to each other via a classical, public, and authenticated channel, and tell each other, which basis they have chosen.\footnote[2]{It is essential to employ an \textit{authenticated} channel for the classical communication after the measurements. If this is not the case, the protocol will be vulnerable to \textit{man-in-the-middle attacks}, where Eve impersonates Bob, when talking to Alice, and vice versa.} They keep only the bits which have been measured for an equal choice. Furthermore, all bits corresponding to uncertain measurement results are discarded and only bits are kept which Bob knows with high certainty. This process of cheosing only favorable events is called postselection and is one of the key ingredients necessary for successful QKD \cite{QKDx}.

Consider an eavesdropper Eve. She has mainly two possibilities for attacks. The \textit{intercept and resend attack} as well as the \textit{beam-splitting attack}.

An intercept and resend attack works in the way that she simply measures the state, prepares a new state corresponding to her measurement result and resends it. However, Eve also has to choose a basis for the measurement to obtain a reasonable result. Then there are two cases: On the one hand, Eve may choose the right basis, probably learns the bit's value and resends the right state (due to the overlap, she may also obtain a wrong measurement result and resend a wrong state. However, this may be noticed by Alice and Bob, as will become clear soon). Alice and Bob can not notice this. On the other hand, she may choose the wrong basis. Then she resends the state also in the other basis and Bob may measure the wrong bit, even though he chooses the initially right basis corresponding to Alice's basis choice. In this case Alice and Bob notice Eve: When they compare the bases they have been applied and find that they have used the same bases but measured too often different bits, they know that they have been eavesdropped. These measurement results which they compare are the control bits. The statistics of these measurement results will reveal any eavesdropper. Of course, these control bits cannot be used for the key anymore, as the classical authenticated channel is public. If Alice and Bob perform a sufficiently large number of control measurements and observe anomalous statistics in the measurements, they can conclude that they have been eavesdropped and can start the key transmission again. Hence, this protocol is secure against intercept and resend attacks. Actually, it is not even necessary to begin with a new key transmission as due to the postselection Alice and Bob share a greater amount of information than Alice and Eve do, anyway. This enables key distillation. It will become clear in the next paragraph when discussing beam-splitting attacks.

So, furthermore, to gain information about the key Eve can tap parts of the signal by a beam splitter with transmission $\eta$ (beam-splitting attack). Then Eve gets an amount of $1-\eta$ of the signal and therefore states of the form $\ket{\sqrt{1-\eta}\mathrm{e}^{in\frac{\pi}{2}}\alpha}^B$ and $\ket{\sqrt{1-\eta}\mathrm{e}^{i(n+2)\frac{\pi}{2}}\alpha}^B$. However, this also yields attenuated states for Bob. When performing a large amount of measurements Bob notices that the attenuation is higher than expected, and hence can conclude that the transmission has been eavesdropped. Nevertheless, especially for very large $\eta$ he may possibly not notice. Even so, Eve will not be able to take possession of the whole key. For her measurements she also randomly chooses the basis $n$, independent of Bob's later choice. It is a crucial point that her measurements are fully independent of Bob's measurements. Furthermore, she also listens to the classical conversation between Alice and Bob, in which they compare the chosen bases. With the control measurements Alice and Bob cannot find Eve. She remains hidden. Therefore, the key corresponding to the chosen bits is actually used. Due to the postselection, Bob has only kept bits which correspond to measurements with very certain results and the right choice of basis. However, it is highly improbable, actually nearly impossible for a large number of events, that in case of a "good" measurement of Bob, also Eve always performs a "good" measurement, for which she knows the bit with high certainty. The postselection of favorable "good" measurements is done by Bob and not by Eve. The postselected events do not, in general, all correspond to "good" measurements of Eve. Therefore, Eve has more erroneous relevant measurements than Bob. Hence, the amount of shared information between Alice and Bob, $I_{AB}$, is larger than the amount of information held between Alice and Eve, $I_{AE}$. This translates into only partial knowledge of the key. Finally, Alice and Bob can perform privacy amplification to eliminate Eve's possible knowledge about the distilled key completely. Concluding, the scheme is also secure against beam-splitting attacks. Note that, as mentioned, the possibility of key distillation due to postselection also occurs in the earlier discussed scenario for intercept and resend attacks. Even though Eve is actually detected in this case, a key can nevertheless be distilled. 
 
In reality, things are slightly more difficult, as the channel from Alice to Bob is not only under the control of Eve but also possibly introduces further losses and noise. These effects may distort the states in such a way that no data can be obtained which is suited for key distillation. However, as noted earlier, QKD can be performed if the data shows nonclassical quantum correlations. Hence, denoting the channel by $\Upsilon^B$, the security of the key depends on whether $\hat{\rho}'_n{}^{AB}=({\mathbb 1}^A\otimes\Upsilon^B)\ket{\psi}^{AB}_n\bra{\psi}$ is entangled or not. Actually, $\hat{\rho}'_n{}^{AB}$ is a hybrid entangled state, and depending on the channel it may also show true hybrid entanglement. The detection of entanglement in states of this form has been discussed in detail in subsection \ref{subsecB}, where the channel was considered to be lossy and noisy. Note that there are further approaches for witnessing entanglement in this class of states, which are based on special EVMs related to the original SV determinants \cite{QKDh2,QKDh3}.
 
Concluding, hybrid entanglement and especially its detection is highly relevant for certain QKD protocols.

\section{Quantum Bus Schemes}\label{sec:qubus}
As explained earlier, entanglement is a key ingredient for most applications in quantum information. However, how is entanglement between, for example, two qubits experimentally established? Nonlinear interactions between the qubits to be entangled are necessary for this task. On the one hand, there are measurement-induced nonlinearities, which are for example heavily exploited in the quantum computation scheme by KLM \cite{KLM}. These approaches are highly non-deterministic, though. On the other hand, weak nonlinearities are actual physical nonlinear effects. They are typically much too weak for the realization of appropriate entangling gates when dealing, for example, with single photons. An elegant way out of this dilemma is provided by the already mentioned \textit{quantum bus} (\textit{qubus}) approach, which utilizes an intense qumode to mediate and enhance the interaction between two qubits to be entangled \cite{qubus1,qubus2,qubus3,qubus4}. In principle, the qumode interacts subsequently with both qubits, resulting in a tripartite qubit-qubit-qumode HE state (see figure \ref{fig:qubus1}). Afterwards, an appropriate measurement on the qumode is performed forcing the qubits to remain in an entangled state.  
\begin{figure}[ht]
\begin{center}
\includegraphics[width=12cm]{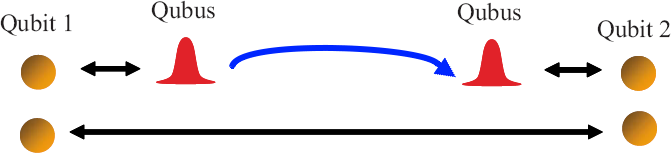}
\caption[The quantum bus as entanglement mediator: It entangles itself subsequently with two qubits. Finally, a measurement on the qubus establishes entanglement between the qubits \cite{qubus4}. This figure is a modified version of Fig. 1a in \cite{qubus4}. The figure was used with permission by Peter van Loock.]{The quantum bus entangles itself with the first qubit, then travels to the second qubit and also interacts with it. Afterwards an appropriate measurement on the qubus is performed to establish entanglement between the qubits \cite{qubus4}. This figure is a modified version of Fig. 1a in \cite{qubus4}. The figure was used with permission by Peter van Loock.}
\label{fig:qubus1}
\end{center}
\end{figure}

This shall be explained in slightly more detail now, according to the references \cite{qubus4,PvL}. Start with a coherent state $\ket{\alpha}$ as qubus and a first qubit which is in the state $\frac{1}{\sqrt{2}}(\ket{0}+\ket{1})$. If they interact dispersively as explained in subsection \ref{sec:Generation}, the resulting hybrid entangled state has the form
\begin{equation}
\frac{1}{\sqrt{2}}\biggl(\ket{0}+\ket{1}\biggr)\otimes\ket{\alpha}\quad\rightarrow\quad\frac{1}{\sqrt{2}}\biggl(\ket{0}\ket{\alpha}+\ket{1}\ket{\alpha\,e^{i\theta}}\biggr),
\end{equation}
up to an irrelevant phase space rotation. The qubus subsystem is now sent through a channel to the second qubit. However, this channel may introduce losses. These are again modeled by coupling to a vacuum environment as illustrated in figure \ref{fig:ADampCh}. Hence, after the channel
\begin{equation}
\begin{aligned}
\frac{1}{\sqrt{2}}\biggl(\ket{0}\ket{\alpha}+\ket{1}\ket{\alpha\,e^{i\theta}}\biggr) \rightarrow & \;\,\rho=\frac{1}{2}\biggl(\ket{0}\bra{0}\otimes\ket{\sqrt{\eta}\alpha}\bra{\sqrt{\eta}\alpha} \\ & + \ket{1}\bra{1}\otimes\ket{\sqrt{\eta}\alpha \,e^{i\theta}}\bra{\sqrt{\eta}\alpha\,e^{i\theta}} \\ & +  A \ket{0}\bra{1}\otimes\ket{\sqrt{\eta}\alpha}\bra{\sqrt{\eta}\alpha\,e^{i\theta}} \\ & + A^\ast \ket{1}\bra{0}\otimes\ket{\sqrt{\eta}\alpha\,e^{i\theta}}\bra{\sqrt{\eta}\alpha} \biggr),
\end{aligned}
\end{equation}
where $\eta$ denotes the transmissivity of the beam splitter, and $A:=\braket{\sqrt{1-\eta}\alpha \,e^{i\theta}|\sqrt{1-\eta}\alpha}$. Now the qumode interacts again dispersively with the second qubit, resulting in another controlled phase rotation of the qumode. This interaction then yields a tripartite HE state of the form 
\begin{align}
\rho & \rightarrow  F\ket{\Psi^+}\bra{\Psi^+}+(1-F)\ket{\Psi^-}\bra{\Psi^-}, \\
\ket{\Psi^\pm} & :=  \frac{1}{\sqrt{2}}\ket{\sqrt{\eta}\alpha}\ket{\Phi_2^\pm}\pm\frac{1}{2}e^{-i\varphi}\ket{\sqrt{\eta}\alpha e^{i\theta}}\ket{10}+\frac{1}{2}e^{i\varphi}\ket{\sqrt{\eta}\alpha e^{-i\theta}}\ket{01}, \\
\ket{\Phi^\pm_2} & =  \frac{1}{\sqrt{2}}(\ket{00}\pm\ket{11}), \\
\varphi & :=  \eta\alpha^2\sin\theta, \\
F & :=  \frac{1}{2}[1+e^{-(1-\eta)\alpha^2(1-\cos\theta)}].
\end{align}
Finally, during the last step of the scheme the qubus disentangles from the qubits, projecting them on two-qubit entangled states. This proceeds via an appropriate measurement on the qubus mode. Homodyne detection along the x-axis is experimentally most efficient, discriminating between $\ket{\sqrt{\eta}\alpha}$ and $\ket{\sqrt{\eta}\alpha e^{\pm i\theta}}$ \cite{hybridRep}. This projects the qubits on approximate versions of $\ket{\Phi^+_2}$ or $\frac{1}{\sqrt{2}}(\ket{10}+\ket{01})$. However, as the distance on the x-axis between the peaks of the coherent states scales as $\sim \alpha\theta^2$ for small $\theta$, the amplitude $\alpha$ has to be very large to compensate for the small $\theta$ and to suppress overlap errors.  Unfortunately, this increases the decoherence induced by imperfect channel transmission, too. A schematic of the whole experimental setup is shown in figure \ref{fig:qubus2}.
\begin{figure}[ht]
\begin{center}
\includegraphics[width=12cm]{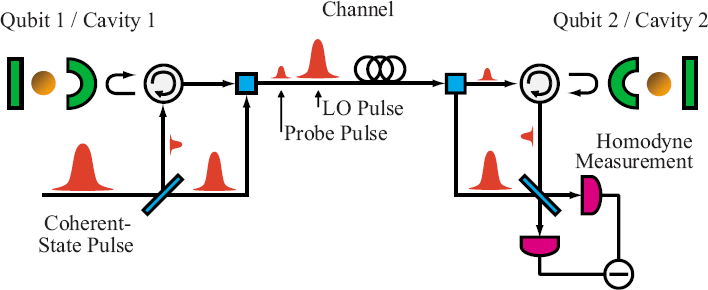}
\caption[Possible experimental realization of the quantum bus scheme, which distributes entanglement between two separated qubits \cite{qubus4,hybridRep}. This figure is taken from \cite{qubus4} with permission by Peter van Loock.]{Possible experimental realization of the quantum bus scheme, which distributes entanglement between two separated qubits. It is important to note that for the homodyne detection a strong local oscillator (LO) pulse has to be taken along \cite{qubus4,hybridRep}. This figure is taken from \cite{qubus4} with permission by Peter van Loock.}
\label{fig:qubus2}
\end{center}
\end{figure}
Note that there are also other ways for measurements. For example homodyne detection along the p-axis on the one hand results in lower overlap errors, yielding higher fidelities than in the x-axis homodyning scheme. On the other hand not all outcomes of the measurement correspond to a two-qubit state which remains entangled. Hence, several outcomes must be discarded, rendering the scheme less efficient. Another approach consists in utilizing POVMs and performing unambiguous state discrimination instead of homodyne detection \cite{qubus4}. In the end, for every measurement scheme there is some trade-off between fidelity and success probability.

Finally, it is worth mentioning that for a qubus implementation inside, for example, a hybrid quantum repeater setup, also hybrid entanglement distillation as well as hybrid entanglement swapping has to be performed \cite{qubus4}. These shall not be discussed here, however, for more information, see the reference.

It can be concluded that not just bipartite hybrid entanglement is experimentally relevant, but also multipartite HE. Quantum bus schemes are quite popular, since they are very important for the experimental realization of several tasks, especially regarding quantum communication. Hence, it is important to also investigate the involved multipartite hybrid entanglement, which has been done briefly in this thesis in the previous chapter.

\chapter{Summary and Conclusion}\label{ch:6}
The central theme of this thesis has been the investigation of hybrid entanglement with a focus on bipartite systems. So far, no rigorous and systematic analysis of hybrid entanglement has been performed. This has been the goal of this thesis, and actually, a lot of new insights were obtained. However, also several other related questions, which are to some extent associated with the subtleties of hybrid entanglement and have therefore arisen during the studies, have been discussed and in most cases successfully answered.

The thesis starts with a brief introduction, {\textbf{chapter \ref{ch:2}}}, explaining some fundamentals of quantum information theory. However, not really basic anymore, in subsection \ref{catwitnessing} general entanglement witnessing with aid of the SV criteria in two-mode cat states is discussed, and a generalized determinant is derived which achieves entanglement detection for states of the form $\frac{1}{\sqrt{{\mathcal N}_\phi}}(\ket{\alpha,\alpha}+\mathrm{e}^{i\phi}\ket{-\alpha,-\alpha})$ for any phase $\phi$. This is quite interesting, but, nonetheless, the witnessing is only demonstrated for pure states and would be much more relevant in the mixed state case. Further investigations should concern the performance of the generalized determinant for the corresponding mixed states obtained when sending the considered two-mode cats through decoherence-causing channels. Additionally, also even more general two-mode CSS states, such as the ones in equations \eqref{eq:generalCats1} and \eqref{eq:generalCats2}, may be analyzed with the SV criteria. Finally, general issues on entanglement witnessing have been discussed, which bring out some very interesting questions to be tackled in the future: Are witnesses which work well for pure states also the right ones for witnessing in the corresponding mixed states, which are obtained by sending the pure states through decoherence-causing channels? Does there always exist an \textit{optimal} witness in the sense that it makes use of \textit{all} entanglement of a given state? If yes, it could be expected that this witness also detects entanglement in the decohered states if still present. Concluding, besides the two-mode cat state witnessing, also several interesting questions regarding entanglement witnessing have been briefly presented and discussed. It is exciting to await future results concerning these open problems.

In {\textbf{chapter \ref{ch:3}}}, the notion of \textit{hybrid} as used in this thesis is defined. Furthermore, some motivations, why it is reasonable to "go hybrid", are discussed. Concerning this matter, the most important point probably comes from quantum computation and the question, how to realize universality. Hence, a very important protocol for hybrid quantum computation, the GKP scheme, is briefly presented and discussed \cite{GKP}. Especially important in this regard are the different notions of nonlinearity. On the one hand, there are measurement-induced nonlinearities, which are heavily utilized by KLM in their teleportation based scheme for DV quantum computation, while, on the other hand, weak nonlinearities are actual physical nonlinearities, which are for example exploited and enhanced in the subsequently discussed qubus schemes (see section \ref{sec:qubus}) \cite{KLM,qubus1,qubus2,qubus3,qubus4}. Additionally, chapter \ref{ch:3} summarizes some applications involving hybrid approaches and provides a lot of reference material. Concluding, this chapter is mainly a motivation for the analysis of the central theme of the thesis, the hybrid entanglement, and serves as an overview over suggested applications which make use of hybrid approaches.

The main topic of the thesis is discussed in {\textbf{chapter \ref{ch:4}}}. It starts with a proof that every correlated hybrid state is non-Gaussian. This fact is not a priori clear and therefore not trivial. Then, the Gram-Schmidt process is discussed and a modified version, which has been called inverse Gram-Schmidt process, is derived. The latter is actually the central tool, when dealing with hybrid entanglement, and gives rise to a new classification scheme of hybrid entangled states: On the one side, states which contain a finite number of qumode states can, with aid of the inverse Gram-Schmidt process, always be expressed in finite orthonormal bases. Therefore, they are effectively DV hybrid entangled states. On the other side, for states possessing an infinite number of qumode states, the inverse Gram-Schmidt process is not applicable anymore. Hence, these states "stay" in their hybrid Hilbert space and are called \textit{truly} hybrid entangled. This classification constitutes the central result of this thesis and it is quite relevant: Since hybrid entangled states are initially always non-Gaussian and live in an infinite-dimensional Hilbert space, exact entanglement quantification is impossible. However, the effectively DV hybrid entangled states can be described by proper density matrices after the application of the inverse Gram-Schmidt process. Hence, all the tools, such as several entanglement measures or, in the pure state case, the Schmidt decomposition, can be applied. It is quite remarkable that many of these complicated non-Gaussian and infinite-dimensional Hilbert space states can be handled so easily in the end. It is worth pointing out that the Gram-Schmidt process can be also applied on certain fully CV systems for which all subsystems are CV. The only requirement is that the superposition sum contains only a finite number of qumodes in each subsystem. Then, also these states can be handled with DV tools. Examples for this kind of CV states are, for example, the pure two-mode cat states, which have been considered in this thesis. 

In subsections \ref{sec:pure}-\ref{sec:mixed2}, a large variety of exemplary states is presented and analyzed. For several states which are either too complicated for DV measures, or which are truly hybrid entangled and therefore stay in the non-Gaussian infinite-dimensional regime, exact entanglement quantification is not possible. However, at least entanglement detection can be achieved with aid of the SV criteria, which therefore are one of the most important tools for studying hybrid entanglement. Most important during in investigations are states of the form $\frac{1}{\sqrt{2}}(\ket{0}\ket{\alpha}+\ket{1}\ket{-\alpha})$, as they are relevant in certain QKD schemes (compare section \ref{sec:QKD}) \cite{QKDh2}. Entanglement quantification of states like this, when sent through a lossy channel, has been demonstrated via the concurrence. Furthermore, true hybrid entanglement is obtained, when taking the channel not only as lossy but also as noisy. Using this example it was shown that the SV criteria also work for the new true hybrid entanglement. Also another example of a true hybrid entangled state has been considered, \eqref{eq:trueHE2}, for which entanglement witnessing was accomplished.

When dealing with true HE, infinite sums or integrals always occur. It is important not to truncate these and by this obtain approximate results. A truncation would actually correspond to the analysis of a different state, strictly speaking not even a truly HE state anymore, but an effectively DV HE state. It appeared that this complicates the analysis of truly hybrid entangled states considerably. This has led to the examination of general issues regarding entanglement witnessing in mixed states with an infinite number of mix terms. It has been shown that, when the global robustness of the truncated state (truncation in the convex combination of mix terms, i.e. pure state projectors) is higher than a certain amount $1-p_N$ depending on where truncation has been performed, the overall state with an infinite number of mix terms is entangled. Additionally, it was derived that, for every entangled state of this form, a truncated state exists such that its global robustness is greater than $1-p_N$. Due to the complicated calculation of robustness monotones these theorems are of limited practical use. Nevertheless, they provide insight into the connections between entanglement witnessing and entanglement quantification. Furthermore, one may actually think of special scenarios, where no witnessing can be performed for the given states, while for any truncated state, robustness calculations can be performed. Hence, these theorems can still be considered as a nice result, especially for the case that, in the future, efficient algorithms for the computation of the robustness are found.

In addition, not only bipartite HE has been investigated, but also multipartite HE (section \ref{sec:multihe}). While in the bipartite case true hybrid entanglement only occurs for certain kinds of mixed states (the ones with an infinite number of mixed terms) it has been shown that in the multipartite case true HE may be also obtained for pure states. Any hybrid system which contains at least two CV subsystems can be truly HE. When the overall system contains only one CV subsystem, true HE only appears for the mentioned mixed states. Additionally, investigations regarding the entanglement structure of tripartite HE states have been performed. It has been worked out that, depending on the overlaps between the involved qumode states, either no entanglement at all, bipartite entanglement between two subsystems, GHZ-like genuine tripartite entanglement, or a mix of the latter is obtained. This is a very interesting finding, as it is quite easy to tune overlaps: Just use coherent states and adjust their amplitudes. Possibly, this continuous switching between different entanglement scenarios via overlap tuning is something which can be achieved experimentally in the future. 

Finally, to round off the thesis and show that hybrid entanglement can be actually experimentally realized and is also physically relevant, {\textbf{chapter \ref{ch:5}}} discusses the generation of hybrid entanglement via the dispersive interaction between a qubit and a qumode in the Jaynes-Cummings-Paul model. Additionally, some applications, which employ hybrid entangled states, are briefly explained. First, via projection measurements on the DV part of qubit-qumode entangled states, Schr\"odinger cat states can be generated. Second, a certain QKD protocol is presented, which exploits hybrid entanglement between qubits and binary coherent states $\{\ket{\pm\alpha}\}$. The necessary requirement for unconditionally secure key distribution then corresponds to inseparability of the mentioned entangled states when subject to the quantum channel, which connects Alice with Bob. Fittingly, in this thesis entanglement witnessing in states of this kind is discussed quite in detail (for the thermal channel). These investigations are therefore of very high relevance. Third, the quantum bus principle is introduced. A qubus mediates a nonlinear interaction between two separated qubits and can therefore serve as an entangling gate, important for several applications in quantum computation and quantum communication. The approach makes use of tripartite hybrid entangled states, showing that also studying multipartite entanglement is very important.  

Altogether, it can be concluded that the thesis provides a lot of new insights into hybrid entanglement and associated topics. The motivation of this work has come from questions concerning the classification of hybrid entanglement, its detection and quantification as well as its generation and its applications. A classification scheme has been derived and the detection of hybrid entanglement via the SV criteria has been demonstrated. Exact entanglement quantification can be either performed like in the DV case, or it cannot be accomplished at all, since the state is non-Gaussian and infinite-dimensional. A generation scheme has been presented and three applications have been discussed.

%Appendices, Reihenfolge noch Auftreten im Text anzupassen
\clearpage
\phantomsection
\addcontentsline{toc}{chapter}{Appendices}

\begin{appendix}

\chapter{Notations}
\section{Abbreviations}
Throughout this thesis several abbreviations are used for the sake of appropriate readability. Here is an overview.
\vspace{1cm}
\begin{table}[!ht]\renewcommand{\arraystretch}{1.3}\addtolength{\tabcolsep}{-1pt}
\begin{center}
\begin{tabular}{@{}*{4}{l}@{}}
\toprule
\textbf{Abbreviation} & \textbf{Meaning}  \\ 
\midrule
CKW & Coffman, Kundu and Wootters \\
CP & Completely Positive  \\
CPTP & Completely Positive, Trace-Preserving  \\
CSS & Coherent State Superposition, Schr\"odinger-Cat State  \\
CV & Continuous Variable(s)  \\
DV & Discrete Variable(s)  \\
EVM & Expectation Value Matrix  \\
GHZ & Greenberger-Horne-Zeilinger \\
GKP & Gottesman, Kitaev and Preskill \\
HE & Hybrid Entanglement / Hybrid Entangled  \\
JCM & Jaynes-Cummings-Paul Model \\
KLM & Knill, Laflamme and Milburn \\
LOCC & Local Operations and Classical Communication  \\
NPT & Negative Partial Transpose  \\
PnCP & Positive but not Completely Positive  \\
POVM & Positive Operator-Valued Measure  \\
PPT & Positive Partial Transpose  \\
PT & Partial(ly) Transpose(d)  \\
QIT & Quantum Information Theory  \\
QKD & Quantum Key Distribution  \\
SV & Shchukin-Vogel \\
TMSS & Two-Mode Squeezed State \\
\bottomrule
\end{tabular}
\end{center}
\end{table}
\newpage

\section{Mathematical
Notation}
In addition to the summary of abbreviations here is also an overview over the mathematical notations used throughout the thesis.
\vspace{1cm}
\begin{table}[!ht]\renewcommand{\arraystretch}{1.3}\addtolength{\tabcolsep}{-1pt}
\begin{center}
\begin{tabular}{@{}*{4}{l}@{}}
\toprule
\textbf{Notation} & \textbf{Meaning}  \\ 
\midrule
$\mathbb R,\;\mathbb C$ & Field of real numbers, field of complex numbers  \\
${\mathbb R}^n$ & Set of real $n$-dimensional vectors \\
$\mathbb N,\;{\mathbb N}_0,\;{\mathbb Z}$ & Natural numbers, natural numbers including zero, integers   \\
${\mathcal H}(\mathbb C)$ & Hilbert space over the field of the complex numbers  \\
${\mathcal H}_d^A$ & $d$-dimensional Hilbert space of system A, qudit space for finite $d$  \\
${\mathcal H}^A_\infty$ & Infinite-dimensional Hilbert space of system A, qumode space  \\
${\mathbb{1}}_n,\;{\mathbb{O}}_n$ & Identity operator in $n$ dimensions, zero operator in $n$ dimensions  \\
$M^\ast$ & Elementwise complex conjugate of matrix $M$: $(M^\ast)_{ij}=M^\ast_{ij}$  \\
$M^T$ & Transpose of matrix $M$: $(M^T)_{ij}=M_{ji}$ \\
$M^\dagger$ & Hermitian adjoint of matrix $M$, complex conjugation \\
 & and transposition: $M^\dagger=(M^T){}^\ast$  \\
$\hat{A}^{\Gamma_B}$ & Partial transposition of operator $\hat{A}$ with respect to system $B$ \\
$\hat{A}\otimes\hat{B}$ & Tensor product of operators $\hat{A}$ and $\hat{B}$  \\
$\hat{A}\oplus\hat{B}$ & Direct sum of operators $\hat{A}$ and $\hat{B}$ \\
$\hat{\sigma}_x$, $\hat{\sigma}_y$, $\hat{\sigma}_z$ & Pauli matrices \\
$\mathrm{tr}[\hat{A}]$ & Complete trace of operator $\hat{A}$: $\mathrm{tr}[\hat{A}]=\sum_n\braket{n|\hat{A}|n}$ \\
$\mathrm{tr}_B[\hat{A}]$, $\mathrm{tr}_{{\mathcal H}^B}[\hat{A}]$ & Partial trace of operator $\hat{A}$ with respect to Hilbert space of \\
 & system $B$: $\mathrm{tr}_B[\hat{A}]=\sum_n\braket{n_B|\hat{A}|n_B}$ \\
$\det[M]$ & Determinant of matrix $M$ \\
$\dim[\mathcal V]$ & Dimension of vector space $\mathcal V$ \\
$\mathrm{Sp}(2n,\mathbb R)$ & Real, symplectic group on ${\mathbb R}^{2n}$ \\
$\hat{S}_{sym}(\hat{x}^n,\hat{p}^m)$ & Symmetrization of operators $\hat{x}^n,\hat{p}^m$ \\
$SEP,\;QS$ & Set of separable quantum states $\hat{\rho}_{sep}$, set of all quantum states $\hat{\rho}$ \\
$:\hat{A}:$ & Normally ordered form of operator $\hat{A}$ \\
$\lfloor x\rfloor$ & Floor of the real number $x$: $\lfloor x\rfloor=\max\{m\in{\mathbb Z}:m\leq x\}$  \\
$\Upsilon\hat{\rho}$ & CPTP map $\Upsilon(\hat{\rho})$ on state $\hat{\rho}$  \\
$({\mathbb 1}^A\otimes\Upsilon^B)\hat{\rho}^{AB}$ & One-sided channel: CPTP map $\Upsilon_B$ on subsystem B of the state $\hat{\rho}^{AB}$  \\
$x\lor y$ & Logical disjunction, "or"-conjunction: $x$ or $y$  \\
$]a,b[$ & Open interval of real numbers between $a$ and $b$: $]a,b[:=\{x\in{\mathbb R}:a<x<b\}$  \\
$\{\alpha_i\}\prec\{\beta_i\}$ & Set $\{\alpha_i\}$ is majorized by set $\{\beta_i\}$  \\
\bottomrule
\end{tabular}
\end{center}
\end{table}

\chapter{Calculations and Formulas}
This part of the appendix provides auxiliary material. For the understanding of the thesis often only theorems themselves are important but not their proofs. These proofs are shifted here. Furthermore, ancillary calculations are presented and some important definitions and formulas are listed.

\section{The Hadamard Lemma}\label{Hamamard}
\begin{lem}[Hadamard's Lemma]
For bounded operators $\hat{X}$ and $\hat{Y}$,
\begin{equation}
\mathrm{e}^{\hat{X}} \hat{Y} \mathrm{e}^{-\hat{X}}=\sum_{m=0}^\infty\frac{[\hat{X},\hat{Y}]_m}{m!},
\end{equation}
with $[\hat{X},\hat{Y}]_m=[\hat{X},[\hat{X},\hat{Y}]_{m-1}]$ and $[\hat{X},\hat{Y}]_0=\hat{Y}$.
\end{lem}
\begin{proof}[\textbf{Proof.}]
Consider
\begin{equation}
F(t)=\mathrm{e}^{t\hat{X}} \hat{Y} \mathrm{e}^{-t\hat{X}}\,,\qquad F(0)=\hat{Y}.
\end{equation}
Then
\begin{equation}
\frac{dF(t)}{dt}=[\hat{X},F].
\end{equation}
This differential equation has the solution
\begin{equation}
F(t)=\sum_{n=0}^\infty \frac{t^n}{n!}[\hat{X},\hat{Y}]_n,
\end{equation}
which can be seen by evaluation of the derivative:
\begin{equation}
\begin{aligned}
\frac{dF(t)}{dt} &= \sum_{n=1}^\infty \frac{n\,t^{n-1}}{n!}[\hat{X},\hat{Y}]_n \\
 &= \sum_{n=0}^\infty \frac{\,t^{n}}{n!}[\hat{X},\hat{Y}]_{n+1} \\
 &= \hat{X}\sum_{n=0}^\infty \frac{\,t^{n}}{n!}[\hat{X},\hat{Y}]_{n}-\sum_{n=0}^\infty \frac{\,t^{n}}{n!}[\hat{X},\hat{Y}]_{n}\,\hat{X} \\
 &= [\hat{X},F].
\end{aligned}
\end{equation}
Evaluated at $t=1$, this yields
\begin{equation}
\mathrm{e}^{\hat{X}} \hat{Y} \mathrm{e}^{-\hat{X}}=\sum_{m=0}^\infty\frac{[\hat{X},\hat{Y}]_m}{m!},
\end{equation}
which completes the proof.
\end{proof}

\section{The Schmidt Decomposition}\label{Schmidtproof}
\begin{thm}[Schmidt decomposition]
Let $\ket{\psi}^{AB}$ be a normalized pure bipartite state in the product Hilbert space ${\mathcal H}^A_{n}\otimes{\mathcal H}^B_{m}$ with $\dim[{\mathcal H}^A_{n}]=n$ and $\dim[{\mathcal H}^B_{m}]=m$. Then there exist orthonormal bases $\{\ket{u_i}^A\}$ and $\{\ket{w_i}^B\}$ such that
\begin{equation}
\ket{\psi}^{AB}=\sum_{i=1}^{r(\psi)} \sqrt{\alpha_i}\ket{u_i}^A\ket{w_i}^B\,,\qquad \alpha_i>0\;\forall \;i,\;\sum_i \alpha_i=1,
\end{equation}
with $r(\psi)\leq min\{n,m\}$. Such a decomposition is called \textup{Schmidt decomposition} of $\ket{\psi}^{AB}$.
\end{thm}
\begin{proof}[\textbf{Proof.}] \cite{Audretsch} 
Expand $\ket{\psi}^{AB}$ in the orthonormal bases $\{\ket{k}^A:1\leq k\leq n\}$ of ${\mathcal H}^A_{n}$ and $\{\ket{i}^B:1\leq i\leq m\}$ of ${\mathcal H}^B_{m}$. Assume $n\leq m$ without loss of generality.
\begin{equation}
\ket{\psi}^{AB}=\sum_{k,i=1}^{n,m}a_{ki}\ket{k}^{A}\ket{i}^B.
\end{equation}
Introduce the relative states
\begin{equation}
\ket{\tilde{w}_k}^{B}:=\sum_{i=1}^{m}a_{ki}\ket{i}^B,
\end{equation}
leading to
\begin{equation}\label{eq:schmidtproofx1}
\ket{\psi}^{AB}=\sum_{k=1}^{n} \ket{k}^{A}\ket{\tilde{w}_k}^{B}.
\end{equation}
The relative states $\{\ket{\tilde{w}_k}^{B}\}$ are in general neither orthogonal nor normalized. Orthogonality can be obtained in the following way:

1. For $\{\ket{k}^A\}$ choose the orthonormal eigenvectors $\{\ket{u_k}^A\}$ of $\hat{\rho}^A$:
\begin{equation}\label{eq:schmidtproof1}
\hat{\rho}^A=\sum_{k=1}^{n} p_k\ket{u_k}^A\bra{u_k}\,,\qquad p_k\geq 0\;\forall \;k,\;\sum_{k=1}^{n} p_k=1.
\end{equation}
Let $p_k>0$ for $1\leq k\leq r$ and $p_k=0$ for $r+1\leq k\leq n$. The eigenvalues $p_k$ can be degenerate. The $\{\ket{u_k}^A\}$ are determined only up to a phase for non-degenerate eigenvalues.

2. The relative states $\{\ket{\tilde{w}_k}^{B}\}$ are now associated with the $\{\ket{u_k}^A\}$. Consider subsystem $\hat{\rho}^A$:
\begin{equation}\label{eq:schmidtproof2}
\begin{aligned}
\hat{\rho}^A &= \mathrm{tr}_B[\,\ket{\psi}^{AB}\bra{\psi}\,] \\
 &=\mathrm{tr}_B[\,\sum_{k,l}^n \ket{u_k}^A\bra{u_l}\otimes\ket{\tilde{w}_k}^{B}\bra{\tilde{w}_l}\,] \\
 &=\sum_{k,l}^n \braket{\tilde{w}_l|\tilde{w}_k} \ket{u_k}^A\bra{u_l}.
\end{aligned}
\end{equation}

Comparison of equation \eqref{eq:schmidtproof1} with equation \eqref{eq:schmidtproof2} leads to the desired orthogonality of the relative states:
\begin{equation}\label{eq:schmidtproofx2}
\braket{\tilde{w}_l|\tilde{w}_k}^B=p_k\delta_{kl}.
\end{equation}
For $k\geq r+1$, the $\{\ket{\tilde{w}_k}^{B}\}$ are null vectors. Combining equations \eqref{eq:schmidtproofx1} and \eqref{eq:schmidtproofx2} as well as identifying $\{p_k\}$ with $\{\alpha_k\}$ completes the proof.
\end{proof}

\section{The Vacuum Density Operator}\label{VacuumDensityOperator}
\begin{lem}
Consider a qumode system with mode operators $\hat{a}$ and $\hat{a}^\dagger$. The vacuum density operator $\ket{0}\bra{0}$ can then be written as
\begin{equation}
\ket{0}\bra{0}=\,:\mathrm{e}^{-\hat{a}^\dagger\hat{a}}:\;,
\end{equation}
where $:\cdots:$ denotes normal ordering.
\end{lem}
\begin{proof}[\textbf{Proof.}]
1. It is known that
\begin{equation}
(x+y)^n=\sum_{k=0}^n \binom n k x^{n-k}y^k.
\end{equation}
With $x=1$ and $y=-1$,
\begin{equation}\label{eq:binom1}
\sum_{k=0}^n (-1)^k \binom n k =\delta_{0n}.
\end{equation}
2. The vacuum density operator $\ket{0}\bra{0}$ is defined by its action on a Fock state $\ket{n}$,
\begin{equation}
\ket{0}\braket{0|n}=\delta_{0n}\ket{0}.
\end{equation}
Consider $:\mathrm{e}^{-\hat{a}^\dagger\hat{a}}:\ket{n}$:
\begin{equation}
\begin{aligned}
:\mathrm{e}^{-\hat{a}^\dagger\hat{a}}:\ket{n} &= \sum_{m=0}^{\infty} \frac{(-1)^m}{m!}\hat{a}^{\dagger^m}\hat{a}^m\,\ket{n} \\
&= \sum_{m=0}^{n} \frac{(-1)^m}{m!}\frac{n!}{(n-m)!}\,\ket{n} \\
&= \sum_{m=0}^{n} (-1)^m \binom n m \,\ket{n} \\
&= \delta_{0n}\ket{0},
\end{aligned}
\end{equation}
where equation \eqref{eq:binom1} has been exploited.
\end{proof}

\section{The Heaviside Step Function}\label{Heaviside}
In subsection \ref{catwitnessing}, the so-called \textit{Heaviside step function} $\Theta_{\frac{1}{2}}$ has been used. However, varying definitions can be found in literature. Therefore, here is the definition used in this thesis:
\begin{equation}
\Theta_{\frac{1}{2}}(x):=\begin{cases}
0 :& x<0, \\
\frac{1}{2} :& x=0, \\
1 :& x>0.
\end{cases}
\end{equation}
Further properties of this function are of no relevance for this thesis. Just note that it is related to the Dirac delta function via differentiation. However, in contrast to the delta function it can be also used as an actual function and is not only distributively defined. Furthermore, smooth approximations of $\Theta_{\frac{1}{2}}$ are given by so-called logistic functions, for example 
\begin{equation}
\Theta_{\frac{1}{2}}(x)=\lim_{k\rightarrow\infty}\Bigl(\frac{1}{2}+\frac{1}{\pi}\arctan(kx)\Bigr).
\end{equation}

\section{Hermite Polynomials}\label{Hermite}
The \textit{Hermite polynomials} $H_n(x)$ (\cite{Bronstein, Temme}) are defined by
\begin{equation}
H_n(x):=(-1)^n\mathrm{e}^{x^2}\frac{d^n}{dx^n}\mathrm{e}^{-x^2}.
\end{equation}
They are polynomials of order $n$. The first five Hermite polynomials are 
\begin{align}
H_0(x) & =1, \\
H_1(x) & =2x, \\ 
H_2(x) & =4x^2-2, \\
H_3(x) & =8x^3-12, \\
H_4(x) & =16x^4-48x^2+12.
\end{align}
However, there is also an explicit expression,
\begin{equation}
H_n(x)=n!\sum_{i=0}^{\lfloor\frac{n}{2}\rfloor}\frac{(-1)^i}{i!(n-2i)!}(2x)^{n-2i},
\end{equation}
where the floor function has been used. It is defined as
\begin{equation}
\lfloor x\rfloor=\max\{m\in{\mathbb Z}:m\leq x\}.
\end{equation}
Additionally, the polynomials can be calculated recursively with the relation
\begin{equation}
H_{n+1}(x)=2xH_n(x)-2nH_{n-1}(x),
\end{equation}
which yields together with $H_0(x)$ and $H_1(x)$ an efficient way for the actual calculation of the polynomials. 

Furthermore, Hermite's polynomials are orthogonal with respect to the weightfunction $\mathrm{e}^{-x^2}$,
\begin{equation}
\int\limits_{-\infty}^\infty H_n(x)H_m(x)\mathrm{e}^{-x^2}dx=2^nn!\sqrt{\pi}\delta_{nm},
\end{equation}
and they form an orthogonal basis of the Hilbert space of functions $f(x)$ for which
\begin{equation}
\int\limits_{-\infty}^\infty |f(x)|^2\mathrm{e}^{-x^2}dx<\infty.
\end{equation}

Beside their stong importance in quantum mechanics, where they are relevant for the construction of the the energy eigenstates of the quantum harmonic oscillator, the Hermite polynomials find applications in the finite element method for the numerical solution of partial differential equations.

\section{Kraus Operators of the Thermal Channel}\label{Apx:KrausOpTh}
In paragraph \ref{subsecB}, the thermal photon noise channel has been investigated (recall figure \ref{fig:NoiseCh}, which explains the modeling of the channel.). Furthermore, a set of Kraus operators, describing the channel, has been derived, which is presented now. Consider the thermal channel
\begin{equation}
\Upsilon^A_{thermal}(\hat{\rho}^A)=\mathrm{tr}_E[\hat{U}^{AE}\;\bigl(\hat{\rho}^A\otimes\hat{\rho}^E_{thermal}\bigr) \;\hat{U}^{{AE}^\dagger}],
\end{equation}
with 
\begin{equation}
\hat{U}^{AE}=e^{\theta(\hat{a}_E^\dagger \hat{a}_A-\hat{a}_A^\dagger \hat{a}_E)},
\end{equation}
and
\begin{equation}
\hat{\rho}^E_{thermal}=\sum_{n=0}^\infty \rho_n^{th.}\ket{n}^{E}\bra{n}\,,\qquad\rho_n^{th.}=\frac{\braket{n_{th.}}^n}{(1+\braket{n_{th.}})^{n+1}}.
\end{equation}
This can be rewritten in the form
\begin{equation}
\Upsilon^A_{thermal}(\hat{\rho}^A)=\sum_{n=0}^\infty\rho_n^{th.}\sum_{m=0}^\infty \hat{K}_{mn}\hat{\rho}^A\hat{K}_{mn}^\dagger,
\end{equation}
with 
\begin{equation}
\hat{K}_{mn}=\braket{m|\hat{U}^{AE}|n}^E.
\end{equation}
Defining $\hat{\tilde{K}}_{mn}:=\sqrt{\rho_n^{th.}}\hat{K}_{mn}$ yields
\begin{equation}
\Upsilon^A_{thermal}(\hat{\rho}^A)=\sum_{n,m=0}^\infty \hat{\tilde{K}}_{mn}\hat{\rho}^A\hat{\tilde{K}}_{mn}^\dagger.
\end{equation}
Now, merge the indices $m$ and $n$ into an index $k={\underline k}=(m,n)$ via an appropriate index ordering procedure. For example consider ${\underline k}_1=(m,n)$, ${\underline k}_2=(\tilde{m},\tilde{n})$ and order such that
\begin{equation}
{\underline k}_1 < {\underline k}_2 \Leftrightarrow
\begin{cases}
|{\underline k}_1|<|{\underline k}_2| & \text{or} \\
|{\underline k}_1|=|{\underline k}_2| & \text{and}\; {\underline k}_1 <' {\underline k}_2,
\end{cases}
\end{equation}
with $|{\underline k_1}|=m+n$, and ${\underline k}_1 <' {\underline k}_2$ denotes that the first nonzero difference $\tilde{m}-m,\tilde{n}-n$ is positive. This is by the way the ordering, which has also been applied in the derivation of the SV criteria in section \ref{subsec:witnessing}. Making use of this ordering,
\begin{equation}
\Upsilon^A_{thermal}(\hat{\rho}^A)=\sum_{k=0}^\infty \hat{\tilde{K}}_{k}\hat{\rho}^A\hat{\tilde{K}}_{k}^\dagger
\end{equation}
is obtained, which corresponds to an operator-sum decomposition of the thermal channel.

The Kraus operators $\hat{K}_{mn}$ are derived in the Fock basis as well as in the coherent state basis. The transformation from $\hat{K}_{mn}$ to the desired operators $\hat{\tilde{K}}_{mn}=\hat{\tilde{K}}_{k}$ is then trivial.

\textbf{1. Fock basis.}

Consider $\bra{m}\hat{U}^{AE}\ket{n}^E\ket{k}^A$:
\begin{equation}
\begin{aligned}
\bra{m}\hat{U}^{AE}\ket{n}^E\ket{k}^A & = \frac{1}{\sqrt{k!n!}}\bra{m}\hat{U}^{AE}(\hat{a}^{E^\dagger})^n(\hat{a}^{A^\dagger})^k\ket{0}^E\ket{0}^A \\
&=\frac{1}{\sqrt{k!n!}}\bra{m}(\sqrt{\eta}\hat{a}^{E^\dagger}-\sqrt{1-\eta}\hat{a}^{A^\dagger})^n(\sqrt{\eta}\hat{a}^{A^\dagger}+\sqrt{1-\eta}\hat{a}^{E^\dagger})^k\ket{0}^E\ket{0}^A \\
&=\frac{1}{\sqrt{k!n!}}\sum_{i=0}^n\sum_{j=0}^k {\binom n i}{\binom k j}\bra{m}\sqrt{\eta}^i(\hat{a}^{E^\dagger})^i(-\sqrt{1-\eta})^{n-i}(\hat{a}^{A^\dagger})^{n-i} \\
&\qquad\cdot\sqrt{\eta}^{k-j}(\hat{a}^{A^\dagger})^{k-j}\sqrt{1-\eta}^j(\hat{a}^{E^\dagger})^j\ket{0}^E\ket{0}^A \\
&=\frac{1}{\sqrt{k!n!}}\sum_{i=0}^n\sum_{j=0}^k {\binom n i}{\binom k j}\sqrt{\eta}^{k-j+i}\sqrt{1-\eta}^{n-i+j}(-1)^{n-i} \\
&\qquad\cdot\braket{m|i+j}^E\sqrt{(i+j)!(n+k-i-j)!}\ket{n+k-i-j}^A \\
&=\frac{1}{\sqrt{k!n!}}\sum_{i=0}^{\min\{n,m\}} {\binom n i}{\binom k {m-i}}\sqrt{\eta}^{k-m+2i}\sqrt{1-\eta}^{n+m-2i}(-1)^{n-i} \\ &\qquad\cdot\sqrt{m!(n+k-m)!}\ket{n+k-m}^A \\
&=\sqrt{\frac{m!(n+k-m)!}{k!n!}}\kappa_{mkn}(\eta)\ket{n+k-m}^A, 
\end{aligned}
\end{equation}
with 
\begin{equation}
\kappa_{mkn}(\eta):=\sum_{i=0}^{\min\{n,m\}} {\binom n i}{\binom k {m-i}}\sqrt{\eta}^{k-m+2i}\sqrt{1-\eta}^{n+m-2i}(-1)^{n-i}.
\end{equation}
Hence,
\begin{equation}
\hat{K}_{mn}^{Fock}=\sum_{k=\max\{0,m-n\}}^\infty \sqrt{\frac{m!(n+k-m)!}{k!n!}}\kappa_{mkn}(\eta)\ket{n+k-m}^A\bra{k}.
\end{equation}

\textbf{2. Coherent state basis.}

Consider $\bra{m}\hat{U}^{AE}\ket{n}^E\ket{\alpha}^A$:
\begin{equation}
\begin{aligned}
\bra{m}\hat{U}^{AE}\ket{n}^E\ket{\alpha}^A & = \frac{1}{\sqrt{n!}}\bra{m}\hat{U}^{AE}(\hat{a}^{E^\dagger})^n\hat{D}^A(\alpha)\ket{0}^E\ket{0}^A \\
&=\frac{1}{\sqrt{n!}}\bra{m}(\sqrt{\eta}\hat{a}^{E^\dagger}-\sqrt{1-\eta}\hat{a}^{A^\dagger})^n\hat{D}^E(\sqrt{1-\eta}\alpha)\hat{D}^A(\sqrt{\eta}\alpha)\ket{0}^E\ket{0}^A \\
&=\frac{1}{\sqrt{n!}}\sum_{i=0}^n{\binom n i}\sqrt{\eta}^i(-\sqrt{1-\eta}\hat{a}^{A^\dagger})^{n-i}\sqrt{\frac{m!}{(m-i)!}}\braket{m-i|\sqrt{1-\eta}\alpha}^E\ket{\sqrt{\eta}\alpha}^A \\
&=\frac{1}{\sqrt{n!}}\sum_{i=0}^{\min\{n,m\}}{\binom n i}\sqrt{\frac{m!}{(m-i)!}}\sqrt{\eta}^i(-\sqrt{1-\eta}\hat{a}^{A^\dagger})^{n-i} \\
&\qquad\cdot\mathrm{e}^{-\frac{1}{2}(1-\eta)|\alpha|^2}\frac{(\sqrt{1-\eta}\alpha)^{m-i}}{\sqrt{(m-i)!}}\ket{\sqrt{\eta}\alpha}^A.
\end{aligned}
\end{equation}
Therefore,
\begin{equation}
\begin{aligned}
\hat{K}_{mn}^{Coh.} & =\frac{1}{\sqrt{n!}}\sum_{i=0}^{\min\{n,m\}}{\binom n i}\sqrt{\frac{m!}{(m-i)!}}\sqrt{\eta}^i(-\sqrt{1-\eta}\hat{a}^{A^\dagger})^{n-i} \\
&\qquad\cdot\int_{\mathbb C}d^2\alpha\,\mathrm{e}^{-\frac{1}{2}(1-\eta)|\alpha|^2}\frac{(\sqrt{1-\eta}\alpha)^{m-i}}{\sqrt{(m-i)!}}\ket{\sqrt{\eta}\alpha}^A\bra{\alpha}\delta(\hat{a}^A-\alpha).
\end{aligned}
\end{equation}
Shorter,
\begin{equation}
\hat{K}_{mn}^{Coh.}=\sum_{i=0}^{\min\{n,m\}}\varkappa_{nmi}(\eta)(\hat{a}^{A^\dagger})^{n-i}\int_{\mathbb C}d^2\alpha\,\varpi_{mi}(\eta,\alpha)\ket{\sqrt{\eta}\alpha}^A\bra{\alpha}\delta(\hat{a}^A-\alpha), 
\end{equation}
with
\begin{equation}
\varkappa_{nmi}(\eta):={\binom n i}\sqrt{\frac{m!\eta^i}{n!(m-i)!}}(-\sqrt{1-\eta})^{n-i},
\end{equation}
and
\begin{equation}
\varpi_{mi}(\eta,\alpha):=\mathrm{e}^{-\frac{1}{2}(1-\eta)|\alpha|^2}\frac{(\sqrt{1-\eta}\alpha)^{m-i}}{\sqrt{(m-i)!}}.
\end{equation}
The delta distribution $\delta(\zeta)$, which has been made use of, is defined by
\begin{equation}
\int_{\mathbb C}d^2\zeta\,\delta(\zeta-\zeta_0)f(\zeta)=f(\zeta_0),
\end{equation}
with $\zeta,\zeta_0\in\mathbb C$. Furthermore, make use of
\begin{equation}
\delta\bigl(f[\hat{a}]\bigr)\ket{\alpha}=\delta\bigl(f[\alpha]\bigr)\ket{\alpha},
\end{equation}
for coherent states $\ket{\alpha}$.

Finally, note that the Kraus operators $\hat{K}_{m0}^{Fock/Coh.}$ for $n=0$ are the Kraus operators of the amplitude damping channel, which can be also found in \cite{Nielsen}.

\section{Integral Identities}\label{Apx:Ints}
In subsection \ref{subsecB}, some important integrals have been used. These shall be calculated here.
\begin{itemize}
\item The integral $\int_{\mathbb R}dx\,\mathrm{e}^{-\frac{x^2}{a}}$.
\begin{equation}\label{Gint1}
\int_{\mathbb R}dx\,\mathrm{e}^{-\frac{x^2}{a}}=\sqrt{a}\int_{\mathbb R}dy\,\mathrm{e}^{-y^2},
\end{equation}
with substitution $x=\sqrt{a}y$. Hence, focus on the so-called \textit{Gaussian integral}
\begin{equation}
\int_{\mathbb R}dx\,\mathrm{e}^{-x^2}.
\end{equation}
On the one hand,
\begin{equation}\label{Gint2}
\int_{\mathbb R}dx\,\mathrm{e}^{-x^2}=\sqrt{\Biggl(\int_{\mathbb R}dx\,\mathrm{e}^{-x^2}\Biggr)^2}.
\end{equation}
On the other hand,
\begin{equation}
\begin{aligned}
\Biggl(\int_{\mathbb R}dx\,\mathrm{e}^{-x^2}\Biggr)^2 &=\int_{\mathbb R}\int_{\mathbb R}dx\,dy\,\mathrm{e}^{-(x^2+y^2)}=\int_{\mathbb R^2}dA\,\mathrm{e}^{-(x^2+y^2)} \\
&=\int\limits_{0}^\infty\int\limits_{0}^{2\pi}dr\,d\phi\,r\mathrm{e}^{-r^2}=2\pi\int\limits_{0}^\infty \frac{ds}{2\sqrt{s}}\,\sqrt{s}\mathrm{e}^{-s} \\
&=\pi\int\limits_{0}^\infty ds\,\mathrm{e}^{-s}=-\pi\Bigl[\mathrm{e}^{-s}\Bigr]^{\infty}_0 \\
&=\pi,
\end{aligned}
\end{equation}
where $dA$ denotes surface integration, which has been performed in polar coordinates. Furthermore, $r=\sqrt{s}$ has been substituted. With this result and equations \eqref{Gint1} and \eqref{Gint2},
\begin{equation}
\int_{\mathbb R}dx\,\mathrm{e}^{-\frac{x^2}{a}}=\sqrt{a\pi}.
\end{equation}

\item The integral $\int_{\mathbb R}dx\,x\mathrm{e}^{-\frac{x^2}{a}}$.

This is easy. While $x$ is an odd function, $\mathrm{e}^{-\frac{x^2}{a}}$ is an even function. Hence, $x\mathrm{e}^{-\frac{x^2}{a}}$ is odd and therefore,
\begin{equation}
\int_{\mathbb R}dx\,x\mathrm{e}^{-\frac{x^2}{a}}=0.
\end{equation}

\item The integral $\int_{\mathbb R}dx\,x^2\mathrm{e}^{-\frac{x^2}{a}}$.
\begin{equation}
\begin{aligned}
\int_{\mathbb R}dx\,x^2\mathrm{e}^{-\frac{x^2}{a}} &=\int_{\mathbb R}dx\; a^2\frac{\partial}{\partial a}\mathrm{e}^{-\frac{x^2}{a}}=a^2\frac{\partial}{\partial a}\int_{\mathbb R}dx\;\mathrm{e}^{-\frac{x^2}{a}} \\
&=a^2\frac{\partial}{\partial a}\sqrt{a\pi}=\frac{\sqrt{a^3\pi}}{2}.
\end{aligned}
\end{equation}
\end{itemize}

\section{Geometric Series}\label{Apx:Series}
In section \ref{subsec:OtherTrueHE} some formulas regarding infinite geometric series have been exploited. These shall be proved here.
\begin{itemize}
\item The series $\sum_{i=1}^\infty x^i$ with $0<x<1$.
\begin{equation}
(1-x)\sum_{i=0}^n x^i=\sum_{i=0}^n x^i-x\sum_{i=0}^n x^i=1-x^{n+1}.
\end{equation}
Hence,
\begin{equation}
\sum_{i=0}^n x^i=\frac{1-x^{n+1}}{1-x}.
\end{equation}
For $n=\infty$,
\begin{equation}
\sum_{i=0}^\infty x^i=\frac{1}{1-x}.
\end{equation}
Finally,
\begin{equation}
\sum_{i=1}^\infty x^i=\sum_{i=0}^\infty x^i-1=\frac{1}{1-x}-1=\frac{x}{1-x}.
\end{equation}
\item The series $\sum_{i=1}^\infty i\,x^i$ with $0<x<1$.
\begin{equation}
\begin{aligned}
\sum_{i=1}^\infty i\,x^i&=\sum_{i=1}^\infty x\frac{\partial}{\partial x}x^i=x\frac{\partial}{\partial x}\sum_{i=1}^\infty x^i \\
&=x\frac{\partial}{\partial x}\frac{x}{1-x}=\frac{x}{(1-x)^2}.
\end{aligned}
\end{equation}
\end{itemize}

\end{appendix}

%List of Figures
% AM ENDE ZU PR�FEN OB RICHTIG SITZT !!! !!! !!! !!! !!!
\cleardoublepage
\phantomsection
\addcontentsline{toc}{chapter}{List of Figures}
\listoffigures

%Bibliogrphy
%\thispagestyle{empty}
%\cleardoublepage
\nocite{*}
\bibliographystyle{alphaurl}
\bibliography{DAbib}

\newcommand{\etalchar}[1]{$^{#1}$}
\begin{thebibliography}{ODTBG07}

\bibitem[AGM06]{QKDbell}
A.~Acin, N.~Gisin, and L.~Masanes.
\newblock From {B}ell's theorem to secure quantum key distribution.
\newblock {\em Phys. Rev. Lett. 97, 120405}, 2006.

\bibitem[AP04]{Isom1}
P.~Arrighi and Ch. Patricot.
\newblock On quantum operations as quantum states.
\newblock {\em Ann. Phys. 311, 1}, pages 26--52, 2004.

\bibitem[Aud07]{Audretsch}
J.~Audretsch.
\newblock {\em Entangled Systems}.
\newblock WILEY-VCH, 2007.

\bibitem[Ban99]{Banaszek}
K.~Banaszek.
\newblock Optimal receiver for quantum cryptography with two coherent states.
\newblock {\em Phys. Lett. A 253, 1}, pages 12--15, 1999.

\bibitem[BB84]{BB84}
C.~H. Bennett and G.~Brassard.
\newblock Quantum cryptography: Public key distribution and coin tossing.
\newblock In {\em Proceedings of the IEEE International Conference on
  Computers, Systems, and Signal Processing, Bangalore, India}, pages 175--179,
  1984.

\bibitem[BBC{\etalchar{+}}93]{Qtel1}
C.~H. Bennett, G.~Brassard, C.~Crepeau, R.~Jozsa, A.~Peres, and W.~K. Wootters.
\newblock Teleporting an unknown quantum state via dual classical and
  {E}instein-{P}odolsky-{R}osen channels.
\newblock {\em Phys. Rev. Lett. 70, 13}, pages 1895--1899, 1993.

\bibitem[Bel64]{Bell}
J.~S. Bell.
\newblock On the {E}instein-{P}odolsky-{R}osen paradox.
\newblock {\em Physics 1}, pages 195--200, 1964.

\bibitem[BH93]{Bmech}
D.~Bohm and B.~J. Hiley.
\newblock {\em The Undivided Universe: Ontological Interpretation of Quantum
  Theory}.
\newblock Routledge Chapman and Hall, New York, 1993.

\bibitem[BPM{\etalchar{+}}97]{Qtel2}
D.~Bouwmeester, J.~Pan, K.~Mattle, M.~Eibl, H.~Weinfurter, and A.~Zeilinger.
\newblock Experimental quantum teleportation.
\newblock {\em Nature 390}, pages 575--579, 1997.

\bibitem[BPvL10]{Nadja}
N.~K. Bernardes, L.~Praxmeyer, and P.~van Loock.
\newblock Rate analysis for a hybrid quantum repeater,
  \href{http://arxiv.org/abs/1010.0106}{arXiv:1010.0106v1 [quant-ph]}. 2010.

\bibitem[Bru02]{Bruss}
D.~Bru{\ss}.
\newblock Characterizing entanglement.
\newblock {\em J. Math. Phys. 43, 9}, pages 4237--4251, 2002.

\bibitem[BSBN02]{Bartlett}
S~.D. Bartlett, B.~C. Sanders, S.~L. Braunstein, and K.~Nemoto.
\newblock Efficient classical simulation of continuous variable quantum
  information processes.
\newblock {\em Phys. Rev. Lett. 88, 097904}, 2002.

\bibitem[BSG{\etalchar{+}}05]{QDot}
A.~S. Bracker, E.~A. Stinaff, D.~Gammon, M.~E. Ware, J.~G. Tischler,
  A.~Shabaev, Al.~L. Efros, D.~Park, D.~Gershon, V.~L. Korenev, and I.~A.
  Merkulov.
\newblock Optical pumping of the electronic and nuclear spin of single
  charge-tunable quantum dots.
\newblock {\em Phys. Rev. Lett. 94, 047402}, 2005.

\bibitem[BSMM05]{Bronstein}
I.~N. Bronstein, K.~A. Semendjajew, G.~Musiol, and H.~M\"uhlig.
\newblock {\em Taschenbuch der Mathematik}.
\newblock Harri Deutsch Verlag, 6. edition, 2005.

\bibitem[BvL05]{PvLBraun}
S.~L. Braunstein and P.~van Loock.
\newblock Quantum information with continuous variables.
\newblock {\em Rev. Mod. Phys. 77, 2}, pages 513--577, 2005.

\bibitem[BW98]{statechar2}
K.~Banaszek and K.~Wodkiewicz.
\newblock Nonlocality of the {E}instein-{P}odolsky-{R}osen state in the
  {W}igner representation.
\newblock {\em Phys. Rev. A 58, 6}, pages 4345--4347, 1998.

\bibitem[CAB{\etalchar{+}}10]{discord2}
D.~Cavalcanti, L.~Aolita, S.~Boixo, K.~Modi, M.~Piani, and A.~Winter.
\newblock Operational interpretations of quantum discord,
  \href{http://arxiv.org/abs/1008.3205}{arXiv:1008.3205v4 [quant-ph]}. 2010.

\bibitem[Che00]{Chefles}
A.~Chefles.
\newblock Quantum state discrimination,
  \href{http://arxiv.org/abs/quant-ph/0010114}{arXiv:quant-ph/0010114v1}. 2000.

\bibitem[Cho75]{Choi1}
M.~Choi.
\newblock Completely positive linear maps on complex matrices.
\newblock {\em Linear Algebra and its Applications 10, 3}, pages 285--290,
  1975.

\bibitem[CKW00]{CKW}
V.~Coffman, J.~Kundu, and W.~K. Wootters.
\newblock Distributed entanglement.
\newblock {\em Phys. Rev. A 61, 052306}, 2000.

\bibitem[CZ95]{qubus2}
J.~I. Cirac and P.~Zoller.
\newblock Quantum computations with cold trapped ions.
\newblock {\em Phys. Rev. Lett. 74, 20}, pages 4091--4094, 1995.

\bibitem[DGCZ00]{Duan}
L.~Duan, G.~Giedke, J.~I. Cirac, and P.~Zoller.
\newblock Inseparability criterion for continuous variable systems.
\newblock {\em Phys. Rev. Lett. 84, 12}, pages 2722--2725, 2000.

\bibitem[DHR02]{LOCCmath}
M.~J. Donald, M.~Horodecki, and O.~Rudolph.
\newblock The uniqueness theorem for entanglement measures.
\newblock {\em J. Math. Phys. 43, 9}, pages 4252--4272, 2002.

\bibitem[DLH{\etalchar{+}}08]{quantumcomm6}
R.~Dong, M.~Lassen, J.~Heersink, Ch. Marquardt, R.~Filip, G.~Leuchs, and U.~L.
  Andersen.
\newblock Experimental entanglement distillation of mesoscopic quantum states.
\newblock {\em Nature Phys. 4}, pages 919 -- 923, 2008.

\bibitem[EBSP04]{quantumcomm4}
J.~Eisert, D.~E. Browne, S.~Scheel, and M.~B. Plenio.
\newblock Distillation of continuous-variable entanglement with optical means.
\newblock {\em Annals of Physics (NY) 311, 431}, 2004.

\bibitem[Ein71]{EinBornLet}
A.~Einstein.
\newblock {\em The Born-Einstein Letters; Correspondence between Albert
  Einstein and Max and Hedwig Born from 1916 to 1955}.
\newblock Walker, New York, 1971.

\bibitem[Eke91]{E91}
A.~K. Ekert.
\newblock Quantum cryptography based on {B}ell's theorem.
\newblock {\em Phys. Rev. Lett. 67, 6}, pages 661--663, 1991.

\bibitem[EPR35]{EPR}
A.~Einstein, B.~Podolsky, and N.~Rosen.
\newblock Can quantum-mechanical description of physical reality be complete?
\newblock {\em Phys. Rev. 47}, pages 777--780, 1935.

\bibitem[ESP02]{ContIssues}
J.~Eisert, C.~Simon, and M.~B. Plenio.
\newblock On the quantification of entanglement in infinite-dimensional quantum
  systems.
\newblock {\em J. Phys. A 35, 17}, pages 3911--3923, 2002.

\bibitem[FC04]{statechar1}
J.~Fiurasek and N.~J. Cerf.
\newblock How to measure squeezing and entanglement of {G}aussian states
  without homodyning.
\newblock {\em Phys. Rev. Lett. 93, 063601}, 2004.

\bibitem[GA90]{Aravind}
J.~J. Gong and P.~K. Aravind.
\newblock Expansion coefficients of a squeezed coherent state in the number
  state basis.
\newblock {\em Am. J. Phys. 58, 10}, pages 1003--1006, 1990.

\bibitem[GAZ{\etalchar{+}}10]{Gabriel}
C.~Gabriel, A.~Ariello, W.~Zhong, T.~G. Euser, N.~Y. Joly, P.~Banzer,
  M.~F\"ortsch, D.~Elser, U.~L. Andersen, Ch. Marquardt, P.~St.J. Russel, and
  G.~Leuchs.
\newblock Hybrid-entanglement in continuous variable systems,
  \href{http://arxiv.org/abs/1007.1322}{arXiv:1007.1322v1 [quant-ph]}. 2010.

\bibitem[GKP01]{GKP}
D.~Gottesman, A.~Kitaev, and J.~Preskill.
\newblock Encoding a qubit in an oscillator.
\newblock {\em Phys. Rev. A 64, 012310}, 2001.

\bibitem[Gla63]{Glauber}
R.~J. Glauber.
\newblock Coherent and incoherent states of the radiation field.
\newblock {\em Phys. Rev. 131}, pages 2766--2788, 1963.

\bibitem[Gou05]{Gour}
G.~Gour.
\newblock Family of concurrence monotones and its applications.
\newblock {\em Phys. Rev. A 71, 012318}, 2005.

\bibitem[Gro96]{Gro96}
L.~K. Grover.
\newblock A fast quantum mechanical algorithm for database search.
\newblock In {\em Proceedings of 28th Annual ACM Symposium on Theory of
  Computing (STOC)}, pages 212--219,
  \href{http://arxiv.org/abs/quant-ph/9605043}{arXiv:quant-ph/9605043}. 1996.

\bibitem[GRTZ02]{QKDrev}
N.~Gisin, G.~Ribordy, W.~Tittel, and H.~Zbinden.
\newblock Quantum cryptography.
\newblock {\em Rev. Mod. Phys. 74, 1}, pages 145--195, 2002.

\bibitem[GT09]{Guehne}
O.~G\"uhne and G.~Toth.
\newblock Entanglement detection.
\newblock {\em Phys. Rep. 474}, pages 1--75, 2009.

\bibitem[Has09]{Hastings}
M.~B. Hastings.
\newblock Superadditivity of communication capacity using entangled inputs.
\newblock {\em Nature Phys. 5}, pages 255--257, 2009.

\bibitem[Hel76]{Helstrom}
C.~W. Helstrom.
\newblock {\em Quantum Detection and Estimation Theory}.
\newblock Mathematics in science and engineering, vol. 123. Academic Press, New
  York, 1976.

\bibitem[HHH96]{PeresHorodeckiCrit}
M.~Horodecki, P.~Horodecki, and R.~Horodecki.
\newblock Separability of mixed states: {N}ecessary and sufficient conditions.
\newblock {\em Phys. Lett. A 223, 1-2}, pages 1--8, 1996.

\bibitem[HHH98]{BoundEnt}
M.~Horodecki, P.~Horodecki, and R.~Horodecki.
\newblock Mixed-state entanglement and distillation: Is there a "bound"
  entanglement in nature?
\newblock {\em Phys. Rev. Lett. 80, 24}, pages 5239--5242, 1998.

\bibitem[HHHH09]{HoroEntIntro}
R.~Horodecki, P.~Horodecki, M.~Horodecki, and K.~Horodecki.
\newblock Quantum entanglement.
\newblock {\em Rev. Mod. Phys. 81, 2}, pages 865--942, 2009.

\bibitem[HML08]{QKDh3}
H.~H\"aseler, T.~Moroder, and N.~L\"utkenhaus.
\newblock Testing quantum devices: {P}ractical entanglement verification in
  bipartite optical systems.
\newblock {\em Phys. Rev. A 77, 032303}, 2008.

\bibitem[HN03]{Robust3}
A.~W. Harrow and M.~A. Nielsen.
\newblock Robustness of quantum gates in the presence of noise.
\newblock {\em Phys. Rev. A 68, 012308}, 2003.

\bibitem[HO09]{DAjason}
J.~H\"olscher-Obermaier.
\newblock Entanglement in measurement-based quantum computation and
  communication, 2009.
\newblock Diploma Thesis, MPL Erlangen.

\bibitem[H{\"o}r09]{DAdominik}
D.~H{\"o}rndlein.
\newblock Non-{G}aussian methods for quantum error correction, 2009.
\newblock Diploma Thesis, MPL Erlangen.

\bibitem[HSD{\etalchar{+}}08]{quantumcomm5}
B.~Hage, A.~Samblowski, J.~DiGuglielmo, A.~Franzen, J.~Fiurasek, and
  R.~Schnabel.
\newblock Preparation of distilled and purified continuous-variable entangled
  states.
\newblock {\em Nature Phys. 4}, pages 915 -- 918, 2008.

\bibitem[HSR03]{Isom2}
M.~Horodecki, P.~W. Shor, and M.~B. Ruskai.
\newblock General entanglement breaking channels,
  \href{http://arxiv.org/abs/quant-ph/0302031}{arXiv:quant-ph/0302031v2}. 2003.

\bibitem[Jam72]{Jamiol}
A.~Jamiolkowski.
\newblock Linear transformations which preserve trace and positive
  semidefiniteness of operators.
\newblock {\em Rep. Math. Phys. 3, 4}, pages 275--278, 1972.

\bibitem[JC63]{JCM1}
E.~T. Jaynes and F.~W. Cummings.
\newblock Comparison of quantum and semiclassical radiation theories with
  applications to the beam maser.
\newblock In {\em Proc IEEE 51}, pages 89--109, 1963.

\bibitem[Joh90]{Sylvester}
C.~R. Johnson, editor.
\newblock {\em Matrix Theory and Applications}. American Mathematical Society,
  1990.

\bibitem[KdMT{\etalchar{+}}08]{Konrad}
T.~Konrad, F.~de~Melo, M.~Tiersch, Ch. Kasztelan, A.~Aragao, and
  A.~Buchleitner.
\newblock Evolution equation for quantum entanglement.
\newblock {\em Nature Phys. 4}, pages 99 -- 102, 2008.

\bibitem[KKI99]{multicomm3}
A.~Karlsson, M.~Koashi, and N.~Imoto.
\newblock Quantum entanglement for secret sharing and secret splitting.
\newblock {\em Phys. Rev. A 59, 1}, pages 162--168, 1999.

\bibitem[KLM01]{KLM}
E.~Knill, R.~Laflamme, and G.~J. Milburn.
\newblock A scheme for efficient quantum computation with linear optics.
\newblock {\em Nature 409}, pages 46--52, 2001.

\bibitem[Kra71]{Kraus}
K.~Kraus.
\newblock General state changes in quantum theory.
\newblock {\em Ann. Phys. 64}, pages 311--335, 1971.

\bibitem[KS98]{Schwetlick}
A.~Kielbasinski and H.~Schwetlick.
\newblock {\em Numerische lineare Algebra. Eine computerorientierte
  Einf\"uhrung}.
\newblock Harri Deutsch Verlag, 1998.

\bibitem[KZ04]{Wignernoncl}
A.~Kenfack and K.~Zyczkowski.
\newblock Negativity of the {W}igner function as an indicator of
  non-classicality.
\newblock {\em J. Opt. B 6, 10}, pages 396--404, 2004.

\bibitem[LBYP07]{Qcom2}
C.~Lu, D.~E. Browne, T.~Yang, and J.~Pan.
\newblock Demonstration of a compiled version of {S}hor's quantum factoring
  algorithm using photonic qubits.
\newblock {\em Phys. Rev. Lett. 99, 250504)}, 2007.

\bibitem[Leo97]{Leon}
U.~Leonhardt.
\newblock {\em Measuring the Quantum State of Light}.
\newblock Cambridge University Press, 1997.

\bibitem[LHA{\etalchar{+}}01]{statechar3}
A.~I. Lvovsky, H.~Hansen, T.~Aichele, O.~Benson, J.~Mlynek, and S.~Schiller.
\newblock Quantum state reconstruction of the single-photon {F}ock state.
\newblock {\em Phys. Rev. Lett. 87, 050402}, 2001.

\bibitem[LJG{\etalchar{+}}06]{cats1}
A.~M. Lance, H.~Jeong, N.~B. Grosse, T.~Symul, T.~C. Ralph, and P.~K. Lam.
\newblock Quantum-state engineering with continuous-variable postselection.
\newblock {\em Phys. Rev. A 73, 041801(R)}, 2006.

\bibitem[LKL04]{QKDh1}
S.~Lorenz, N.~Korolkova, and G.~Leuchs.
\newblock Continuous-variable quantum key distribution using polarization
  encoding and post selection.
\newblock {\em Appl. Phys. B 79}, pages 273--277, 2004.

\bibitem[Llo00]{SLloyd}
S.~Lloyd.
\newblock Hybrid quantum computing,
  \href{http://arxiv.org/abs/quant-ph/0008057}{arXiv:quant-ph/0008057v1}. 2000.

\bibitem[LMSN08]{Louis}
S.~G.~R. Louis, W.~J. Munro, T.~P. Spiller, and K.~Nemoto.
\newblock Loss in hybrid qubit-bus couplings and gates.
\newblock {\em Phys. Rev. A 78, 022326}, 2008.

\bibitem[LRH08]{cats3}
A.~P. Lund, T.~C. Ralph, and H.~L. Haselgrove.
\newblock Fault-tolerant linear optical quantum computing with small-amplitude
  coherent states.
\newblock {\em Phys. Rev. Lett. 100, 030503}, 2008.

\bibitem[LWL{\etalchar{+}}07]{Qcom3}
B.~P. Lanyon, T.~J. Weinhold, N.~K. Langford, M.~Barbieri, D.~F.~V. James,
  A.~Gilchrist, and A.~G. White.
\newblock Experimental demonstration of a compiled version of {S}hor's
  algorithm with quantum entanglement.
\newblock {\em Phys. Rev. Lett. 99, 250505}, 2007.

\bibitem[Mer07]{Mermin}
N.~D. Mermin.
\newblock {\em Quantum Computer Science: An Introduction}.
\newblock Cambridge University Press, 2007.

\bibitem[MGVT02]{Mancini}
S.~Mancini, V.~Giovannetti, D.~Vitali, and P.~Tombesi.
\newblock Entangling macroscopic oscillators exploiting radiation pressure.
\newblock {\em Phys. Rev. Lett. 88, 120401}, 2002.

\bibitem[MJPV99]{multicomm1}
M.~Murao, D.~Jonathan, M.~B. Plenio, and V.~Vedral.
\newblock Quantum telecloning and multiparticle entanglement.
\newblock {\em Phys. Rev. A 59, 1}, pages 156--161, 1999.

\bibitem[MP06]{SVcomment}
A.~Miranowicz and M.~Piani.
\newblock Comment on "{I}nseparability criteria for continuous bipartite
  quantum states".
\newblock {\em Phys. Rev. Lett. 97, 058901}, 2006.

\bibitem[MPHH09]{SVgeneralize}
A.~Miranowicz, M.~Piani, P.~Horodecki, and R.~Horodecki.
\newblock Inseparability criteria based on matrices of moments.
\newblock {\em Phys. Rev. A 80, 052303}, 2009.

\bibitem[MPV00]{multicomm2}
M.~Murao, M.~B. Plenio, and V.~Vedral.
\newblock Quantum-information distribution via entanglement.
\newblock {\em Phys. Rev. A 61, 032311}, 2000.

\bibitem[NC00]{Nielsen}
M.~A. Nielsen and I.~L. Chuang.
\newblock {\em Quantum Computation and Quantum Information}.
\newblock Cambridge University Press, 2000.

\bibitem[Nie99]{NielsenTrans}
M.~A. Nielsen.
\newblock Conditions for a class of entanglement transformations.
\newblock {\em Phys. Rev. Lett. 83, 2}, pages 436--439, 1999.

\bibitem[Nie06]{NielsenCluster}
M.~A. Nielsen.
\newblock Cluster-state quantum computation.
\newblock {\em Rep. Math. Phys. 57, 1}, pages 147--161, 2006.

\bibitem[NLDS09]{Neves}
L.~Neves, G.~Lima, A.~Delgado, and C.~Saavedra.
\newblock Hybrid photonic entanglement: Realization, characterization, and
  applications.
\newblock {\em Phys. Rev. A 80, 042322}, 2009.

\bibitem[ODTBG07]{quantumcomm2}
A.~Ourjoumtsev, A.~Dantan, R.~Tualle-Brouri, and P.~Grangier.
\newblock Increasing entanglement between {G}aussian states by coherent photon
  subtraction.
\newblock {\em Phys. Rev. Lett. 98, 030502}, 2007.

\bibitem[OJTBG07]{cats2}
A.~Ourjoumtsev, H.~Jeong, R.~Tualle-Brouri, and P.~Grangier.
\newblock Generation of optical "{S}chr\"odinger cats" from photon number
  states.
\newblock {\em Nature 448}, pages 784--786, 2007.

\bibitem[OKW00]{quantumcomm1}
T.~Opatrny, G.~Kurizki, and D.~Welsch.
\newblock Improvement on teleportation of continuous variables by photon
  subtraction via conditional measurement.
\newblock {\em Phys. Rev. A 61, 032302}, 2000.

\bibitem[OTBG06]{statechar4}
A.~Ourjoumtsev, R.~Tualle-Brouri, and P.~Grangier.
\newblock Quantum homodyne tomography of a two-photon {F}ock state.
\newblock {\em Phys. Rev. Lett. 96, 213601}, 2006.

\bibitem[Ott01]{Subtra}
A.~Otte.
\newblock Separabilit\"at in {Q}uantennetzwerken,
  \href{http://elib.uni-stuttgart.de/opus/volltexte/2001/921/pdf/Diss.pdf}{http://elib.uni-stuttgart.de/\~{}}.
  2001.
\newblock Dissertation, University of Stuttgart.

\bibitem[OZ01]{discord1}
H.~Ollivier and W.~H. Zurek.
\newblock Quantum discord: A measure of the quantumness of correlations.
\newblock {\em Phys. Rev. Lett. 88, 017901}, 2001.

\bibitem[Pau63]{JCM2}
H.~Paul.
\newblock Induzierte {E}mission bei starker {E}instrahlung.
\newblock {\em Ann. Phys. 11}, pages 411--412, 1963.

\bibitem[Per96]{Peres}
A.~Peres.
\newblock Separability criterion for density matrices.
\newblock {\em Phys. Rev. Lett. 77, 8}, pages 1413--1415, 1996.

\bibitem[Pom96]{Pom}
C.~Pomerance.
\newblock A tale of two sieves.
\newblock {\em Notices of the AMS 43, 12}, pages 1473--1485, 1996.

\bibitem[PV06]{MeasuresIntro}
M.~B. Plenio and S.~Virmani.
\newblock An introduction to entanglement measures,
  \href{http://arxiv.org/abs/quant-ph/0504163v3}{arXiv:quant-ph/0504163v3}.
  2006.

\bibitem[RB01]{Raussendorf}
R.~Raussendorf and H.~J. Briegel.
\newblock A one-way quantum computer.
\newblock {\em Phys. Rev. Lett. 86, 22}, pages 5188--5191, 2001.

\bibitem[RFSdG03]{Raymer}
M.~G. Raymer, A.~C. Funk, B.~C. Sanders, and H.~de~Guise.
\newblock Separability criterion for separate quantum systems.
\newblock {\em Phys. Rev. A 67, 052104}, 2003.

\bibitem[RGL06]{QKDh2}
J.~Rigas, O.~G\"uhne, and N.~L\"utkenhaus.
\newblock Entanglement verification for quantum-key-distribution systems with
  an underlying bipartite qubit-mode structure.
\newblock {\em Phys. Rev. A 73, 012341}, 2006.

\bibitem[SBW90]{cats4}
C.~M. Savage, S.~L. Braunstein, and D.~F. Walls.
\newblock Macroscopic quantum superpositions by means of single-atom
  dispersion.
\newblock {\em Opt. Lett. 15, 11}, pages 628--630, 1990.

\bibitem[Sch35]{SCat}
E.~Schr\"odinger.
\newblock The present situation in quantum mechanics.
\newblock {\em Naturwissenschaften}, 1935.

\bibitem[Sch01]{Schleich}
W.~P. Schleich.
\newblock {\em Quantum Optics in Phase Space}.
\newblock WILEY-VCH, 2001.

\bibitem[Sho97]{Sho97}
P.~W. Shor.
\newblock Polynomial-time algorithms for prime factorization and discrete
  logarithms on a quantum computer.
\newblock {\em SIAM J. Computing 26}, pages 1484--1509, 1997.

\bibitem[Sho04]{ShoAddi}
P.~W. Shor.
\newblock Equivalence of additivity questions in quantum information theory.
\newblock {\em Comm. Math. Phys. 246, 3}, pages 453--472, 2004.

\bibitem[Sim00]{Simon}
R.~Simon.
\newblock Peres-{H}orodecki separability criterion for continuous variable
  systems.
\newblock {\em Phys. Rev. Lett. 84, 12}, pages 2726--2729, 2000.

\bibitem[SMD94]{SimonH}
R.~Simon, N.~Mukunda, and B.~Dutta.
\newblock Quantum-noise matrix for multimode systems: \textit{{U}(n)}
  invariance, squeezing, and normal forms.
\newblock {\em Phys. Rev. A 49, 3}, pages 1567--1583, 1994.

\bibitem[SNB{\etalchar{+}}06]{qubus1}
T.~P. Spiller, K.~Nemoto, S.~L. Braunstein, W.~J. Munro, P.~van Loock, and
  G.~J. Milburn.
\newblock Quantum computation by communication.
\newblock {\em New J. Phys. 8, 30}, 2006.

\bibitem[SRLL02]{QKDx}
Ch. Silberhorn, T.~C. Ralph, N.~L\"utkenhaus, and G.~Leuchs.
\newblock Continuous variable quantum cryptography: {B}eating the 3 db loss
  limit.
\newblock {\em Phys. Rev. Lett. 89, 167901}, 2002.

\bibitem[Ste03]{Robust2}
M.~Steiner.
\newblock Generalized robustness of entanglement.
\newblock {\em Phys. Rev. A 67, 054305}, 2003.

\bibitem[Sti55]{Stinespring}
W.~F. Stinespring.
\newblock Positive functions on c*-algebras.
\newblock {\em Proc. Amer. Math. Soc. 6, 211}, 1955.

\bibitem[Sud63]{Sudarshan}
E.~C.~G. Sudarshan.
\newblock Equivalence of semiclassical and quantum mechanical descriptions of
  statistical light beams.
\newblock {\em Phys. Rev. Lett. 10, 7}, pages 277--279, 1963.

\bibitem[SV05a]{SV}
E.~Shchukin and W.~Vogel.
\newblock Inseparability criteria for continuous bipartite quantum states.
\newblock {\em Phys. Rev. Lett. 95, 230502}, 2005.

\bibitem[SV05b]{SVmeasure}
E.~Shchukin and W.~Vogel.
\newblock Nonclassical moments and their measurement.
\newblock {\em Phys. Rev. A 72, 043808}, 2005.

\bibitem[SZ97]{Scully}
M.~O. Scully and M.~S. Zubairy.
\newblock {\em Quantum Optics}.
\newblock Cambridge University Press, 1997.

\bibitem[TdMB08]{Tiersch}
M.~Tiersch, F.~de~Melo, and A.~Buchleitner.
\newblock Entanglement evolution in finite dimensions.
\newblock {\em Phys. Rev. Lett. 101, 170502}, 2008.

\bibitem[TI97]{Bau3}
L.~N. Trefethen and D.~Bau III.
\newblock {\em Numerical Linear Algebra}.
\newblock Society for Industrial Mathematics, 1997.

\bibitem[TNNT{\etalchar{+}}09]{quantumcomm3}
H.~Takahashi, J.~S. Neergaard-Nielsen, M.~Takeuchi, M.~Takeoka, K.~Hayasaka,
  A.~Furusawa, and M.~Sasaki.
\newblock Non-{G}aussian entanglement distillation for continuous variables,
  \href{http://arxiv.org/abs/0907.2159}{arXiv:0907.2159v1 [quant-ph]}. 2009.

\bibitem[TT96]{Temme}
M.~M. Temme and N.~M. Temme.
\newblock {\em Special Functions: An Introduction to the Classical Functions of
  Mathematical Physics}.
\newblock Wiley-Interscience, 1996.

\bibitem[UJA{\etalchar{+}}04]{Qtel3}
R.~Ursin, T.~Jennewein, M.~Aspelmeyer, R.~Kaltenbaek, M.~Lindenthal,
  P.~Walther, and A.~Zeilinger.
\newblock Communications: Quantum teleportation across the {D}anube.
\newblock {\em Nature 430}, 2004.

\bibitem[vEH01]{EntNonOrth2}
S.~J. van Enk and O.~Hirota.
\newblock Entangled coherent states: Teleportation and decoherence.
\newblock {\em Phys. Rev. A 64, 022313}, 2001.

\bibitem[vL11]{PvL}
P.~van Loock.
\newblock Optical hybrid approaches to quantum information.
\newblock {\em Laser \& Photonics Reviews 5, 2}, pages 167--200, 2011.

\bibitem[vLLMN08]{qubus4}
P.~van Loock, N.~L\"utkenhaus, W.~J. Munro, and K.~Nemoto.
\newblock Quantum repeaters using coherent-state communication.
\newblock {\em Phys. Rev. A 78, 062319}, 2008.

\bibitem[vLLS{\etalchar{+}}06]{hybridRep}
P.~van Loock, T.~D. Ladd, K.~Sanaka, F.~Yamaguchi, K.~Nemoto, W.~J. Munro, and
  Y.~Yamamoto.
\newblock Hybrid quantum repeater using bright coherent light.
\newblock {\em Phys. Rev. Lett. 96, 240501}, 2006.

\bibitem[vLMN{\etalchar{+}}08]{qubus3}
P.~van Loock, W.~J. Munro, K.~Nemoto, T.~P. Spiller, T.~D. Ladd, S.~L.
  Braunstein, and G.~J. Milburn.
\newblock Hybrid quantum computation in quantum optics.
\newblock {\em Phys. Rev. A 78, 022303}, 2008.

\bibitem[VSB{\etalchar{+}}01]{Qcom1}
L.~M.~K. Vandersypen, M.~Steffen, G.~Breyta, C.~S. Yannoni, M.~H. Sherwood, and
  I.~L. Chuang.
\newblock Experimental realization of {S}hor's quantum factoring algorithm
  using nuclear magnetic resonance.
\newblock {\em Nature 414, 6866}, pages 883--887, 2001.

\bibitem[VT99]{Robust1}
G.~Vidal and R.~Tarrach.
\newblock Robustness of entanglement.
\newblock {\em Phys. Rev. A 59, 1}, pages 141--155, 1999.

\bibitem[VW02]{Neg}
G.~Vidal and R.~F. Werner.
\newblock Computable measure of entanglement.
\newblock {\em Phys. Rev. A 65, 032314}, 2002.

\bibitem[Wan01]{XWang}
X.~Wang.
\newblock Continuous-variable and hybrid quantum gates.
\newblock {\em J. Phys. A 34, 44}, pages 9577--9584, 2001.

\bibitem[Wan02]{EntNonOrth1}
X.~Wang.
\newblock Bipartite entangled non-orthogonal states.
\newblock {\em J. Phys. A 35, 1}, pages 165--173, 2002.

\bibitem[WHR{\etalchar{+}}09]{Wallquist}
M.~Wallquist, K.~Hammerer, P.~Rabl, M.~Lukin, and P.~Zoller.
\newblock Hybrid quantum devices and quantum engineering.
\newblock {\em Phys. Scr., T 137, 014001}, 2009.

\bibitem[Wig32]{Wigner1}
E.~P. Wigner.
\newblock On the quantum correction for thermodynamic equilibrium.
\newblock {\em Phys. Rev. 40}, pages 749--759, 1932.

\bibitem[Wil36]{Williamson1}
J.~Williamson.
\newblock On the algebraic problem concerning the normal forms of linear
  dynamical systems.
\newblock {\em Amer. J. Math. 58, 1}, pages 141--163, 1936.

\bibitem[Woo98]{Woot1}
W.~K. Wootters.
\newblock Entanglement of formation of an arbitrary state of two qubits.
\newblock {\em Phys. Rev. Lett. 80, 10}, pages 2245--2248, 1998.

\bibitem[Woo01]{Woot2}
W.~K. Wootters.
\newblock Entanglement of formation and concurrence.
\newblock {\em Quant. Inf. Comp. 1, 1}, pages 27--44, 2001.

\bibitem[WTC{\etalchar{+}}08]{Wittmann}
C.~Wittmann, M.~Takeoka, K.~N. Cassemiro, M.~Sasaki, G.~Leuchs, and U.~L.
  Andersen.
\newblock Demonstration of near-optimal discrimination of optical coherent
  states.
\newblock {\em Phys. Rev. Lett. 101, 210501}, 2008.

\end{thebibliography}

%\begin{thebibliography}{99}
%\bibitem[Gro96]{Gro96} L. K. Grover. {\sl A fast quantum mechanical algorithm for database search}, Proceedings of 28th Annual ACM Symposium on Theory of Computing (STOC), pages 212-219, May 1996. \href{http://arxiv.org/abs/quant-ph/9605043}{quant-ph/9605043}.
%\bibitem[Sho94]{Sho97} P.W. Shor. {\sl Polynomial-time algorithms for prime factorization and discrete logarithms on a quantum computer}, SIAM J. Computing 26, pp. 1484-1509 (1997).
%\end{thebibliography}

\chapter*{Acknowledgments}
\addcontentsline{toc}{chapter}{Acknowledgments}
\phantom{oben}
\vspace{2cm}

First, I want to thank Peter van Loock, the supervisor of this work. He took a lot of time for me and was always available for questions and discussions. Without his help and his inexhaustible wealth of great ideas this thesis would not exist in this form. Thank you, \textit{Peter}!

\vspace{1cm}

Furthermore, I want to thank all the other people from the "optical quantum information theory" group. Discussions with them were interesting and fruitful, midday lunch breaks and vespertine drinks always fun. The comfortable atmosphere in the group contributed significantly to the successful realization of this thesis. Special thanks go to Ricardo, who helped me with the proofreading. Thank you, \textit{Ricardo}, \textit{Nadja}, \textit{Seckin}, \textit{Denis}, \textit{Chriszandro}, \textit{ShengLi} and \textit{Mathias}!

\vspace{1cm}

Finally, I want to thank my \textit{family}, which has always supported me and actually enabled my studies. Thank you!

\end{document}